\documentclass[11pt]{article}
\usepackage{amsfonts}
\usepackage{amssymb}
\usepackage{amstext}
\usepackage{amsmath}
\usepackage{xspace}
\usepackage{theorem}
\usepackage{graphicx}
\usepackage{graphics}
\usepackage{colordvi}
\usepackage{colordvi}
\usepackage{url}
\usepackage{subfigure}

\advance \oddsidemargin by -0.8in
\newcommand{\myparskip}{3pt}
\parskip \myparskip

        {\hspace*{\fill}$\Box$\par\vspace{4mm}}

\newcommand{\MP}{{\sf Minimum Planarization}\xspace}
\newcommand{\MCN}{{\sf Minimum Crossing Number}\xspace}

\newcommand{\Gr}[1]{G^{(#1)}}
\newcommand{\Hr}[1]{H^{(#1)}}
\newcommand{\edges}[1]{E^{(#1)}}

\newcommand{\optcro}[1]{\mathsf{OPT}_{\mathsf{cr}}(#1)}

\newcommand{\algsc}{\ensuremath{{\mathcal{A}}_{\mbox{\textup{\footnotesize{ALN}}}}}\xspace}
\newcommand{\alphasc}{\ensuremath{\alpha_{\mbox{\textup{\footnotesize{ALN}}}}}}
\newcommand{\betaFCG}{\beta_{\mathrm{FCG}}}
\newcommand{\fout}{F_{\mathsf{out}}}

\newcommand{\bad}{\mathsf{IRG}}
\newcommand{\G}{{\mathbf{G}}}
\renewcommand{\H}{{\mathbf{H}}}

\newcommand{\norm}[1]{\lVert #1\rVert}

\newcommand{\event}{{\cal{E}}}

\newcommand{\opt}{\mathsf{OPT}}
\newcommand{\OPT}{\mathsf{OPT}}

\newcommand{\set}[1]{\left\{ #1 \right\}}
\newcommand{\sse}{\subseteq}

\newcommand{\tset}{{\mathcal T}}

\newcommand{\pset}{{\mathcal{P}}}
\newcommand{\dset}{{\mathcal{D}}}

\newcommand{\qset}{{\mathcal{Q}}}
\newcommand{\lset}{{\mathcal{L}}}

\newcommand{\cset}{{\mathcal{C}}}
\newcommand{\fset}{{\mathcal{F}}}

\newcommand{\jset}{{\mathcal{J}}}
\newcommand{\xset}{{\mathcal{X}}}
\newcommand{\wset}{{\mathcal{W}}}

\newcommand{\yset}{{\mathcal{Y}}}
\newcommand{\rset}{{\mathcal{R}}}

\newcommand{\hset}{{\mathcal{H}}}
\newcommand{\sset}{{\mathcal{S}}}
\newcommand{\zset}{{\mathcal{Z}}}
\newcommand{\gset}{{\mathcal{G}}}

\newcommand{\nots}{\overline S}
\newcommand{\notr}{\overline R}

\newcommand{\be}{\begin{enumerate}}
\newcommand{\ee}{\end{enumerate}}
\newcommand{\bd}{\begin{description}}
\newcommand{\ed}{\end{description}}
\newcommand{\bi}{\begin{itemize}}
\newcommand{\ei}{\end{itemize}}

\numberwithin{equation}{section}
\numberwithin{figure}{section}

\newtheorem{theorem}{Theorem}[section]
\newtheorem{lemma}[theorem]{Lemma}
\newtheorem{observation}{Observation}[section]
\newtheorem{corollary}{Corollary}[section]
\newtheorem{claim}[theorem]{Claim}

\newtheorem{definition}{Definition}[section]
\newenvironment{proof}{\par \smallskip{\bf Proof:}}{\hfill\stopproof}
\def\stopproof{\square}
\def\square{\vbox{\hrule height.2pt\hbox{\vrule width.2pt height5pt \kern5pt
\vrule width.2pt} \hrule height.2pt}}

\renewcommand{\phi}{\varphi}
\newcommand{\eps}{\epsilon}

\newcommand{\half}{\ensuremath{\frac{1}{2}}}

\newcommand{\poly}{\operatorname{poly}}

\newcommand{\reals}{{\mathbb R}}

\newcommand{\expect}[2][]{\text{\bf E}_{#1}\left [#2\right]}
\newcommand{\prob}[2][]{\text{\bf Pr}_{#1}\left [#2\right]}

\newenvironment{properties}[2][0]
{
\begin{enumerate} \setcounter{enumi}{#1}}{\end{enumerate}}

\newcommand{\mynote}[1]{{\sc\bf{[#1]}}}

\newcommand{\dmax}{d_{\mbox{\textup{\footnotesize{max}}}}}

\newcommand{\out}{\operatorname{out}}
\newcommand{\cro}{\operatorname{cr}}

\newcommand{\supp}{\mbox{supp}}

\begin{document}


\title{An Algorithm for the Graph Crossing Number Problem}
\author{Julia Chuzhoy\thanks{Toyota Technological Institute, Chicago, IL
60637. Email: {\tt cjulia@ttic.edu}. Supported in part by NSF CAREER award CCF-0844872.}}

\begin{titlepage}

\maketitle
\thispagestyle{empty}
\begin{abstract}
We study the \MCN problem: given an $n$-vertex graph $G$, the goal is to find a drawing of $G$ in the plane with minimum number of edge crossings. This is one of the central problems in topological graph theory, that has been studied extensively over the past three decades. The first non-trivial efficient algorithm for the problem, due to Leighton and Rao, achieved an $O\left (n\log^4n\right )$-approximation for bounded degree graphs. This algorithm has since been improved by poly-logarithmic factors, with the best current approximation ratio standing on $O\left (n\cdot \poly(d)\cdot \log^{3/2}n\right )$ for graphs with maximum degree $d$. In contrast, only APX-hardness is known on the negative side.

In this paper we present an efficient randomized algorithm to find a drawing of any $n$-vertex graph $G$ in the plane with $O\left (\opt^{10}\cdot \poly(d\cdot \log n)\right )$ crossings, where $\opt$ is the number of crossings in the optimal solution, and $d$ is the maximum vertex degree in $G$. This result implies 
an $\tilde{O}\left (n^{9/10}\cdot \poly(d)\right )$-approximation for \MCN, thus breaking the long-standing $\tilde{O}(n)$-approximation barrier for bounded-degree graphs. 
\end{abstract}

\end{titlepage}

\section{Introduction}

A drawing of a graph $G$ in the plane is a mapping, in which every vertex of $G$ is mapped into a point in the plane, and every edge into a continuous curve connecting the images of its endpoints. 
We assume that no three curves meet at the same point, and no curve contains an image of any vertex other than its endpoints.
A \emph{crossing} in such a drawing is a point where the images of two edges intersect, and the \emph{crossing number} of a graph $G$, denoted by $\optcro{G}$, is the smallest number of crossings achievable by any drawing of $G$ in the plane. The goal in the \MCN problem is to find a drawing of the input graph $G$ with minimum number of crossings.
We denote by $n$ the number of vertices in $G$, and by $\dmax$ its maximum vertex degree.

The concept of the graph crossing number dates back to 1944, when P\'al Tur\'an has posed the question of determining the crossing number of the complete bipartite graph $K_{m,n}$. This question was motivated by improving the performance of workers at a brick factory, where Tur\'an has been working at the time (see Tur\'an's account in \cite{turan_first}). Later, Anthony Hill (see~\cite{Guy-complete-graphs}) has posed the question of computing the crossing number of the complete graph $K_n$, and Erd\"{o}s and Guy~\cite{erdos_guy73} noted that
\emph{``Almost all questions one can ask about crossing numbers remain unsolved.''}
Since then, the problem has become a subject of intense study, with hundreds of papers written on the subject (see, e.g. the extensive bibliography maintained by Vrt'o \cite{vrto_biblio}.) 
Despite this enormous stream of results and ideas, some of the most basic questions about the crossing number problem remain unanswered.
For example, the crossing number of $K_{11}$ was established just a few years ago (\cite{K11}), while the answer for $K_t, t\geq 13$, remains elusive.
We note that in general $\optcro{G}$ can be as large as $\Omega(n^4)$, for example for the complete graph. In particular, one of the famous results in this area, due to Ajtai et al.~\cite{ajtai82} and Leighton~ \cite{leighton_book} states that if $|E(G)|\geq 4n$, then $\optcro{G}=\Omega(|E(G)|^3/n^2)$.

In this paper we focus on the algorithmic aspect of the problem.
The first non-trivial algorithm for \MCN was obtained by Leighton and Rao \cite{LR},
who combined their breakthrough result on balanced separators with the techniques of Bhatt and Leighton~\cite{bhatt84} for VLSI design, to obtain an algorithm that finds a drawing of any bounded-degree $n$-vertex graph with at most $O(\log^4 n) \cdot (n + \optcro{G})$ crossings.
This bound was later improved to $O(\log^3 n) \cdot (n+\optcro{G})$ by Even, Guha and Schieber \cite{EvenGS02}, and the new approximation algorithm of Arora, Rao and Vazirani~\cite{ARV} for Balanced Cut  gives a further improvement to $O(\log^2 n) \cdot (n+\optcro{G})$, thus implying an $O(n \cdot \log^2 n)$-approximation for \MCN on bounded-degree graphs. This result can also be extended to general graphs with maximum vertex degree $\dmax$, where the approximation factor becomes  $O(n\cdot\poly(\dmax)\cdot \log^2n)$. Chuzhoy, Makarychev and Sidiropoulos~\cite{CMS10} have recently improved this result to an $O(n\cdot \poly(\dmax)\cdot \log^{3/2} n)$-approximation.
On the negative side, the problem was shown to be NP-complete by Garey and Johnson \cite{crossing_np_complete}, and remains NP-complete even on cubic graphs~\cite{Hlineny06a}.
More surprisingly, even in the very restricted case, where the input graph $G$ is obtained by adding a single edge to a planar graph, the problem is still NP-complete~\cite{cabello_edge}.
The NP-hardness proof of ~\cite{crossing_np_complete}, combined with the inapproximability result for Minimum Linear-Arrangement \cite{Ambuhl07}, implies that there is no PTAS for \MCN unless NP has randomized subexponential time algorithms.

To summarise, although current lower bounds do not rule out the possibility of a constant-factor approximation for the problem, the state of the art, prior to this work, only gives an $\tilde O(n\cdot \poly(\dmax))$-approximation. In view of this glaring gap in our understanding of the problem, a natural question is whether we can obtain good algorithms for the case where the optimal solution cost is low --- arguably, the most interesting setting for this problem. 
A partial answer was given by Grohe~\cite{Grohe04}, who showed that the problem is fixed-parameter tractable. Specifically, Grohe designed an exact $O(n^2)$-time algorithm, for the case where the optimal solution cost is bounded by a constant. Later, Kawarabayashi and Reed \cite{KawarabayashiR07} have shown a linear-time algorithm for the same setting. Unfortunately, the running time of both algorithms depends super-exponentially on the optimal solution cost.

Our main result is an efficient randomized algorithm, that, given any $n$-vertex graph with maximum degree $\dmax$, 
produces a drawing of $G$ with $O\left ((\optcro{G})^{10}\cdot\poly(\dmax\cdot \log n)\right )$ crossings with high probability. In particular, we obtain an $O\left (n^{9/10}\cdot \poly(\dmax\cdot \log n)\right )$-approximation for general graphs, and an 
$\tilde{O}(n^{9/10})$-approximation for bounded-degree graphs, thus breaking the long standing barrier of $\tilde{O}(n)$-approximation for this setting.

We note that many special cases of the \MCN problem have been extensively studied, with better approximation algorithms known for some.
Examples include $k$-apex graphs~\cite{crossing_apex,CMS10}, bounded genus graphs~\cite{BorozkyPT06,crossing_genus,crossing_projective,crossing_torus,CMS10} and minor-free graphs~\cite{WoodT06}.
Further overview of work on \MCN can be found in the expositions of Richter and Salazar \cite{richter_survey}, Pach and T\'{o}th \cite{pach_survey}, Matou\v{s}ek \cite{matousek_book}, and Sz\'{e}kely \cite{szekely_survey}.

\noindent {\bf Our results and techniques.}
Our main result is summarized in the following theorem.

\begin{theorem}\label{theorem: main-crossing-number}
There is an efficient randomized algorithm, that, given any $n$-vertex graph $G$ with maximum degree $\dmax$, finds a drawing of $G$ in the plane with $O\left ((\optcro{G})^{10}\cdot \poly(\dmax\cdot \log n)\right )$ crossings with high probability.
\end{theorem}

Combining this theorem with the algorithm of Even et al.~\cite{EvenGS02}, we obtain the following corollary.

\begin{corollary}\label{corollary: main-approx-crossing-number}
There is an efficient randomized $O\left (n^{9/10}\cdot \poly(\dmax\cdot \log n)\right )$-approximation algorithm for \MCN.
\end{corollary}

We now give an overview of our techniques. Instead of directly solving the \MCN problem, it is more  convenient to work with a closely related problem -- \MP.
In this problem, given a graph $G$, the goal is to find a minimum-cardinality subset $E^*$ of edges, such that the graph $G\setminus E^*$ is planar.
The two problem are closely related, and this connection was recently formalized by~\cite{CMS10}, in the following theorem:

\begin{theorem}[\cite{CMS10}]\label{thm:CMS10}
Let $G=(V,E)$ be any $n$-vertex graph of maximum degree $d_{\max}$, and suppose we are given a subset $E^*\subset E$ of edges, $|E^*|=k$, such that $G\setminus E^*$ is planar.
Then there is an efficient algorithm to find a drawing of $G$ in the plane with at most $O\left(d_{\max}^3 \cdot k \cdot (\optcro{G} + k)\right)$ crossings.
\end{theorem}

Therefore, in order to solve the \MCN problem, it is sufficient to find a good solution to the \MP problem on the same graph.
We note that an $O(\sqrt {n\log n}\cdot \dmax)$-approximation algorithm for the \MP problem follows easily from  the Planar Separator theorem of Lipton and Tarjan~\cite{planar-separator} (see e.g.~\cite{CMS10}), and we are not aware of any other algorithmic results for the problem. 
 Our main technical result is the proof of the following theorem, which, combined with Theorem~\ref{thm:CMS10}, implies Theorem~\ref{theorem: main-crossing-number}.

\begin{theorem}\label{thm:main}
There is an efficient randomized algorithm, that, 
given an $n$-vertex graph $G=(V,E)$ with maximum degree $\dmax$, finds a subset $E^*\sse E$ of edges, such that $G\setminus E^*$ is planar, and  with high probability $|E^*| = O\left ((\optcro{G})^5\poly(\dmax \cdot \log n)\right )$.
\end{theorem}


We now describe our main ideas and techniques. Given an optimal solution $\phi$ to the \MCN problem on graph $G$, we say that an edge $e\in E(G)$ is \emph{good} iff it does not participate in any crossings in $\phi$. For convenience, we consider a slightly more general version of the problem, where, in addition to the graph $G$, we are given a simple cycle $X\sse G$, that we call the \emph{bounding box}, and our goal is to find a drawing of $G$, such that the edges of $X$ do not participate in any crossings, and all vertices and edges of $G\setminus X$ appear on the same side of the closed curve to which $X$ is mapped. In other words, if $\gamma_X$ is the simple closed curve to which $X$ is mapped, and $F_1,F_2$ are the two faces into which $\gamma_X$ partitions the plane, then one of the faces $F\in \set{F_1,F_2}$ must contain the drawings of all the edges and vertices of $G\setminus X$. We call such a drawing \emph{a drawing of $G$ inside the bounding box $X$}. Since we allow $X$ to be empty, this is indeed a generalization of the \MCN problem. In fact, from Theorem~\ref{thm:CMS10}, it is enough to find what we call a \emph{weak solution} to the problem, namely, a small-cardinality subset $E^*$ of edges with $E^*\cap E(X)=\emptyset$, such that there is a planar drawing of the remaining graph $G\setminus E^*$ inside the bounding box $X$. Our proof consists of three major ingredients that we describe below.

The algorithm is iterative. Throughout the algorithm, we gradually remove some edges from the graph, and gradually build a planar drawing of the remaining graph. One of the central notions we use is that of \emph{graph skeletons}. A skeleton $K$ of graph $G$ is simply a sub-graph of $G$, that contains the bounding box $X$, and has a unique planar drawing (for example, it may be convenient to think of $K$ as being $3$-vertex connected). Given a skeleton $K$, and a small subset $E'$ of edges (that we eventually remove from the graph), we say that $K$ is an \emph{admissible skeleton} iff all the edges of $K$ are good, and every connected component of $G\setminus (K\cup E')$ only contains a small number of vertices (say, at most $(1-1/\rho)n$, for some balance parameter $\rho$). Since $K$ has a unique planar drawing, and all its edges are good, we can find its unique planar drawing efficiently, and it must be identical to the drawing $\phi_K$ of $K$ induced by the optimal solution $\phi$. Let $\fset$ be the set of faces in this drawing. Since $K$ only contains good edges, for each connected component $C$ of $G\setminus (K\cup E')$, all edges and vertices of $C$ must be drawn completely inside one of the faces $F_C\in \fset$ in $\phi$. Therefore, if, for each such connected component $C$, we can identify the face $F_C$ inside which it needs to be embedded, then we can recursively solve the problems induced by each such component $C$, together with the bounding box formed by the boundary of $F_C$. In fact, given an admissible skeleton $K$, we show that we can find a good assignment of the connected components of $G\setminus (K\cup E')$ to the faces of $\fset$, so that, on the one hand, all resulting sub-problems have solutions of total cost at most $\optcro{G}$, while, on the other hand, if we combine weak solutions to these sub-problems with the set $E'$ of edges, we obtain a feasible weak solution to the original problem. The assignment of the components to the faces of $\fset$ is done by reducing the problem to an instance of the Min-Uncut problem. We defer the details of this part to later sections, and focus here on finding an admissible skeleton $K$.

Our second main ingredient is the use of well-linked sets of vertices, and well-linked balanced bi-partitions. Given a set $S$ of vertices, let $G[S]$ be the sub-graph of $G$ induced by $S$, and let $\Gamma(S)$ be the subset of vertices of $S$ adjacent to the edges in $E(S,\nots)$. Informally, we say that $S$ is $\alpha$-well-linked, iff every pair of vertices in $\Gamma(S)$ can send one flow unit to each other, with overall congestion bounded by $\alpha|\Gamma(S)|$. We say that a bi-partition $(S,\nots)$ of the vertices of $G$ is $\rho$-balanced and $\alpha$-well-linked, iff $|S|,|\nots|\geq n/\rho$, and both $S$ and $\nots$ are $\alpha$-well-linked. Suppose we can find a $\rho$-balanced, $\alpha$-well linked bi-partition of $G$ (it is convenient to think of $\rho,\alpha=\poly(\dmax \cdot \log n)$). In this case, we show a randomized algorithm, that w.h.p. constructs an admissible skeleton $K$, as follows. Let $\pset,\pset'$ be the collections of the flow-paths in $G[S]$ and $G[\nots]$ respectively, guaranteed by the well-linkedness of $S$ and $\nots$. Since the congestion on all edges is relatively low, only a small number of paths in $\pset\cup \pset'$ contain bad edges. Therefore, if we choose a random collection of paths from $\pset$ and $\pset'$ with appropriate probability, the resulting skeleton $K$, obtained from the union of these paths, is unlikely to contain bad edges. Moreover, we can show that w.h.p., every connected component of $G\setminus K$ only contains a small number of edges in $E(S,\nots)$. It is still possible that some connected component $C$ of $G\setminus K$ contains many vertices of $G$. However, only one such component $C$ may contain more than $n/2$ vertices. Let $E'$ be the subset of edges in $E(S,\nots)$, that belong to $C$. Then, since the original cut $(S,\nots)$ is $\rho$-balanced, once we remove the edges of $E'$ from $C$, it will decompose into small enough components. This will ensure that all connected components of $G\setminus (K\cup E')$ are small enough, and $K$ is admissible.

Using these ideas, given an efficient algorithm for computing $\rho$-balanced $\alpha$-well-linked cuts, we can obtain an algorithm for the \MCN problem. Unfortunately, we do not have an efficient algorithm for computing such cuts. We can only compute such cuts in graphs that do not contain a certain structure, that we call \emph{nasty vertex sets}. Informally, a subset $S$ of vertices is a nasty set, iff $|S|>>|E(S,\nots)|^2$, and the sub-graph $G[S]$ induced by $S$ is planar. We show an algorithm, that, given any graph $G$, either produces a $\rho$-balanced $\alpha$-well linked cut, or finds a nasty set $S$ in $G$. Therefore, if $G$ does not contain any nasty sets, we can compute the $\rho$-balanced $\alpha$-well-linked bi-paritition of $G$, and hence obtain an algorithm for \MCN.
Moreover, given {\bf any} graph $G$, if our algorithm fails to produce a good solution to \MCN on $G$, then w.h.p. it returns a nasty set of vertices in $G$.

The third major component of our algorithm is handling the nasty sets. Suppose we are given a nasty set $S$, and assume for now that it is also $\alpha$-well-linked for some parameter $\alpha=\poly(\log n)$.  Let $\Gamma(S)$ denote the endpoints of the edges in $E(S,\nots)$ that belong to $S$, and let $|\Gamma(S)|=z$. Recall that $|S|>>z^2$, and $G[S]$ is planar. Intuitively, in this case we can use the $z\times z$ grid to ``simulate'' the sub-graph $G[S]$. More precisely, we replace the sub-graph $G[S]$ with the $z\times z$ grid $Z_S$, and identify the vertices of the first row of the grid with the vertices in $\Gamma(S)$. We call the resulting graph the \emph{contracted graph}, and denote it by $G_{|S}$. Notice that the number of vertices in $G_{|S}$ is smaller than that in $G$.
When $S$ is not well-linked, we perform a simple well-linked decomposition procedure to partition $S$ into a collection of well-linked subsets, and replace each one of them with a grid separately. 
Given a drawing of the resulting contracted graph $G_{|S}$, we say that it is a \emph{canonical drawing} if the edges of the newly added grids do not participate in any crossings. Similarly, we say that a planarizing subset $E^*$ of edges is a weak canonical solution for $G_{|S}$, iff the edges of the grids do not belong to $E^*$.
We show that the crossing number of $G_{|S}$ is bounded by $\poly(\dmax\cdot \log n)\optcro{G}$, and this bound remains true even for canonical drawings. On the other hand, we show that given any weak canonical solution $E^*$ for $G_{|S}$, we can efficiently find a weak solution of comparable cost for $G$. Therefore, it is enough to find a weak feasible canonical solution for graph $G_{|S}$.
However, even the contracted graph $G_{|S}$ may still contain nasty sets. We then show that, given any nasty set $S'$ in $G_{|S}$, we can find another subset $S''$ of vertices in the original graph $G$, such that the contracted graph $G_{|S''}$ contains fewer vertices than $G_{|S}$. The crossing number of $G_{|S''}$ is again bounded by $\poly(\dmax\cdot \log n)\optcro{G}$ even for canonical drawings, and a weak canonical solution to $G_{|S''}$ gives a weak solution to $G$ as before.

Our algorithm then consists of a number of stages. In each stage, it starts with the current contracted graph $G_{|S}$ (where in the first stage, $S=\emptyset$, and $G_{|S}=G$). It then either finds a good weak canonical solution for problem $G_{|S}$, thus giving a feasible solution to the original problem, or returns a nasty set $S'$ in graph $G_{|S}$. We then construct a new contracted graph $G_{|S''}$, that contains fewer vertices than $G_{|S}$, and becomes the input to the next stage.



\noindent{\bf Organization.} We start with some basic definitions, notation, and general results on cuts and flows in Section~\ref{sec: Prelims}. We then present a more detailed algorithm overview in Section~\ref{sec: overview}. Section~\ref{sec: graph contraction} is devoted to the graph contraction step, and the rest of the algorithm appears in Sections~\ref{sec: alg} and~\ref{sec: iteration}. For convenience, the list of all main parameters appears in Section~\ref{sec: param-list} of Appendix. Our conclusions appear in Section~\ref{sec: conclusions}. 

\label{------------------------------------------------Preliminaries--------------------------------------------}
\section{Preliminaries and Notation}\label{sec: Prelims}
 In order to avoid confusion, throughout the paper, we denote the input graph by $\G$, with $|V(\G)|=n$, and maximum vertex degree $\dmax$. 
 In statements regarding general arbitrary graphs, we will denote them by $G$, to distinguish them from the specific graph $\G$. 

\noindent{\bf General Notation.}
We use the words ``drawing'' and ``embedding'' interchangeably. 
Given any graph $G$, a drawing $\phi$ of $G$, and any sub-graph $H$ of $G$, we denote by $\phi_H$ the drawing of $H$ induced by $\phi$, and by $\cro_{\phi}(G)$ the number of crossings in the drawing $\phi$ of $G$. Notice that we can assume w.l.o.g. that no edge crosses itself in any drawing. For any pair $E_1,E_2\sse E(G)$ of subsets of edges, we denote by $\cro_{\phi}(E_1,E_2)$ the number of crossings in $\phi$ in which the images of edges of $E_1$ intersect the images of edges of $E_2$, and by $\cro_{\phi}(E_1)$ the number of crossings in $\phi$ in which the images of edges of $E_1$ intersect with each other. Given two disjoint sub-graphs $H_1,H_2$ of $G$, we will sometimes write $\cro_{\phi}(H_1,H_2)$ instead of $\cro_{\phi}(E(H_1),E(H_2))$, and $\cro_{\phi}(H_1)$ instead of $\cro_{\phi}(E(H_1))$.
If $G$ is a planar graph, and $\phi$ is a drawing of $G$ with no crossings, then we say that $\phi$ is a \emph{planar} drawing of $G$.
For a graph $G=(V,E)$, and subsets $V'\sse V$, $E'\sse E$ of its vertices and edges respectively, we denote by $G[V']$, $G\setminus V'$, and $G\setminus E'$ the sub-graphs of $G$ induced by $V'$, $V\setminus V'$, and $E\setminus E'$, respectively.

\begin{definition}
Let $\gamma$ be any closed simple curve, and let $F_1,F_2$ be the two faces into which $\gamma$ partitions the plane. Given any drawing $\phi$ of a graph $G$, we say that $G$ is \emph{embedded inside $\gamma$}, iff one of the faces $F\in\set{F_1,F_2}$ contains the images of all edges and vertices of $G$ (the images of the vertices of $G$ may lie on $\gamma$). Similarly, if $C\sse G$ is a simple cycle, then we say that $G$ is embedded inside $C$, iff the edges of $C$ do not participate in any crossings, and $G\setminus E(C)$ is embedded inside $\gamma_C$ -- the simple closed curve to which $C$ is mapped.
\end{definition}

Given a graph $G$ and a bounding box $X$, we define the problem $\pi(G,X)$, that we use extensively.

\begin{definition}
Given a graph $G$ and a simple (possibly empty) cycle $X\sse G$, called the \emph{bounding box}, a \emph{strong solution} for problem $\pi(G,X)$, is a drawing $\psi$ of $G$, in which $G$ is embedded inside the bounding box $X$, and its cost is the number of crossings in $\psi$. A \emph{weak solution} to  problem $\pi(G,X)$ is a subset $E'\sse E(G)\setminus E(X)$ of edges, such that $G\setminus E'$ has a \emph{planar drawing}, in which it is embedded inside the bounding box $X$. \end{definition}

Notice that in order to prove Theorem~\ref{thm:main}, it is enough to find a weak solution for problem $\pi(\G,X_0)$, where $X_0=\emptyset$, of cost $O\left ((\optcro{\G})^5\poly(\dmax \cdot \log n)\right )$.


\begin{definition} For any graph $G=(V,E)$, a subset $V'\sse V$ of vertices is called a $c$-separator, iff $|V'|=c$, and the graph $G\setminus V'$ is not connected. We say that $G$ is $c$-connected iff it does not contain $c'$-separators, for any $0<c'<c$.
 \end{definition}

We will use the following four well-known results:

\begin{theorem} (Whitney~\cite{Whitney})\label{thm:Whitney} Every 3-connected planar graph has a unique planar drawing.
\end{theorem}

\begin{theorem}(Hopcroft-Tarjan~\cite{planar-drawing})\label{thm:planar drawing}
For any graph $G$, there is an efficient algorithm to determine whether $G$ is planar, and if so, to find a planar drawing of $G$.
\end{theorem}

\begin{theorem}(Ajtai et al.~\cite{ajtai82}, Leighton \cite{leighton_book})\label{thm: large average degree large crossing number}
Let $G$ be any graph with $n$ vertices and $m\geq 4n$ edges. Then $\optcro{G}=\Omega(m^3/n^2)=\Omega(n)$.
\end{theorem}

\begin{theorem}(Lipton-Tarjan~\cite{planar-separator})\label{thm: planar separator}
Let $G$ be any $n$-vertex planar graph. Then there is a constant $q$, and an efficient algorithm to partition the vertices of $G$ into three sets $A,B,C$, such that $|A|,|C|\geq n/3$, $|B|\leq q\sqrt{n}$, and there are no edges in $G$ connecting the vertices of $A$ to the vertices of $C$.
\end{theorem}

\subsection{Well-linkedness}

\begin{definition} Let $G=(V,E)$ be any graph, and $J\sse V$ any subset of its vertices. We denote by $\out_G(J)=E_G(J,V\setminus J)$, and we call the edges in $\out_G(J)$ the \emph{terminal edges for $J$}. For each terminal edge $e=(u,v)$, with $u\in J$, $v\not \in J$, we call $u$ the \emph{interface vertex} and $v$ the \emph{terminal vertex} for $J$. We denote by $\Gamma_G(J)$ and $T_G(J)$ the sets of all interface and terminal vertices for $J$, respectively, and we omit the subscript $G$ when clear from context (see Figure~\ref{fig: terminal interface vertices}).
\end{definition}

\begin{figure}[h]
\scalebox{0.3}{\rotatebox{0}{\includegraphics{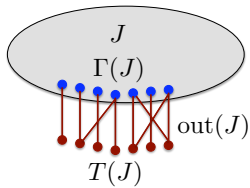}}} \caption{Terminal vertices and edges for set $J$ are red; interface vertices are blue. \label{fig: terminal interface vertices}}
\end{figure}

\begin{definition} Given a graph $G$, a subset $J$ of its vertices, and a parameter $\alpha>0$, we say that $J$ is $\alpha$-well-linked, iff for any partition $(J_1,J_2)$ of $J$, if we denote by $T_1=\out(J_1)\cap \out(J)$, and by $T_2=\out(J_2)\cap \out(J)$, then
$|E(J_1,J_2)|\geq \alpha\cdot \min\set{|T_1|,|T_2|}$
\end{definition}

Notice that if  $G$ is a connected graph and $J\subset V(G)$ is $\alpha$-well-linked for any $\alpha>0$, then $G[J]$ must be connected. Finally, we define $\rho$-balanced $\alpha$-well-linked bi-partitions.

\begin{definition} Let $G$ be any graph, and let $\rho>1, 0<\alpha\leq 1$ be any parameters. We say that a bi-partition $(S,\nots)$ of $V(G)$ is $\rho$-balanced and $\alpha$-well-linked, iff $|S|,|\nots|\geq |V(G)|/\rho$ and both $S$ and $\nots$ are $\alpha$-well-linked.
\end{definition}

\subsection{Sparsest Cut and Concurrent Flow}\label{subsec: sparsest cut}
In this section we summarize some well-known results on graph cuts and flows that we use throughout the paper.
We start by defining the non-uniform sparsest cut problem. Suppose we are given a graph $G=(V,E)$, with weights $w_v$ on vertices $v\in V$. Given any partition $(A,B)$ of $V$, the \emph{sparsity} of the cut $(A,B)$ is $\frac{|E(A,B)|}{\min\set{W(A),W(B)}}$, where $W(A)=\sum_{v\in A}w_v$ and $W(B)=\sum_{v\in B}w_v$. In the non-uniform sparsest cut problem, the input is a graph $G$ with weights on vertices, and the goal is to find a cut of minimum sparsity. Arora, Lee and Naor~\cite{sparsest-cut} have shown an $O(\sqrt{\log n}\cdot \log\log n)$-approximation algorithm for the non-uniform sparsest cut problem. We denote by $\algsc$ this algorithm and by $\alphasc=O(\sqrt{\log n}\cdot\log\log n)$ its approximation factor.
We will usually work with a special case of the sparsest cut problem, where we are given a subset $T\sse V$ of vertices, called terminals, and the vertex weights are $w_v=1$ for $v\in T$, and $w_v=0$ otherwise.

A problem dual to sparsest cut is the maximum concurrent multicommodity flow problem.
Here, we need to compute the maximum value $\lambda$, such that $\lambda/|T|$ flow units can be simultaneously sent in $G$ between every pair of terminals with no congestion. The flow-cut gap is the maximum possible ratio, in any graph, between the value of the minimum sparsest cut and the maximum concurrent flow. The value of the flow-cut gap in undirected graphs, that we denote by $\betaFCG $ throughout the paper, is $\Theta(\log n)$~\cite{LR, GVY,LLR,Aumann-Rabani}.
In particular, if the value of the sparsest cut is $\alpha$, then every pair of terminals can send $\frac{\alpha}{|T|\cdot \betaFCG }$ flow units to each other with no congestion.

Let $G$ be any graph, let $S$ be a subset of vertices of $G$, and let $0 < \alpha < 1$, such that $S$ is $\alpha$-well-linked. 
We now define the sparsest cut and the concurrent flow instances corresponding to $S$, as follows.
For each edge $e\in \out(S)$, we sub-divide the edge by adding a new vertex $t_e$ to it. Let $G'$ denote the resulting graph, and let $T$ denote the set of all vertices $t_e$ for $e\in \out_G(S)$.
Consider the graph $H=G'[S]\cup \out_{G'}(S)$. We can naturally define an instance of the non-uniform sparsest cut problem on $H$, where the set of terminals is $T$. The fact that $S$ is $\alpha$-well-linked is equivalent to the value of the sparsest cut in the resulting instance being at least $\alpha$.
We obtain the following simple well-known consequence:

\begin{observation}\label{observation: existence of flow in well-linked instance}
Let $G$, $S$, $H$, and $T$ be defined as above, and let $0<\alpha<1$, such that $S$ is $\alpha$-well-linked. Then every pair of vertices in $T$ can send one flow unit to each other in $H$, such that the maximum congestion on any edge is at most $\betaFCG |T|/\alpha$.
Moreover, if $M$ is any partial matching on the vertices of $T$, then we can send one flow unit between every pair $(u,v)\in M$ in graph $H$, with maximum congestion at most $2\betaFCG /\alpha$.
\end{observation}
\begin{proof}
The first part is immediate from the definition of the flow-cut gap. Let $F$ denote the resulting flow. In order to obtain the second part, for every pair $(u,v)\in M$, 
 $u$ will send $1/|T|$ flow units to every vertex in $T$, and $v$ will collect $1/|T|$ flow units from every vertex in $T$, via the flow $F$. It is easy to see that every flow-path is used at most twice.
\end{proof}

For convenience, when given an $\alpha$-well-linked subset $S$ of vertices in a graph $G$, we will omit the subdivision of the edges in $\out(S)$, and we will say that the edges $e\in \out(S)$ send flow to each other, instead of the corresponding vertices $t_e$.

We will also use the algorithm of Arora, Rao and Vazirani~\cite{ARV} for balanced cut, summarized below.

\begin{theorem}[Balanced Cut~\cite{ARV}]\label{thm: ARV}
Let $G$ be any $n$-vertex graph, and suppose there is a partition of the vertices of $G$ into two sets, $A$ and $B$, with $|A|,|B|\geq \eps n$ for some constant $\eps>0$, and 
$|E(A,B)|=c$. Then there is an efficient algorithm to find a partition $(A',B')$ of the vertices of $G$, such that $|A'|,|B'|\geq \epsilon' n$ for some constant $0<\epsilon'<\epsilon$, and $|E(A',B')|\leq O(c\sqrt{\log n})$.
\end{theorem}

\subsection{Canonical Vertex Sets and Solutions}

As already mentioned in the Introduction, we will perform a number of graph contraction steps on the input graph $\G$, where in each such graph contraction step, a sub-graph of $\G$ will be replaced with a grid. So in general, if $H$ is the current graph, we will also be given a collection $\zset$ of disjoint subsets of vertices of $H$, such that for each $Z\in \zset$, $H[Z]$ is the $k_Z\times k_Z$ grid, for some $k_Z\geq 2$. We will also ensure that $\Gamma_H(Z)$ is precisely the set of the vertices in the first row of the grid $H[Z]$, and the edges in $\out_H(Z)$ form a matching between $\Gamma_H(Z)$ and $T_H(Z)$. Given such a graph $H$, and  a collection $\zset$ of vertex subsets, we will be looking for solutions in which the edges of the grids $H[Z]$ do not participate in any crossings. This motivates the following definitions of canonical vertex sets and canonical solutions.

Assume that we are given a graph $G$ and a collection $\zset$ of disjoint subsets of vertices of $G$, such that each subset $Z\in\zset$ is $1$-well-linked (but some vertices of $G$ may not belong to any subset $Z\in \zset$).

\begin{definition}\label{definiton: canonical subset} We say that a subset $J\sse V$ of vertices is \emph{canonical} for $\zset$ iff for each $Z\in\zset$, either $Z\sse J$, or $Z\cap J=\emptyset$.
\end{definition}

We next define canonical drawings and canonical solutions w.r.t. the collection $\zset$ of subsets of vertices:

\begin{definition}\label{definition: canonical drawing} Let $G=(V,E)$ be any graph, and $\zset$ any collection of disjoint subsets of vertices of $G$. We say that a drawing $\phi$ of $G$ is \emph{canonical} for $\zset$ iff for each $Z\in \zset$, no edge of $G[Z]$ participates in crossings.
Similarly, we say that a solution $E^*$ to the \MP problem on $G$ is \emph{canonical} for $\zset$, iff for each $Z\in \zset$, no edge of $G[Z]$ belongs to $E^*$.
\end{definition}

\begin{definition}
Given a graph $G$, a simple cycle $X\sse G$ (that may be empty), and a collection $\zset$ of disjoint subsets of vertices of $G$, a \emph{strong solution} to problem $\pi(G,X,\zset)$ is a drawing $\psi$ of $G$, in which the edges of $E(X)\cup\left(\bigcup_{Z\in\zset}E(G[Z])\right )$ do not participate in any crossings, and $G$ is embedded inside the bounding box $X$. The cost of the solution is the number of edge crossings in $\psi$. A \emph{weak solution} to  problem $\pi(G,X,\zset)$ is a subset $E'\sse E(G)\setminus E(X)$ of edges, such that graph $G\setminus E'$ has a \emph{planar drawing} inside the bounding box $X$, and for all $Z\in \zset$, $E'\cap E(G[Z])=\emptyset$. \end{definition}

We will sometimes use the above definition for problem $\pi(G',X,\zset)$, where $G'$ is a sub-graph of $G$. That is, some sets $Z\in \zset$ may not be contained in $G'$, or only partially contained in it. We can then define $\zset'$ to contain, for each $Z\in \zset$, the set $Z\cap V(G')$. We will sometimes use the notion of weak or strong solution to problem $\pi(G',X,\zset)$ to mean weak or strong solutions to $\pi(G',X,\zset')$, to simplify notation.

\subsection{Cuts in Grids}

The following simple claim about grids and its corollary are used throughout the paper.

\begin{claim}\label{claim: cut of grids}
Let $Z$ be the $k\times k$ grid, for any integer $k\geq 2$, and let $\Gamma$ denote the set of vertices in the first row of $Z$. Let $(A,B)$ be any partition of the vertices of $Z$, with $A,B\neq \emptyset$. Then $|E(A,B)|\geq\min\set{|A\cap \Gamma|, |B\cap \Gamma|}+1$.
\end{claim}

\begin{proof}
Let $\Gamma_A=\Gamma\cap A$, $\Gamma_B=\Gamma\cap B$, and assume w.l.o.g. that $|\Gamma_A|\leq |\Gamma_B|$.
If $\Gamma_A=\emptyset$, then the claim is clearly true.
Otherwise, there is some vertex $t\in \Gamma_A$, such that a vertex $t'$ immediately to the right or to the left of $t$ in the first row of the grid belongs to $\Gamma_B$. Let $e=(t,t')$ be the corresponding edge in the first row of $Z$.
We can find a collection of $|\Gamma_A|$ edge-disjoint paths, connecting vertices in $\Gamma_A$ to vertices in  $\Gamma_B$, that do not include the edge $e$, as follows: assign a distinct row of $Z$ (different from the first row) to each vertex in $\Gamma_A$. Route each such vertex inside its column to its designated row, and inside this row to the column corresponding to some vertex in $\Gamma_B$. If we add the path consisting of the single edge $e$, we will obtain a collection of $|\Gamma_A|+1$ edge-disjoint paths, connecting vertices in $\Gamma_A$ to vertices in $\Gamma_B$. All these paths have to be disconnected by the above cut.
\end{proof}

\begin{corollary}\label{corollary: canonical s-t cut} Let $G$ be any graph, $\zset$ any collection of disjoint subsets of vertices of $G$, such that for each $Z\in \zset$, $G[Z]$ is the $k_Z\times k_Z$ grid, for $k_Z\geq 2$. Moreover, assume that each vertex in the first row of $Z$ is adjacent to exactly one edge in $\out_G(Z)$, and no other vertex of $Z$ is adjacent to edges in $\out_G(Z)$. Let $s,t$ be any pair of vertices of $G$, that do not belong to any set $Z\in \zset$, and let $(A,B)$ be the minimum $s$--$t$ cut in $G$. Then both sets $A$ and $B$ are canonical w.r.t. $\zset$.
\end{corollary}
\begin{proof}
Assume for contradiction that some set $Z\in \zset$ is split between the two sides, $A$ and $B$. Let $\Gamma=\Gamma(Z)$ denote the set of vertices in the first row of $Z$, and let $\Gamma_A=\Gamma\cap A$, $\Gamma_B=\Gamma\cap B$. Assume w.l.o.g. that $|\Gamma_A|\leq |\Gamma_B|$.
Then by Claim~\ref{claim: cut of grids} $|E(A\cap Z,B\cap Z)|>|\Gamma_A|$, and so the value of the cut $(A\setminus Z,B\cup Z)$ is smaller than the value of the cut $(A,B)$, a contradiction.
\end{proof}

\begin{claim}\label{claim: cutting the grid}
Let $Z$ be the $k\times k$ grid, for any integer $k\geq 2$, and let $\Gamma$ be the set of vertices in the first row of $Z$. Suppose we are given any partition $(A,B)$ of $V(Z)$, denote $\Gamma_A=\Gamma\cap A$, $\Gamma_B=\Gamma\cap B$, and assume that $|\Gamma_B|\leq |\Gamma_A|$. Then $|B|\leq 4|E(A,B)|^2$.
\end{claim}
\begin{proof}
Denote $M=|E(A,B)|$. Let $C_A$ denote the set of columns associated with the vertices in $\Gamma_A$, and similarly, $C_B$ is the set of columns associated with the vertices in $\Gamma_B$. Notice that $(C_A,C_B)$ define a partition of the columns of $Z$. We consider three cases.

The first case is when no column is completely contained in $A$. In this case, for every column in $C_A$, at least one edge must belong to $E(A,B)$, and so $M\geq |\Gamma_A|\geq k/2$. Since $|B|\leq |Z|\leq k^2$, the claim follows. From now on we assume that there is some grid column, denoted by $c$, that is completely contained in $A$.

The second case is when some grid column $c'$ is completely contained in $B$. In this case, it is easy to see that $M\geq k$ must hold, as there are $k$ edge-disjont paths connecting vertices of $c$ to vertices of $c'$ in $Z$. So $|B|\leq |Z|\leq k^2\leq M^2$, as required.

Finally, assume that no column is contained in $B$. Let $C'_B$ be the set of columns that have at least one vertex in $B$. Clearly, $M\geq |C'_B|$. Let $M'$ be the maximum number of vertices in any column $c'\in C'_B$, which are contained in $B$. Then $M\geq M'$ must hold, since there are $M'$ edge-disjoint paths between the vertices of column $c$, and the vertices of $c'\cap B$. On the other hand, $|B|\leq |C'_B|\cdot M'\leq M^2$.
\end{proof}

\subsection{Well-linked Decompositions}

The next theorem summarizes well-linked decomposition of graphs, which has been used extensively in graph decomposition (e.g., see~\cite{CSS,Raecke}). For completeness we provide its proof in Appendix.

\begin{theorem}[Well-linked decomposition]\label{thm: well-linked} Given any graph $G=(V,E)$, and any subset $J\sse V$ of vertices, we can efficiently find a partition $\jset$ of $J$, such that each set $J'\in \jset$ is $\alpha^*$-well-linked for $\alpha^*=\Omega(1/(\log^{3/2} n\log\log n))$, and $\sum_{J'\in \jset}|\out(J')|\leq 2\out(J)$.
\end{theorem}

We now define some additional properties that set $J$ may possess, that we use throughout the paper. We will then show that if a set $J$ has any collection of these properties, then we can find a well-linked decomposition $\jset$ of $J$, such that every set $J'\in \jset$ has these properties as well.

\begin{definition} Given a graph $G$ and any subset $J\sse V(G)$ of its vertices, we say that $J$ has property (P1) iff the vertices of $T(J)$ are connected in $G\setminus J$. We say that it has property (P2) iff there is a planar drawing of $J$ in which all interface vertices $\Gamma(J)$ lie on the boundary of the same face, that we refer to as the \emph{outer face}. We denote such a planar drawing by $\pi(J)$. If there are several such drawing, we select any of them arbitrarily.
\end{definition}

The next theorem is an extension of Theorem~\ref{thm: well-linked}, and its proof appears in Appendix.

\begin{theorem}\label{thm: well-linked-general} Suppose we are given any graph $G=(V,E)$, a subset $J\sse V$ of vertices, and a collection $\zset$ of disjoint subsets of vertices of $V$, such that each set $Z\in \zset$ is $1$-well-linked. Then we can efficiently find a partition $\jset$ of $J$, such that each set $J'\in \jset$ is $\alpha^*$-well linked for $\alpha^*=\Omega(1/(\log^{3/2} n\log\log n))$, and $\sum_{J'\in \jset}|\out(J')|\leq 2\out(J)$. Moreover, if $J$ has any combination of the following three properties: (1) property (P1); (2) property (P2); (3) it is a canonical set for $\zset$, then  each set $J'\in \jset$ will also have the same combination of these properties.\end{theorem}

Throughout the paper, we use $\alpha^*$ to denote the parameter from Theorem~\ref{thm: well-linked-general}. 

\label{--------------------------------sec: high level overview----------------------------------}
\section{High Level Algorithm Overview}\label{sec: overview}

In this section we provide a high-level overview of the algorithm.
We start by defining the notion of nasty vertex sets.

\begin{definition}
Given a graph $G$,
we say that a subset $S\sse V(G)$ of vertices is \emph{nasty} iff it has properties (P1) and (P2), and $|S|\geq \frac{2^{16}\cdot \dmax^6}{(\alpha^*)^2}\cdot|\Gamma(S)|^2$, where $\alpha^*$ is the parameter from Theorem \ref{thm: well-linked}.
\end{definition}

Note that we do not require that $G[S]$ is connected.

For the sake of clarity, let us first assume that the input graph $\G$ contains no nasty sets. Our algorithm then proceeds as follows. We use a balancing parameter $\rho=O(\optcro{G}\cdot\poly(\dmax\cdot \log n))$ whose exact value is set later. The algorithm  has $O(\rho\cdot\log n)$ iterations. At the beginning of each iteration $h$, we are given a collection $G_1,\ldots, G_{k_h}$ of $k_h\leq \optcro{G}$ disjoint sub-graphs of $\G$, together with bounding boxes $X_i\sse G_i$ for all $i$. We are guaranteed that w.h.p., there is a strong solution to each problem $\pi(G_i,X_i)$, of total cost at most $\optcro{\G}$.
In the first iteration, $k_1=1$, and the only graph is $G_1=\G$, whose bounding box is $X_0=\emptyset$. 

We now proceed to describe each iteration. The idea is to find a \emph{skeleton} $K_i$ for each graph $G_i$, with $X_i\sse K_i$, such that $K_i$ only contains good edges --- that is, edges that do not participate in any crossings in the optimal solution $\phi$, and $K_i$ has a unique planar drawing, in which $X_i$ serves as the bounding box. Therefore, we can efficiently find the drawing $\phi_{K_i}$ of the skeleton $K_i$, induced by the optimal drawing $\phi$. We then decompose the remaining graph $G_i\setminus E(K_i)$ into \emph{clusters}, by removing a small subset of edges from it, so that, on the one hand, for each such cluster $C$, we know the face $F_C$ of $\phi_{K_i}$ where we should embed it, while on the other hand, different clusters $C,C'$ do not interfere with each other, in the sense that we can find an embedding of each one of these clusters separately, and their embeddings do not affect each other. For each such cluster $C$, we then define a new problem $\pi(C,\gamma(F_C))$, where $\gamma(F_C)$ is the boundary of the face $F_C$. We will ensure that all resulting sub-problems have strong solutions whose total cost is at most $\optcro{\G}$. In particular, there are at most $\optcro{G}$ resulting sub-problems, for which $\emptyset$ is not a feasible weak solution. Therefore, in the next iteration we will need to solve at most $\optcro{G}$ new sub-problems.
The main challenge is to find $K_i$, such that the number of vertices in each such cluster $C$ is bounded by roughly $(1-1/\rho)|V(G_i)|$, so that the number of iterations is indeed bounded by $O(\rho \log n)$. We need this bound on the number of iterations, since the probability of successfully constructing the skeletons in each iteration is only $(1-1/\rho)$. Roughly speaking, we are able to build the skeleton as required, if we can find a $\rho$-balanced $\alpha$-well-linked bipartition of the vertices of $G_i$, where $\alpha=1/\poly(\dmax \cdot \log n)$. We are only able to find such a partition if no nasty sets exist in $\G$. More precisely, we show an efficient algorithm, that either finds the desired bi-partition, or returns a nasty vertex set.

In order to obtain the whole algorithm, we therefore need to deal with nasty sets. We do so by performing a graph contraction step, which is formally defined in the next section. Informally,
given a nasty set $S$, we find a partition $\xset$ of $S$, such that for every pair $X,X'\in \xset$, the graphs $\G[X],\G[X']$ share at most one interface vertex and no edges. Each such graph $\G[X]$ is also $\alpha^*$-well-linked, has properties (P1) and (P2), and $\sum_{X\in \xset}|\Gamma(X)|\leq O(|\Gamma(S)|)$. We then replace each sub-graph $\G[X]$ of $\G$ by a grid $Z_X$, whose interface is $\Gamma(X)$. After we do so for each $X\in \xset$, we denote by $\G_{|S}$ the resulting contracted graph. Notice that we have replaced $\G[S]$ by a much smaller graph, whose size is bounded by $O(|\Gamma(S)|^2)$. Let $\zset$ denote the collection of sets $V(Z_X)$ of vertices, for $X\in \xset$.
We then show that the cost of the optimal solution to problem $\pi(\G_{|S},\emptyset,\zset)$ is at most $\poly(\dmax\cdot\log n)\optcro{G}$. Therefore, we can restrict our attention to canonical solutions only. We also show that it is enough to find a weak solution to problem $\pi(\G_{|S},\emptyset,\zset)$, in order to obtain a weak solution for the whole graph $\G$.
Unfortunately, we do not know how to find a nasty set $S$, such that the corresponding contracted graph $\G_{|S}$ contains no nasty sets. Instead, we do the following. Let $H=\G_{|S}$ be the current graph, which is a result of the graph contraction step on some set $S$ of vertices, and let $\zset$ be the corresponding collection of sub-sets of vertices representing the grids. Suppose we can find a nasty canonical set $R$ in the graph $H$. We show that this allows us to find a new set $S'$ of vertices in $\G$, such that the contracted graph $\G_{|S'}$ contains fewer vertices than $\G_{|S}$.

Returning to our algorithm, let $\G_{|S}$ be the current contracted graph. We show that with high probability, the algorithm either returns a weak solution for $\G_{|S}$ of cost $O\left ((\optcro{G})^5\poly(\dmax \cdot \log n)\right )$, or it returns a nasty canonical subset $S'$ of $\G_{|S}$. In the former case, we can recover a good weak solution for the original graph $\G$. In the latter case, we find a subset $S''$ of vertices in the original graph $\G$, and perform another contraction step on $\G$, obtaining a new graph $\G_{|S''}$, whose size is strictly smaller than that of $\G_{|S}$. We then apply the algorithm to graph $\G_{|S''}$. Since the total number of graph contraction steps is bounded by $n$, after $n$ such iterations, we are guaranteed w.h.p. to obtain a weak feasible solution of cost $O\left ((\optcro{G})^5\poly(\dmax \cdot \log n)\right )$ to $\pi(\G,\emptyset)$, thus satisfying the requirements of Theorem~\ref{thm:main}.
We now turn to formal description of the algorithm. One of the main ingredients is the graph contraction step, summarized in the next section.

\section{Graph Contraction Step}\label{sec: graph contraction}
The input to the graph contraction step consists of the input graph $\G$, and a subset $S\sse V(\G)$ of vertices, for which properties (P1) and (P2) hold. It will be convenient to think of $S$ as a nasty set, but we do not require it.

Let $\cset=\set{G_1,\ldots,G_q}$ be the set of all connected components of $\G[S]$. For each $1\leq i\leq q$, let $\Gamma_i=V(G_i)\cap \Gamma(S)=\Gamma(V(G_i))$ be the set of the interface vertices of $G_i$. 
The goal of the graph contraction step is to find,  for each $1\leq i\leq q$, 
a partition $\xset_i$ of the set $V(G_i)$, that has the following properties. Let $\xset=\bigcup_{i=1}^q\xset_i$.


\begin{properties}{C}
\item Each set $X\in \xset$ is $\alpha^*$-well-linked, and has properties (P1) and (P2).
Moreover, there is a planar drawing $\pi'(X)$ of $\G[X]$, and a simple closed curve $\gamma_X$, such that $\G[X]$ is embedded inside $\gamma_X$ in $\pi'(X)$, and the vertices of $\Gamma(X)$ lie on $\gamma_X$.\label{property: subsets-first}

\item For each $X\in\xset$, either $|\Gamma(X)|=2$, or there is a partition $(C^*_X,R_1,\ldots,R_t)$ of $X$, such that $\G[C^*_X]$ is $2$-connected and $\Gamma(X)\sse C^*_X$. Moreover, for each $1\leq t'\leq t$, there is a vertex $u_{t'}\in C^*_X$, whose removal from $\G[X]$ separates the vertices of $R_{t'}$ from the remaining vertices of $X$.
\label{property: structure of X}

\item For each pair $X,X'\in \xset$, the two sets of vertices are completely disjoint, except for possibly sharing one interface vertex, $v\in \Gamma(X)\cap \Gamma(X')$. 
\label{property: disjointness}

\item For each $1\leq i\leq q$, if $\Gamma_i'=\bigcup_{X\in\xset_i}\Gamma(X)$, then $|\Gamma_i'|\leq 9|\Gamma_i|$. \label{property: subsets - last}

\item For each $X\in \xset$, $|X|\geq (\alpha^*|\Gamma(X)|)^2/64\dmax^2$.\label{property: size-last}
\end{properties}

 For each set $X\in \xset$, we now define a new graph $Z'_X$, that will eventually replace the sub-graph $\G[X]$ in $\G$. Intuitively, we need $Z'_X$ to contain the vertices of $\Gamma(X)$ and to be $1$-well-linked w.r.t. these vertices. We also need it to have a unique planar embedding where the vertices of $\Gamma(X)$ lie on the boundary of the same face, and finally, we need the size of the graph $Z'_X$ to be relatively small, since this is a graph contraction step. The simplest graph satisfying these properties is a grid of size $|\Gamma(X)|\times |\Gamma(X)|$. 
 
 Specifically, we first define a graph $Z_X$ as follows: if $|\Gamma_X|=1$, then $Z_X$ consists of a single vertex, and if $|\Gamma_X|=2$, then $Z_X$ consists of a single edge. Otherwise, $Z_X$ is a grid of size $|\Gamma(X)|\times |\Gamma(X)|$. In order to obtain the graph $Z'_X$, we add the set $\Gamma(X)$ of vertices to $Z_X$, and add a matching between the vertices of the first row of the grid and the vertices of $\Gamma(X)$. This is done so that the order of the vertices of $\Gamma(X)$ along the first row of the grid is the same as their order along the curve $\gamma_X$ in the drawing $\pi'(X)$. We refer to these new edges as the \emph{matching edges}. For the cases where $|\Gamma_X|=1$ and $|\Gamma_X|=2$, we obtain $Z'_X$ by adding the vertices of $\Gamma(X)$ to $Z_X$, and adding an arbitrary matching between $\Gamma_X$ and the vertices of $Z_X$. (See Figure~\ref{fig: grids}).

\begin{figure}[h]
\centering
\subfigure[General case]{
	\scalebox{0.4}{\includegraphics{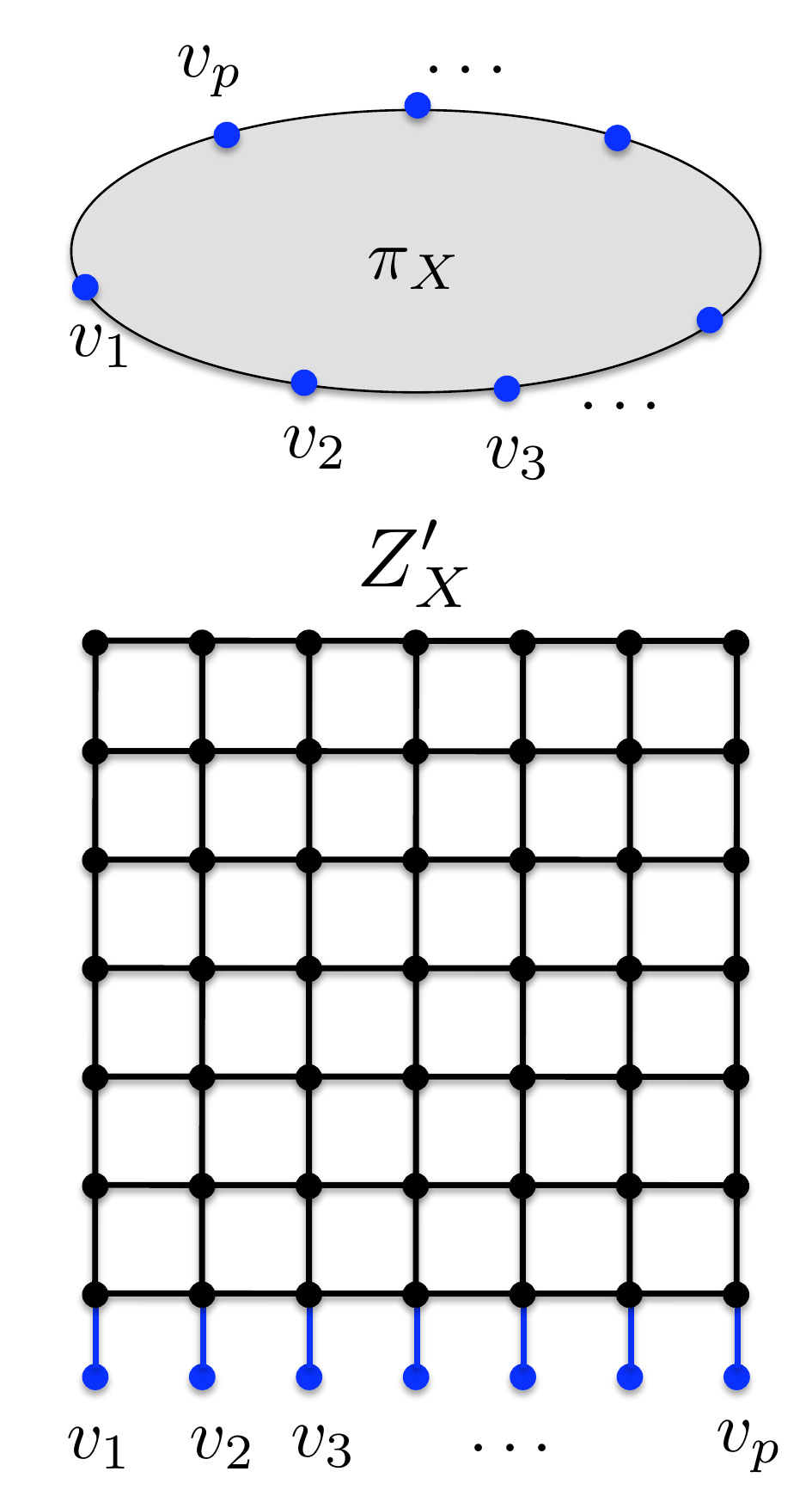}}
	\label{fig: grid}
}
\hfill
\subfigure[$|\Gamma(X)|=1$]{
	\scalebox{0.3}{\includegraphics{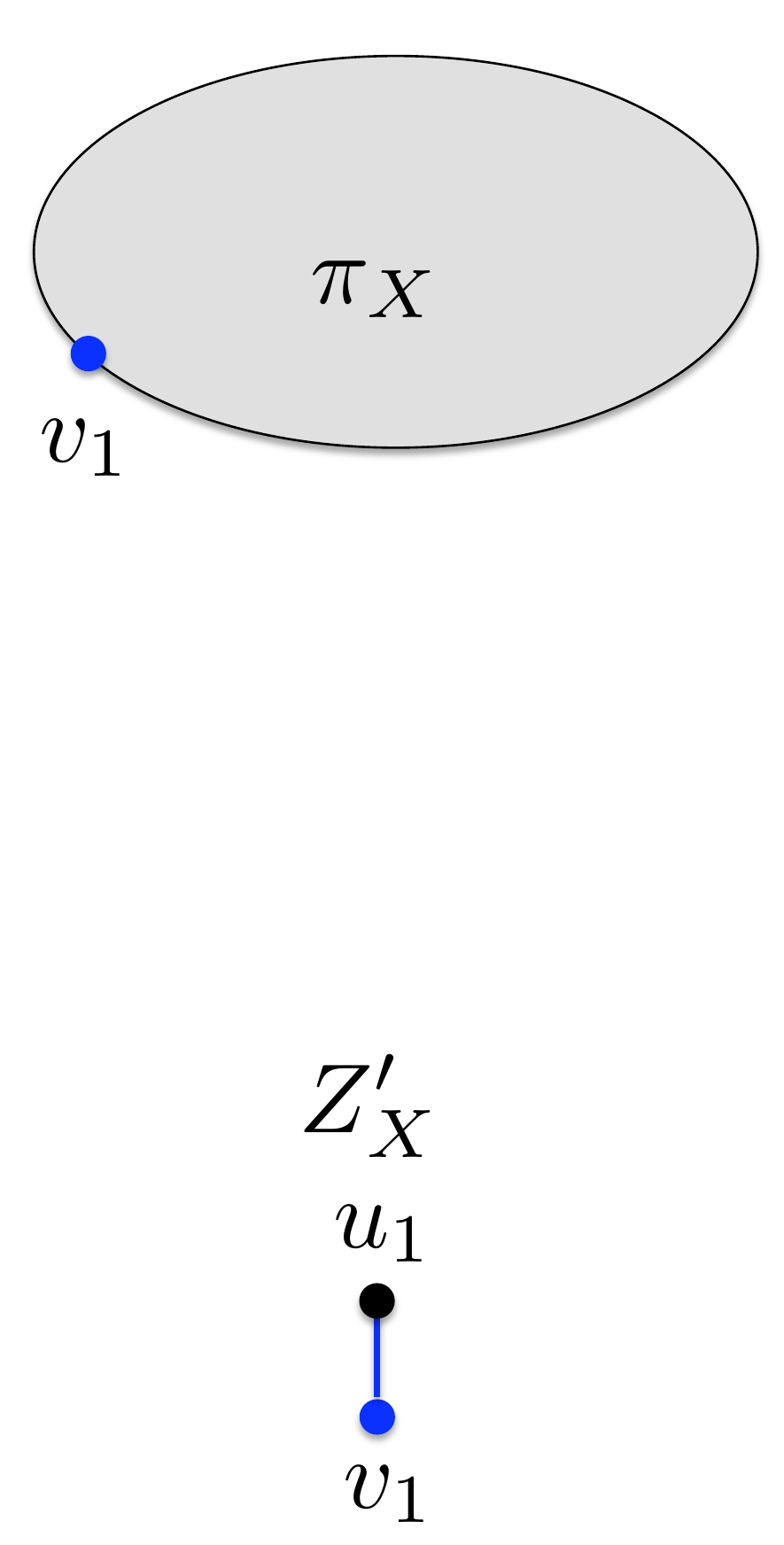}}
	\label{fig: grid 1 interface}
}
\hfill
\subfigure[$|\Gamma(X)|=2$]{
	\scalebox{0.3}{\includegraphics{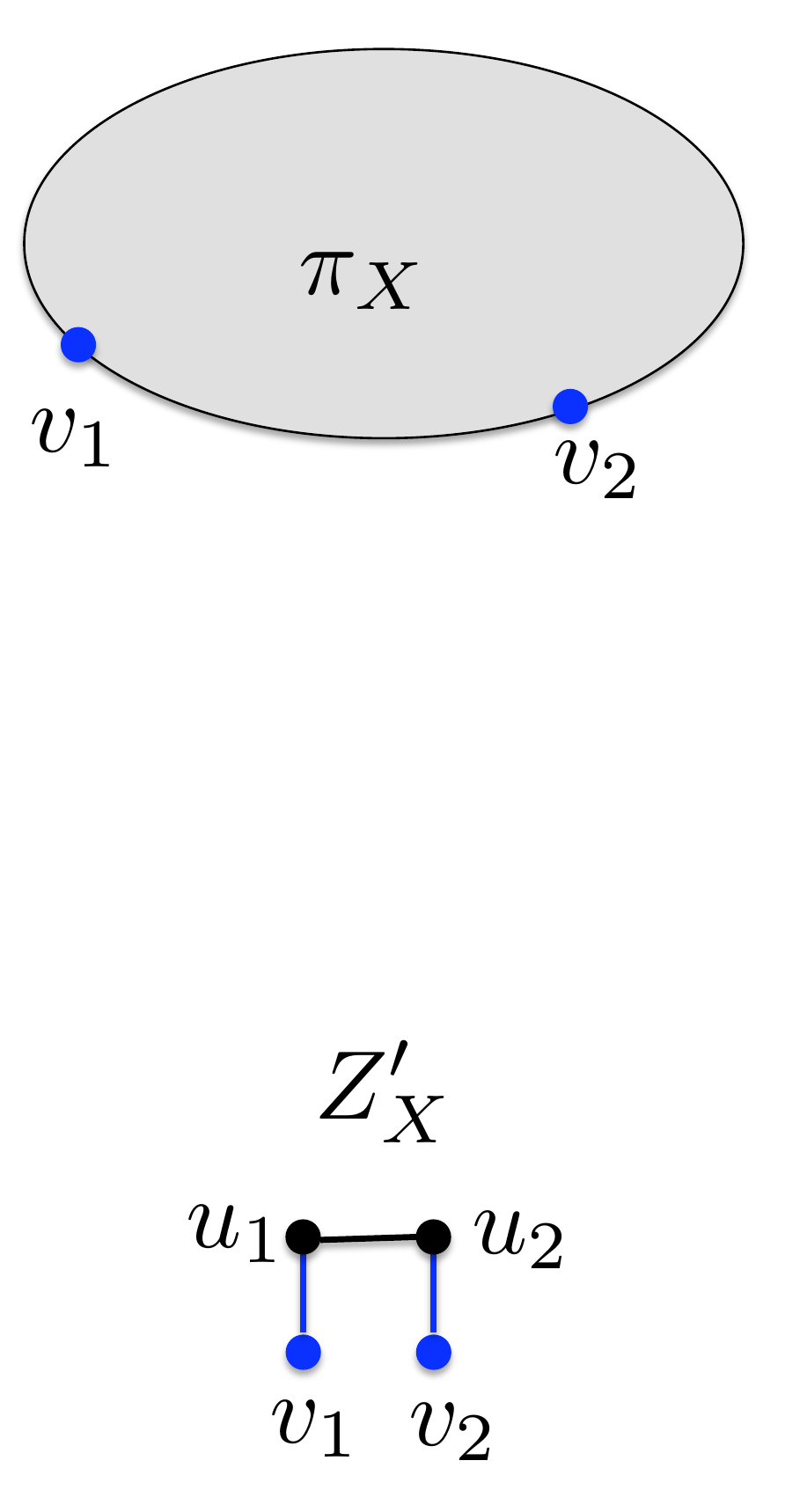}}
	\label{fig: grid 2 interface}
}

\caption{Graph $Z'_X$. The matching edges and the interface vertices are blue; the grid $Z_X$ is black.}\label{fig: grids}
\end{figure}



The contracted graph $\G_{|S}$ is obtained from $\G$, by replacing, for each $X\in\xset$, the subgraph $\G[X]$ of $\G$, with the graph $Z'_X$. This is done as follows: first, delete all vertices and edges of $\G[X]$, except for the vertices of $\Gamma(X)$, from $\G$, and add the edges and the vertices of $Z'_X$ instead. Next, identify the copies of the interface vertices $\Gamma[X]$ in the two graphs. Let $H=\G_{|S}$ denote the resulting contracted graph.
Notice that 

\begin{equation}\label{eq: final size}
\sum_{i=1}^q\sum_{X\in \xset_i}|V(Z'_X)|\leq \sum_{i=1}^q\sum_{X\in \xset_i}2|\Gamma(X)|^2\leq \sum_{i=1}^q2|\Gamma_i'|^2\dmax^2\leq 162\dmax^2|\Gamma|^2
\end{equation}

(we have used the fact that a vertex may belong to the interface of at most $\dmax$ sets  $X\in \xset_i$, and Property~(\ref{property: subsets - last})). Therefore, if the initial vertex set $S$ is nasty, then we have indeed reduced the graph size, as $|V(H)|<|V(\G)|$.

We now define a collection $\zset$ of subsets of vertices of $H$, as follows:
$\zset=\set{V(Z_X)\mid X\in \xset}$. 
Notice that these sets are completely disjoint, as $Z_X$ does not contain the interface vertices $\Gamma(X)$. 
Moreover, for each $Z\in \zset$, $H[Z]$ is a grid, $\Gamma_H(Z)$ consists of the vertices in the first row of the grid, and $\out_H(Z)$ consists of the set of the matching edges, each of which connects a vertex in the first row of the grid $Z$ to a distinct vertex in $T_H(Z)$.
Using Definitions~\ref{definiton: canonical subset} and \ref{definition: canonical drawing},
we can now define canonical subsets of vertices, canonical drawings and canonical solutions to the \MP problem on $H$, with respect to $\zset$.
Our main result for graph contraction is summarized in the next theorem, whose proof appears in Appendix.

\begin{theorem}\label{thm: graph contraction}
Let $S\sse V(\G)$ be any subset of vertices with properties (P1) and (P2), and let $\set{G_1,\ldots,G_q}$ be the set of all connected components of graph $\G[S]$. Then
for each $1\leq i\leq q$, we can efficiently find a partition $\xset_i$ of  $V(G_i)$, such that the resulting partition $\xset=\bigcup_{i=1}^q\xset_i$ of $S$ has properties~(\ref{property: subsets-first})--(\ref{property: size-last}). Moreover, there is a canonical drawing of the resulting contracted graph $H=\G_{|S}$ with $O(d_{\max}^9 \cdot \log^{10} n \cdot (\log\log n)^4 \cdot \optcro{\G})$ crossings.
\end{theorem}

The next claim shows, that in order to find a good solution to the \MP problem on $\G$, it is enough to solve it on $\G_{|S}$. 

\begin{claim}\label{claim: enough to solve contracted graph}
Let $S$ be any subset of vertices of $\G$, $\xset$ any partition of $S$ with properties~(\ref{property: subsets-first})--(\ref{property: size-last}), $H=\G_{|S}$ the corresponding contracted graph and $\zset$ the collection of grids $Z_X$ for $X\in \xset$. Then given any canonical solution $E^*$ to the \MP problem on $H$, we can efficiently find a solution  of cost $O(\dmax)|E^*|$
to \MP on $\G$.\end{claim}

\begin{proof}
Partition set $E^*$ of edges into two subsets: $E^*_1$ contains all edges that belong to sub-graphs $Z'_X$ for $X\in \xset$, and $E^*_2$ contains all remaining edges. Notice that since $E^*$ is a canonical solution, each edge $e\in E^*_1$ must be a matching edge for some graph $Z'_X$. Also from the construction of the contracted graph $H$, all edges in $E^*_2$ belong to $E(\G)$.

Consider some set $X\in \xset$, and let $\Gamma'(X)\sse \Gamma(X)$ denote the subset of the interface vertices of $Z'_X$, whose matching edges belong to $E^*_1$. Let $\Gamma'=\bigcup_{X\in \xset}\Gamma'(X)$. We now define a subset $E^{**}_1$ of edges of $\G$ as follows: for each vertex $v\in \Gamma'$, add all edges incident to $v$ in $\G$ to $E^{**}_1$. Finally, we set $E^{**}=E^{**}_1\cup E^*_2$. Notice that $E^{**}$ is a subset of edges of $\G$, and $|E^{**}|=|E_1^{**}|+|E_2^*|\leq \dmax|E_1^*|+|E_2^*|\leq \dmax |E^*|$. In order to complete the proof of the claim, it is enough to show that $E^{**}$ is a feasible solution to the \MP problem on $\G$.

Let $\G'=\G\setminus E^{**}$, let $H'=H\setminus E^*$, and let $\psi$ be a planar drawing of $H'$. It is now enough to construct a planar drawing $\psi'$ of $\G'$. In order to do so, we start from the planar drawing $\psi$ of $H'$. We then consider the sets $X\in \xset$ one-by-one. For each such set, we replace the drawing of $Z'_X\setminus \Gamma'(X)$ with a drawing of $\G[X]\setminus \Gamma'(X)$. The drawings of the vertices in $\Gamma(X)$ are not changed by this procedure. After all sets $X\in \xset$ are processed, we will obtain a planar drawing of graph $\G'$ (that may also contain drawings of some edges in $E^{**}$, that we can simply erase).

Consider some such set $X\in \xset$. Let $G$ be the current graph (obtained from $H'$ after a number of such replacement steps), and let $\psi$ be the current planar drawing of $G$. Observe that the grid $Z_X$ has a unique planar drawing. We say that a planar drawing of graph $Z'_X\setminus \Gamma'(X)$ is \emph{standard} in $\psi$, iff we can draw a simple closed curve $\gamma'_X$, such that $Z_X$ is embedded completely inside $\gamma'_X$; no other vertices or edges of $G$ are embedded inside $\gamma'_X$; the only edges that $\gamma'_X$ intersects are the matching edges of $Z'_X\setminus \Gamma'(X)$, and each such matching edge is intersected  exactly once by $\gamma'_X$ (see Figure~\ref{fig: standard drawing}).

\begin{figure}[h]
\scalebox{0.4}{\rotatebox{0}{\includegraphics{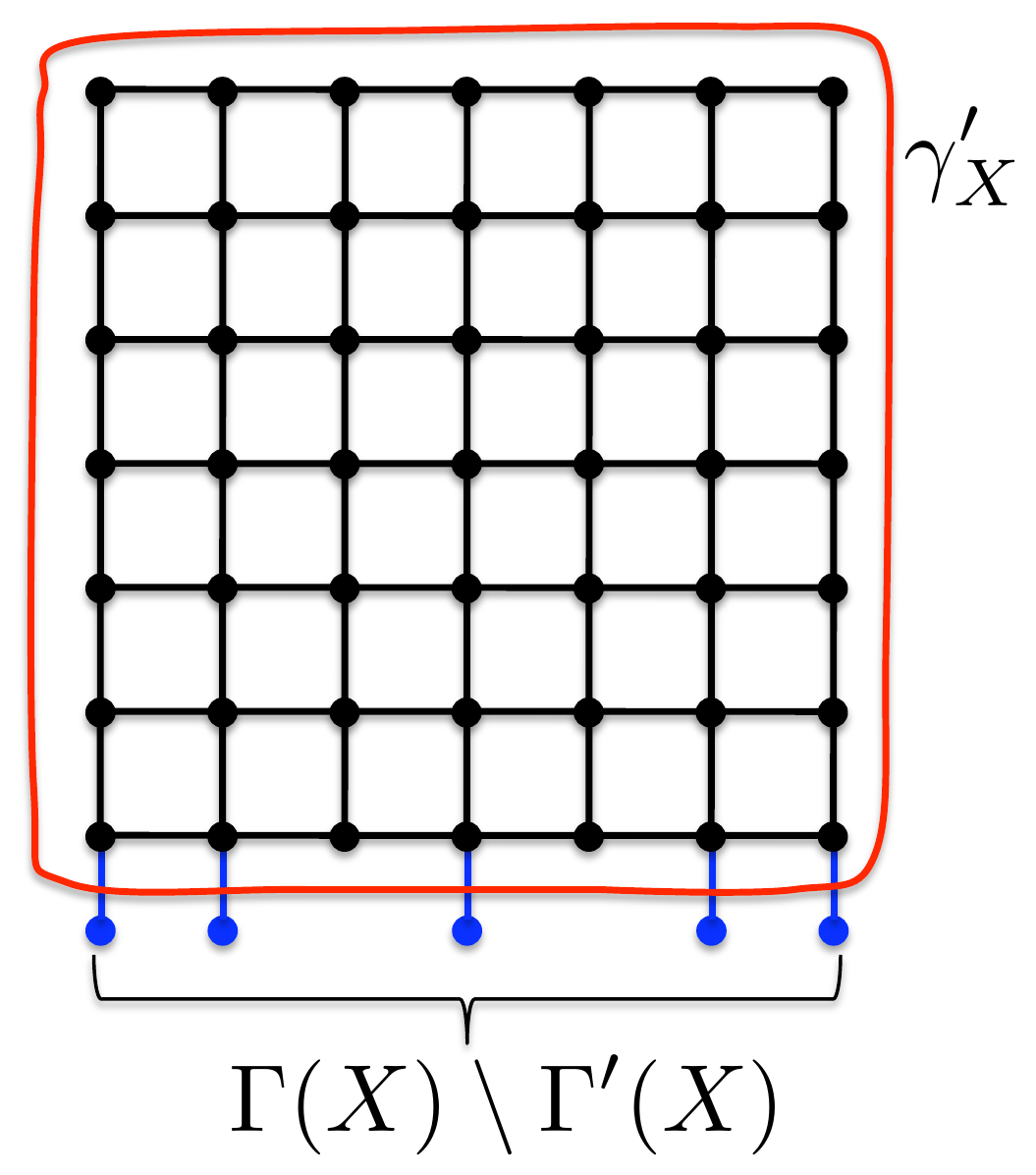}}}\caption{A standard drawing of $Z'_X\setminus \Gamma'(X)$ \label{fig: standard drawing}}
\end{figure}

It is possible that the drawing of $Z'_X\setminus \Gamma'(X)$ in $\psi$ is not standard.
However, since $\psi$ is planar, this can only happen for the following three reasons: (1) some connected component $C$ of the current graph $G$ is embedded inside some face of the grid $Z_X$: in this case we can simply move the drawing of $C$ elsewhere; (2) there is some subset $C$ of $V(G)$, and a vertex $v\in \Gamma(X)\setminus \Gamma'(X)$, such that $\Gamma_G(C)=v$, and $G[C]$ is embedded inside one of the faces of the grid $Z_X$ incident to the other endpoint of the matching edge of $v$; and (3) there is some subset $C$ of $V(G)$, and two consecutive vertices $u,v\in \Gamma(X)\setminus \Gamma'(X)$, such that $\Gamma_G(C)=\set{u,v}$, and $G[C]$ is embedded inside the unique face of the grid $Z_X$ incident to the other endpoints of the matching edges of $u$ and $v$  (See Figure \ref{fig: any to standard drawing}).  In the latter two cases, we simply move the drawing of $C$ right outside the grid, so that the corresponding matching edges now cross the curve $\gamma'(X)$.

\begin{figure}[h]
\centering
\subfigure{
	\scalebox{0.4}{\includegraphics{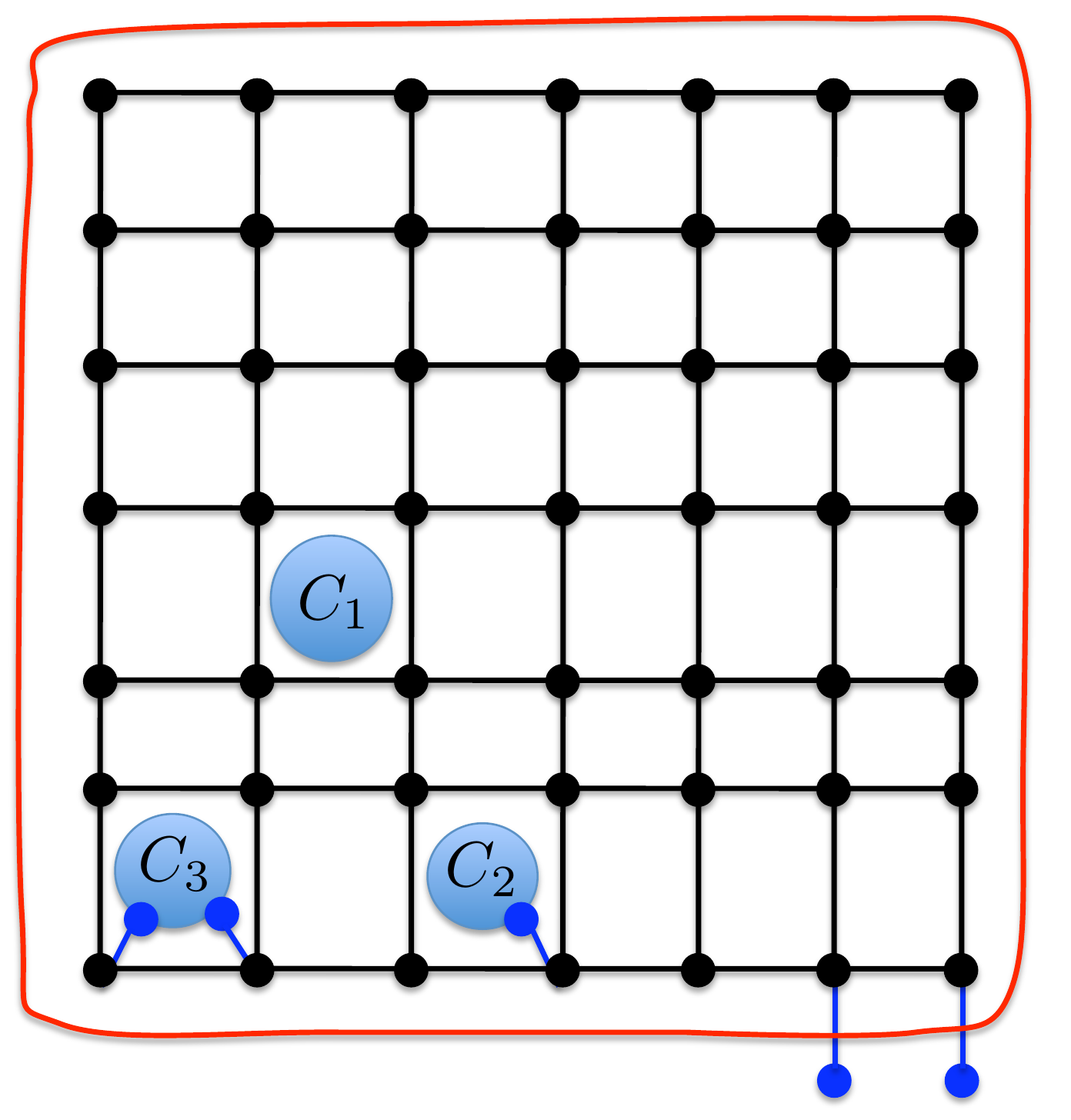}}
	\label{fig: any to standard before}
}
\hfill
\subfigure{
	\scalebox{0.4}{\includegraphics{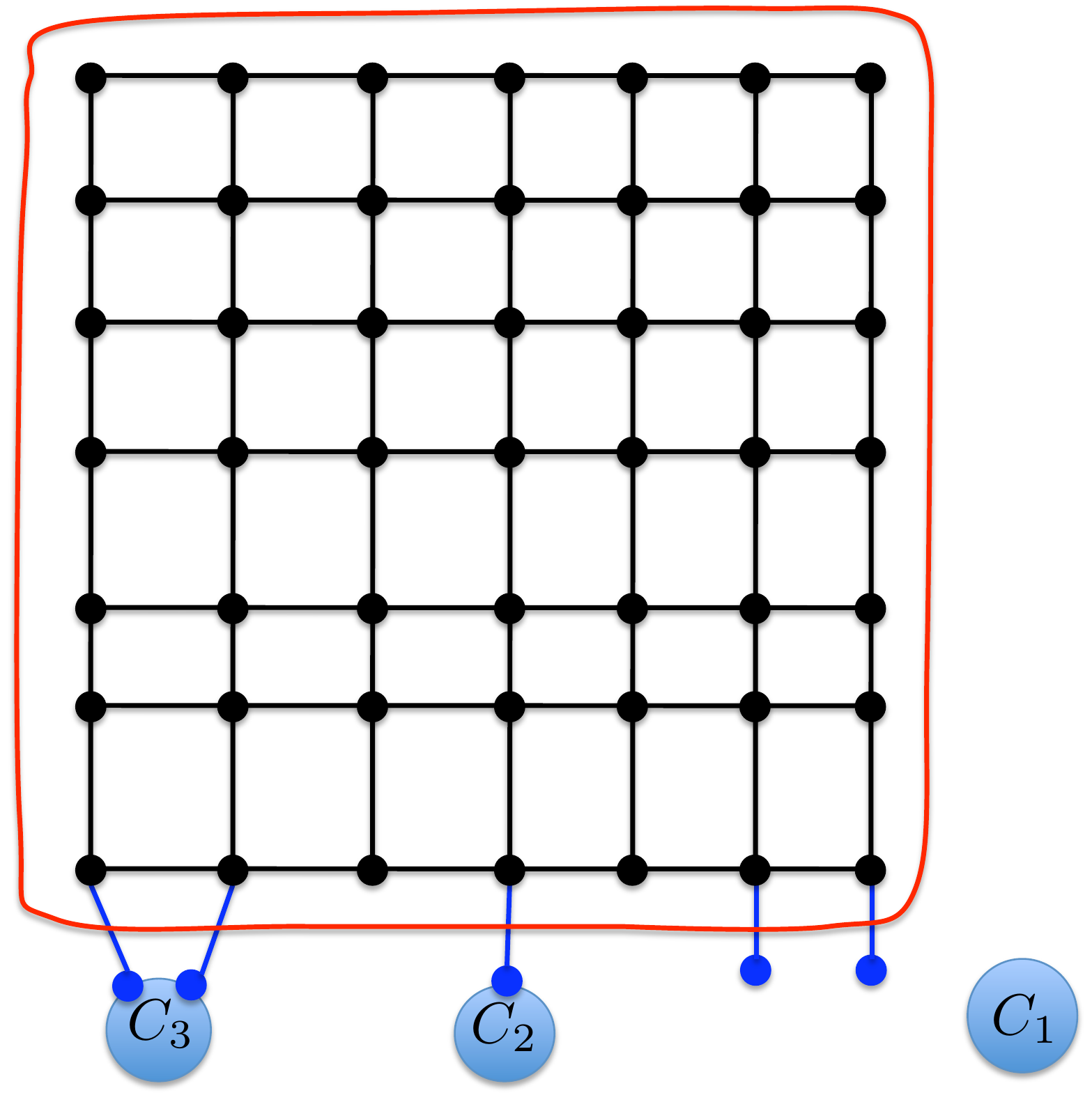}}
	\label{fig: any to standard after}
}
\caption{Transforming drawing $\psi$ to obtain a standard drawing of $Z'_X\setminus \Gamma'(X)$. Cases 1, 2 and 3 are illustrated by clusters $C_1$, $C_2$ and $C_3$, respectively. \label{fig: any to standard drawing}}
\end{figure}

To conclude, we can transform the current planar drawing $\psi$ of the graph $G$ into another planar drawing $\tilde{\psi}$, such that the induced drawing of $Z'_X\setminus \Gamma'(X)$ is standard. We can now draw a simple closed curve $\gamma''(X)$, such that $Z'_X\setminus \Gamma'(X)$ is embedded inside $\gamma''(X)$, no other vertices or edges are embedded inside $\gamma''(X)$, and the set of vertices whose drawings lie on $\gamma''(X)$ is precisely $\Gamma(X)\setminus \Gamma'(X)$. Notice that the ordering of the vertices of $\Gamma(X)\setminus \Gamma'(X)$ along this curve is exactly the same as their ordering along the curve $\gamma(X)$ in the planar embedding $\pi'(X)$ of $\G[X]$, guaranteed by Property (\ref{property: subsets-first}). Let $\pi''(X)$ be the drawing of $\G[X]\setminus \Gamma'(X)$ induced by $\pi'(X)$. We can now simply replace the drawing of $Z'_X\setminus \Gamma'(X)$ with the drawing $\pi''(X)$ of $\G[X]\setminus \Gamma'(X)$, identifying the curves $\gamma_X$ and $\gamma''_X$, and the drawings of the vertices in $\Gamma(X)\setminus \Gamma'(X)$ on them. The resulting drawing remains planar, and the drawings of the vertices in $\Gamma(X)$ do not change.
\end{proof}

Finally, we show that if we find a nasty canonical set in $\G_{|S}$, then we can contract $\G$ even further. The proof of the following theorem appears in Appendix.

\begin{theorem}\label{thm: nasty canonical set to contraction} Let $S$ be any subset of vertices of $\G$, $\xset$ any partition of $S$ with properties~(\ref{property: subsets-first})--(\ref{property: size-last}), $H=\G_{|S}$ the corresponding contracted graph, and $\zset$ the corresponding collection of grids $Z_X$ for $X\in \xset$. Then given any  nasty canonical vertex set $R\sse V(H)$, we can efficiently find a subset $S'\sse V(\G)$ of vertices, and a partition $\xset'$ of $S'$, such that properties~(\ref{property: subsets-first})--(\ref{property: size-last}) hold for $\xset'$, and if $H'=\G_{|S'}$ is the corresponding contracted graph, then $|V(H')|<|V(H)|$.
Moreover,  there is a canonical drawing $\phi'$ of $H'$ with $\cro_{\phi'}(H') = O(d_{\max}^9 \cdot \log^{10} n \cdot (\log\log n)^4 \cdot \optcro{\G})$.
\end{theorem}

Notice that Claim~\ref{claim: enough to solve contracted graph} applies to the new contracted graph as well.

\label{---------------------------------------------------------sec algorithm---------------------------------------------------}
\section{The Algorithm}\label{sec: alg}

The algorithm consists of a number of stages. In each stage $j$, we are given as input a subset $S$ of vertices of $\G$, the contracted graph $\H=\G_{|S}$, and the collection $\zset$ of disjoint sub-sets of vertices of $\H$, corresponding to the grids $Z_X$ obtained during the contraction step. The goal of stage $j$ is to either produce a nasty canonical set $R$ in $\H$, or to find a weak feasible solution to problem $\pi(\H,\emptyset,\zset)$. We prove the following theorem.

\begin{theorem}\label{thm: alg summary}
There is an efficient randomized algorithm, that, given a contracted graph $\H$, a corresponding collection $\zset$ of disjoint subsets of vertices of $\H$, and a bound $\opt'$ on the cost of the strong optimal solution to problem $\pi(\H,\emptyset,\zset)$, with probability at least $1/\poly(n)$, produces either a nasty canonical subset $R$ of vertices of $\H$, or a weak feasible solution $E^*$, $|E^*|\leq O((\opt')^{5}\poly(\dmax\cdot\log n))$ for problem $\pi(\H,\emptyset,\zset)$. (Here, $n=|V(\G)|$).
\end{theorem}

We prove this theorem in the rest of this section, but we first show how Theorems~\ref{thm:main}, \ref{theorem: main-crossing-number} and Corollary~\ref{corollary: main-approx-crossing-number} follow from it.
We start with proving Theorem~\ref{thm:main}, by showing an efficient randomized algorithm to find a subset $E^*\sse E(\G)$ of edges, such that $\G\setminus E^*$ is planar, and $|E^*|\leq O((\optcro{G})^5\cdot\poly(\dmax\cdot \log n))$.
We assume that we know the value $\optcro{\G}$, by using the standard practice of guessing this value, running the algorithm, and then adjusting the guessed value accordingly. It is enough to ensure that whenever the guessed value $\opt\geq \optcro{\G}$, the algorithm indeed returns a subset $E^*$ of edges, $|E^*|\leq O(\opt^{5}\poly(\dmax\cdot \log n))$, such that $\G\setminus E^*$ is a planar graph w.h.p. Therefore, from now on we assume that we are given a value $\opt\geq \optcro{\G}$. The algorithm consists of a number of stages. The input to stage $j$ is a contracted graph $\H$, with the corresponding family $\zset$ of vertex sets. In the input to the first stage, $\H=\G$, and $\zset=\emptyset$. In each stage $j$, we run the algorithm from Theorem~\ref{thm: alg summary} on the current contracted graph $\H$, and the family $\zset$ of vertex subsets.
From Theorem~\ref{thm: graph contraction}, there is a strong feasible solution to problem $\pi(\H,\emptyset,\zset)$ of cost $O(\opt \cdot\poly(\log n\cdot \dmax))$, and so we can set the parameter $\opt'$ to this value.
 Whenever the algorithm returns a nasty canonical set $R$ in graph $\H$, we terminate the current stage, and compute a new contracted graph $\H'$, guaranteed by Theorem~\ref{thm: nasty canonical set to contraction}. Graph $\H'$, together with the corresponding family $\zset'$ of vertex subsets, becomes the input to the next stage. Alternatively, if, after $\poly(n)$ executions of the algorithm from Theorem~\ref{thm: alg summary}, no nasty canonical set is returned, then with high probability, one of the algorithm executions has returned a weak feasible solution $E^*$, $|E^*|\leq O(\opt^{5}\poly(\dmax\cdot\log n))$ for problem $\pi(\H,\emptyset,\zset)$. From Claim~\ref{claim: enough to solve contracted graph}, we can recover from this solution a planarizing set $E^{**}$ of edges for graph $\G$, with $|E^{**}|=O(\opt^{5}\poly(\dmax\cdot\log n))$. Since the size of the contracted graph $\H$ goes down after each contraction step, the number of stages is bounded by $n$, thus implying Theorem~\ref{thm:main}. Combining Theorem~\ref{thm:main} with Theorem~\ref{thm:CMS10} immediately gives Theorem~\ref{theorem: main-crossing-number}. Finally, we obtain Corollary~\ref{corollary: main-approx-crossing-number} as follows. Recall that the algorithm of Even et al.~ \cite{EvenGS02} computes a drawing of any $n$-vertex bounded degree graph $\G$ with $O(\log^2 n) \cdot (n+\optcro{G})$ crossings. It was shown in~\cite{CMS10}, that this algorithm can be extended to arbitrary graphs, where the number of crossings becomes $O(\poly(\dmax)\cdot \log^2 n) \cdot (n+\optcro{G})$. We run their algorithm, and the algorithm presented in this section, on graph $\G$, and output the better of the two solutions. If $\optcro{\G}<n^{1/10}$, then our algorithm is an $O(n^{9/10}\poly(\dmax\cdot \log n))$-approximation; otherwise, the algorithm of~\cite{EvenGS02} gives an  $O(n^{9/10}\poly(\dmax\cdot \log n))$-approximation.

 The remainder of this section is devoted to proving Theorem~\ref{thm: alg summary}.
Recall that we are given the contracted graph $\H$, and a collection $\zset$ of vertex-disjoint subsets of $V(\H)$. For each $Z\in \zset$, $\H[Z]$ is a grid, and $E(Z,V(\H)\setminus Z)$ consists of a set $M_Z$ of matching edges. Each such edge connects a vertex in the first row of $Z$ to a distinct vertex in $T_{\H}(Z)$, and these edges form a matching between the first row of $Z$ and $T_{\H}(Z)$. Abusing the notation, we denote the bound on the cost of the strong optimal solution to $\pi(\H,\emptyset,\zset)$ by $\opt$ from now on, and the number of vertices in $\H$ by $n$. 
For each $Z\in \zset$, we use $Z$ to denote both the set of vertices itself, and the grid $\H[Z]$.
We assume throughout the rest of the section that $\opt\cdot \dmax^6< \sqrt n$: otherwise, if
$\opt\cdot \dmax^6\geq \sqrt n$, then the set $E'$ of all edges of $\H$ that do not participate in grids $Z\in \zset$, is a feasible weak canonical solution for problem $\pi(\H,\emptyset,\zset)$. It is easy to see that $|E'|\leq O(\opt^2\poly(\dmax))$: this is clearly the case if $|E'|\leq 4n$; otherwise, if $|E'|>4n$, then by Theorem~\ref{thm: large average degree large crossing number}, $\opt=\Omega(n)$, and so $|E'|=O(n^2)=O(\opt^2)$.

We use two parameters: $\rho=O(\opt\poly(\dmax\cdot \log n))$ and $m^*=O(\opt^3\cdot \poly(\dmax \cdot \log n))$, whose exact values we set later.
The algorithm consists of $2\rho\log n$ iterations. The input to iteration $h$ is a collection $G_1,\ldots,G_{k_h}$ of $k_h\leq \opt$ sub-graphs of $\H$, together with bounding boxes $X_i\sse G_i$ for all $1\leq i\leq k_h$. We denote $H_i=G_i\setminus V(X_i)$ and $n(H_i)=|V(H_i)|$. Additionally, we have collections $\edges1,\ldots,\edges{h-1}$ of edges of $\H$, where for each $1\leq h'\leq h-1$, set $\edges{h'}$ has been computed in iteration $h'$. 
We say that $(G_1,X_1),\ldots,(G_{k_h},X_{k_h})$, and $\edges1,\ldots,\edges{h-1}$  is a \emph{valid} input to iteration $h$, iff the following invariants hold:

\label{start invariants for the whole algorithm-------------------------------------------------------------}

\begin{properties}{V}
\item For all $1\leq i,j\leq k_h$, graphs $H_i$ and $H_j$ are completely disjoint. \label{invariant 1: disjointness}

\item For all $1\leq i\leq k_h$, $G_i\sse \H\setminus (\edges 1,\ldots,\edges{h-1})$, and $H_i$ is the sub-graph of $\H$ induced by $V(H_i)$. In particular, no edges $e\sse V(H_i)$ belong to $\edges 1,\ldots,\edges{h-1}$. Moreover, every edge $e\in E(\H)$ belongs to either $\bigcup_{h'=1}^h\edges{h'}$ or to $\bigcup_{i=1}^{k_h}G_i$.
\label{invariant 2: proper subgraph}

\item For all $Z\in \zset$, for all $1\leq i\leq k_h$, either $Z\cap V(H_i)=\emptyset$, or $Z\sse V(H_i)$. Let $\zset_i=\set{Z\in \zset\mid Z\sse V(H_i)}$.\label{invariant 2: canonical}

\item For all $1\leq i\leq k_h$, there is a strong solution $\phi_i$ to $\pi(G_i,X_i,\zset_i)$, with $\sum_{i=1}^{k_h}\cro_{\phi_i}(G_i)\leq \opt$. \label{invariant 3: there is a cheap solution}

\item If we are given {\bf any} weak solution  $E_i'$ to problem $\pi(G_i,X_i,\zset_i)$, for all $1\leq i\leq k_h$, and denote $\tilde{E}^{(h)}=\bigcup_{i=1}^{k_h}E_i$, then $\edges1\cup\cdots\edges {h-1}\cup \tilde{E}^{(h)}$ is a feasible weak solution to problem $\pi(\H,\emptyset,\zset)$.\label{invariant 4: any weak solution is enough}

\item For each $1\leq h'<h$, and $1\leq i\leq k_h$, the number of edges in $\edges{h'}$ incident on vertices of $H_i$ is at most $m^*$, and $|\edges{h'}|\leq \opt \cdot m^*$. Moreover, no edges in grids $Z\in \zset$ belong to $\bigcup_{h'=1}^{h-1}\edges {h'}$.\label{invariant 4.5: number of edges removed so far}

\item Let $n_h=(1-1/\rho)^{(h-1)/2}\cdot n$. For each $1\leq i\leq k_h$, either  $n(H_i)\leq n_h$, or $X_i=\emptyset$ and $n(H_i)\leq n_{h-1}$. \label{invariant 5: bound on size}
\end{properties}
\label{end invariants for the whole algorithm-------------------------------------------------------------}

The input to the first iteration consists of a single graph, $G_1=\H$, with the bounding box $X_1=\emptyset$. It is easy to see that all invariants hold for this input. We end the algorithm at iteration $h^*$, where $n_{h^*}\leq (m^*\cdot \rho\cdot \log n)^2$. Clearly, $h^*\leq 2\rho\log n$, from Invariant~(\ref{invariant 5: bound on size}). Let $\gset$ be the set of all instances that serve as input to iteration $h^*$. We need the following theorem, whose proof appears in Appendix.

\begin{theorem}\label{thm: stopping condition}
There is an efficient algorithm, that, given any problem $\pi(G,X,\zset')$, where $V(G\setminus X)$ is canonical for $\zset'$, and $\pi(G,X,\zset')$ has a strong solution of cost $\overline{\opt}$, finds a weak feasible solution to $\pi(G,X,\zset')$ of cost $O(\overline{\opt}\cdot \sqrt{n'}\cdot \poly(\dmax\cdot \log n')+\overline{\opt}^3)$, where $n'=|V(G\setminus X)|$, and $\dmax$ is the maximum degree in $G$.
\end{theorem}

For each $1\leq i\leq k_{h^*}$, let $\edges{h^*}_i$ be the weak solution from Theorem~\ref{thm: stopping condition}, and let $\edges{h^*}=\bigcup_{i=1}^{k_{h^*}}\edges{h^*}_i$. Let $\opt_i$ denote the cost of the strong optimal solution to $\pi(G_i,X_i,\zset_i)$. Then $|\edges{h^*}|=\sum_{i=1}^{k_{h^*}}O(\opt_i\cdot \sqrt{n(H_i)}\cdot\poly(\dmax\cdot \log n)+\opt_i^3)$. Since $n(H_i)\leq n_{h^*-1}\leq 2n_{h^*}$ for all $i$, this is bounded by $\sum_{i=1}^{k_{h^*}}O(\opt_i\cdot m^*\cdot \rho\cdot \poly(\dmax\log n)+\opt_i^3)\leq O(\opt\cdot m^*\cdot \rho\cdot \poly(\dmax \log n)+\opt^3)$, as $\sum_{i=1}^{k_{h^*}}\opt_i\leq \opt$ from Invariant~(\ref{invariant 3: there is a cheap solution}). The final solution is $E^*=\bigcup_{h=1}^{h^*}\edges{h}$, and 

\[\begin{split}
|E^*|&\leq \sum_{h=1}^{h^*-1}|\edges{h}|+|\edges{h^*}|\\
&\leq (2\rho\log n)(\opt\cdot m^*)+O(\opt\cdot m^*\cdot \rho\cdot \poly(\dmax \cdot \log n)+\opt^3)\\
&=O(\opt^5\poly(\dmax\cdot \log n)).
\end{split}\]

We say that the execution of iteration $h$ is \emph{successful}, iff it either produces a valid input to the next iteration, together with the set $\edges{h}$ of edges, or finds a nasty canonical set in $\H$. We show how to execute each iteration, so that it is successful with probability at least $(1-1/\rho)$, if all previous iterations were successful. If any iteration returns a nasty canonical set, then we stop the algorithm and return this vertex set as an output. Since there are at most $2\rho\log n$ iterations, the probability that all iterations are successful is at least $(1-1/\rho)^{2\rho\log n}\geq 1/\poly(n)$. 
In order to complete the proof of Theorem~\ref{thm: alg summary}, it is now enough to show an algorithm for executing each iteration, such that, given a valid input to the current iteration, the algorithm either finds a nasty canonical set in $\H$, or returns a valid input to the next iteration, with probability at least $\frac 1 \rho$.  We do so in the next section.

\section{Iteration Execution}\label{sec: iteration}
\label{------------------------------------------------iteration execution-----------------------------------------------------------------------}

Throughout this section, we denote $n=|V(\H)|$, $\phi$ is the optimal canonical solution for the \MCN problem on $\H$, and $\opt$ is its cost.
We start by setting the values of the parameters $\rho$ and $m^*$. The value of the parameter $\rho$ depends on two other parameters, that we define later. Specifically, we will define two functions $\lambda: \mathbb{N}\rightarrow \reals$, $N:\mathbb{N}\rightarrow \reals$:

\[\lambda(n')=\Omega\left(\frac 1{\log n'\cdot \dmax^2}\right )\] 

and

\[N(n')=O(\dmax\sqrt{n'\log n'})\]

for all $n'>0$. Also, recall that $\alpha^*=\Omega\left(\frac 1 {\log^{3/2}n\cdot \log\log n}\right )$ is the well-linkedness parameter from Theorem~\ref{thm: well-linked-general}. We need the value of $\rho$ to satisfy the following two inequalities:

\begin{equation}\label{eq: value of rho 1}
\forall 0<n'\leq n\quad \quad \rho>\frac{25\cdot 2^{24}\dmax^6\cdot N^2(n')}{n'\cdot \lambda^2(n')\cdot (\alpha^*)^2}
\end{equation}

\begin{equation}\label{eq: value of rho 2}
\forall 0<n'\leq n\quad \quad  \rho >\frac{9\opt}{\lambda(n')}
\end{equation}

Substituting the values of $N(n'),\lambda(n')$ and $\alpha^*$ in the above inequalities, we get that it is sufficient to set:

\[\rho=\Theta(\log n\cdot \dmax^2)\max\set{\dmax^{10}\log^5 n(\log\log n)^2,\opt}=O\left (\opt\cdot \poly(\dmax\log n)\right ).\]

The value of parameter $m^*$ is:

\[m^*=O\left (\frac{\opt^2\cdot\rho\cdot\log^2n\cdot \dmax^2\cdot \betaFCG }{\alpha^*}\right )=O\left (\opt^3\cdot \poly(\dmax\cdot \log n\right ))\]

We now turn to describe each iteration $h$. Our goal is to either find a nasty canonical subset of vertices in $\H$, or produce a feasible input to the next iteration, $h+1$.
 Throughout the execution of iteration $h$, we construct a set $\gset_{h+1}$ of new problem instances, for which Invariants~(\ref{invariant 1: disjointness})--(\ref{invariant 5: bound on size}) hold. We do not need to worry about the number of the instances in $\gset_{h+1}$ being bounded by $\opt$, since, from Invariant~(\ref{invariant 3: there is a cheap solution}), the number of instances in $\gset_{h+1}$, which do not have a solution of cost $0$, is bounded by $\opt$. Since we can efficiently identify such instances, they will then become the input to the next iteration. 
We will also gradually construct the set $\edges h$ of edges, that we remove from the problem instance in this iteration. The iteration is executed on each one of the graphs $G_i$ separately. We fix one such graph $G_i$, for $1\leq i\leq k_h$, and focus on executing iteration $h$ on $G_i$.
We need a few definitions.

\begin{definition}
Given any graph $H$, we say that a simple path $P\sse H$ is a $2$-path, iff the degrees of all inner vertices of $P$ are $2$. We say that it is a maximal $2$-path iff it is not contained in any other $2$-path.
\end{definition}

\begin{definition} We say that a connected graph $H$ is \emph{rigid} iff either $H$ is a simple cycle, or, after we replace every maximal $2$-path in $H$ with an edge, we obtain a $3$-vertex connected graph, with no self-loops or parallel edges.
\end{definition}

Observe that if $H$ is rigid, then it has a unique planar drawing. We now define the notion of a valid skeleton.

\begin{definition}
Assume that we are given an instance $\pi=\pi(G,X,\zset')$ of the problem, and let $\phi'$ be the optimal strong solution for this instance. Given a subset $\tilde{E}$ of edges of $G$, and a sub-graph $K\sse G$, we say that $K$ is a \emph{valid skeleton} for $\pi,\tilde{E},\phi'$, iff the following conditions hold:

\begin{itemize}
\item Graph $K$ is rigid, and the edges of $K$ do not participate in crossings in $\phi'$. Moreover, the set $V(K)$ of vertices is canonical for $\zset'$.

\item $X\sse K$, and no edges of $\tilde{E}$ belong to $K$.

\item Every connected component of $G\setminus (K\cup \tilde E)$ contains at most $n_{h+1}$ vertices.
\end{itemize}
\end{definition}

Notice that if $K$ is a valid skeleton, then we can efficiently find the drawing $\phi_K'$ induced by $\phi'$ -- this is the unique planar drawing of $K$. Each connected component $C$ of $G\setminus (K\cup \tilde E)$ must then be embedded entirely inside some face $F_C$ of $\phi'$. Once we determine the face $F_C$ for each such component $C$, we can solve the problem recursively on these components, where for each component $C$, the bounding box becomes the boundary of $F_C$. This is the main idea of our algorithm. In fact, we will be able to find a valid skeleton $K_i$ for each instance $\pi(G_i,X_i,\zset_i)$ and drawing $\phi_i$, for $1\leq i\leq k_h$, w.h.p., but we cannot ensure that this skeleton will contain the bounding box $X_i$. If there is a large collection of edge-disjoint paths, connecting $K_i$ to $X_i$ in $G_i$, we can still connect $X_i$ to $K_i$, by choosing a small subset of these paths at random. This will give the desired final valid skeleton that contains $X_i$. However, if there is only a small number of such paths, then we cannot find a single valid skeleton that contains $X_i$ (in particular, it is possible that all edges incident on $X_i$ participate in crossings in $\phi_i$, so such a skeleton does not exist). However, in the second case, we can find a small subset $E'_i$ of edges, whose removal disconnects $X_i$ from many vertices of $G_i$. In particular, after we remove $E'_i$ from $G_i$, graph $G_i$ will decompose into two connected components: one containing $X_i$, and at most $n_{h+1}$ other vertices, and another that does not contain $X_i$. The first component is denoted by $G_i^X$, and the second by $G_i'$. The sub-instance defined by $G_i'$ is now completely disconnected from the rest of the graph, and it has no bounding box, so we can add it directly to $\gset_{h+1}$. For the sub-instance $G_i^X$, we show that $X_i$ is a valid skeleton. The edges in $E_i'$ are then added to $\edges h$. We now define these notions more formally.

Recall that for each $i: 1\leq i\leq k_h$, problem $\pi(G_i,X_i,\zset_i)$ is guaranteed to have a strong feasible solution $\phi_i$ of cost at most $\opt_i$. For each such instance, we will find two subsets of edges $E'_i$, and $E''_i$, where $|E'_i|=O(\opt^2\cdot \rho\cdot \dmax)$, and $|E''_i|=O\left (\frac{\opt^2\cdot\rho\cdot \log^2n\cdot\dmax^2\cdot \betaFCG}{\alpha^*}\right)$, that will be added to $\edges h$.
 
Assume first that $X_i\neq \emptyset$. So by Invariant~(\ref{invariant 5: bound on size}), $|V(G_i\setminus X_i)|\leq n_h$.
The graph $G_i\setminus E'_i$ consists of two connected sub-graphs: $G_i^X$, that contains the bounding box $X_i$, and the remaining graph $G'_i$. We will find a subset $E''_i$ of edges and a skeleton $K_i$ for graph $G_i^X$, such that w.h.p.,  $K_i$ is a valid skeleton for the instance $\pi(G_i^X,X_i,\zset_i)$, the set $E''_i$ of edges, and the solution $\phi_i$. Therefore, each one of the connected components of $G_i^X\setminus (K_i\cup E''_i)$ contains at most $n_{h+1}$ vertices. We will process these components, to ensure that we can solve them independently, and then add them to set $\gset_{h+1}$, where they will serve as input to the next iteration. The remaining graph, $G'_i$, contains at most $n_h$ vertices from Invariant~(\ref{invariant 5: bound on size}), and has no bounding box. So we can add $\pi(G_i,\emptyset,\zset_i)$ to $\gset_{h+1}$ directly.

If $X_i=\emptyset$, then we will ensure that $E'_i=\emptyset$, $G'_i=\emptyset$ and $G_i^X=G_i$. Recall that in this case, from Invariant~(\ref{invariant 5: bound on size}), $|V(G_i)|\leq n_{h-1}$. We will find a valid skeleton $K_i$ for $\pi(G_i,X_i,\zset_i),E''_i,\phi_i$, and then process the connected components of $G_i\setminus (K_i\cup E''_i)$ as in the previous case, before adding them to set $\gset_{h+1}$.

The algorithm consists of three steps. Given a graph $G_i\in \set{G_1,\ldots,G_{k_h}}$ with the bounding box $X_i$, the goal of the first step is to either produce a nasty canonical vertex set in the whole contracted graph $\H$, or to find a $\rho$-balanced $\alpha^*$-well-linked partition $(A,B)$ of $V(G_i)$, where $A$ and $B$ are canonical, and $|E(A,B)|$ is small. The goal of the second step is to find the sets $E'_i,E''_i$ of edges and a valid skeleton $K_i$ for instance $\pi(G^X_i,X_i,\zset_i)$. In the third step, we produce a new collection of instances, from the connected components of graphs $G_i\setminus (E''_i\cup K_i)$, which, together with the graphs $G_i'$, for $1\leq i\leq k_h$, are then added to $\gset_{h+1}$, to become the input to the next iteration.

\subsection{Step 1: Partition}\label{sec: step 1}

Throughout this step, we fix some graph $G\in \set{G_1,\ldots,G_{k_h}}$. We denote by $X$ its bounding box, and let $H^0=G\setminus V(X)$. Notice that graph $H^0$ is not necessarily connected. We denote by $H$ the largest connected component of $H^0$, and by $\hset$ the set of the remaining connected components. We focus on $H$ only in the current step. Let $n'=|V(H)|$. If $n'\leq (m^*\cdot \rho\cdot \log n)^2$, then we can simply proceed to the third step, as the size of every connected component of $H^0$ is bounded by $n'\leq n_{h^*}\leq n_{h+1}$. We then define $E'=E''=\emptyset$, $G^X=G$, $G'=\emptyset$, and we use $X$ as the skeleton $K$ for $G$. It is easy to see that it is a valid skeleton.  Therefore, we assume from now on that:

\begin{equation}\label{eq: upper bound on rho in terms of n'}
n'\geq (m^*\cdot \rho\cdot \log n)^2
\end{equation}

 Recall that from Invariant~(\ref{invariant 2: canonical}), $H$ is canonical w.r.t. $\zset$, so we define $\zset'=\set{Z\in \zset: Z\sse H}$.  Throughout this step, whenever we say that a set is canonical, we mean that it is canonical w.r.t. $\zset'$. 
 
Recall that the goal of the current step is to produce a partition $(A,B)$ of the vertices of $H$, such that $A$ and $B$ are both canonical, the partition is $\rho$-balanced and $\alpha^*$-well-linked, and $|E(A,B)|$ is small, or to find a nasty canonical vertex set in $\H$. In fact we will define 4 different cases. The first two cases are the easy cases, for which it is easy to find a suitable skeleton, even though we do not obtain a $\rho$-balanced $\alpha^*$-well-linked bi-partition. The third case will give the desired bi-partition $(A,B)$, and the fourth case will produce a partition with slightly different, but still sufficient properties. We then show that if none of these four cases happen, then we can find a nasty canonical set in $\H$.

The first case is when there is some grid $Z\in \zset'$ with $|Z|\geq  n'/2$. If this case happens, we continue directly to the second step (this is the simple case where eventually the skeleton will be simply $Z$ itself, after we connect it to the bounding box).
In the rest of this step we assume that for each $Z\in \zset'$, $|Z|< n'/2$.
The initial partition is summarized in the next theorem, whose proof appears in Appendix.

\begin{theorem}\label{thm: initial partition}
Assume that for each $Z\in \zset'$, $|Z|< n'/2$. Then we can efficiently find a partition $(A,B)$ of $V(H)$, such that:
\begin{itemize}
\item Both $A$ and $B$ are canonical.

\item $|A|, |B|\geq \lambda n'$, for $\lambda=\Omega\left(\frac{1}{\log n'\cdot \dmax^2}\right )$ and
 $|E(A,B)|\leq O(\dmax\sqrt {n'\log n'})$.

\item Set $A$ is $\alpha^*$-well-linked.
\end{itemize}
\end{theorem}

We say that Case 2 happens iff $|E(A,B)|\leq \frac{10^7\opt^2\cdot \rho\cdot\log^2n\cdot \dmax^2\cdot \betaFCG }{\alpha^*}$. If Case 2 happens, we continue directly to Step 2 (this is also a simple case, in which the eventual skeleton is the bounding box $X$ itself, and $E''=E(A,B)$).

Let $N=\Theta(\dmax\sqrt{n'\log n'})$, so that $|E(A,B)|\leq N$. Notice that set $B$ has property (P1) in $H$, since set $A$ is connected.
Our next step is to use Theorem~\ref{thm: well-linked-general} to produce an $\alpha^*$-well-linked decomposition $\cset$ of $B$, where each set of $C\in \cset$ has property (P1) and is canonical w.r.t. $\zset'$, with $\sum_{C\in \cset}|\out_H(C)|\leq 2N$. It is easy to see that the decomposition will give a slightly stronger property than (P1): namely, for each $C\in \cset$, for every edge $e\in \out_H(C)$, there is a path $P\sse H\setminus C$, connecting $e$ to some vertex of $A$. We will use this property later.

We are now ready to define the third case. This case happens if there is some set $C\in \cset$, with $|C|\geq n'/\rho$. So if Case 3 happens, we have found two disjoint sets $A,C$ of vertices of $H$, with $|A|,|C|\geq n'/\rho$, both sets being canonical w.r.t. $\zset'$	 and $\alpha^*$-well-linked. In the next lemma, whose proof appears in Appendix, we show that we can expand this partition to the whole graph $H$. 

\begin{lemma}\label{lemma: decomposition for Case 2}
If Case 3 happens, then we can efficiently find a partition $(A',B')$ of $V(H)$, such that $|A'|,|B'|\geq n'/\rho$, both sets are canonical w.r.t. $\zset'$, and $\alpha^*$-well-linked w.r.t. $\out_H(A'),\out_H(B')$, respectively.
\end{lemma}

If Case 3 happens, we continue directly to the second step. We assume  that Case 3 does not happen from now on.

Notice that the above decomposition is done in the graph $H$, that is, the sets $C\in \cset$ are well-linked w.r.t. $\out_H(C)$, and $\sum_{C\in \cset}|\out_H(C)|\leq 2N$. Property (P1) is also only ensured for $T_H(C)$, and not necessarily for $T_G(C)$. For each $C\in \cset$, let $\out^X(C)=\out_G(C)\setminus \out_H(C)$, that is, $\out^X(C)$ contains all edges connecting $C$ to the bounding box $X$. We do not have any bound on the size of $\out^X(C)$, and $C$ is not guaranteed to be well-linked w.r.t. these edges. The purpose of the final partitioning step is to take care of this.
This step is only performed if $X\neq \emptyset$.

We perform the final partitioning step on each cluster $C\in \cset$ separately. We start by setting up an $s$-$t$ min-cut/max-flow instance, as follows. We construct a graph $\tilde{C}$, by starting with $H[C]\cup \out_G(C)$, and identifying all vertices in $T_H(C)$ into a source $s$, and all vertices in $T_G(C)\setminus T_H(C)$ into a sink $t$. Let $F$ be the maximum $s$-$t$ flow in $\tilde{C}$, and let $(\tilde{C}_1,\tilde{C}_2)$ be the corresponding minimum $s$-$t$ cut, with $s\in \tilde{C}_1,t\in \tilde{C}_2$. From Corollary~\ref{corollary: canonical s-t cut}, both $\tilde C_1$ and $\tilde C_2$ are canonical.
We let $C_1$ be the set of vertices of $\tilde{C}_1$, excluding $s$, and $C_2$ is the set of vertices of $\tilde{C}_2$, excluding $t$. Notice that both $C_1$ and $C_2$ are also canonical. We say that $C_1$ is a cluster of type $1$, and $C_2$ is cluster of type $2$. Recall that we have computed a max-flow $F$ connecting $s$ to $t$ in $\tilde{C}$. Since all capacities are integral, and all capacities of edges in $H[C]$ are unit, $F$ consists of a collection $\pset$ of edge-disjoint paths in the graph $H[C]\cup\out_G(C)$. Each such path $P$ connects an edge in $\out_H(C)$ to an edge in $\out^X(C)$. Path $P$ consists of two consecutive segments: one is completely contained in $C_1$, and the other is completely contained in $C_2$. If the first segment is non-empty, then it defines a path $P_1\sse H[C_1]\cup \out_G(C_1)$, connecting an edge in $\out_H(C)$, to an edge in $E(\tilde C_1,\tilde C_2)$. Similarly, if the second segment is non-empty, then it defines a path $P_2\sse H[C_2]\cup \out_G(C_2)$, connecting an edge in $E(\tilde C_1,\tilde  C_2)$ to an edge in $\out^X(C)$. Every edge in $E(C_1,C_2)$ participates in one such path $P_1\sse H[C_1]\cup\out_G(C_1)$, and one such path $P_2\sse H[C_2]\cup\out_G(C_2)$. Similarly, if $e\in \out^X(C)\cap \out_G(C_1)$, then it is also an endpoint of exactly one path $P_1\sse H[C_1]\cup\out_G(C_1)$, and if $e\in \out_G(C_2)\setminus \out^X(C)$, then it is an endpoint of exactly one such path $P_2\sse H[C_2]\cup\out_G(C_2)$.

\begin{figure}[h]
\scalebox{0.5}{\rotatebox{0}{\includegraphics{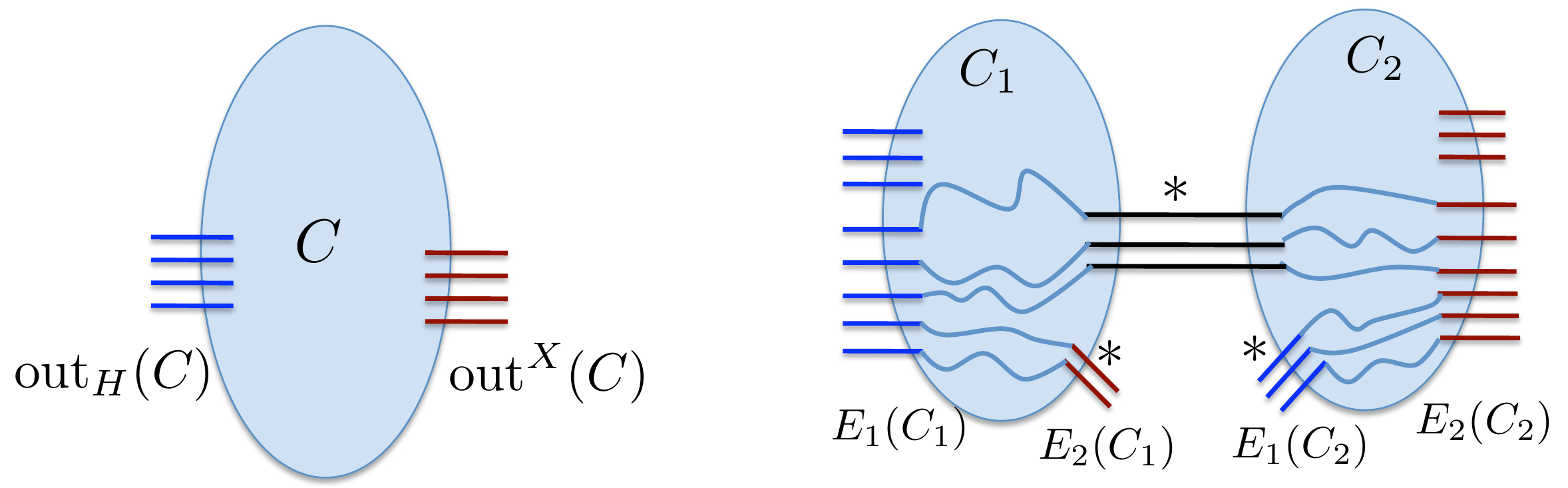}}} \caption{Partition of cluster $C$. Edges of $\out_H(C)$ are blue, edges of $\out^X(C)$ are red; edges participating in the min-cut are marked by $*$. The black edges belong to both $E_2(C_1)$ and $E_1(C_2)$.} \label{fig: step 1 last partition}
\end{figure}

For the cluster $C_1$, let $E_1(C_1)=\out_H(C_1)\cap \out_H(C)$, and $E_2(C_1)=\out_G(C_1)\setminus \out_H(C)$. All edges in $E_2(C_1)$ belong to either $E(C_1,C_2)$ or $\out^X(C)$. By the above discussion, we have a collection $\pset(C_1)$ of edge disjoint paths in $H[C_1]\cup\out_G(C_1)$,  each path connecting an edge in $E_1(C_1)$ to an edge in $E_2(C_1)$, and every edge in $E_2(C_1)$ is an endpoint of a path in $\pset(C_1)$. 
An important property of cluster $C_1$ that we will use later is that if $C_1\neq \emptyset$, then $E_1(C_1)\neq \emptyset$. All edges in $E_1(C_1)$ can reach set $A$ in graph $H\setminus C_1$, and all edges in $E_2(C_1)$ can reach the set $V(X)$ of vertices in the graph $G\setminus C_1$. Moreover, if $E_2(C_1)\neq \emptyset$, then there is a path $P(C_1)$, connecting a vertex of $C_1$ to a vertex of $X$, such that $P(C_1)$ only contains vertices of $C_2$. In particular, it does not contain vertices of any other type-1 clusters.


Similarly, for the cluster $C_2$, let $E_2(C_2)=\out_G(C_2)\cap \out^X(C)$, and $E_1(C_2)=\out_G(C_2)\setminus \out^X(C_2)$. All edges in $E_1(C_2)$ belong to either $E(C_1,C_2)$, or to $\out_H(C)$. From the above discussion, we have a set $\pset(C_2)$ of edge-disjoint paths in $H[C_2]\cup\out_G(C_2)$, each such path connecting an edge in $E_1(C_2)$ to an edge in $E_2(C_2)$, and every edge in $E_1(C_2)$ is an endpoint of one such path. 

Let $\tset_1$ be the set of all non-empty clusters of type $1$, and $\tset_2$ the set of clusters of type $2$. For the case where $X=\emptyset$, all clusters $C\in \cset$ are type-$1$ clusters, and $\tset_2=\emptyset$.
We are now ready to define the fourth case. We say that Case 4 happens, iff clusters in $\tset_2$ contain at least $\lambda n'/2$ vertices altogether.
Notice that Case 4 can only happen if $X\neq \emptyset$. The proof of the next lemma appears in Appendix.

\begin{lemma}\label{lemma: decomposition for Case 3}
If Case 4 happens, then we can find a partition $(A',B')$ of $V(H)$, such that $|A'|,|B'|\geq n'/\rho$, both $A'$ and $B'$ are canonical, and $A'$ is $\alpha^*$-well-linked w.r.t. $E(A',B')$. Moreover, if we denote by $\out^X(B')=\out_G(B')\setminus E(A',B')$, then there is a collection $\pset$ of edge-disjoint paths in graph $H[B']\cup \out_G(B')$, connecting the edges in $E(A',B')$ to edges in $\out^X(B')$, such that each edge $e\in E(A',B')$ is an endpoint of exactly one such path.\end{lemma}

We will show below that for cases 1---4, we can successfully construct a skeleton and produce an input to the next iteration, with high probability. In the next theorem, whose proof appears in Appendix, we show that if none of these cases happen, then we can efficiently find a nasty canonical set.

\begin{theorem}\label{thm: case 4}
If none of the cases 1--4 happen, then we can efficiently find a nasty canonical set in the original contracted graph $\H$.
\end{theorem}

\subsection{Step 2: Skeleton Construction}\label{subsection: skeleton construction}
Let $(G,X)\in\set{(G_1,X_1),\ldots,(G_{k_h},X_{k_h})}$, let $\phi'$ be the strong solution to problem $\pi(G,X,\zset')$, guaranteed by Invariant~(\ref{invariant 3: there is a cheap solution}), and let $\opt'$ denote its cost. Recall that $H$ is the largest connected component in $G\setminus X$, and $\zset'=\set{Z\in \zset: Z\sse V(H)}$. We say that an edge $e\in E(G)$ is \emph{good} iff it does not participate in any crossings in $\phi'$. Recall that for each $Z\in \zset'$, all edges of $G[Z]$ are good.
In the second step we define the subsets $E',E''$ of edges, the two sub-graphs $G^X$ and $G'$ of $G$, and construct a valid skeleton $K$ for $\pi(G^X,X,\zset'), E''$ and $\phi'$, for Cases 1---4.
We define a set $T\sse E(G)$ of edges, that we refer to as ``terminals'' for the rest of this section, as follows. For Case 1, $T=\emptyset$. For Case 2, $T=E(A,B)$, where $(A,B)$ is the partition of $H$ from Theorem~\ref{thm: initial partition}. For Cases 3 and 4, $T=E(A',B')$, where $(A',B')$ are the partitions of $H$ given by Lemmas~\ref{lemma: decomposition for Case 2} and \ref{lemma: decomposition for Case 3}, respectively. For convenience, we rename $(A',B')$ as $(A,B)$ for these two cases. Since the partition $(A,B)$ of $H$ is canonical for cases 2--4, we are guaranteed that $T$ does not contain any edges of grids $Z\in \zset'$.

The easiest case is Case 2. The skeleton $K$ for this case is simply the bounding box $X$, and we set $E''=T$. Recall that $|T|\leq  \frac{10^7\opt^2\cdot \rho\cdot\log^2n\cdot \dmax^2\cdot \betaFCG }{\alpha^*}$ for this case. Since $|A|,|B|\geq n'/\rho$, it is easy to verify that $X$ is a valid skeleton for $G$, $\phi'$ and $E''$. In particular, $|A|,|B|\leq n'(1-\rho)\leq n_{h-1}(1-\rho)\leq n_{h+1}$. We set $E'=\emptyset$, $G^X=G$, and $G'=\emptyset$. From now on we focus on Cases 1, 3 and 4.

We first build an initial skeleton $K'$ of $G$, and a subset $E''$ of edges, such that $K'$ has all the required properties, except that it is possible that $X\not\sse K'$. Specifically, we will ensure that $K'$ only contains good edges, is rigid, and every connected component of $H\setminus (K'\cup E'')$ contains at most $n_{h+1}$ vertices. In the end, we will either connect $K'$ to $X$, or find a small subset $E'$ of edges, separating the two sets.

The initial skeleton $K'$ for Case 1 is simply the grid $Z\in\zset'$ with $|Z|\geq n'/2$, and we set $E''=\emptyset$. Observe that $K'$ is good, rigid, canonical, and every connected component of $H\setminus K'$ contains at most $n'/2\leq n_{h-1}/2\leq n_{h+1}$ vertices. The construction of the initial skeleton for Cases 3 and 4 is summarized in the next theorem, whose proof is deferred to the Appendix.

\begin{theorem}\label{thm: initial skeleton for Cases 3 and 4}
Assume that Cases 3 or 4 happen. Then we can efficiently construct a skeleton $K'\sse G$, such that with probability at least $\left (1-\frac 1{2\rho\cdot \opt}\right )$, $K'$ is good, rigid, and every connected component of $H\setminus K'$ contains at most $O\left ( \frac{\opt^2\cdot \rho\cdot \log^2 n\cdot \dmax^2 \cdot \betaFCG }{\alpha^*}\right )$ terminals.
\end{theorem}

Let $\cset$ be the set of all connected components of $H\setminus K'$. Observe that at most one of the components may contain more than $n'/2$ vertices. Let $C$ denote this component, and let $E''$ be the set of terminals contained in $C$, $E''=T\cap E(C)$. Let $\cset'$ be the set of all connected components of $C\setminus E''$. Then for each $C'\in \cset'$, $|V(C')|\leq n'(1-\rho)$ must hold: otherwise, $V(C')$ must contain vertices that belong to both $A$ and $B$, and so $E(C')$ must contain at least one terminal. Therefore, the size of every connected component of $H\setminus (K'\cup E'')$ is bounded by $n'(1-\rho)\leq n_{h-1}(1-\rho)\leq n_{h+1}$ from Invariant~(\ref{invariant 5: bound on size}). Recall that the terminals do not belong to the grids $Z\in \zset'$.

Observe that it is possible that $V(K')$ is not canonical.
Consider some grid $Z\in\zset$, such that $V(Z)\cap V(K')\neq \emptyset$. If $Z\cap K'$ is a simple path, then we will deal with such grids at the end of the third step. Let $\zset''(G)$  denote the set of all such grids. Assume now that $Z\cap K'$ is not a simple path. Since graph $K'$ is rigid, it must be the case that there are at least three matching edges from $\out_G(Z)$ that belong to $K'$. In this case, we can simply add the whole grid $Z$ to the skeleton $K'$, and still the new skeleton $K'$ remains good and rigid, and every connected component of $H\setminus (K'\cup E'')$ contains at most $n_{h+1}$ vertices. So from now on we assume that if $V(Z)\cap V(K')\neq \emptyset$ for some $Z\in \zset$, then $Z\cap K'$ is a simple path, and so $Z\in \zset''(G)$. We denote by $K^+$ the union of $K'$ with all the grids in $\zset''(G)$. Clearly, $K^+$ is connected, canonical, but it is not necessarily rigid.

Consider Cases 1, 3 and 4. If $X=\emptyset$, then we define $E'=\emptyset$, $G^X=G$, $G'=\emptyset$ and the final skeleton $K=K'$. It is easy to see that $K$ is a valid skeleton for $\pi(G^X,X,\zset'\setminus\zset''(G))$, $E''$ and $\phi'$.

Otherwise, if $X\neq \emptyset$, we now try to connect the skeleton $K'$ to the bounding box $X$ (observe that some of the vertices of $X$ may already belong to $K'$). In order to do so, we will try to find a set $\pset'$ of $24\opt^2\rho$ vertex-disjoint paths in $G\setminus E''$, connecting the vertices of $X$ to the vertices of $K^+$ (where some of these paths can be simply vertices in $X\cap K^+$). We distinguish between three cases.

The first case is when such a collection of paths does not exist in $G\setminus E''$. Then there must be a set $V'\sse V(G)$ of at most $24\opt^2\rho$ vertices, whose removal from $G\setminus E''$ separates $X$ from $K^+$. Therefore, the size of the edge min-cut separating $X$ from $K^+\setminus X$ in $G\setminus E''$ is at most $24\opt^2\rho\dmax$. 
Observe that both $K^+$ and $X$ are canonical w.r.t. $\zset'$, and the vertices in $V(X)\cap V(K^+)$ cannot belong to sets $Z\in \zset'$, by the definition of $\zset'$. 
Therefore, from Corollary~\ref{corollary: canonical s-t cut}, there is a subset $E'$ of at most $24\opt^2\rho\dmax$ edges (canonical edge min-cut), whose removal partitions graph $G\setminus E''$ into two connected sub-graphs, $G^X$  containing $X$, and $G'=G\setminus V(G^X)$, and moreover, $V(G_X)$ and $V(G')$ are both canonical, and the edges of $E'$ do not belong to any grids $Z\in \zset'$. We add the instance $\pi(G',\emptyset,\zset')$ directly to $\gset_{h+1}$. From Invariant~(\ref{invariant 5: bound on size}), since $X\neq \emptyset$, $|V(G')|\leq n_h$, and since the bounding box of the new instance is $\emptyset$, it is a valid input to the next iteration. For graph $G^X$, we use $X$ as its skeleton. Observe that every connected component of $G^X\setminus (X\cup E'')$ must be either a sub-graph of some connected component of $H\setminus (K'\cup E'')$  (and then its size is bounded by $n_{h+1}$), or it must belong to $\hset^0$ (and then its size is bounded by $n_{h-1}/2\leq n_{h+1}$). Therefore, $X$ is a valid skeleton for $\pi(G^X,X,\zset'\setminus \zset''(G))$, $E''$, and $\phi'$.

The second case is when there is some grid $Z\in \zset''(G)$, such that for any collection $\pset'$ of $24\opt^2\rho$ vertex-disjoint paths, connecting the vertices of $X$ to the vertices of $K^+$ in $G$, at least half the paths contain vertices of $\Gamma(Z)$ as their endpoints. Recall that only $2$ edges of $\out_H(Z)$ belong to $K'$. Then there is a collection $E'$ of at most $12\dmax\opt^2\rho+2$ edges in $G\setminus E''$, whose removal separates $V(X)\cup Z$ from $V(K^+)\setminus (Z\cup X)$. Again, we can ensure that the edges of $E'$ do not belong to the grids $Z\in \zset'$. Let $G^X$ denote the resulting subgraph that contains $X$, and $G'=G\setminus G^X$. Then both $G^X$ and $G'$ are canonical as before, and we can add the instance  $\pi(G',\emptyset,\zset')$ to $\gset_{h+1}$, as before. In order to build a valid skeleton for graph $G^X$, we consider the subset $\pset''\sse \pset'$ of $12\opt^2\rho$ vertex-disjoint paths, connecting the vertices of $X$ to the vertices of $\Gamma(Z)$, and we randomly choose three such paths. We then let the skeleton $K$ of $G^X$ consist of the union of $X$, $Z$, and the three selected paths. It is easy to see that the resulting graph $K$ is rigid, and with probability at least $(1-\frac 1{2\rho\cdot \opt})$, it only contains good edges. Moreover, every connected component of $G^X\setminus (K\cup E'')$ is either a sub-graph of a connected component of $H\setminus (K'\cup E'')$ (and  may contain at most $n_{h+1}$ vertices), or it belongs to $\hset^0$ (and then its size is bounded by $n_{h+1}$). Therefore, $K$ is a valid skeleton for $\pi(G^X,X,\zset'\setminus \zset''(G))$, $E''$, and $\phi'$.

The third case is when we can find the desired collection $\pset'$ of paths, and moreover, for each grid $Z\in \zset''(G)$, at most half the paths in $\pset'$ contain vertices of $\Gamma(Z)$. We then randomly select three paths from $\pset'$, making sure that at most two paths containing vertices of $\Gamma(Z)$ are selected for any grid $Z\in \zset''(G)$. Since at most $2\opt$ of the paths in $\pset'$ are bad, with probability at least $1-1/(2\opt\rho)$, none of the selected paths is bad. We then define $K$ to be the union of $K'$, $X$, and the three selected paths. 
Additionally, if, for some grid $Z\in \zset''(G)$, one or two of the selected paths contain vertices in $\Gamma(Z)$, then remove $Z$ from $\zset''(G)$, and add it to $K$. It is easy to verify that the resulting skeleton is rigid, and it only contains good edges. Moreover, every connected component of $G\setminus (K\cup E'')$, is either a sub-graph of a connected component of $H\setminus (K'\cup E'')$, or it is a sub-graph of one of the graphs in $\hset^0$. In the former case, its size is bounded by $n_{h+1}$ as above, while in the latter case, its size is bounded by $|V(G\setminus X)|/2\leq n_{h-1}/2<n_{h-1}(1-\rho)\leq n_{h+1}$. We set $E'=\emptyset$, $G^X=G$, and $G'=\emptyset$.

To summarize this step, we have started with the instance $\pi(G,X,\zset')$, and defined two subsets $E',E''$ of edges, with $|E'|\leq O(\opt^2\dmax\rho)$ and $|E''|\leq O\left ( \frac{\opt^2\cdot \rho\cdot \log^2 n\cdot \dmax^2 \cdot \betaFCG }{\alpha^*}\right )$, whose removal disconnects $G$ into two connected sub-graphs: $G^X$ containing $X$, and $G'$. Moreover, both sets $V(G^X)$, $V(G')$ are canonical, and $E',E''$ do not contain edges belonging to grids $Z\in\zset'$. We have added instance $\pi(G',\emptyset,\zset')$ to $\gset_{h+1}$, and we have defined a skeleton $K$ for $G^X$. We have shown that $K$ is a valid skeleton for $\pi(G^X,X,\zset'\setminus \zset''(G))$, $E''$, and $\phi'$. The probability that this step is successful for a fixed graph $G\in\set{G_1,\ldots,G_{k_h}}$ is at least $(1-1/(\rho\cdot  \opt))$, and so the probability that it is successful across all graphs is at least $(1-1/\rho)$.

We can assume w.l.o.g. that every edge in set $E'$ has one endpoint in $G^X$ and one endpoint in $G'$: otherwise, this edge does not separate $G^X$ from $G'$, and can be removed from $E'$. Similarly, we can assume w.l.o.g. that for every edge $e\in E''$, the two endpoints of $e$ either belong to distinct connected components of $G^X\setminus (K\cup E'')$, or one endpoint belongs to $G^X$, and the other to $G'$. We will use these facts later, to claim that Invariant~(\ref{invariant 2: proper subgraph}) holds for the resulting instances.

\subsection{Step 3: Producing Input to the Next Iteration}\label{sec: step 3}
Recall that so far, for each $1\leq i\leq k_h$, we have found two collections $E_i',E_i''$ of edges, two sub-graphs $G_i^X$ and $G_i'$ with $X_i\sse G_i^X$, and a valid skeleton $K_i$ for $\pi(G_i^X,X_i,\zset\setminus \zset''(G_i))$, $\phi_i$, $E''_i$. The sets $E_i'\cup E_i''$ do not contain any edges of the grids $Z\in \zset$, and each edge in $E_i'\cup E''_i$ either connects a vertex of $G_i^X$ to a vertex of $G_i'$, or vertices of two distinct connected components of $G_i^X\setminus (K_i\cup E''_i)$. Recall that $G_i'$ contains at most $n_{h}$ vertices, and there are no edges in $G_i\setminus (E_i'\cup E_{i}'')$ connecting the vertices of $G_i'$ to those of $G_i^X$. Let $\cset'$ denote the set of all connected components of $G_i^X\setminus (K_i\cup E_i'')$. Then for each $C\in \cset'$, $|V(C)|\leq n_{h+1}$. 

Since graph $K_i$ is rigid, we can find the planar drawing $\phi_i(K_i)$ of $K_i$ induced by $\phi_i$ efficiently. Since all edges of $K_i$ are good for $\phi_i$, each connected component $C\in\cset'$ is embedded inside a single face $F_C^*$ of $\phi_i$. Intuitively, we would like to find this face $F_C^*$ for each such connected component $C$, and then solve the problem recursively on $C$, together with the bounding box $\gamma(F_C^*)$ --- the boundary of the face $F_C^*$. Apart from the difficulty in identifying the face $F_C^*$, a problem with this approach is that it is not clear that we can solve the problems induced by different connected components separately. For example, if both $C$ and $C'$ need to be embedded inside the same face $F$, then even if we find weak solutions for problems $\pi(C,\gamma(F),\zset)$ and $\pi(C',\gamma(F),\zset')$, it is not clear that these two solutions can be combined together to give a feasible weak solution for the whole problem, since the drawings of $C\cup \gamma(F)$ and $C'\cup \gamma(F)$ may interfere with each other. We will define below the condition under which the two clusters are considered independent and can be solved separately. We will then find an assignment of each cluster $C$ to one of the faces of $\phi_i(K_i)$, and find a further partition of each cluster $C\in \cset'$, such that all resulting clusters assigned to the same face are independent, and their corresponding problems can therefore be solved separately.

We now focus on some graph $G=G_i^X\in \set{G_1^X,\ldots,G_{k_h}^X}$, and we denote its bounding box by $X$, its skeleton $K_i$ by $K$, and the two sets $E_i',E_i''$ of edges by $E'$ and $E''$ respectively. We let $\phi'$ denote the drawing of $G_i^X$ induced by the drawing $\phi_i$, guaranteed by Invariant~(\ref{invariant 3: there is a cheap solution}).
As before, $\cset'$ is the set of all connected components of $G\setminus (K\cup E'')$.

While further partitioning the clusters $C\in\cset'$ to ensure independence, we may have to remove edges that connect the vertices of $C$ to the skeleton $K$. However, such edges do not strictly belong to the cluster $C$. We next perform a simple transformation of the graph $G\setminus (E'\cup E'')$ in order to take care of this technicality.

Consider the graph $G\setminus (E'\cup E'')$. We perform the following transformation: let $e=(v,x)$ be any edge in $E(G)\setminus (E'\cup E'')$, such that $x\in K$, $v\not\in K$. We add an artificial vertex $z_e$, that subdivides $e$ into two edges: an artificial edge $(x,z_e)$, and a non-artificial edge $(v,z_e)$. We denote $x_{z_e}=x$. Similarly, if $e=(x,x')$ is any edge in $E(G)\setminus (E'\cup E'')$, with $x,x'\in K$, then we add two artificial vertices $z_e,z'_e$, that subdivide $e$ into three edges, artificial edges $(x,z_e)$, and $(z'_e,x')$, and a non-artificial edge $(z_e,z'_e)$. We denote $x_{z_e}=x$, and $x_{z_{e}'}=x'$. If edge $e$ belonged to any grids $Z\in \zset$ (which can happen if $Z\in \zset''(G)$), then we consider all edges obtained from sub-divviding $e$ also a part of $Z$. Let $\tilde{G}$ denote the resulting graph, $\Gamma$ the set of all these artificial vertices, and let $E_{\tilde{G}}(\Gamma,K)$ be the set of all artificial edges in $\tilde G$. Let $\tilde{\phi}$ be the drawing of $\tilde{G}$ induced by $\phi'$. Notice that we can assume w.l.o.g. that the edges of $E_{\tilde G}(\Gamma,K)$ do not participate in any crossings in $\tilde{\phi}$. We use this assumption throughout the current section. 

For any sub-graph $C$ of $\tilde G\setminus K$, we denote by $\Gamma(C)=\Gamma\cap V(C)$, and $\out_K(C)$ is the subset of artificial edges adjacent to the vertices of $C$, that is, $\out_K(C)=E_{\tilde G}(\Gamma_C,K)$. We also denote by $C^+=C\cup \out_K(C)$, and by $\delta(C)$ the set of endpoints of the edges in $\out_K(C)$ that belong to $K$. 
Let $\cset$ the set of all connected components of $\tilde G\setminus K$. We next formally define the notion of independence of clusters. Eventually, we will find a further partition of each one of the clusters $C\in \cset$, so that the resulting clusters are independent, and can be solved separately in the next iteration. 

Let $\phi_K'$ be the drawing of $K$ induced by $\phi'$. Recall that this is the unique planar drawing of $K$, that can be found efficiently. Let $\fset$ be the set of faces of $\phi_K'$. For each face $F\in \fset$, let $\gamma(F)$ denote the set of edges and vertices lying on its boundary. Since $K$ is rigid, $\gamma(F)$ is a simple cycle. Since all edges of $K$ are good for $\phi'$, for every component $C\in \cset$, $C^+$ is embedded completely inside some face $F^*_C$ of $\fset$ in the drawing $\tilde {\phi}$, and so $\delta(C)\sse \gamma(F)$ must hold. Therefore, there are three possibilities: either there is a unique face $F_C\in \fset$, such that $\delta(C)\sse \gamma(F_C)$. In this case we say that $C$ is of type 1, and  $F_C=F^*_C$ must hold; or there are two faces $F_1(C),F_2(C)$, whose both boundaries contain $\delta(C)$, so  $F^*_C\in\set{F_1(C),F_2(C)}$. In this case we say that $C$ is of type 2. The third possibility is that $|\delta(C)|\leq 1$. In this case we say that $C$ is of type 3, and we can embed $C$ inside any face whose boundary contains the vertex $\delta(C)$. The embedding of such clusters does not affect other clusters. For convenience, when $C$ is of type 1, we denote $F_1(C)=F_2(C)=F_C$, and if it is of type 3, then we denote $F_1(C)=F_2(C)=F$, where $F$ is any face of $\fset$ whose boundary contains $\delta(C)$.

We now formally define when two clusters $C,C'\in \cset$ are independent. Let $C,C'\in \cset$ be any two clusters, such that there is a face $F\in \fset$, with $\delta(C),\delta(C')\sse \gamma(F)$. The set $\delta(C)$ of vertices defines a partition $\Sigma$ of $\gamma(F)$ into segments, where every segment $\sigma\in \Sigma$ contains two vertices of $\delta(C)$ as its endpoints, and does not contain any other vertices of $\delta(C)$. Similarly, the set $\delta(C')$ of vertices defines a partition $\Sigma'$ of $\gamma(F)$.

\begin{definition}
We say that the two clusters $C,C'$ are \emph{independent}, iff $\delta(C)$ is completely contained in some segment $\sigma'\in \Sigma'$. Notice that in this case, $\delta(C')$ must also be completely contained in some segment $\sigma\in \Sigma$.
\end{definition}

Our goal in this step is to assign to each cluster $C\in \cset$, a face $F(C)\in\set{F_1(C),F_2(C)}$, and to find a partition $\qset(C)$ of the vertices of the cluster $C$. Intuitively, each such cluster $Q\in\qset(C)$ will become an instance in the input to the next iteration, with $\gamma(F(C))$ as its bounding box.
Suppose we are given such an assignment $F(C)$ of faces, and the partition $\qset(C)$ for each $C\in \cset$. We will use the following notation.
For each $C\in \cset$,
let $E^*(C)$ denote the set of edges cut by $\qset(C)$, that is, $E^*(C)=\bigcup_{Q\neq Q'\in \qset(C)}E_{\tilde{G}}(Q,Q')$, and let $E^*=\bigcup_{C\in \cset}E^*(C)$. For each $Q\in \qset(C)$, we denote by $X_Q=\gamma(F(C))$, the boundary of the face inside which $C$ is to be embedded. For each face $F\in \fset$, we denote by $\qset(F)=\bigcup_{C:F(C)=F}\qset(C)$ the set of all clusters to be embedded inside $F$, and we denote by $\qset=\bigcup_{C\in \cset}\qset(C)$. 
Abusing the notation, for each cluster $Q\in \qset$, we will refer to $Q$ both as the set of vertices, and as the sub-graph $\tilde G[Q]$ induced by it.
As before, we denote  $Q\cup \out_K(Q)$ by $Q^+$. 
The next theorem shows that it is enough to find an assignment of every cluster $C\in \cset$ to a face $F(C)\in \set{F_1(C),F_2(C)}$, and a partition $\qset(C)$ of the vertices of $C$, such that all the resulting clusters assigned to every face of $\fset$ are independent.

\begin{theorem}\label{thm: no conflict case}
Suppose we are given, for each cluster $C\in \cset$, a face $F(C)\in \set{F_1(C),F_2(C)}$, and a partition $\qset(C)$ of the vertices of $C$.
Moreover, assume that for every face $F\in \fset$, every pair $Q,Q'\in \qset(F)$ of clusters is independent, and for each $Z\in \zset$, $E^*\cap E(Z)=\emptyset$. 
Then:

\begin{itemize}
\item For each $Q\in \qset$, there is a strong solution to the problem $\pi(Q^+\cup X_Q,X_Q,\zset)$, such that the total cost of these solutions, over all $Q\in \qset$, is bounded by $\cro_{\tilde \phi}(\tilde G)\leq \cro_{\phi'}(G)$.

\item For each $Q\in \qset$, let $E^{**}_Q$ be any feasible weak solution to the problem $\pi(Q^+\cup X_Q,X_Q,\zset)$, and let $E^{**}=\bigcup_{Q\in \qset}E^{**}_Q$. Then $E'\cup E''\cup E^*\cup E^{**}$ is a feasible weak solution to problem $\pi(G,X,\zset)$.
\end{itemize}
\end{theorem}

We remark that this theorem does not require that the sets $C\in \cset$ are canonical vertex sets.

\begin{proof}
Fix some $Q\in \qset$, and let $\tilde{\phi}_{Q^+}$ be the drawing of $Q^+\cup X_Q$ induced by $\tilde{\phi}$. Recall that the edges of the skeleton $K$ do not participate in any crossings in $\tilde \phi$, and every pair $Q,Q'\in \qset$ of graphs is completely disjoint. Therefore, $\sum_{Q\in \qset}\cro_{\tilde \phi_{Q^+}}(Q^+)\leq \cro_{\tilde \phi}(\tilde G)$.  Observe that every edge of $\tilde G$ belongs either to $K$, or to $E^*$, or to $Q^+$ for some $Q\in \qset$. Therefore, it is now enough to show that for each $Q\in \qset$, $\tilde \phi_{Q^+}$ is a feasible strong solution to problem $\pi(Q^+\cup X_Q,X_Q,\zset)$. Since $\phi'$ is canonical, so is $\tilde \phi_{Q^+}$. It now only remains to show that $Q^+$ is completely embedded on one side (that is, inside or outside) of the cycle $X_Q$ in $\tilde \phi_{Q^+}$. Let $C\in \cset$, such that $Q\in \qset(C)$. Recall that $C$ is a connected component of $\tilde G\setminus K$. Since $K$ is good, $C$ is embedded completely inside one face in $\fset$. In particular, since $X_Q$ is the boundary of one of the faces in $\fset$, all vertices and edges of $C$ (and therefore of $Q$) are completely embedded on one side of $X_Q$. Therefore, $X_Q$ can be viewed as the bounding box in the embedding $\tilde \phi_{Q^+}$.

We now prove the second part of the theorem. For each $Q\in \qset$, let $E^{**}_Q$ be any feasible weak solution to the problem $\pi(Q^+\cup X_Q,X_Q,\zset)$, and let $E^{**}=\bigcup_{Q\in \qset}E^{**}_Q$. We first show that $E'\cup E''\cup E^*\cup E^{**}$ is a feasible weak solution to the problem $\pi(\tilde G,X,\zset)$. 

Let $F\in \fset$ be any face of $\phi'_K$. For each $Q\in \qset(F)$, let $\tilde Q=Q\setminus E^{**}_Q$, and let $\tilde Q^+=Q^+\setminus E^{**}_Q$. Since $E^{**}_Q$ is a weak solution for instance $\pi(Q^+\cup X_Q,X_Q,\zset)$, there is a planar drawing $\psi_{Q}$ of $\tilde{Q}^+\cup X_Q$, inside the bounding box $X_Q=\gamma(F)$. It is enough to show that for each face $F\in \fset$, we can find a planar embedding of graphs $\tilde{Q}^+$, for all $Q\in\qset(F)$ inside $\gamma(F)$.

Fix an arbitrary ordering $\qset(F)=\set{Q_1,\ldots,Q_r}$. We now gradually construct a planar drawing of the graphs $\tilde{Q}_j^+$ inside $\gamma(F)$. For convenience, we will also be adding new artificial edges to this drawing. We perform $r$ iterations, and the goal in iteration $j: 1\leq j\leq r$ is to add the graph $\tilde{Q}_j^+$ to the drawing. We will maintain the following invariant: at the beginning of every iteration $j$, for each $j'\geq j$, there is a face $F'$ in the current drawing, such that $\delta(Q_j)\sse \gamma(F')$.

In the first iteration, we simply use the drawing $\psi_{Q_1}$ of $\tilde{Q}_1^+\cup \gamma(F)$. The vertices of $\delta(Q_1)$ define a partition $\Sigma_1$ of $\gamma(F)$ into segments, such that every segment contains two vertices of $\delta(Q_1)$ as its endpoints, and no other vertices of $\delta(Q_1)$. For each such segment $\sigma$, we add a new artificial edge $e_{\sigma}$ connecting its endpoints to the drawing. All such edges can be added without creating any crossings.
 Since every pair of clusters in $\qset(F)$ is independent, for each graph $Q_j$, $j> 1$, the vertices of $\delta(Q_j)$ are completely contained in one of the resulting segments $\sigma\in \Sigma_1$. The face $F'$ of the current drawing, whose boundary consists of $\sigma$ and $e_{\sigma}$ then has the property that $\delta(Q_j)\sse \gamma(F')$.
 
 Consider now some iteration $j+1$, and let $F'$ be the face of the current drawing, such that $\delta(Q_{j+1})\sse \gamma(F')$. We add the drawing $\psi_{Q_{j+1}}$ of $\tilde Q_{j+1}^+\cup \gamma(F)$, with $\gamma(F')$ replacing $\gamma(F)$ as the bounding box. We can do so since $\delta(Q_j)\sse \gamma(F')$. We can therefore add this drawing, so that no crossings with edges that already belong to the drawing are introduced. The bounding box $\gamma(F')$ is then sub-divided into the set $\Sigma'$ of sub-segments, by the vertices of $\delta(Q_j)$. Again, for each such segment $\sigma'$, we add an artificial edge $e_{\sigma'}$, connecting its endpoints, to the drawing, inside the face $F'$, such that no crossings are introduced. Since there are no conflicts between clusters in $\qset(F)$, for each $Q_{j'}$, with $j'>j+1$, such that $\delta(Q_{j'})\sse \gamma(F')$, there is a segment $\sigma'\in \Sigma'$, containing all vertices of $\delta(Q_{j'})$. The corresponding new face $F''$, formed by $\sigma'$ and the edge $e_{\sigma'}$ will then have the property that $\delta(Q_{j'})\sse \gamma(F'')$.
 
 We have thus shown that $\tilde{G}\setminus (E^*\cup E^{**})$ has a planar drawing. The same drawing induces a planar drawing for $G\setminus (E'\cup E''\cup E^*\cup E^{**})$.
\end{proof}

In the rest of this section, we will show an efficient algorithm to find the assignment of the faces of $\fset$ to the clusters $C\in \cset$, and the partition $\qset(C)$ of each such cluster, satisfying the requirements of Theorem~\ref{thm: no conflict case}. Our goal is also to ensure that $|E^*|$ is small, as these edges are eventually removed from the graph. If two clusters $C,C'\in \cset$, with $\delta(C),\delta(C')\sse \gamma(F)$ for some $F\in \fset$ are not independent, then we say that they have a conflict. The process of partitioning both clusters into sub-clusters to ensure that the sub-clusters are independent is called conflict resolution. The next theorem shows how to perform conflict resolution for a pair of clusters. The proof of this theorem is due to Yury Makarychev~\cite{Yura}. We provide it here for completeness.

\begin{theorem}\label{thm: conflict resolution for 2 clusters}
Let $C,C'\in \cset$, such that both $C$ and $C'$ are embedded inside the same face $F\in \fset$ in $\tilde{\phi}$. Then
we can efficiently find a subset $E_{C,C'}\sse E(C)$ of edges, $|E_{C,C'}|\leq 30 \cro_{\tilde\phi}(E(C),E(C'))$, such that if $\cset'$ denotes the collection of all connected components of $C\setminus E_{C,C'}$, then for every cluster $Q\in \cset'$, $Q$ and $C'$ are independent. Moreover, $E_{C,C'}$ does not contain any edges of the grids $Z\in \zset$.
\end{theorem}

\begin{proof}

We say that a set $\tilde{E}$ of edges is valid iff it satisfies the condition of the theorem. For simplicity, we will assign weights $w_e$ to edges as follows: edges that belong to grids $Z\in \zset$ have infinite weight, and all other edges have weight $1$.
We first claim that there is a valid set of weight at most $\cro_{\tilde{\phi}}(C, C')$.
Indeed, let $\tilde{E}$ be the set of edges of $C$, that are crossed by the edges of $C'$ in $\tilde \phi$. Clearly, 
$|\tilde E| \leq \cro_{\tilde{\phi}}(C,C')$, and this set does not contain any edges in grids $Z\in \zset$, or edges adjacent to the vertices of $K$ (this was our assumption when we defined $\tilde\phi$). Let $\cset'$ be the set of all connected components of $C\setminus \tilde E$, and  consider some cluster $Q\in \cset'$. Assume for contradiction, that $Q$ and $C'$ are not independent. Then there are four vertices $a,b,c,d\in \gamma(F)$, whose ordering along $\gamma(F)$ is $(a,b,c,d)$, and $a,c\in \delta(Q)$, while $(b,d)\in \delta(C')$. But then there must be a path $P\sse Q\cup\out_K(Q)$ connecting $a$ to $c$, and a path $P'\sse C'\cup \out_K(C')$, connecting $b$ to $d$, as both $Q$ and $C'$ are connected graphs. Moreover, since $Q$ and $C'$ are completely disjoint, the two paths must cross in $\tilde{\phi}$. Recall that we have assumed that the artificial edges adjacent to $K$ do not participate in any crossings in $\tilde{\phi}$. Therefore, the crossing is between an edge of $Q$ and an edge of $C'$. This is impossible, since we have removed all edges that participate in such crossings from $C$.

We now show how to  \textit{efficiently} find a valid set $\tilde E$ of edges, of weight at most
$30 \cro_{\tilde{\phi}}(E(C),E(C'))$. 
Let $\Sigma' = \{\sigma'_1,\sigma'_2,\dots, \sigma'_k\}$ be the set of segments of $\gamma(F)$, 
defined by $\delta(C')$, in the circular order. Throughout the rest of the proof we identify $k+1$ and $1$. 

Consider the set $\Gamma(C)$ of vertices. We partition this set into a number of subsets, as follows. 
For $1\leq i\leq k$, let $\Gamma_i\sse \Gamma(C)$ denote the subset of vertices $z\in \Gamma(C)$, for which $x_z$ lies strictly 
inside the segment $\sigma_i'$. Let $\Gamma_{i,i+1}\sse \Gamma(C)$ denote the subset of vertices $z\in \Gamma(C)$, for which $x_z$ is the vertex separating segments $\sigma'_i$ and $\sigma'_{i+1}$.

We now restate the problem of finding a valid cut $E_{C,C'}$ as an assignment problem.
We need to assign each vertex of $C$ to one of the segments $\sigma'_1, \dots, \sigma'_k$ so that
\begin{itemize}
\item every vertex in $\Gamma_i$ is assigned to the segment $\sigma'_i$;
\item every vertex in $\Gamma_{i,i+1}$ is assigned to either  $\sigma'_i$ or $\sigma'_{i+1}$. 
\end{itemize}

We say that an edge of $C$ is cut by such an assignment, iff its endpoints are assigned to different segments.
Given any such assignment, whose weight is finite, let $\tilde E$ be the set of cut edges. We prove that set $\tilde E$ is valid. Since the weight of $\tilde E$ is finite, it cannot contain edges of grids $Z\in \zset$.
Let $\cset'$ be the collection of all connected components of $C\setminus \tilde E$. It is easy to see that for each $Q\in \cset'$, $Q$ and $C'$ is independent. This is since for all edges in $\out_K(Q)$, their endpoints that belong to $K$ must all be contained inside a single segment $\sigma'$ of $\Sigma'$.

On the other hand, every finite-weight valid set $\tilde E$ of edges corresponds to a valid assignment. Let $\cset'$ be the set of all connected components of $C\setminus \tilde E$, and let $Q\in \cset'$. Since there are no conflicts between $Q$ and $C'$, all vertices of $\delta(Q)$ that serve as endpoints of the set $\out_K(Q)$ of edges, must be contained inside a single segment $\sigma'\in \Sigma'$. If the subset of $\delta(Q)$ contains a single vertex, there can be two such segments of $\Sigma'$, and we choose any one of them arbitrarily; if this subset of $\delta(Q)$ is empty, then we choose an arbitrary segment of $\Sigma'$. We then assign all vertices of $Q$ to $\sigma'$. Since $\tilde E$ does not contain any edges  that are adjacent to the vertices of $K$ (as such edges are not part of $E(C)$), we are guaranteed that every vertex in $\Gamma_i$ is assigned to the segment $\sigma'_i$, and every vertex in $\Gamma_{i,i+1}$ is assigned to either $\sigma'_i$ or $\sigma'_{i+1}$, for all $1\leq i\leq k$.

We now show how to approximately solve the assignment problem, and therefore the original problem,
using linear programming. We will ensure that the weight of the solution $E_{C,C'}$ is at most $30$ times the optimum, and so $|E_{C,C'}|\leq 30 \cro_{\tilde \phi}(E(C),E(C'))$.   

For each vertex $u$ of $C$
and segment $\sigma'_i$ we introduce an indicator variable $y_{u,i}$, for assigning $u$ to segment $\sigma'_i$. All variables for vertex $u$ form a vector 
$y_u = (y_{u,1}, \dots, y_{u,k}) \in {\mathbb R}^k$. 
We denote the standard basis of ${\mathbb R}^k$ by $e_1,\dots, e_k$.
In the intended integral solution, $y_u = e_i$ if $u$ is assigned to $\sigma'_i$; that is,
$y_{u,i} = 1$ and $y_{u,j} = 0$ for $j\neq i$. Equip the space ${\mathbb R}^k$ with the $\ell_1$ norm
$\|y_u\|_1 = \sum_{i=1}^k |y_{u,i}|$. We solve the following linear program.
\begin{align*}
\text{minimize } &&\frac{1}{2} \sum_{e=(u,v)\in E(C)} w_e\cdot  \|y_u - y_v\|_1\\
\text{subsject to }&& \\
&& \|y_u\|_1 = 1 &&& \forall u \in V(C);\\
&& y_{u,i} = 1 &&& \forall 1\leq i\leq k, \forall u \in \Gamma_i;\\
&& y_{u,i} + y_{u,i+1} = 1 &&& \forall 1\leq i\leq k,\forall u \in \Gamma_{i,i+1};\\
&& y_{u,i} \geq 0 &&& \forall u \in V(C), \forall 1\leq i\leq k.
\end{align*}

Let $\OPT_{LP}$ be the value of the optimal solution of the LP. For all $1\leq i\leq k$, $r\in (1/2,3/5)$, define balls $B_i^r = \{u: y_{u,i} \geq r\}$ and 
$B_{i,i+1}^r = \{u: u\not\in B_i^r\cup B_{i+1}^r; y_{u,i}+y_{u,i+1} \geq 5r/3\}$.
Note that since, for each $u\in V(C)$, at most one coordinate $y_{u,i}$ can be greater than $\half$, whenever $r\geq \half$, the balls $B_i^r$ and $B_j^r$ are
disjoint for all $i\neq j$. Similarly, balls $B_{i,i+1}^r$ and $B_{j,j+1}^{r}$ are disjoint for $i\neq j$ when $r \geq 1/2$: this is since, if $u\in B_{i,i+1}^r$, then $y_{u,i}+y_{u,i+1}\geq 5/6$ must hold, while $y_{u,i},y_{u,i+1}<\half$. Therefore, $y_{u,i},y_{u,i+1}>1/3$ must hold, and there could be at most two coordinates $1\leq j\leq k$, for which $y_{u,j}>1/3$.

For each value of $r: 1/2\leq r/\leq 3/5$, we let $E^r$ denote all edges that have exactly one endpoint in the balls $B_i^r$, and $B_{i,i+1}^r$, for all $1\leq i\leq k$. We choose $r\in (1/2,3/5)$ that minimizes $|E^r|$, and we let $E_{C,C'}$ denote the set $E^r$ for this value of $r$. 
 We assign all
vertices in balls $B_i^r$ and $B_{i,i+1}^r$ to the segment $\sigma'_i$. We assign all unassigned vertices
to an arbitrary segment. We need to verify that this assignment is valid; that is, vertices 
from $\Gamma_i$ are assigned to $\sigma'_i$ and vertices from $\Gamma_{i,i+1}$
are assigned to either $\sigma'_i$ or $\sigma'_{i+1}$, for all $1\leq i\leq k$. Indeed, if $u\in\Gamma_i$, then $y_{u,i} = 1$, and so $u\in B^r_i$; similarly, if $u\in\Gamma_{i,i+1}$ then $y_{u,i} + y_{u,i+1} = 1$, and so $u\in B^r_i\cup B^r_{i+1}\cup B^r_{i,i+1}$.

Finally, we need to show that the cost of the assignment is at most $30 \OPT_{LP}$. In fact, 
we show that if we choose $r\in (1/2,3/5)$ uniformly at random, then the expected cost is at most $30\opt_{LP}$.
Consider an edge $e=(u,v)$. We compute the probability that $e\in B^r_i$, for each $1\leq i\leq k$. This is the probability that $y_{u,i}\geq r$, but $y_{v,i}<r$ (or vice versa if $y_{v,i}<y_{u,i}$). This probability is bounded by $10|y_{u,i}-y_{v,i}|$. Similarly, the probability that $u\in B^r(i,i+1)$ but $v\not \in B^r(i,i+1)$ is bounded by the probability that $y_{u,i}+y_{u,i+1}\geq 5r/3$, but $y_{v,i}+y_{v,i+1}< 5r/3$, or vice versa. This probability is at most $6\cdot \frac 5 3 ((y_{u,i}+y_{u,i+1})-(y_{v,i}+y_{v,i+1}))\leq 10(|y_{u,i}-y_{v,i}|+|y_{u,i+1}-y_{v,i+1}|)$. Therefore, overall, the probability that $e=(u,v)$ belongs to the cut is at most:

\[\sum_{i=1}^k10 |y_{u,i}-y_{v,i}|+\sum_{i=1}^k 10(|y_{u,i}-y_{v,i}|+|y_{u,i+1}-y_{v,i+1}|)\leq 10 \norm{y_u-y_v}_1+20\norm{y_u-y_v}_1=30\norm{y_u-y_v}_1\]

\end{proof}

We now show how to find the assignment $F(C)$ of faces of $\fset$ to all clusters $C\in \cset$, together with the partition $\qset(C)$ of the vertices of $C$. We will reduce this problem to an instance of the min-uncut problem. Recall that the input to the min-uncut problem is a collection $X$ of Boolean variables, together with a collection $\Psi$ of constraints. Each constraint $\psi\in \Psi$ has non-negative weight $w_{\psi}$, and involves exactly two variables of $X$. All constraints $\psi\in \Psi$ are required to be of the form $x\neq y$, for $x,y\in X$.
The goal is to find an assignment to all variables of $X$, to minimize the total weight of unsatisfied constraints. Agarwal et. al.~\cite{ACMM} have shown an $O(\sqrt{\log n})$-approximation algorithm for Min Uncut.

Fix any pair $C,C'\in \cset$ of clusters, and a face $F\in \fset$, such that $\delta(C),\delta(C')\sse \gamma(F)$.  Let $E'_{C,C'}$ denote the union of the sets $E_{C,C'}$ and $E_{C',C}$ of edges from Theorem~\ref{thm: conflict resolution for 2 clusters}, and let $w_{C,C'}=|E'(C,C')|$, $w_{C,C'}\leq 60\cro_{\tilde{\phi}}(E(C),E(C'))$. 

For each face $F\in \fset$, we denote by $\cset(F)\sse \cset$ the set of all clusters $C\in \cset$, of the first type, for which $\delta(C)\sse \delta(F)$. Recall that for each such cluster, $F^*_C=F$ must hold.
Let $E'(F)=\bigcup_{C,C'\in \cset(F)}E'_{C,C'}$, and let $w_F=|E'(F)|$.

Let $\pset$ be the set of all maximal $2$-paths in $K$. For every path $P\in \pset$, we denote by $\cset(P)\sse \cset$ the set of all type-2 clusters $C$, for which $\delta(C)\sse P$. Let $F_1(P),F_2(P)$ be the two faces of $K$, whose boundaries contain $P$. Recall that for each $C\in \cset(P)$, $F^*_C\in \set{F_1(P),F_2(P)}$.

For every $C\in \cset(P)$, $F\in \set{F_1(C),F_2(C)}$, let $w_{C,F}=\sum_{C'\in \cset(F)}w_{C,C'}$. If we decide to assign $C$ to face $F_1(P)$, then we will pay $w_{C,F_1(P)}$ for this assignment, and similarly, if $C$ is assigned to face $F_2(P)$, we will pay  $w_{C,F_2(P)}$.

We now set up an instance of the min-uncut problem, as follows. The set of variables $X$ contains, for each path $P\in \pset$, for each $F\in \set{F_1(P),F_2(P)}$, a Boolean variable $y_{P,F}$, and for each path $P\in \pset$ and cluster $C\in \cset(P)$ a Boolean variable $y_C$. Intuitively, if $y_C=y_{P,F_1(P)}$, then $C$ is assigned to $F_1(P)$, and if $y_C=y_{P,F_2(P)}$, then $C$ is assigned to $F_2(P)$. The set $\Psi$ of constraints contains constraints of three types: first, for each path $P\in \pset$, we have the constraint $y_{P,F_1(P)}\neq y_{P,F_2(P)}$ of infinite weight. For each $P\in \pset$, for each pair $C,C'\in \cset(P)$ of clusters, there is a constraint $y_C\neq y_{C'}$, of weight $w_{C,C'}$. Finally, for each $P\in \pset$, $F\in \set{F_1(P),F_2(P)}$, and for each $C\in \cset(P)$, we have a constraint $y_C\neq y_{P,F}$ of weight $w_{C,F}$.

\begin{claim}
There is a solution to the min-uncut problem, whose cost, together with $\sum_{F\in \fset}w_F$, is bounded by $60\cro_{\tilde{\phi}}(G)$.
\end{claim}

\begin{proof}
We simply consider the optimal solution $\tilde{\phi}$. 
For each path $P\in \pset$, we assign $y_{P,F_1(P)}=0$ and $y_{P,F_2(P)}=1$. For each cluster $C\in \cset(P)$, if $F^*(C)=F_1(P)$, then we set $y_C=y_{P,F_1(P)}$, and otherwise we set $y_C=y_{P,F_2(P)}$. From Theorem~\ref{thm: conflict resolution for 2 clusters}, for every pair $C,C'$ of clusters with $F^*_C=F^*_{C'}$, $w_{C,C'}\leq 60\cro_{\tilde{\phi}}(C,C')$.
\end{proof}

We can therefore find an $O(\sqrt{\log n})$-approximate solution to the resulting instance of the min-uncut problem, using the algorithm of~\cite{ACMM}. This solution naturally defines an assignment of faces to clusters. Namely, if $C$ is a type-1 cluster, then we let $F(C)=F$, where $F$ is the unique face with $\delta(C)\sse \gamma(F)$. If $C$ is a type-2 cluster, and $C\in \cset(P)$, for some path $P\in \pset$, then we assign $C$ to $F_1(P)$ if $y_C=y_{P,F_1(P)}$, and we assign it to $F_2(C)$ otherwise. If $C$ is a type-3 cluster, then we assign it to any face that contains the unique vertex in $\delta(C)$.
 
For each face $F$, let $\cset'(F)$ denote all clusters $C$ that are assigned to $C$. Let $\tilde E(F)$ denote the union of the sets $E'_{C,C'}$ of edges for all $C,C'\in \cset'(F)$, and let $\tilde E=\bigcup_{F\in\fset}\tilde E(F)$.

For each cluster $C\in \cset$, we now obtain a partition $\qset'(C)$ of its vertices that corresponds to the connected components of graph $C\setminus \tilde E$. For each $Q\in \qset'(C)$, we let $Q$ denote both the set of vertices in the connected component of $C\setminus \tilde{E}$, and the sub-graph of $\tilde{G}$ induced by $Q$.
From Theorem~\ref{thm: conflict resolution for 2 clusters}, we are guaranteed that for every face $F\in \fset$, for all $C,C'\in \cset'_F$, if $Q\in \qset'(C)$ and $Q'\in \qset'(C')$, then $Q$ and $Q'$ are independent.

It is however possible that for some $C\in \cset$, there is a pair $Q,Q'\in \qset'(C)$ of clusters, such that there is a conflict between $Q$ and $Q'$. In order to avoid this, we perform the following grouping procedure: For each $F\in\fset$, for each $C\in \cset'_F$, while there is a pair $Q,Q'\in \qset(C)$ of clusters that are not independent, remove $Q,Q'$ from $\qset'(C)$, and replace them with $Q\cup Q'$.
For each $C\in \cset$, let $\qset(C)$ be the resulting partition of the vertices of $C$. Clearly, each pair $Q,Q'\in \qset(C)$ is independent.

\begin{claim}
For each $F\in \fset$, for each pair $C,C'\in \cset'(F)$ of clusters, and for each $Q\in \qset(C), Q'\in \qset(C')$, clusters $Q$ and $Q'$ are independent.\end{claim}

\begin{proof}
Consider the partitions $\qset'(C)$, $\qset'(C')$, as they change throughout the grouping procedure. Before we have started the grouping procedure, every pair $Q\in \qset'(C)$, $Q'\in \qset'(C')$ was independent. Consider the first step in the grouping procedure, such that this property held for $\qset'(C),\qset'(C')$ before this step, but does not hold anymore after this step. Assume w.l.o.g. that the grouping step was performed on a pair $Q_1,Q_2\in \qset'(C)$. Since no other clusters in $\qset'(C)$ or $\qset'(C')$ were changed, there must be a cluster $Q'\in \qset'(C')$, such that both pairs $Q_1,Q'$ and $Q_2,Q'$ are independent, but $Q_1\cup Q_2$ and $Q'$ are not independent. We now show that this is impossible.

Let $\Sigma$ be the partitioning of $\gamma(F)$ defined by the vertices of $\delta(Q')$. Since $Q_1$ and $Q'$ are independent, there is a segment $\sigma\in \Sigma$, such that $\delta(Q_1)\sse \sigma$. Similarly, since $Q_2$ and $Q'$ are independent, there is a segment $\sigma'\in \Sigma$, such that $\delta(Q_2)\sse \sigma'$. However, since $Q_1$ and $Q_2$ are not independent, $\sigma=\sigma'$ must hold. But then all vertices of $\delta(Q_1\cup Q_2)$ are contained in the segment $\sigma\in \Sigma$, contradicting the fact that $(Q_1\cup Q_2)$ and $Q'$ are not independent.
\end{proof}

To summarize, we have shown how to find an assignment $F(C)\in \set{F_1(C),F_2(C)}$ for every cluster $C\in \cset$, and a partition $\qset(C)$ of the vertices of every cluster $C$, such that for every face $F\in \fset$, every pair $Q,Q'\in \qset(F)$ of clusters is independent. Moreover, if $E^*$ denotes the subset of edges $E_{\tilde{G}}(Q,Q')$ for all $Q,Q'\in \qset$, then we have ensured that $|E^*|\leq O(\sqrt{\log n})\cro_{\tilde {\phi}}(\tilde{G})=O(\sqrt{\log n})\cro_{\phi'}(G)$, and set $E^*$ does not contain edges of grids $Z\in \zset$, or artificial edges. Therefore, the conditions of Theorem~\ref{thm: no conflict case} hold.

We now define the set $\edges h(G)$, as follows: $\edges h(G)=E'\cup E''\cup E^*$.
Recall that $|E'|\leq O(\opt^2\rho\dmax)$, $|E''|\leq O\left(\frac{\opt^2\cdot\rho\cdot\log^2n\cdot\dmax^2\cdot\betaFCG }{\alpha^*}\right )$, and $|E^*|\leq O(\sqrt{\log n})\cro_{\phi'}(G)$. Therefore, 

\[|\edges h(G)|\leq  O\left(\frac{\opt^2\cdot\rho\cdot\log^2n\cdot\dmax^2\cdot\betaFCG }{\alpha^*}\right )+O(\sqrt{\log n})\cro_{\phi'}(G)\leq m^*.\]

 We also set $\edges h=\bigcup_{G\in\set{G_1^X,\ldots,G_{k_h}^X}}\edges h(G)$, so $|\edges h|\leq m^*\cdot \opt$ as required.

We now define a collection $\gset_{h+1}$ of instances. Recall that for all $1\leq i\leq k_h$, this collection already contains the instance $\pi(G_i',\emptyset,\zset)$. Let $G=G_i^X$, and $Q\in \qset$.
Let $Q'$ denote the subset of vertices of $Q$ without the artificial vertices, and let $H_Q$ be the sub-graph of $\H$ induced by $Q\cup X_Q$. We then add the instance $\pi_G(H_Q,X_Q,\zset)$ to $\gset_{h+1}(G)$. 
This finishes the definition of the set $\gset'_{h+1}$. From Theorem~\ref{thm: no conflict case}, for each $1\leq i\leq k_h$, there is a strong solution to each resulting sub-instance of $G_i^X$, such that the total cost of these solutions is at most $\cro_{\phi_i}(G_i^X)$. Clearly, $\phi_i$ also induces a strong solution to instance $\pi(G_i',\emptyset,\zset)$ of cost $\cro_{\phi_i}(G_i')$. Therefore, there is a strong solution for each instance in $\gset_{h+1}$ of total cost at most $\sum_{i=1}^{k_h}\cro_{\phi_i}(G_i)\leq \opt$, and so the number of instances in $\gset_{h+1}$, for which $\emptyset$ is not a feasible weak solution is bounded by $\opt$. We let $\gset'_{h+1}\sse \gset_{h+1}$ denote the set of all instances for which $\emptyset$ is not a feasible solution. Observe that we can efficiently verify whether $\emptyset$ is a feasible solution for a given instance, so we can compute $\gset'_{h+1}$ efficiently. We now claim that $\gset'_{h+1}$ is a valid input to the next iteration, except that it may not satisfy Invariant~(\ref{invariant 2: canonical}) due to the grid sets $\zset''(G)$ -- we deal with this issue at the end of this section.

 We have already established Invariant~(\ref{invariant 3: there is a cheap solution}) in the above discussion.
 Also, from Theorem~\ref{thm: no conflict case}, if we find a weak feasible solution $\tilde{E}_H$ for each instance $H\in \gset_{h+1}$, then the union of these solutions, together with $\edges h$, gives a weak feasible solution to all instances $\pi(G_i,X_i,\zset_i)$ for $1\leq i\leq k_h$, thus giving Invariant~(\ref{invariant 4: any weak solution is enough}).
In order to establish Invariant~(\ref{invariant 4.5: number of edges removed so far}), observe that the number of edges in $\edges h$, incident on any new instance is bounded by the maximum number of edges in $\edges h$ that belong to any original instance $G_1,\ldots,G_{k_h}$, which is bounded by $m^*$, and the total number of edges in $\edges h$ is bounded by $m^*\cdot \opt$. Invariant~(\ref{invariant 5: bound on size}) follows from the fact that for each $1\leq i\leq k_h$, $|V(G_i')|\leq n_h$, and these graphs have empty bounding boxes. All sub-instances of $G_i^X$ were constructed by further partitioning the clusters in $G_i^X\setminus (K_i\cup E'')$, and each such cluster contains at most $n_{h+1}$ vertices. Invariant~(\ref{invariant 1: disjointness}) is immediate, as is Invariant~(\ref{invariant 2: proper subgraph}) (recall that we have ensured that all edges in $E',E'',E^*$ connect vertices in distinct sub-instances). Finally, if we assume that $\zset''(G)=\emptyset$ for all $G\in \set{G_1,\ldots,G_{k_h}}$, then the resulting sub-instances are canonical, as we have ensured that the edges in sets $E',E'',E^*$ do not belong to the grids $Z\in \zset$, thus giving Invariant~(\ref{invariant 2: canonical}). Therefore, we have shown how to produce a valid input to the next iteration, for the case where $\zset''(G)=\emptyset$ for all $G\in \set{G_1,\ldots,G_{k_h}}$. It now only remains to show how to deal with the grids in sets $\zset''(G)$.

\paragraph{Dealing with grids in sets $\zset''(G)$}
Let $G\in \set{G_1^X,\ldots,G_{k_h}^X}$, and let $Z\in \zset''(G)$. Recall that this means that $Z\cap K$ is a simple path, that we denote by $P_Z$, and in particular, $K$ contains exactly two edges of $\out(Z)$, that we denote by $e_Z$ and $e'_Z$.

The difficulty in dealing with such grids is that, on the one hand, we need to ensure that all new sub-instances are canonical, so we would like to add such grids to the skeleton $K$. On the other hand, since $Z\cap K$ is a simple path, graph $K\cup Z$ is not rigid, and has 2 different planar drawings (obtained by ``flipping'' $Z$ around the axis $P_Z$), so we cannot claim that we can efficiently find the optimal drawing $\phi'_{K\cup Z}$ of $K\cup Z$. Our idea in dealing with this problem is that we use the conflict resolution procedure to establish which face of the skeleton $K$ each such grid $Z\in\zset''(G)$ must be embedded in. Once this face is established, we can simply add $Z$ to $K$. Even though the resulting skeleton is not rigid, its drawing is now fixed.

More specifically, let $Z\in \zset''(G)$ be any such grid, and let $v,v'$ be the two vertices in the first row of $Z$ adjacent to the edges $e_z$ and $e'_z$, respectively. We start by replacing the path $P_Z$ in the skeleton $K$, with the unique path connecting $v$ and $v'$ that only uses the edges of the first row of $Z$. Let $P'_Z$ denote this path. We perform this transformation for each $Z\in \zset''(G)$. The resulting skeleton $K$ is still rigid and good. It is now possible that the size of some connected component of $G\setminus (K\cup E'')$ becomes larger. However, since we eventually add all vertices of all such grids $Z\in \zset''(G)$ to the skeleton, this will not affect the final outcome of the current iteration.

We then run the conflict resolution procedure exactly as before, and obtain the collection $\gset_{h+1}$ of new instances as before. Consider some such instance $\pi(H_Q,X_Q,\zset)$, and assume that $H_Q$ is a sub-graph of $G\in\set{G_1^X,\ldots,G^X_{k_h}}$. Let $Q=H_Q\setminus X_Q$. From the above discussion, $Q$ is canonical w.r.t. $\zset\setminus \zset''(G)$. The only problem is that for some grids $Z\in \zset''(G)$, $Q$ may contain the vertices of $Z\setminus P'_Z$. This can only happen if $P'_Z$ belongs to the bounding box $X_Q$. Recall that we are guaranteed that there is a strong solution to instance $\pi(H_Q,X_Q,\zset)$, and the total cost of all such solutions over all instances in $\gset_{h+1}$ is at most $\opt$. In particular, the edges of $Z$ do not participate in crossings in this solution. Therefore, we can simply consider the grid $Z$ to be part of the skeleton, remove its vertices from $Q$, and update the bounding box of the resulting instance if needed.
In other words, the conflict resolution procedure, by assigning every cluster $C\in \cset$ to a face of $\fset$, has implicitly defined a drawing of the graph $K\cup(\bigcup_{Z\in\zset''(G)}Z)$. Even though this drawing may be different from the drawing induced by $\phi'$, we are still guaranteed that the resulting sub-problems all have strong feasible solutions of total cost bounded by $\opt$. The final instances in $\gset'_{h+1}$ are now guaranteed to satisfy all Invariants~(\ref{invariant 1: disjointness})--(\ref{invariant 5: bound on size}).

\section{Conclusions} \label{sec: conclusions}
We have shown an efficient randomized algorithm to find a drawing of any graph $\G$ in the plane with at most $O\left ((\optcro{\G})^{10}\poly(\dmax\cdot \log n)\right )$ crossings. We did not make an effort to optimize the powers of $\opt,\dmax$ and $\log n$ in this guarantee, or the constant hidden in the $O(\cdot)$ notation, and we believe that they can be improved. We hope that the technical tools developed in this paper will help obtain better algorithms for the \MCN problem. A specific possible direction is obtaining efficient algorithms for $\rho$-balanced $\alpha$-well-linked bi-partitions. In particular, an interesting open question is whether there is an efficient algorithm, that, given an $n$-vertex graph $G$ with maximum degree $\dmax$, finds a $\rho$-balanced $\alpha$-well-linked bi-partition of $G$, for $\rho,\alpha=\poly(\dmax\cdot \log n)$. In fact it is not even clear whether such a bi-partition exists in every graph.
We note that the dependence of $\rho$ on $\dmax$ is necessary, for example, in the star graph. This question appears to be interesting in its own right, and its positive resolution would greatly simplify our algorithm and improve its performance guarantee. We also note that if we only require that one of the two sets in the bi-partition is well-linked, then there is an efficient algorithm for finding such bi-partitions, similarly to the proof of Theorem~\ref{thm: initial partition}.

\paragraph{Acknowledgements.} The author thanks Yury Makarychev and Anastasios Sidiropoulos for many fruitful discussions, and for reading earlier drafts of the paper.
\label{------------------------------------------------Start Appendix------------------------------------------------------------------------}

\bibliography{alg-v2}
\bibliographystyle{alpha}
\newpage

\appendix

\section{Parameter List}\label{sec: param-list}

\begin{tabular}{|l|l|l|}\hline
&&\\
$\alpha^*$&$\Omega(1/(\log^{3/2} n\cdot \log\log n))$&Well-linkedness parameter in the\\
&& well-linked decomposition, Theorems~\ref{thm: well-linked},\ref{thm: well-linked-general}\\ \hline
&&\\
&$|S|\geq \frac{2^{16}\cdot \dmax^6}{(\alpha^*)^2}\cdot|\Gamma(S)|^2$&Requirement for nasty set $S$\\ \hline
&&\\
$\betaFCG $&$O(\log n)$&Flow-cut gap for undirected graphs \\
&&\\ \hline
&&\\
$\alphasc$&$O(\sqrt{\log n}\log\log n)$&approximation factor of algorithm $\algsc$ 
\\&& for non-uniform sparsest cut.\\ \hline
 &&\\
 $\beta^*$&$\beta_{FCG}\cdot \dmax/\alpha^*=O(\log^{5/2}n\cdot \log\log n\cdot \dmax)$&Parameter for Property (P3) for graphs\\
 && $H^{(3)}(X)$ in Step 3 of Graph Contraction.\\ \hline
&&\\
 $N$& $O(\dmax\sqrt {n'\log n'})$&number of cut edges in the initial partition in \\
 && Step 1 of the algorithm\\ \hline
 &&\\
 $\eps$&constant&balance parameter guaranteed by ARV for\\ 
 &&balanced cut\\ \hline
&&\\
$\lambda$&$\frac{\eps^2n'}{25N^2}=\Omega\left(\frac{1}{\log n'\cdot \dmax^2}\right )$& balance parameter in the initial partition\\
&& in Step 1 of the algorithm\\ \hline
 &&\\
 $\rho$&$\Theta(\log n\cdot \dmax^2)\max\set{\dmax^{10}\log^5 n(\log\log n)^2,\opt}$&Balance parameter in final partition\\
 &&in Step 1 of the algorithm\\ \hline
 &&\\
 $\edges h$&$|\edges h|\leq \opt\cdot m^*$& edges removed from the graph\\
 &&in iteration $h$\\ \hline
 &&\\
 $m^*$&$O\left (\frac{\opt^2\cdot\rho\cdot\log^2n\cdot \dmax^2\cdot \betaFCG }{\alpha^*}\right )$& maximum number of edges in $\edges{h'}$\\
 &\quad$=O\left (\opt^3\cdot \poly(\dmax\cdot \log n\right ))$&incident on any graph $H_i$ in the input\\
 && to iteration $h$, for all $h'<h$.\\ \hline
 &&\\
 $n'$&$\geq (m^*\cdot \rho\cdot \log n)^2$&number of vertices in a sub-instance $H$\\
 &&considered in current iteration \\ \hline
 &&\\
  $E'_i$&$|E_i'|=O(\opt^2\cdot \rho\cdot \dmax)$& Set of edges separating $G_i^X$ from $G_i'$\\
  && in each iteration of the algorithm\\ \hline
  &&\\
  $E_i''$&$|E''_i|=O\left (\frac{\opt^2\cdot\rho\cdot \log^2n\cdot\dmax^2\cdot \betaFCG}{\alpha^*}\right)$& Edges of $G_i$ for which we find\\
  &&a valid skeleton\\ \hline
\end{tabular}

\section{Auxiliary Claim}
The next simple claim is used extensively throughout the paper.
\begin{claim}\label{claim: numbers}
Suppose we are given any collection $\rset$ of non-negative numbers, and each number $x\in \rset$ is associated with another number, $y_x>0$. Moreover, assume that the following three conditions hold:

\begin{enumerate}
\item For each $x\in \rset$, $x\leq \beta y_x^2$.

\item For each $x\in \rset$, $x\leq M$.

\item $\sum_{x\in \rset}y_x\leq S$
\end{enumerate}
for some parameters $\beta, M,S>0$. Then $\sum_{x\in \rset}x\leq 2S\sqrt{\beta M}+M/4$.
\end{claim}

Before providing the proof of this claim, we give some intuition as to why we need it. The general setting in which we will use this claim, is when we are given some collection $\xset$ of subsets of vertices of some graph $G$. For each set $X\in \xset$, we will have the number $x=|X|$, and $y_x=|\Gamma(X)|$. What the above claim essentially says is that if none of the sets $X\in \xset$ is nasty or too large, and if the total number of the interface vertices in all such sets $X$ is small, then the total number of vertices contained in sets $X\in \xset$ cannot be too large.

\begin{proof}
We will perform a number of transformations on the set $\rset$, until we obtain a set for which we can bound $\sum_{x\in\rset} x$ easily. After each such transformation, the above three conditions will continue to hold, and $\sum_{x\in \rset}x$ will not decrease. It will then be enough to bound $\sum_{x\in \rset}x$ in the final set $\rset$. We perform one of the following three steps, while possible:

\begin{itemize}
\item If there is a number $x\in \rset$ with $y_x>\sqrt{\frac M{\beta}}$: remove $x$ from $\rset$, and add two new numbers, $x'=x$, $x''=0$, with $y_{x'}=\sqrt{\frac M{\beta}}$, and $y_{x''}=y_x-y_{x'}$. Since $x\leq M$ and $x\leq \beta y_x^2$, it is easy to see that all three conditions continue to hold, and $\sum_{x\in \rset}x$ does not decrease.

\item If the first step is not applicable, but there is a number $x\in \rset$ with $x<\beta y_x^2$, replace it with $x'=\beta y_x^2$, and $y_{x'}=y_x$. Since $y_x\leq \sqrt{\frac M {\beta}}$, it is easy to see that all three conditions continue to hold, and $\sum_{x\in \rset}x$ does not decrease.

\item If the above two steps are not applicable, but there are two numbers $x,x'\in \rset$, with $y_x,y_{x'}\leq \half \sqrt{\frac M{\beta}}$, replace the two numbers $x,x'$ with a single number $x''=\beta(y_x+y_{x'})^2$, and set $y_{x''}=y_x+y_{x'}$. Notice that all three conditions continue to hold, and moreover, $x''=\beta(y_x+y_{x'})^2\geq \beta(y_x^2+y_{x'}^2)= x+x'$.
\end{itemize}

When none of the above steps is applicable, there is at most one number $x\in \rset$ with $y_x< \half \sqrt{\frac M{\beta}}$ (and in this case $x\leq \beta y_x^2\leq M/4$), and since $\sum_{x\in \rset}y_x\leq S$, we have that $|\rset|\leq \frac{2S\sqrt{\beta}}{\sqrt M}+1$. As each number $x\leq M$, we have that $\sum_{x\in \rset}x\leq 2S\sqrt{\beta M}+M/4$.
\end{proof}

\section{Well-Linked Decomposition}
\label{sec: appendix-well-linked-proofs}
\label{------------------------------------------------C: well linked decomposition------------------------------------------------------------------------}

\subsection{Proof of Theorem~\ref{thm: well-linked}}
We use the $\alphasc$-approximation algorithm $\algsc$ for the non-uniform sparsest cut problem (see Section~\ref{subsec: sparsest cut} for definitions). We set $\alpha^*=1/(8\alphasc\log n)=\Omega(1/(\log^{3/2}n\log \log n))$,

Throughout the algorithm, we maintain a partition $\jset$ of the input set $J$ of vertices. At the beginning, $\jset$ consists of the subsets of $J$ defined by the connected components of $G[J]$. 

Let $J'\in \jset$ be any set in the current partition, and let $G'=G[J']\cup \out(J')$ be the corresponding sub-graph of $G$. We set up a non-uniform sparsest cut problem on $G'$, where the weight of every vertex in $T(J')$ equals the number of edges in $\out(J')$ incident on it, and the weights of all other vertices are $0$. If $J'$ is not $\alpha^*$-well-linked, then there is a cut of sparsity at most $\alpha^*$ in $G'$. In this case, we can apply algorithm $\algsc$, to obtain a partition $(J'_1,J'_2)$ of $V(G')$ of sparsity at most $\alpha^*\cdot\alphasc\leq 1/(8\log n)$. Let $A=J'_1\setminus T(J)$ and $B=J'_2\setminus T(J)$. Notice that $(A,B)$ is a partition of $J'$. Moreover, if we denote $T_1=\out(J')\cap \out(A)$ and $T_2=\out(J')\cap \out(B)$, then $|E(A,B)|\leq \min\set{|T_1|,|T_2|}/(8\log n)$.   We then replace $J'$ with $A$ and $B$  in $\jset$. For accounting purposes, we charge the edges in $E(A,B)$ to the edges in $\out(J')$, as follows. If $|A|<|B|$, then we evenly charge the edges in $\out(J')\cap \out(A)$ for the edges in $E(A,B)$. Since $|E(A,B)|\leq |T_1|/(8\log n)$, the charge to every edge is at most $1/(8\log n)$. Otherwise, if $|B|\leq |A|$, we charge the edges of $\out(J')\cap \out(B)$ for the edges in $E(A,B)$. Again, the charge to every edge is at most $1/(8\log n)$.

We continue this procedure, until for every subset $J'\in\jset$, algorithm \algsc returns a cut of sparsity greater than $1/(8\log n)$. We are then guaranteed that every set $J'\in \jset$ is $\alpha^*$-well-linked.

In order to bound $\sum_{J'\in\jset}|\out(J')|$, we use the above charging scheme. Notice that every edge can be  charged at most $2\log n$ times (since each time an edge $e=(u,v)$ is charged for a cluster to which $u$ belongs, the size of this cluster decreases by at least factor $2$, and the same holds for $v$). Therefore, the total amount charged to any edge $e\in \out(J)$ is at most $\frac 1 4$. However, this only refers to the \emph{direct} charge. For example, some edge $e'\not \in \out(J)$, that has first been charged to the edges in $\out(J)$, can in turn be charged for other edges. We call such charging \emph{indirect}. If we sum up the indirect charge for every edge $e\in \out(J)$, we obtain a geometric series, and so the total direct and indirect amount charged to every edge $e\in \out(J)$ is at most $1$. Therefore, $\sum_{J'\in \jset}|\out(J')|\leq 2|\out(J)|$.

\subsection{Proof of Theorem~\ref{thm: well-linked-general}}
We first show how to handle property (P1) and the property that $J$ is canonical. We deal with property (P2) later.

The decomposition procedure is similar to the one in the proof of Theorem~\ref{thm: well-linked}, and the proof is by induction. Initially, $\jset$ is the partition of $J$ induced by the connected components of the graph $G[J]$. Clearly, if (P1) holds for $J$, it has to hold for every set in $\jset$, and if $J$ is canonical, every set in $\jset$ is also canonical (since for each $Z\in \zset$, $G[Z]$ is connected).

Assume that we have some set $J'\in \jset$, for which algorithm \algsc has returned a cut of sparsity at most $1/(8\log n)$, and let $(A,B)$ denote the resulting partition of $J'$. We first show that if $A$ or $B$ are not canonical, but $J'$ is canonical, then we can efficiently find a cut of even smaller sparsity in graph $G[J']\cup \out(J')$. We can then replace $(A,B)$ with the corresponding new partition and continue, until we obtain a partition $(A',B')$ of $J'$, where both $A'$ and $B'$ are canonical, and the sparsity of the corresponding cut is at most $1/(8\log n)$.

Denote $T=\out(J')$, $T_1=T\cap \out(A)$ and $T_2=T\cap \out(B)$, and we assume w.l.o.g. that $|T_1|\leq |T_2|$. We refer to the edges in $T$ as terminals. Let $E'=E(A,B)$. Then $|E'|/|T_1|\leq 1/(8\log n)$.

If $J'$ is canonical, but one of the sets $A$,$B$ is not, then there must be some set
 $Z\in \zset$, that is being split between the two sides. Let $\Upsilon=\out(Z)$, $\Upsilon_1'=\out(Z)\cap T_1$, $\Upsilon_2'=\out(Z)\cap T_2$. Let $\Upsilon_1$ denote all edges in $\out(Z)$ whose endpoint inside $Z$ belongs to $A$, but the edge itself is not in $T_1$ (notice that the other endpoint may be inside either $A$ or $B$). Similarly, let $\Upsilon_2$ denote all edges in $\out(Z)$ whose endpoint inside $Z$ belongs to $B$, but the edge is not in $T_2$ (see Figure~\ref{fig: well linked}).

\begin{figure}[h]
\scalebox{0.4}{\rotatebox{0}{\includegraphics{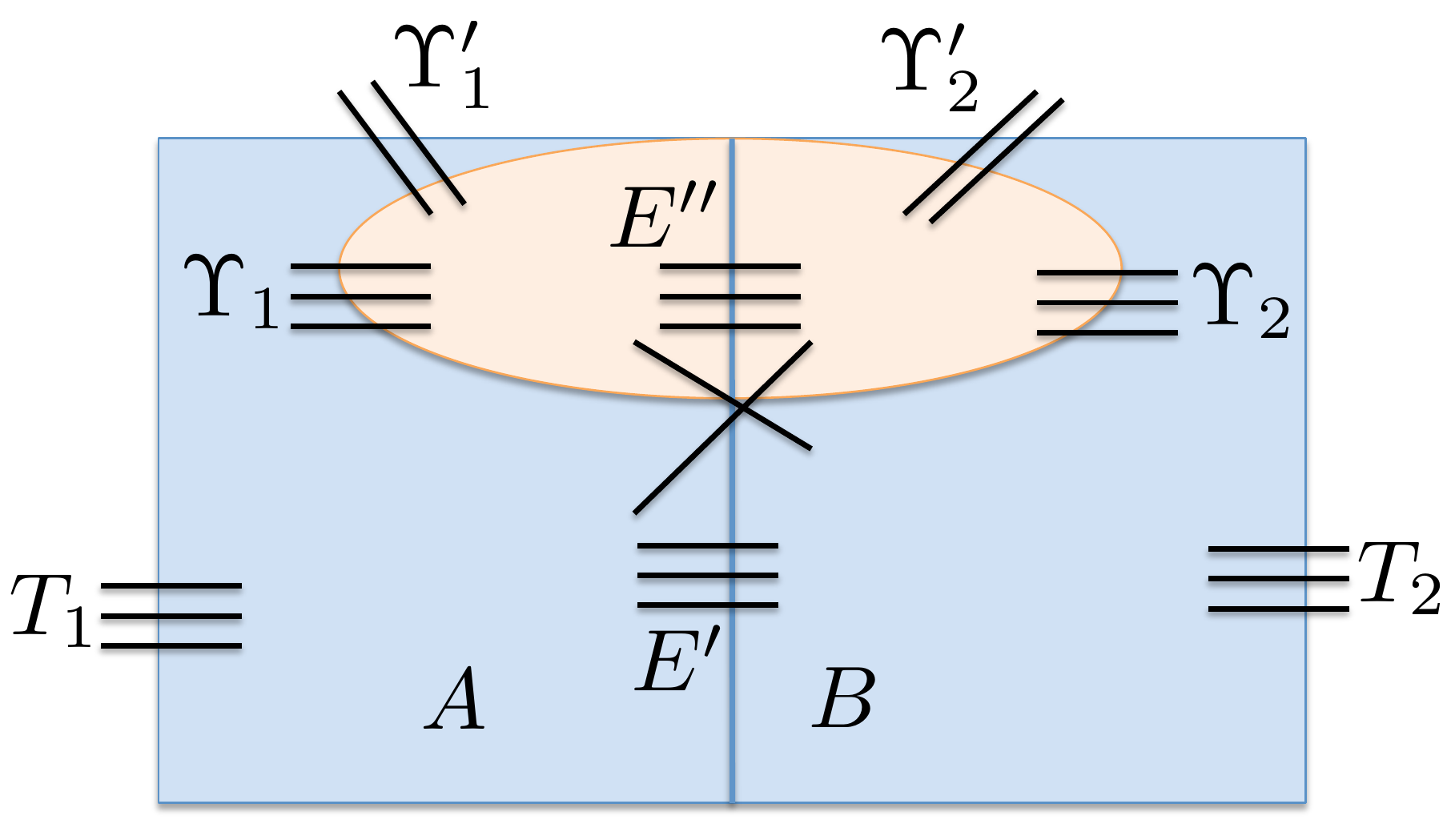}}}\caption{Canonical well-linked decomposition\label{fig: well linked}}
\end{figure}

Let $E''=E(Z\cap A,Z\cap B)$.
We consider two cases. In the first case, $|\Upsilon_1'|+|\Upsilon_1|\leq |\Upsilon_2'|+|\Upsilon_2|$, and therefore $|E''|\geq |\Upsilon_1'|+|\Upsilon_1|$ (since $Z$ is $1$-well-linked). In this case we move $Z$ completely to $B$. Observe that the cut size goes down by at least $|E''|-|\Upsilon_1|$, and the size of the new cut is at most $|E'|-|E''|+|\Upsilon_1|\leq |E'|-|\Upsilon_1'|$. The number of terminals on the smaller side becomes $|T_1|-|\Upsilon'_1|$. Therefore, the new sparsity is at most $\frac{|E'|-|\Upsilon_1'|}{|T_1|-|\Upsilon_1'|}<\frac{|E'|}{|T_1|}$, since for all $c<a<b$, $\frac{a-c}{b-c}<\frac a b$.

In the second case, $|\Upsilon_1'|+|\Upsilon_1|> |\Upsilon_2'|+|\Upsilon_2|$, and therefore $|E''|\geq |\Upsilon_2'|+|\Upsilon_2|$. We then move all vertices in $Z$ to $A$. Observe that the cut size goes down by at least $|E''|-|\Upsilon_2|$, and the size of the new cut is at most $|E'|-|E''|+|\Upsilon_2|\leq |E'|-|\Upsilon_2'|$. The number of terminals on one side becomes $|T_1|+|\Upsilon_2|$, and on the other side $|T_2|-|\Upsilon_2|$.
If the second quantity is smaller than the first, then the new sparsity is at most $\frac{|E'|-|\Upsilon_2'|}{|T_2|-|\Upsilon_2'|}<\frac{|E'|}{|T_2|}\leq\frac{|E'|}{|T_1|}$. Otherwise, the new sparsity is at most $\frac{|E'|-|\Upsilon_2'|}{|T_1|+|\Upsilon_2'|}<\frac{|E'|}{|T_1|}$. In any case, the sparsity goes down.

We continue this process, until we obtain a partition $(A,B)$ of $J'$, where $A$, $B$ are canonical and the cut sparsity remains at most $1/(8\log n)$.

We next show that if $J'$ had property (P1), then we can find a cut $(A',B')$ of $J'$, such that both $A',B'$ will still have property (P1), and the sparsity of the cut is at most $1/(8\log n)$. Moreover, if the cut $(A,B)$ we are starting from is canonical, then in the new cut, both sets $A'$ and $B'$ are canonical.
Since $J'$ had property (P1), the vertices of $T(J')$ are connected in the graph $G\setminus J'$. Assume that set $A$ does not have property (P1).
Then there are two edges $e=(x,y),e'=(x',y')$, $e,e'\in \out(A)$, $x,x'\in A$, such that there is no path connecting $y$ to $y'$ in graph $G\setminus A$. 
 Since $J'$ had property (P1), this can only happen if graph $G[B]$ has a connected component $C$, such that $\out(C)\sse E(A,B)$, and exactly one of the vertices $y,y'$ belongs to $C$. In this case, we can move all vertices of $C$ to $A$, and this will only decrease the cut sparsity, since $\out(C)\cap \out(J')=\emptyset$. Moreover, if $A$ and $B$ were both canonical, then both $B\setminus C, A\cup C$ will remain canonical, as for each $Z\in\zset$, $G[Z]$ is a connected graph. We can continue this process until both sets $A$ and $B$ have property (P1). Eventually, if the set $J'$ we start from is not $\alpha^*$-well-linked, we obtain a partition $(A,B)$ of $J'$ that induces a cut of sparsity at most $1/(8\log n)$. Moreover, if $J'$ was canonical, or had property (P1), or both, then both sets $A$ and $B$ will have the same properties.

Let $\jset$ be the final partition of $J$. We have shown that if $J$ was canonical, or had property (P1), or both, then every set $J'\in \jset$ will have exactly the same properties. It now only remains to show how to handle property (P2).

Assume that the original set $J$ had property (P2). Construct a new graph $H$ as follows: $H$ consists of the sub-graph $G[J]$ of $G$, and an additional vertex $s$, that connects with an edge to every vertex in $\Gamma_G(J)$. Notice that set $J$ has properties (P1) and (P2) in graph $H$. We can then find the decomposition $\jset$ of $J$ as above, such that every set $J'\in\jset$ has property (P1) in graph $H$, and if $J$ is canonical, then each set $J'\in \jset$ is also canonical.
 We now show that every set $J'\in \jset$ has property (P2) as well. Since $J$ had property (P2), there is a planar drawing $\psi$ of the graph $H$. This drawing induces a planar drawing $\psi'$ of $J'$, for each $J'\in \jset$. All vertices in $\Gamma(J')$ are connected in graph $(H\setminus J')\cup \out(J')$. Therefore, in $\psi'$, all vertices of $\Gamma(J')$ lie on the boundary of the same face, and so $J'$ has property (P2). Finally, notice that if, additionally, $J$ had property (P1) in $G$, then each set $J'\in\jset$ also has this property.

The number of edges $\sum_{J'\in \jset}|\out(J')|$ is bounded by $2\out(J)$ as in the proof of Theorem~\ref{thm: well-linked}.

\section{Graph Contraction: Proofs of Theorems}
\subsection{Proof of Theorem~\ref{thm: graph contraction}}
\label{------------------------------------------------E: graph contraction------------------------------------------------------------------------}

Fix some $i: 1\leq i\leq q$. We start by showing how to find the partition $\xset_i$ of $V(G_i)$, for which properties~(\ref{property: subsets-first})--(\ref{property: size-last}) hold. The difficult part will then be to show that there exists a required drawing of the resulting contracted graph.

We start with some fixed index $1\leq i\leq q$ and the graph $G_i$. Our first step is to use Theorem~\ref{thm: well-linked-general}, to find a partition $\wset_i$ of $V(G_i)$, into disjoint subsets, such that each set $Y\in \wset_i$ is $\alpha^*$-well-linked and has properties (P1) and (P2). We are also guaranteed that $\sum_{Y\in \wset_i}|\out(Y)|\leq 2|\Gamma(G_i)|$. We denote $\wset_i=\set{Y^i_1,\ldots,Y^i_{p_i}}$.


For each $1\leq j\leq p_i$, we now define a further decomposition $\yset_j^i$ of the set $Y_j^i$ of vertices, and in the end, we will set $\xset_i=\bigcup_{j=1}^{p_i}\yset_j^i$. We now fix some set $Y_j^i\in \wset_i$, and focus on defining the decomposition $\yset_j^i$ of $Y_j^i$. We will omit the subscript and the superscript $i$ from now on when clear from context.

First, we decompose $Y_j$ into a collection $\cset_j$ of maximal $2$-connected components (some components may consist of a single edge). Let $S_j^{(1)}$ be the set of all vertices $u\in V(Y_j)$, such that $u$ is a $1$-separator for $Y_j$. 
Let $\tset$ be the tree whose vertex set is: $S_j^{(1)}\cup \set{v_C\mid C\in \cset_j}$, and there is an edge between $v_C$ and $u\in S_j^{(1)}$ iff $u\in C$. If there is a vertex $u\in S_j^{(1)}\cap \Gamma(Y_j)$, then we root $\tset$ at $u$, and denote $t_0=u$. Otherwise, there must be a cluster $C\in \cset_j$ with $C\cap \Gamma(Y_j)\neq \emptyset$. We then root $\tset$ at $v_C$, and we set $t_0$ to be any vertex in $C\cap \Gamma(Y_j)$.
For each node $x$ of the tree, we denote by $\tset(x)$ the sub-tree of $\tset$ rooted at $x$. We also denote by $\Upsilon(x)$ the union of all clusters $C\in \cset_j$, such that $v_C\in \tset(x)$. Also, for any sub-tree $\tset'$ of $\tset$, we denote by $\Upsilon(\tset')$ the union of all clusters $C\in \cset_j$, such that $v_C\in \tset'$.

We now mark some vertices of the tree $\tset$, as follows. For each cluster $C\in \cset_j$ that contains at least one vertex in $\Gamma(Y_j)$, we mark the vertex $v_C$. We then go over the tree in the bottom-up fashion. Consider the current vertex $x$ of $\tset$. If $x$ has at least two children, say $y$ and $y'$, whose sub-trees contain marked vertices, then we mark $x$ as well. Otherwise we do not mark $x$. For each marked vertex $x$, remove the edge connecting it to its father from the tree $\tset$. Also, if $y$ is a child of a marked vertex $x$, and the sub-tree $\tset(y)$ contains a marked vertex, we remove the edge $(x,y)$ from the tree. Consider the resulting forest.  We say that a tree $\tset'$ of this forest is a trivial tree iff it consists of a single marked vertex  $u\in S^{(1)}_j$. For each non-trivial tree $\tset'$ in this forest, add the set $\Upsilon(\tset')$ of vertices to $\yset_j$.

This finishes the definition of the partition $\yset_j^i$ of $Y_j^i$. Notice that every vertex $v\in Y_j^i$ belongs to some set $X\in\yset_j^i$. Observe also that every pair $X,X'\in \yset_j^i$ of vertex subsets is completely disjoint, except that it may share one interface vertex, and so the graphs $\G[X]$ and $\G[X']$ are edge disjoint. We now set $\xset_i=\bigcup_{j=1}^{p_i}\yset_j^i$. Since all subsets of vertices in the partition $\wset_i$ of $V(G_i)$ were completely vertex disjoint, this establishes Property~(\ref{property: disjointness}). We prove that this partition has properties~(\ref{property: subsets-first}), (\ref{property: structure of X}), (\ref{property: subsets - last}) and (\ref{property: size-last}) below. But first, we need to establish some properties of the sets $X\in \yset_j^i$, that will be useful later.

\paragraph{Structure of sets $\mathbf{X\boldsymbol{\in \yset_j^i}}$.} Fix some $1\leq i\leq q$, $1\leq j\leq p_i$ (we will omit the index $i$ now). Consider some set $X\in \yset_j$, and the corresponding sub-tree $\tset'$ of $\tset$, such that $X=\Upsilon(\tset')$. Let $r$ be the root of the tree $\tset'$, and let $y_1,\ldots,y_t$ be the children of $r$ that belong to $\tset'$. If $r$ is a vertex of the form $v_C$, then we let $C^*_X$ denote the cluster of $\cset_j$, such that $v_{C^*_X}=r$. Otherwise, $C^*_X=\set{r}$. For each $1\leq t'\leq t$, let $R'_{t'}$ denote the sub-set of vertices of $\Upsilon(y_{t'})$ that belong to $X$. Then $(C^*_X,R'_1,\ldots,R'_{t})$ define a partition of $X$ (except that each set $R_{t'}$ shares a single vertex with $C^*_X$). Observe that $\Gamma(X)$ consists of three types of vertices: $\Gamma_1(X)=\Gamma(X)\cap \Gamma(Y_j)$ are the original interface vertices of $Y_j$; $\Gamma_2(X)$ contains a single vertex that is common to $r$ and its parent (if $r\in S^{(1)}_j$, then $\Gamma_2(X)=\set{r}$); $\Gamma_3(X)$ contains all remaining interface vertices.
Note that $\Gamma_1(X), \Gamma_2(X), \Gamma_3(X)$ is not necessarily a strict partition of $\Gamma(X)$, in the sense that some vertices of $\Gamma(X)$ may belong to several subsets.

If $r$ is a marked vertex, then we say that $X$ is of type 1; otherwise it is of type 2. 
Assume first that $r$ is a marked vertex. Since for each $1\leq t'\leq t$, we did not remove the edge $(r,y_{t'})$ from the tree, $\tset(y_{t'})$ does not contain any marked vertices, and so $\tset(y_{t'})\sse \tset'$, as no edges have been removed from it. Therefore, all interface vertices $\Gamma(X)$ must belong to $C^*_X$ in this case.

Assume now that $r$ is not a marked vertex. Then $C^*_X$ does not contain any vertices of $\Gamma(Y_j)$, and $\tset'$ does not contain any marked vertices. We claim that in this case $|\Gamma_3(X)|\leq 1$, and so $|\Gamma(X)|\leq 2$.
 Assume for contradiction that $|\Gamma_3(X)|\geq 2$.
Two cases are possible. The first case is when there is some node $v_C\in \tset'$, such that $C$ contains two vertices, $a,b\in \Gamma_3(X)$.
But that means that $v_C$ had two children that were marked vertices, and so it should have been marked itself, a contradiction.
The second case is when there are two {\bf distinct} nodes $x,y\in \tset'$, which either belong to $\Gamma_3(X)$, or their cluster contains a vertex in $\Gamma_3(X)$. 
 Then $x$ and $y$ each have a child, that is a marked vertex of $\tset$. We denote these vertices by $x'$ and $y'$. Let $z$ be the lowest common ancestor of $x'$ and $y'$ in $\tset$. Then $z\in \tset'$, and yet $z$ is a marked vertex, a contradiction. It follows that if $X$ is a type-2 cluster, then $|\Gamma(X)|\leq 2$. 
 It is now immediate to see that each set $X\in \xset$ has Property~(\ref{property: structure of X}). The only change is that, in order to obtain the partition $(C^*_X,R_1,\ldots,R_t)$, we start with the sets $(C^*_X,R'_1,\ldots,R'_t)$, and for each $1\leq t'\leq t$, we remove the unique vertex that $R'_{t'}$ shares with $C^*_X$, to obtain $R_{t}$. This vertex then serves as the separator $u_{t'}$.

\begin{figure}[h]
\scalebox{0.3}{\includegraphics{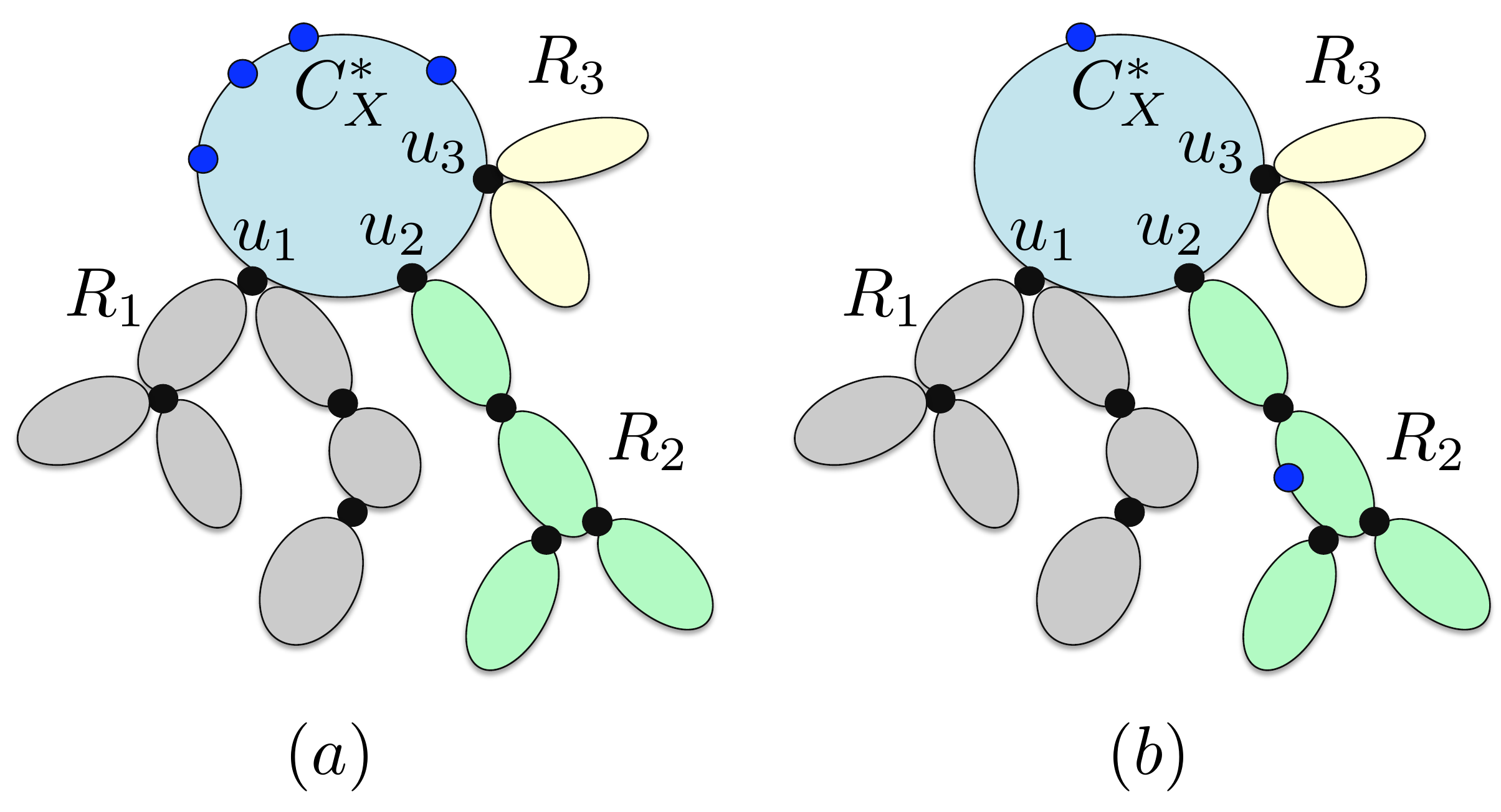}} \caption{Structure of $X$. (a) type-1 clusters; (b) type-2 clusters. Vertices of $\Gamma(X)$ are in blue.} \label{fig: structure of X}
\end{figure}

We are now ready to establish Properties~(\ref{property: subsets-first}), (\ref{property: subsets - last}) and (\ref{property: size-last}) for the resulting partition.
First, it is easy to see that each set $X\in \yset_j^i$ still has properties (P1) and (P2): Recall that $Y_j^i$ had both properties. Consider any planar drawing of $Y_j^i$, where all the interface vertices $\Gamma(Y_j^i)$ lie on the outer face. This induces the required planar drawing for each one of the clusters $X\in\yset_j$, giving property (P2). 
It is also easy to see that we can draw a closed simple curve $\gamma_X$, such that the images of all vertices in $\Gamma(X)$ lie on $\gamma_X$, and all other vertices and edges of $\G[X]$ are drawn inside $\gamma_X$. This follows from the fact that $X$ has property (P2), and whenever $|\Gamma(X)|>2$, all vertices of $\Gamma(X)$ belong to the $2$-connected sub-graph of $\G[X]$ induced by $C_X^*$.

We now turn to establish property (P1).
Consider some cluster $X\in \yset_j^i$. Recall the decomposition of $\Gamma(X)$ into three sets: $\Gamma_1(X)=\Gamma(X)\cap \Gamma(Y_j)$, $\Gamma_2(X)$ contains a single vertex common to the root of the tree $\tset'$ and its parent in $\tset$, and $\Gamma_3(X)$ contains all other vertices in $\Gamma(X)$. In order to establish property (P1), it is enough to show that for each vertex $v\in \Gamma_2(X)\cup \Gamma_3(X)$, for each edge $(v,u)\in E$ with $u\not\in X$, there is a path $P$ connecting $u$ to some vertex $t\in \Gamma(Y_j)$ in the graph $\G[Y_j]\setminus X$.
Consider first the vertex $v\in \Gamma_2(X)$. For each $u\not \in X$, with $(v,u)\in E$, there is a path connecting $u$ to the vertex $t_0\in \Gamma(Y_j)$ (recall that this is the vertex of $\Gamma(Y_j)$ corresponding to the root node of the tree $\tset$). Consider now some vertex $v\in \Gamma_3(X)$, and let $x$ be the node of $\tset'$, such that either $x=v$, or $v\in C$ and $v_C=x$. Let $u\not\in X$, such that $(u,v)\in E$, and let $y$ be the child node of $x$ in $\tset$, such that either $u=y$, or $u\in {C'}$, while $y=v_{C'}$. Then $y\not\in\tset'$, and $\tset(y)$ contains a marked vertex. Therefore, $\Upsilon(y)$ contains a vertex $t\in \Gamma(Y_j)$, and there is a path $P$ connecting $u$ to $t$ in the sub-graph of $\G$ induced by $\Upsilon(y)$. This establishes property (P1).
Next, we prove that each set $X\in \yset_j$ is $\alpha^*$-well-linked.

\begin{claim}\label{claim: well linkedness of X}
Each cluster $X\in \yset_j$ is $\alpha^*$-well-linked. 
\end{claim}

\begin{proof}
Consider some cluster $X\in \yset_j$, and let $\tset'$ be the corresponding sub-tree of $\tset$, such that $X=\Upsilon(\tset')$.
First, if $X$ is a type-2 cluster, then $|\Gamma(X)|\leq 2$, and since $\G[X]$ is connected, $X$ is $\alpha^*$-well-linked.

Assume now that $X$ is a type-1 cluster, that is, its root vertex $r$ is marked.  Then  $\Gamma(X)\sse C^*_X$. If $r\in S^{(1)}_j$, then $\Gamma(X)$ consists of a single vertex, $r$, and is therefore well-linked. Assume now that $r\not\in S^{(1)}_j$, and let $u\in S^{(1)}_j$ be the parent of $r$ in $\tset$.
Let $u_1,\ldots,u_k$ be the children of $r$ in $\tset$ that {\bf do not} belong to $\tset'$. Notice that $u_1,\ldots,u_k\in S^{(1)}_j$. Then, for each $1\leq k'\leq k$, there is a marked vertex in $\tset(u_{k'})$, and a vertex $t_{k'}\in \Gamma(Y_j)\cap \Upsilon(u_{k'})$. Moreover, for $k'\neq k''$, $t_{k'}\neq t_{k''}$. Let $t_u=t_0$, the vertex in $\Gamma(Y_j)$ associated with the root of the tree $\tset$.

 Then $\Gamma(X)=\Gamma_1(X)\cup\set{u,u_1,\ldots,u_k}$, where $\Gamma_1(X)\sse \Gamma(Y_j)$. Assume for contradiction that $X$ is not $\alpha^*$-well-linked. We will show that this implies that $Y_j$ is not $\alpha^*$-well-linked.
Let $(A,B)$ be any partition of $X$, such that, if we denote $T_A=\Gamma(X)\cap A$ and $T_B=\Gamma(X)\cap B$, then $|T_A|\leq |T_B|$, and $|E(A,B)|<\alpha^*|T_A|$.

We extend $(A,B)$ to a partition $(A',B')$ of the whole set $Y_j$ of vertices, as follows: if $u\in A$, then we add all vertices in $Y_j\setminus \Upsilon(r)$ to $A$, and otherwise we add them to $B$. Also, for each $1\leq k'\leq k$, if $u_{k'}\in A$, then we add all vertices in $\Upsilon(u_{k'})$ to $A$, and otherwise, we add them to $B$. Let $(A',B')$ be the resulting partition of $Y_j$. Then $E(A',B')=E(A,B)$.  

Let $T'_{A'}=\Gamma(Y_j)\cap A'$, and $T'_{B'}=\Gamma(Y_j)\cap B'$. Then $|T'_{A'}|\geq |T_A|$ must hold, as for each vertex $u_{k'}\in A$, we have added the vertex $t_{k'}$ to $T_{A'}$, and if $u\in T_A$, then vertex $t_u$ has been added to $T'_{A'}$. It then follows that $|E(A',B')|=|E(A,B)|<\alpha^*|T_A|\leq \alpha^*|T'_{A'}|$, and so $Y_j$ is not $\alpha^*$-well-linked, a contradiction.
\end{proof}

We have thus established Property~(\ref{property: subsets-first}). In order to establish Property~(\ref{property: subsets - last}),  we need to bound $\sum_{X\in \yset_j}|\out(X)|$. Notice that the number of marked vertices in $\tset$ is at most $2|\Gamma(Y_j)|$. Therefore, the number of edges removed from $\tset$ is at most $4|\Gamma(Y_j)|$. Since $\sum_{X\in \yset_j}|\Gamma(X)|$ is bounded by twice the number of the edges removed from the tree $\tset$ plus $|\Gamma(Y_j)|$, we get that  $\sum_{X\in \yset_j}|\Gamma(X)|\leq 9|\Gamma(Y_j)|$.
We establish Property~(\ref{property: size-last}) in the next claim.

\begin{claim} \label{claim: size of X} Let $X$ be any subset of vertices of $\G$ that has properties (P1), (P2), and is $\alpha$-well-linked for any $\alpha$. Then $|X|\geq \alpha^2 |\Gamma(X)|^2/64\dmax^2$.
\end{claim}
\begin{proof}
Let $\pi_X$ be the drawing of $X$ in which the interface vertices lie on the boundary of the outer face. Let $S_1,S_2,S_3,S_4\sse \Gamma(X)$ be collections of $z=\lfloor |\Gamma(X)|/4\rfloor $ interface vertices each, such that for each $1\leq h\leq 4$, the interface vertices of $S_h$ lie contiguously along the boundary of the outer face of $\pi_X$, and the ordering among these sets is $(S_1,S_2,S_3,S_4)$. Since $X$ is $\alpha$-well-linked, there is a collection $\pset$ of $\alpha z/\dmax$ {\bf vertex-disjoint} paths in $\G[X]$, connecting vertices in $S_1$ to vertices in $S_3$. Similarly, there is a collection $\pset'$ of $\alpha z/\dmax$ vertex-disjoint paths connecting vertices in $S_2$ to vertices in $S_4$ in $\G[X]$. Since the drawing $\pi_X$ has no crossings, every pair $P\in \pset$, $P'\in \pset'$ of paths has to share a vertex. Therefore, $X$ must contain at least $\alpha^2z^2/\dmax^2$ vertices. Since $z\geq |\Gamma(X)|/8$, the claim follows.
\end{proof}

Combining Claims~\ref{claim: well linkedness of X} and \ref{claim: size of X} establishes Property~(\ref{property: size-last}). Therefore, so far we have proved  properties~(\ref{property: subsets-first})--(\ref{property: size-last}) for partition $\xset_i$, for each $1\leq i\leq q$. It now only remains to show that there is a canonical drawing $\phi'$ of the resulting contracted graph $H=G_{|S}$.

\subsection{Existence of the Canonical Drawing}
The goal of this section is to prove the following theorem.

\begin{theorem}\label{thm: canonical drawing}
Suppose we are given a collection $\xset$ of subsets of vertices of $\G$, for which Properties~(\ref{property: subsets-first}), (\ref{property: structure of X}) and (\ref{property: disjointness}) hold.
 Let $H$ be the contracted graph, in which each sub-graph $\G[X]$ is replaced by $Z'_X$. Then there is a canonical  drawing $\phi'$ of $H$, such that $\cro_{\phi'}(H) = O(d_{\max}^9 \cdot \log^{10} n \cdot (\log\log n)^4 \cdot \optcro{\G})$.
\end{theorem}

The rest of this section is devoted to proving Theorem~\ref{thm: canonical drawing}.
We start with a high-level overview of the proof.
For each set $X\in \xset$, let $\Gr0(X)$ denote the graph $\G[X]$. Recall that $\Gr0(X)$ has a planar drawing $\pi(X)$, in which all interface vertices $\Gamma(X)$ lie on the boundary of the outer face. On the other hand, the optimal drawing $\phi$ of $\G$ also induces a drawing $\phi_X$ of $\Gr0(X)$. Following~\cite{CMS10}, we will define the sets of irregular vertices and edges of the graph $\Gr0(X)$ to be the vertices and edges whose ``local'' drawing is different in $\pi(X)$ and $\phi_X$. 

Next, we show that there is a ``nice'' drawing $\phi''$ of $\G$, in which, roughly speaking, for each sub-graph $\Gr0(X)$, only the edges incident on the vertices in $\Gamma(X)$ participate in crossings. We show that such a nice drawing for $\G$ immediately gives a canonical drawing for the contracted graph $H$. We then bound $\cro_{\phi''}(\G)$ by roughly $\cro_{\phi}(\G)$ plus the number of irregular edges and vertices in each set $X\in \xset$. Using the result of~\cite{CMS10}, this number is in turn roughly bounded by $O(\cro_{\phi}(\G))$, but only if the graphs $\Gr0(X)$ are $3$-vertex connected. Therefore, in order to apply this argument, we need to first perform some surgery on the graphs $\Gr0(X)$, to get rid of all $1$-vertex cuts, and most $2$-vertex cuts. The rest of the proof consists of three parts. The first part is the surgery that we perform on graphs $\Gr0(X)$. In the second part we define irregular vertices and edges for each set $X\in \xset$ and bound their number. This is similar to what is done in~\cite{CMS10}, except that we will need to deal with the few $2$-vertex separators that still remain in the graphs. In the third step we show a nice drawing of the resulting graph, that will give a canonical drawing of the contracted graph.
\subsubsection{Part 1: Surgery}

The goal of this part is to get rid of all $1$-vertex cuts, and most $2$-vertex cuts in sets $X\in \xset$. We do so in four simple steps. At the end of each step $h: 1\leq h\leq 4$, we will obtain a graph $\Gr h(X)$, for each $X\in \xset$, that will replace the graph $\Gr0(X)$ in $\G$. Graph $\Gr h(X)$ will contain the set $\Gamma(X)$ of vertices, and in order to obtain the graph $\Gr h$ from $\G$, we simply replace each subgraph $\Gr0(X)$ with $\Gr h(X)$ for each $X\in \xset$, using the vertices $\Gamma(X)$ as the interface in this replacement procedure. Therefore, once we define the graphs $\Gr h(X)$ for all $X\in \xset$, the graph $\Gr h$ is fixed. We will ensure that the final graph $\Gr 4$, that we obtain after the fourth step, is precisely the contracted graph $H=\G_{|S}$. We will also ensure that the maximum vertex degree $\dmax$ does not increase throughout these steps.

From now on, for each set $X\in \xset$, $\Gamma(X)$ will denote the set of interface vertices of $X$ in the graph $\G$. This set will not change as we obtain new graphs by transforming $\G$. For any graph $H'$, given a subset $S\sse V(H')$ of vertices, we will use $\Gamma_{H'}(S)$ to denote the set of interface vertices of $S$ in $H'$. For convenience, when $H''$ is a sub-graph of $H'$, we will sometimes write $\Gamma_{H'}(H'')$ instead of $\Gamma_{H'}(V(H''))$, and $T_{H'}(H'')$ instead of $T_{H'}(V(H''))$.
\paragraph{Step 1: getting rid of things we do not need}
Let $X\in \xset$. Assume first that $|\Gamma(X)|\leq 2$. Then we set $\Gr1(X)=\Gr2(X)=\cdots =\Gr4(X)=Z'(X)$ (see Figure~\ref{fig: grids}). It is easy to see that replacing $\Gr0(X)$ with $Z'(X)$ does not increase the number of edges participating in crossings in the resulting graph. Moreover, for any drawing of the final graph $\Gr4$, it is easy to obtain a drawing in which only the matching edges of $Z'_X$ will participate in crossings (and not the unique non-matching edge). Therefore, from now on we can simply ignore sets $X$ for which $|\Gamma(X)|\leq 2$. Let $\xset'\sse \xset$ denote the collection of sets $X$ for which $|\Gamma(X)|>2$. We will restrict our attention to sets $X\in \xset'$ from now on.

Consider some set $X\in \xset'$. Using Property~(\ref{property: structure of X}), we can decompose $X$ into $C^*_X,R_1,\ldots,R_t$, such that $\Gamma(X)\sse C^*(X)$, and for each set $R_{t'}$, for $1\leq t'\leq t$, there is a $1$-separator $u_{t'}$ in $\Gr0(X)$, whose removal separates $R_{t'}$ from the remaining vertices of $X$. We simply erase the sets $R_1,\ldots,R_t$ from $X$. In other words, we replace $X$ with $C^*_X$. Let $\Gr1(X)$ denote this resulting graph, and let $\Gr 1$ denote the whole resulting graph obtained from $\G$, by replacing $\G[X]$ with $\Gr1(X)$ for all $X\in \xset$.

Notice that the set of the interface vertices $\Gamma_{\Gr1}(\Gr 1(X))=\Gamma(X)$, even though $C^*(X)$ had additional interface vertices in the old graph $\Gr0(X)$. Therefore, $\Gr1(X)$ retains properties\footnote{Notice that property (P1) is now defined w.r.t. $\Gr 1$ and the set $T_{\Gr 1}(\Gr 1(X))$ of terminal vertices.} (P1) and (P2), it contains all vertices in $\Gamma(X)$, and it is $\alpha^*$-well-linked w.r.t. them, since $\Gr0(X)$ had all these properties. Also, $\Gr1(X)$ does not contain any $1$-vertex cuts. 
 Clearly, $\optcro{\Gr1}\leq \optcro{\G}$, since, in a sense, $\Gr1$ is a sub-graph of $\G$. Obviously, the maximum vertex degree did not increase in this step.


\paragraph{Step 2: introducing new interface vertices}
Let $X$ be any set in $\xset'$. For each interface vertex $v\in \Gamma(X)$, we create a new copy $v_X$ of $v$, that will replace $v$ in $\Gr1(X)$. We call $v_X$ a \emph{new interface vertex} for $X$, and $v$ an \emph{old interface vertex}. In order to obtain the new graph $\Gr2(X)$ from $\Gr1(X)$, we replace each old interface vertex $v$ with a new interface vertex $v_X$, and add a \emph{matching edge} $(v,v_X)$ to the graph (see Figure~\ref{fig: Step 2}). Therefore, if $v$ is an interface vertex that has been shared by $k$ clusters $X_1,\ldots,X_k\in \xset'$, then each graph $\Gr2(X_{k'})$, for $1\leq k'\leq k$ will now contain a new copy $v_{X_{k'}}$ of $v$, and all these copies are connected to $v$ via matching edges. We denote by $\Gamma'(X)$ the set of new interface vertices, and by $M(X)$ the set of matching edges of $\Gr2(X)$. We also denote by $\Hr2(X)$ the graph $\Gr2(X)$ without the old interface vertices and the matching edges. Notice that the sets of the interface vertices are now defined as follows: $\Gamma_{\Gr2}(\Gr 2(X))=\Gamma(X)$, and $\Gamma_{\Gr 2}(\Hr 2(X))=\Gamma'(X)$.

Let $\Gr2$ be the resulting graph. Notice that for each $X\in \xset$, $\Hr2(X)$ still has properties (P1) and (P2), it is $\alpha^*$-well linked w.r.t. the new interface vertices, and does not contain any $1$-vertex cuts. The maximum vertex degree did not increase.

\begin{figure}[h]
\scalebox{0.3}{\includegraphics{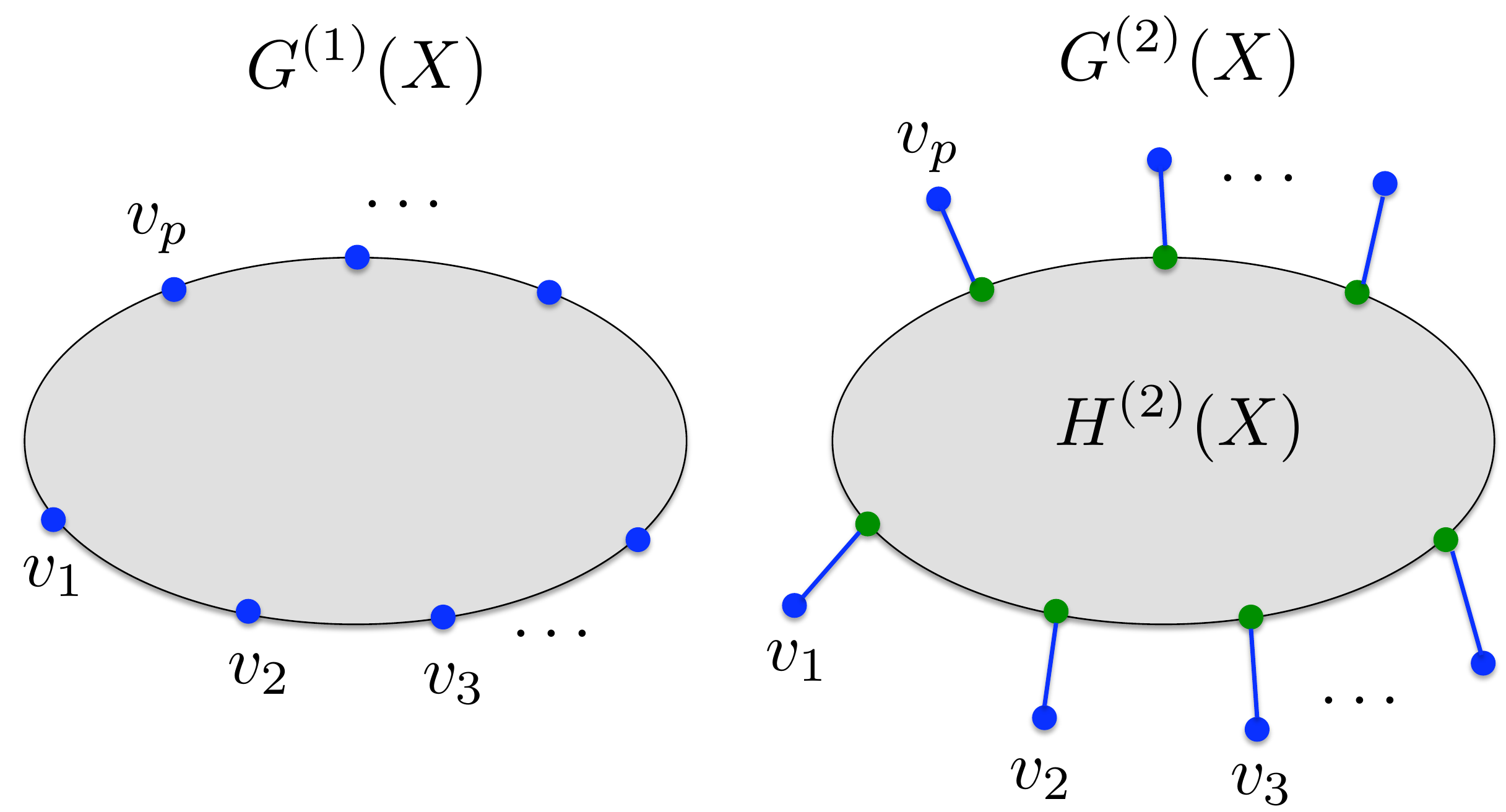}} \caption{Step 2: (a) before, (b) after. Old interface vertices and matching edges are blue, new interface vertices are green.} \label{fig: Step 2}
\end{figure}

Notice that graph $\Gr2(X)$ looks very similar to $Z'_X$: namely, $\Hr2(X)$ is $\alpha^*$-well-linked w.r.t. the new interface vertices, and it has properties (P1) and (P2). Moreover, if we replace $\Hr2(X)$ with the grid $Z_X$, we will obtain precisely $Z'_X$. In particular, if we prove that there is a good drawing of the graph $\Gr2$, such that for each $X\in \xset$, only the matching edges of $\Gr2(X)$ participate in crossings, this would imply the existence of the required canonical drawing for the contracted graph $H$. This is indeed what we do in the rest of the proof. For now, we need to prove the following lemma.

\begin{lemma}
$\opt_{cr}(\Gr2)\leq \frac{20 \dmax^3}{\alpha^*}\cdot \opt_{cr}(\Gr1)\leq \frac{20 \dmax^3}{\alpha^*}\cdot \opt_{cr}(\G)$. 
\end{lemma}
\begin{proof}
Let $\phi_1$ be the optimal drawing of $\Gr1$, and let $\pi'_X$, for $X\in\xset'$, be the planar embedding of $\Gr1(X)$, with all interface vertices lying on the boundary of the outer face $\fout$. (Recall that $\pi'_X$ may be different from the drawing of $X$ induced by $\phi_1$). Notice that $\Gr1(X)$ is $2$-vertex connected, and so the boundary of every face in $\pi'_X$ is a simple cycle.

\begin{claim}\label{claim:x}
Let $F$ be any face, other than the outer face $\fout$, in the drawing $\pi'_X$ of $\Gr1(X)$. Then there are at most $5\dmax/\alpha^*$ interface vertices of $\Gamma(X)$ on the boundary of $F$ in $\pi_X'$.
\end{claim}

\begin{proof}
Assume otherwise, and let $\Gamma^*\sse \Gamma(X)$ be the set of interface vertices lying on the boundary of some face $F\neq \fout$ of $\pi'_X$. Denote $\Gamma^*=\set{v_1,\ldots,v_p}$, where $p>5\dmax/ \alpha^*$. Let $\gamma$ be the boundary of the outer face of $\pi_X'$, and let $\gamma'$ be the boundary of $F$. Then all vertices of $\Gamma^*$ lie on $\gamma\cap \gamma'$. Assume that they appear in the order $(v_1,\ldots,v_p)$ on $\gamma$. Then the two vertices $v_1,v_{\lceil p/2\rceil}$ are a $2$-vertex cut in $\Gr1(X)$, that separate two sets of more than $2\dmax/\alpha^*$ of interface vertices from each other (see Figure~\ref{fig:claim}). This is impossible since the degree of each vertex is at most $\dmax$, and the interface vertices are $\alpha^*$ well-linked in $\Gr1(X)$.

\begin{figure}[h]
\scalebox{0.3}{\rotatebox{0}{\includegraphics{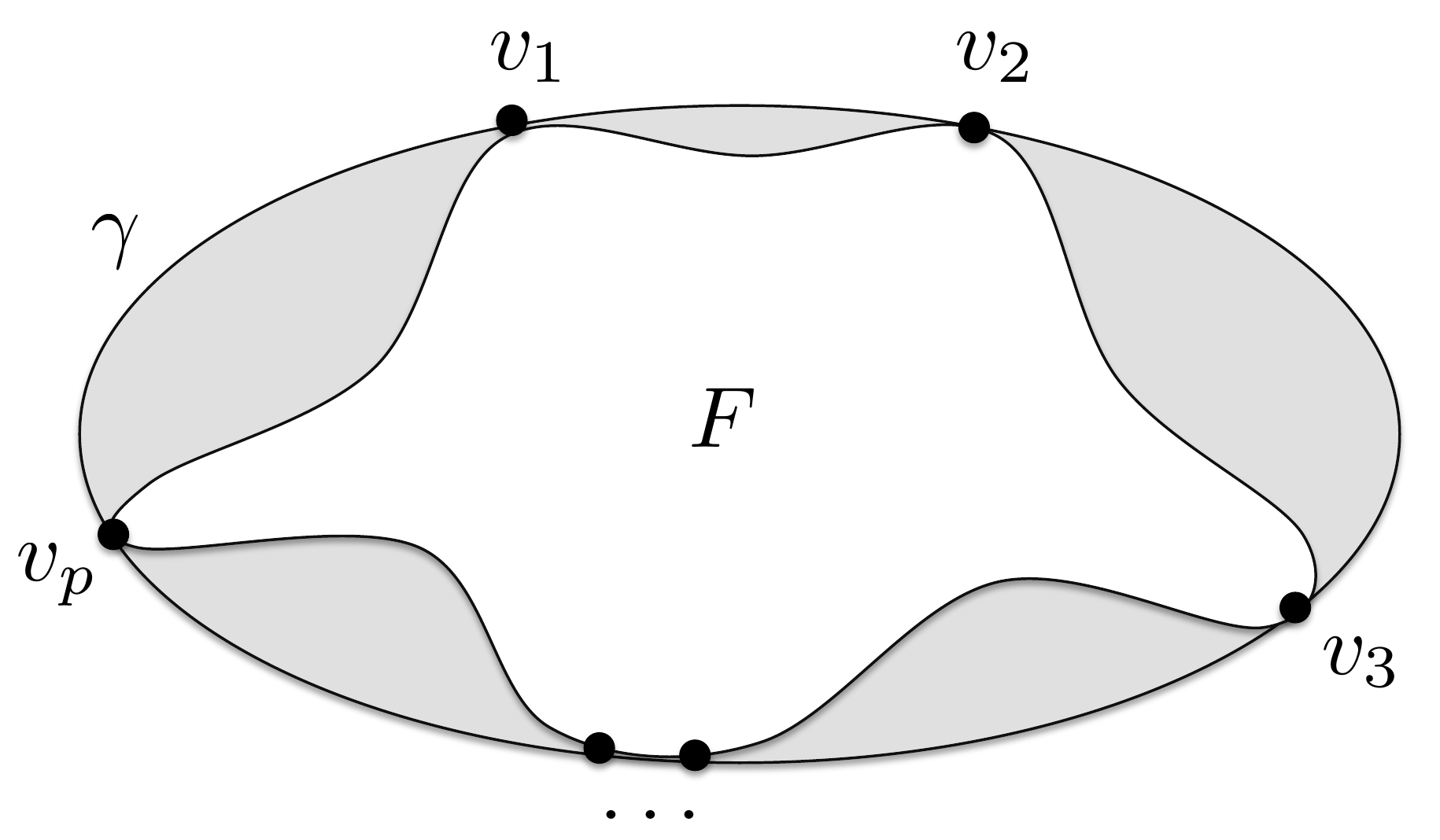}}} \caption{Illustration for Claim~\ref{claim:x}.\label{fig:claim}}
\end{figure}

\end{proof}

Consider now the optimal drawing $\phi_1$ of $\Gr1$, fix some $X\in \xset'$, and let $u\in \Gamma(X)$. We partition the edges incident on $u$ in $\Gr1$ into two subsets: $E_u'$ are the edges that belong to $\Gr1(X)$, and $E_u''$ are the remaining edges. If the images of edges $E_u'$, as they enter vertex $u$, appear consecutively in $\phi_1$, then splitting $u$ into two vertices $u,u_X$ and adding the edge $(u,u_X)$  does not create any new crossings. If the edges in $E_u'$ do not appear consecutively in $\phi_1$, then we add $u$ to a set $V^*(X)$. In this case we may have to pay up to $\dmax^2$ for splitting $u$ into two vertices. It is therefore enough to show that for all $X\in \xset'$:

\[|V^*(X)|\leq \frac{10\dmax}{\alpha^*}\cro_{\phi_1}(E(\Gr1(X)),E(\Gr1))\]

Consider again some vertex $u\in V^*(X)$. Since the edges of $E_u'$ do not appear consecutively in $\phi_1$, there must be two edges $e_1,e_1'\in E'_u$, and two edges $e_2,e_2'\in E_u''$, such that their ordering is $(e_1,e_2,e_1',e_2')$ in $\phi_1$, and moreover, there is a face $F\neq \fout$ in $\pi'_X$, such that $e_1,e_1'$ lie on the boundary of this face. Let $e_2=(u,v)$, $e_2'=(u,v')$. Since graph $\Gr1(X)$ has property (P1), there must be a path $P$ connecting $v$ to $v'$ in $\Gr1\setminus \Gr1(X)$, and since the boundary of $F$ is a simple cycle, denoted by $C$, the image of path $P$ in $\phi_1$ must cross the image of at least one edge of $C$. We charge this crossing for $u$. Since the face $F$ may only contain at most $5\dmax/\alpha^*$ interface vertices, and every edge of $\Gr1(X)$ participates in at most two faces of $\pi'_X$, each such crossing will be charged at most $10\dmax/\alpha^*$ times for $X$.
\end{proof}

\paragraph{Step 3: getting rid of $2$-cuts}

Fix some $X\in \xset'$, and consider the graph $\Hr2(X)$. We will construct a graph $\Hr3(X)$, by deleting some edges and vertices from $\Hr2(X)$, and contracting some $2$-paths. Graph $\Gr3(X)$ is then obtained by adding the matching edges back to $\Hr3(X)$ (or, equivalently, replacing $\Hr2(X)$ with $\Hr3(X)$ in $\Gr2(X)$). This will give the graph $\Gr3$.

\begin{definition} Let $H'$ be any graph and $T$ any subset of vertices of $H'$. We say that $H'$ has property (P3) w.r.t. $T$ and parameter $\beta$, iff there is a flow $\fset$ in $H'$, in which every pair $v,v'\in T$ sends one flow unit to each other, and the congestion on every vertex is at most $|T|\cdot \beta$.
\end{definition}

The two notions, of well-linkedness and property (P3), are closely related to each other.
Specifically, from Observation~\ref{observation: existence of flow in well-linked instance},
if $H'$ is any graph, and $S$ is an $\alpha$-well-linked set in $H'$, then $H'[S]$ has property (P3) w.r.t. $\Gamma(S)$, with parameter $\betaFCG\cdot \dmax/\alpha$. On the other hand, if $H'[S]$ has property (P3) w.r.t. $\Gamma(S)$ and some parameter $\beta$, then $S$ is $1/(2\beta\dmax)$ well-linked. Indeed, consider any partition $(S_1,S_2)$ of $S$, and let $T_1=\out(S)\cap \out(S_1)$, and $T_2=\out(S)\cap \out(S_2)$. Similarly, let $\Gamma_1=\Gamma(S)\cap \Gamma(S_1)$, and $\Gamma_2=\Gamma(S)\cap \Gamma(S_2)$. Assume w.l.o.g. that $|\Gamma_1|\leq |\Gamma_2|$. Then $|\Gamma_1|\geq |T_1|/\dmax$, and from Property (P3), we can send at least $|\Gamma_1|\cdot |\Gamma_2|$ flow units between the two sets, with congestion at most $|\Gamma(S)|\beta$ on vertices and edges. Therefore, $|E(S_1,S_2)|\geq \frac{|\Gamma_1|\cdot |\Gamma_2|}{|\Gamma(S)|\beta}\geq \frac{|T_1|}{2\beta\dmax}$. Using the same reasoning, it is easy to show that for the special case where every vertex in $\Gamma(S)$ is incident on exactly one edge in $\out(S)$, and $H'[S]$ has property (P3) w.r.t. $\Gamma(S)$ and parameter $\beta$, set $S$ is $1/(2\beta)$-well linked.

So far we have only worked with well-linkedness, but for the analysis of Step 3, property (P3) is more convenient.

As $\Hr2(X)$ is $\alpha^*$-well-linked, it also has property (P3) w.r.t. $\Gamma'(X)$ with parameter $\beta^*= \beta_{FCG}\cdot \dmax/\alpha^*=O(\log^{5/2}n\cdot \log\log n\cdot \dmax)$. We will ensure that throughout this step, the graph will retain this property, and so will the final graph $\Hr3(X)$. This in turn will imply that $\Hr3(X)$ is well-linked. 

\begin{definition} Given any graph $H'$, let $(u,v)$ be any $2$-vertex cut in $H'$, and let $C_1,\ldots,C_{q'}$ be the connected components of $H'\setminus \set{u,v}$. For each $C_r$, $1\leq r\leq q'$, let $C'_r$ be the sub-graph of $H'$ induced by $V(C_r)\cup\set{u,v}$. We then define $\cset_{u,v}=\set{C'_1,\ldots,C'_{q'}}$. For each $1\leq r\leq q'$, the vertices of $V(C'_r)\setminus\set{u,v}$ are called the \emph{inner vertices} of $C'_r$.
\end{definition}

Notice that if $H'$ has property (P3) w.r.t. some set $S$ and parameter $\beta^*$, and $|S|\geq 12\beta^*$, then there can be at most one cluster $C\in \cset_{u,v}$ that contains more than $4\beta^*$ vertices of $S$: assume otherwise, and let $C,C'\in \cset_{u,v}$ be two clusters, containing more than $4\beta^*$ vertices of $S$ each (notice that it is possible that $u,v\in S$). Assume w.l.o.g. that $|C'\cap S|\leq |C\cap S|$, so that $|S\setminus C|\geq (|S|-2)/2$. Then because of Property (P3), the total amount of flow leaving the cluster $C'$ is at least $|C'\cap S|\cdot \frac{|S|-2} 2\geq (4\beta^*+1)\frac{|S|-2} 2= 2\beta^*|S|+\frac {|S|}2-4\beta-1>2\beta^*|S|$, and all these flow-paths have to contain either $u$ or $v$. This is impossible, since the load on each one of these vertices is at most $\beta^*|S|$.
We call the unique cluster $C\in {\cal C}_{u,v}$ with $|C\cap S|>4\beta^*$ \emph{the main cluster} of $\cset_{u,v}$ (if such cluster exists).

We start with the following simple operation on the graph $\Hr2(X)$: while there is a $2$-cut $(u,v)$, with some cluster $C\in\cset_{u,v}$, such that $C$ does not contain vertices of $\Gamma'(X)$ as inner vertices, we replace $C$ with an edge $(u,v)$. Notice that this procedure has no influence on properties (P1),(P2) and (P3). Let $H'$ be the resulting graph. Then $H'$ has properties (P1), (P2) and property (P3), w.r.t. $\Gamma'(X)$ and parameter $\beta^*$. Let $\pi_{H'}$ denote the planar drawing of $H'$, in which all vertices of $\Gamma'(X)$ lie on the boundary $\gamma$ of the outer face. Then for any $2$-separator $(u,v)$ in $H'$, each cluster $C\in \cset_{u,v}$ must contain at least one vertex $x\in\Gamma'(X)$ as its inner vertex (that is, $x\neq u,v$). Since all vertices in $\Gamma'(X)$ lie on the boundary $\gamma$ of the outer face of $\pi_{H'}$, it follows that $\cset_{u,v}$ may only contain two clusters, $C_{u,v},C'_{u,v}$, and both vertices $u$ and $v$ have to lie on $\gamma$. Assume that $|\Gamma'(X)|\geq 12\beta^*$, and assume w.l.o.g. that $C_{u,v}$ is the main cluster of $\cset_{u,v}$. Let $\Gamma'_{u,v}=\Gamma'(X)\cap C'_{u,v}$. Recall that $|\Gamma'_{u,v}|\leq 4\beta^*$. We say that a $2$-cut $(u,v)$ is \emph{maximal}, iff $\Gamma'_{u,v}$ is not contained in any other set $\Gamma'_{u',v'}$ for any other $2$-separator $(u',v')$ of $H'$. Let $\rset$ denote the set of all maximal clusters $C'_{u,v}$. We now replace each cluster $C'_{u,v}\in \rset$ with a path $Q_{u,v}$, whose only vertices are $\Gamma'_{u,v}$, and they appear on $Q_{u,v}$ in exactly the same order as on the boundary $\gamma$ of the outer face of $\pi_{H'}$ (see Figure~\ref{fig: Step 3}).
Observe that each such path $Q_{u,v}$ must contain at least one vertex of $\Gamma'(X)$ as its inner vertex, and so all vertices of $Q_{u,v}$ lie on $\gamma$.
We denote by $M'=\set{(u,v)\mid C'_{u,v}\in \rset}$, and $S^{(2)}=\set{v\mid\exists u: (u,v)\in M'}$. Let $H''$ be the resulting graph. We do not allow $H''$ to contain parallel edges, and if such edges have been introduced, we simply remove them. This procedure will not affect property (P3), since the congestion is measured on vertices. 

For the case where $|\Gamma'(X)|<12\beta^*$, we let $H''$ be simply a cycle, whose only vertices are the vertices of $\Gamma'(X)$, that appear on the cycle in the same order as on the boundary of the outer face of $\pi_{H'}$.

\begin{figure}[h]
\scalebox{0.3}{\includegraphics{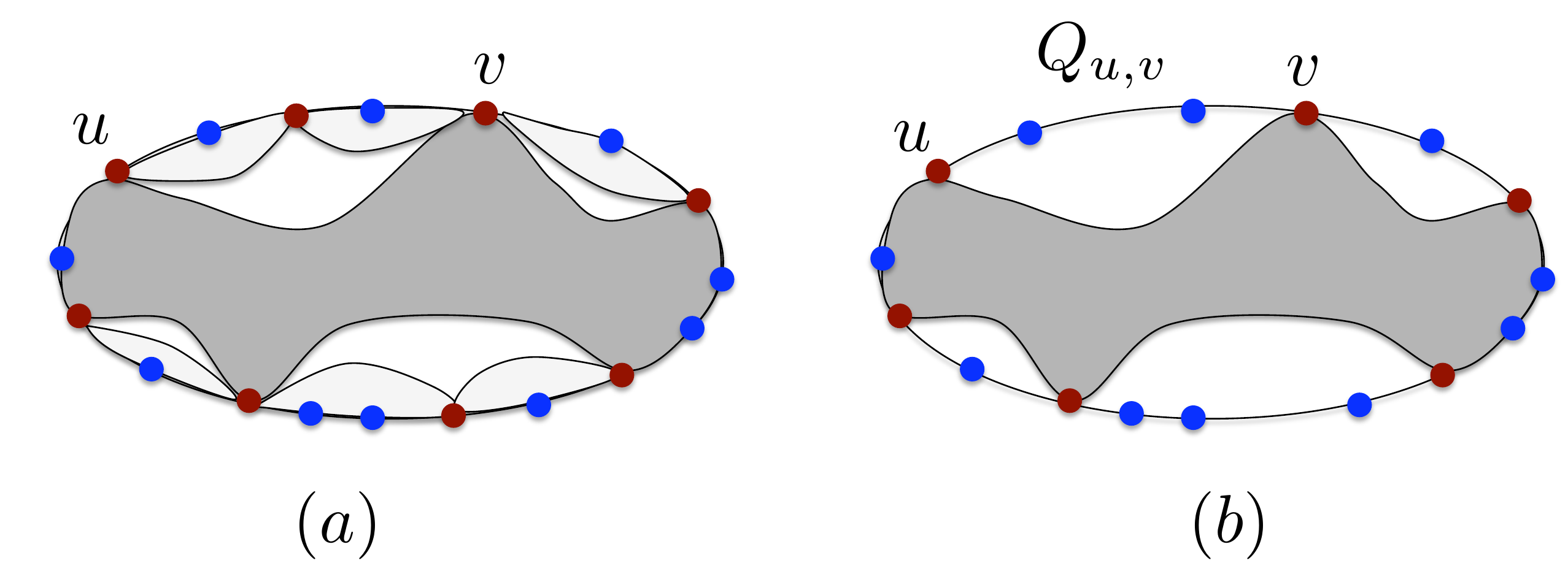}} \caption{Step 3: (a) before, (b) after. New interface vertices are blue, 2-separators are red.} \label{fig: Step 3}
\end{figure}

\begin{claim}\label{claim: graph H3 has property P3} Graph $H''$ has property (P3) w.r.t. $\Gamma'(X)$, and parameter $12\beta^*$. Therefore, it is $1/(24\beta^*)$-well-linked.
\end{claim}
\begin{proof}
If $|\Gamma'(X)|\leq 12\beta^*$, the claim is clearly true as $H''$ is connected. We therefore assume that $|\Gamma'(X)|>12\beta^*$.
Consider any cluster $C'_{u,v}\in \rset$. Since $H'$ had property (P3) with parameter $\beta^*$, the total amount of flow in $\fset$, routed on flow-paths  that contain either $u$ or $v$, and visit inner vertices of $C'_{u,v}$ is  at most $2\beta^*|\Gamma'(X)|$. All such flow-paths can be re-routed via the path $Q_{u,v}$, causing total load of at most $2\beta^*|\Gamma'(X)|$ on the vertices of $Q_{u,v}$. Additionally, for every pair $(t,t')\in \Gamma'_{u,v}$ sends one flow unit to each other. We again re-route all such flow-paths along $Q_{u,v}$. As $|\Gamma'_{u,v}|\leq 4\beta^*$, the total number of such pairs is at most $16(\beta^*)^2\leq 4\beta^*\cdot |\Gamma'(X)|$, and so the total load on any vertex becomes at most $6\beta^*|\Gamma'(X)|$.
\end{proof}

We now denote $\Hr3(X)=H''$, and $\Gr3(X)$ is the graph obtained from $\Hr3(X)$ after we add all matching edges back to it. Let $\Gr3$ be the resulting whole graph. Notice that $\Hr3(X)$ is obtained from $\Hr2(X)$ by performing a series of steps, where each step either deletes vertices or edges from the graph, or contracts a $2$-path. Such operations do not increase the crossing number of the graph, and so $\optcro{\Gr3}\leq \optcro{\Gr2}\leq  \frac{20 \dmax^3}{\alpha^*}\cdot \opt_{cr}(\G)$.
It is also easy to see that both $\Gr3(X)$ and $\Hr3(X)$ retain properties (P1) and (P2), have no $1$-vertex separators, and we have established above that $\Hr3(X)$ is $1/(12\beta^*)$-well-linked.

In general, this concludes the third step. Notice however that $\Hr3(X)$ still contains $2$-separators, and we will need to deal with them when bounding the number of irregular vertices and edges. We establish a few more structural properties of the graph $H''=\Hr3(X)$, in the next three observations. These properties will be used when bounding the number of irregular $2$-separator vertices in $\Hr3(X)$.

\begin{observation} \label{observation: characterization of 2-cuts} Assume that $|\Gamma'(X)|>12\beta^*$.
Let $(x,y)$ be any $2$-vertex cut in $H''$. Then $x,y\in Q_{u,v}$, for some $(u,v)\in M'$.
\end{observation}
\begin{proof}
Assume otherwise. We consider three cases. The first case is when none of the vertices $x,y$ is an inner vertex on any path $Q_{u,v}$ for any $(u,v)\in M'$. Then $(x,y)$ was a $2$-separator in $H'$. Since $(x,y)\not\in M'$, it was not a maximal cut. This means that either $x$ or $y$ must be an inner vertex in some $C'_{u,v}\in \rset$, and this is a contradiction to $(x,y)$ surviving in $H''$.

The second case is when exactly one of the vertices $x,y$ is an inner vertex on some path $Q_{u,v}$, for some $(u,v)\in M'$. Assume w.l.o.g. that it is $x$. Then $y\neq u,v$, and moreover, either $(y,v)$ is a $2$-cut in $H''$, or there is an edge $(y,v)$ in $H''$, that lies on the boundary $\gamma'$ of the outer face of the drawing of $H''$. Similarly, either $(y,u)$ is a $2$-cut in $H'$, or there is an edge $(y,u)$ in $H''$ that lies on $\gamma'$. If $(y,v)$ is a $2$-cut in $H''$, but $(y,v)\not\in M'$, then we obtain Case 1. Similarly, if $(y,u)$ is a $2$-cut in $H''$, but $(y,u)\not\in M'$, then we obtain Case 1. Therefore, we can assume that for each $z\in \set{u,v}$, either there is an edge $(y,z)$ in $H''$ that belongs to $\gamma'$, or $(y,z)\in M'$. In either case, the segment of $\gamma'$ between $y$ and $z$, which does not contain $x$, can contain at most $4\beta^*$ vertices of $|\Gamma'(X)|$. Since $C'_{u,v}$ also contains at most $4\beta^*$ vertices of $\Gamma'(X)$, this means that $|\Gamma'(X)|\leq 12\beta^*$, a contradiction.

The third case is when $x$ is an inner vertex on some path $Q_{u,v}$ and $y$ is an inner vertex on some path $Q_{u',v'}$, where $(u,v),(u',v')\in M'$. If $(u,v)=(u',v')$, then we again get two clusters that contain all vertices in $\Gamma'(X)$, but have at most $4\beta^*$ vertices of $\Gamma'(X)$ each, contradicting that $|\Gamma'(X)|>12\beta^*$.
So assume w.l.o.g. that $u'\not\in \set{u,v}$. Then $(u',x)$ is a $2$-cut in $H''$, and we obtain the second case.
\end{proof}

Let $\psi$ be the unique planar drawing of $H''$, in which the vertices of $\Gamma'(X)$ lie on the boundary $\gamma$ of the outer face $F_{out}$. Since $H''$ is $2$-vertex connected, the boundary of every face of $\psi$ is a simple cycle. Observe that each path $Q_{u,v}$, for $(u,v)\in M'$ appears consecutively on $\gamma$. We therefore have two types of internal faces in the drawing $\psi$: a face of the first type contains some path $Q_{u,v}$ for $(u,v)\in M'$ on its boundary, and a face of the second type does not contain any such path.

\begin{observation}\label{observation: 2-cuts on face} Assume that $|\Gamma'(X)|>12\beta^*$,
let $F\neq  F_{out}$ be any face of $\psi$, and let $\gamma_F$ be its boundary.  If $F$ is of the first type, then  the number of vertices on $\gamma(F)\cap \gamma$ is at most $8\beta^*$, and if $F$ is of the second type, then the number of vertices on $\gamma\cap \gamma_F$ is at most $3$.
\end{observation}
\begin{proof}
It is easy to see that if $x,y\in \gamma\cap \gamma_F$, then $(x,y)$ is a $2$-separator in $H''$, unless there is an edge $(x,y)$ in $H''$ that belongs to $\gamma$ (see Figure~\ref{fig:claim}). 
Assume first that $F$ is of the first type. Let $Q_{u,v}$ be the path it contains, for $(u,v)\in M'$, let $y\in \Gamma'(X)$ be any {\bf inner} vertex of $Q_{u,v}$ (which must exist for all $Q_{u,v}$), and let $x\not\in Q_{u,v}$ be any additional vertex on $\gamma(F)\cap \gamma$. Then $(x,y)$ is a $2$-separator in $H''$, and yet they do not belong to the same path $Q_{u',v'}$ for $(u',v')\in M'$, which is impossible from Observation~\ref{observation: characterization of 2-cuts}.

Assume now that $F$ is of the second type.
Suppose there are four vertices $v_1,\ldots,v_4$ in $\gamma\cap \gamma_F$, and
assume that these vertices appear on $\gamma_F$ in this order.  
Every pair of vertices $(v_i,v_j)$ that do not appear consecutively in this order is a $2$-separator for $H'$, and hence must belong to $M'$, which is impossible, because then we would get two paths $Q_{v_1,v_3}$ and $Q_{v_2,v_4}$, that must both be contained in $\gamma$.
\end{proof}

\begin{observation}\label{observation: num fo 2-cuts between subterminals} Assume that $|\Gamma'(X)|>12\beta^*$,
let $t,t'\in \Gamma'(X)$ be a pair of interface vertices, such that one of the two segments $\sigma$ of $\gamma$, connecting $t$ and $t'$, does not contain any other vertices of $\Gamma'(X)$. Then $|\sigma\cap (S^{(2)}\setminus \set{t,t'})|\leq 2$.
\end{observation}
\begin{proof}
Assume otherwise, and let $v_1,v_2,v_3$ be three vertices of $S^{(2)}\setminus\set{t,t'}$ lying on $\sigma$ in this order. We claim that either $(v_1,v_2)\in M'$, or $(v_2,v_3)\in M'$: otherwise, from the definition of $S^{(2)}$, there must be some other vertex $v'\in S^{(2)}$, such that $(v_2,v')\in M'$, $v'\neq v_1,v_2$. This is impossible, because then we would have a path $Q_{v_2,v'}\sse\gamma$. Assume w.l.o.g. that $(v_1,v_2)\in M'$. But then $\sigma$ must contain $Q_{v_1,v_2}$, which in turn must contain vertices of $\Gamma'(X)$, a contradiction.
\end{proof}

\subsubsection{Part 2: Irregular Vertices and Edges}
 We start by defining irregular vertices and edges, and bounding their number as in ~\cite{CMS10}. We deal with $2$-separators later.
Suppose we are given any graph $G$, and a pair $\phi,\psi$ of drawings of $G$.
\begin{definition}
We say that a vertex $x$ of $G$ is irregular iff its degree is more than $2$, 
and the circular ordering of the edges incident on it, as their images enter $x$, is different in $\phi$ and $\psi$ (ignoring the orientation). 
We denote the set of irregular vertices by $\bad_V(\phi, \psi)$, and we call all other vertices \emph{regular}. \end{definition}

\begin{definition}
For any pair $(x,y)$ of vertices in $G$,
we say that a path $P$, connecting $x$ to $y$ in $G$ is irregular iff $x$ and $y$ have degree at least $3$, all other
vertices on $P$ have degree $2$ in $G$, vertices $x$ and $y$ are {\bf regular}, but their
orientations differ in $\phi$ and $\psi$. That is, the orderings of the edges adjacent to $x$ and to $y$, as their images enter these vertices, are identical in both drawings, but the pairwise orientations are different: for one of the two vertices, the orientations are identical in both drawings (say clock-wise), while for the other vertex, the orientations are opposite (one is clock-wise, and the other is counter-clock-wise). An edge $e$ is an irregular edge iff it is the first or the last edge on an irregular path. In particular, if the irregular path only consists of edge $e$, then $e$ is an irregular edge.
We denote the set of irregular edges by $\bad_E(\phi, \psi)$, and all other edges are called regular.
\end{definition}

\begin{lemma}\label{lem:bad3} (\cite{CMS10})
Let $G$ be any planar $2$-vertex connected graph, let $S_2$ be the set of vertices participating in $2$-vertex separators in $G$, and $E_2$ the set of edges adjacent to the vertices of $S_2$. Let $\phi$ be an arbitrary drawing of $G$ and $\psi$ be a planar drawing of $G$.
Then
\[
|\bad_V(\psi, \phi)\setminus S_2| + |\bad_E(\psi, \phi)\setminus E_2| = O(\cro_{\phi}(G)).
\]
\end{lemma}

We now fix some set $X\in \xset'$, and we define the set of irregular vertices and edges for $X$. Recall that both $\Hr3(X)$ and $\Gr3(X)$ have properties (P1) and (P2) w.r.t. $\Gamma'(X)$ and $\Gamma(X)$ respectively, and there is a set of matching edges in $\Gr3(X)$ connecting $\Gamma'(X)$ to $\Gamma(X)$. Let $\psi_X$ denote the unique planar drawing of $\Hr3(X)$, in which the vertices of $\Gamma'(X)$ lie on the boundary of the outer face $F_{out}$, and let $\psi^+_X$ denote the extension of $\psi_X$ to include the drawing of the matching edges inside $F_{out}$. That is, $\psi^+_X$ is a planar drawing of $\Gr3(X)$, with all vertices of $\Gamma(X)$ lying on the boundary of the outer face. 
Let $\phi$ be the optimal drawing of $\Gr3$, and let $\phi_X,\phi_X^+$ be the drawings of the graphs $\Hr3(X)$ and $\Gr3(X)$, respectively, induced by $\phi$. From now on, we denote by $F_{out}$ the outer face of $\psi_X$, and by $\gamma$ its boundary.

We are now ready to define the set of irregular vertices and edges for $X$. We say that a vertex $v\in V(\Hr3(X))$ is an irregular vertex, iff $v\in \bad_V(\phi^+_X,\psi^+_X)$. Notice that we require that $v\in V(\Hr3(X))$, that is, $v\not\in \Gamma(X)$, but if $v\in \Gamma'(X)$, then we need to take its matching edge into consideration, so the set of irregular vertices is defined w.r.t. the extended drawings $\phi^+_X$ and $\psi^+_X$. Let $\bad_V(X)$ denote the set of all irregular vertices for $X$. Similarly, we say that an edge $e\in E(\Hr3(X))$ is an irregular edge, iff $e\in \bad_E(\phi^+_X,\psi^+_X)$. Again, we do not include the matching edges in the set of irregular edges, but we define the irregular edges w.r.t. the extended drawings $\phi^+_X$, $\psi^+_X$. Let $\bad_E(X)$ denote the set of all irregular edges for $X$. We bound the number of irregular vertices and edges for $X$ in the next two lemmas.

\begin{lemma}\label{lemma: bad vertices for X}
For each $X\in \xset'$, $|\bad_V(X)|\leq O(\beta^*)\cro_{\phi}(\Gr3(X),\Gr3)$.
\end{lemma}
\begin{proof}
We fix some $X\in \xset'$. The vertices of $\Hr3(X)$ can be partitioned into three types.

The first type is the vertices that do not participate in any $2$-separators in $\Hr3(X)$, and do not belong to $\Gamma'(X)$. If $v$ is such a vertex, and $v\in \bad_V(X)$, then $v\in \bad_V(\phi_X,\psi_X)$ must hold. Since $\Hr3(X)$ does not contain any $1$-vertex cuts, by Lemma~\ref{lem:bad3}, the number  of such irregular vertices is bounded by $O(\cro_{\phi_X}(\Gr3(X),\Gr3(X)))\leq O(\cro_{\phi}(\Gr3(X),\Gr3))$.

The second type is vertices that serve as $2$-separators in $\Hr3(X)$, but do not belong to $\Gamma'(X)$. Recall that all such vertices belong to $S^{(2)}$, and for each such vertex $v$, there is a vertex $u\in S^{(2)}$, such that $(u,v)\in M'$. 

Consider the boundary $\gamma$ of the outer face $F_{\out}$ of the planar drawing $\psi_X$, and recall that $Q_{u,v}\sse \gamma$. Let $x, y\in \Gamma'(X)$ be the two new interface vertices lying closest to $v$, on both sides of $v$ on $\gamma$, with $y\in Q_{u,v}$. Then no other new terminal vertices appear between $x$ and $y$ on $\gamma$.
Since $\Hr3(X)$ has property (P1), there is a path $P$ connecting $x$ to $y$ in $\Gr3\setminus \Hr3(X)$. Let $\tilde{\psi}$ be the drawing $\psi_X$ of $\Hr3(X)$, together with the path $P$, which is drawn inside the outer face $F_{out}$ of $\psi_X$. Then $\tilde{\psi}$ is a planar drawing of $\Hr3(X)\cup P$. Path $P$ has split $F_{out}$ into two subfaces, and we denote by $F_v$ the sub-face containing $v$ (see Figure~\ref{fig: bad 2-sep}).

\begin{figure}[h]
\scalebox{0.4}{\rotatebox{0}{\includegraphics{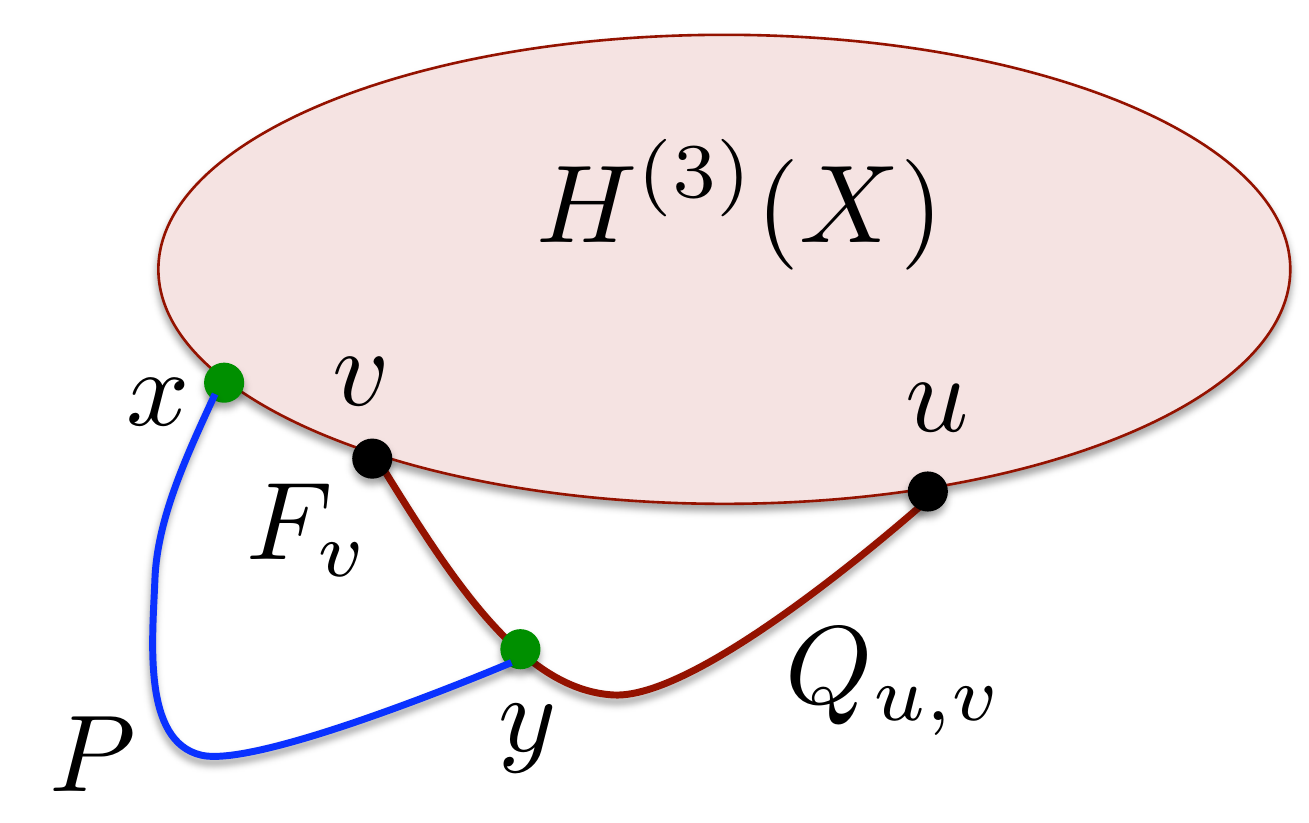}}} \caption{Illustration for Lemma~\ref{lemma: bad vertices for X}.\label{fig: bad 2-sep}}
\end{figure}

 Let $W$ denote the set of vertices and edges lying on the boundary of all faces $F$ of $\tilde{\psi}$, such that $v\in \gamma(F)$. Notice that except for $F_v$, all other such faces are proper faces of $\psi_X$, that are distinct from $F_{out}$.  Graph $W$ is homeomorphic to a wheel, with path $P$ being one of the edges of the wheel, that are not adjacent to $v$. Therefore, $W$ has a unique planar drawing, which is identical to the drawing induced by $\tilde{\psi}$. It is easy to see that if $v$ is an irregular vertex, then the edges of $W$ must cross. Moreover, at least one crossing has to involve some edge $e$ of $\Hr3(X)\cap W$ (that is, $e\not\in P$). We charge this crossing of $e$ for $v$.

We now claim that for each edge $e\in E(\Hr3(X))$, each crossing in which $e$ participates is charged at most $O(\beta^*)$ times. 
First, if $|\Gamma'(X)|\leq 12\beta^*$, then $|S^{(2)}|\leq 2|\Gamma'(X)|\leq 24\beta^*$ must hold (because of the paths $Q_{u,v}$ connecting every pair $(u,v)\in M'$ and containing vertices of $\Gamma'(X)$), and then $e$ may only be charged at most $24\beta^*$ times. We now assume that $|\Gamma'|>12\beta^*$.

Let $F$ be a face of $\psi_X$ on whose boundary $\gamma(F)$ edge $e$ lies, and assume first that $F\neq F_{out}$. Edge $e$ can only be charged for those vertices of $\gamma(F)$, that belong to $S^{(2)}$. Since all such vertices must lie on $\gamma\cap \gamma_F$, by Observation~\ref{observation: 2-cuts on face}, there can be at most $8\beta^*$ such vertices. Assume now that $F=F_{out}$. Then $e$ can only be charged for vertices $v\in S^{(2)}$ if $e\in \gamma(F_v)$. If $x$ and $y$ denote the vertices of $\Gamma'(X)$ lying immediately to the left and to the right of $e$ on $\gamma$, then $v$ must lie between $x$ and $y$ for this to happen. From Observation~\ref{observation: num fo 2-cuts between subterminals}, there are at most $2$ such vertices $v\in \sset^{(2)}$. In total, taking into account both faces on whose boundary $e$ lies, we get that $e$ can be charged at most $O(\beta^*)$ times.
Therefore, the number of irregular vertices of this type is bounded by $O(\beta^*)\cro_{\phi}(\Gr3(X),\Gr3)$.

Finally, the third type is the new interface vertices of $\Gamma'(X)$. Fix one such vertex $t$, and let $t'\in \Gamma(X)$ be its corresponding old interface vertex. Let $t_l$ and $t_r$ be the two new interface vertices lying immediately to the left and to the right of $t$ on $\gamma$, and let $t_l'$ and $t_r'$ be their old interface vertices, respectively. Since all vertices of $\Gamma(X)$ are connected in $\Gr3\setminus\Gr3(X)$, there are two paths: $P$ connecting 
$t_l'$ to $t'$, and $P'$ connecting $t_r'$ to $t'$ in $\Gr3\setminus\Gr3(X)$. We can choose $P$ and $P'$, so that the vertices that they share form one consecutive segment on both paths (See Figure~\ref{fig: bad 2-sep-2}).
We extend the two paths by using the matching edges, so that $P$ connects  $t_l$ to $t$, and $P'$ connects $t_r$  to $t$. 
Let $\tilde{\psi}$ be the planar drawing of $\psi_X\cup P\cup P'$, obtained by adding the drawings of $P$ and $P'$ inside the outer face of $\psi_X$, so that they do not cross.
Observe that the drawings of $P$ and $P'$ partition $F_{out}$ into three sub-faces. We denote by $F_t$ and $F'_t$ the two sub-faces whose boundaries contain $t$.

\begin{figure}[h]
\scalebox{0.4}{\rotatebox{0}{\includegraphics{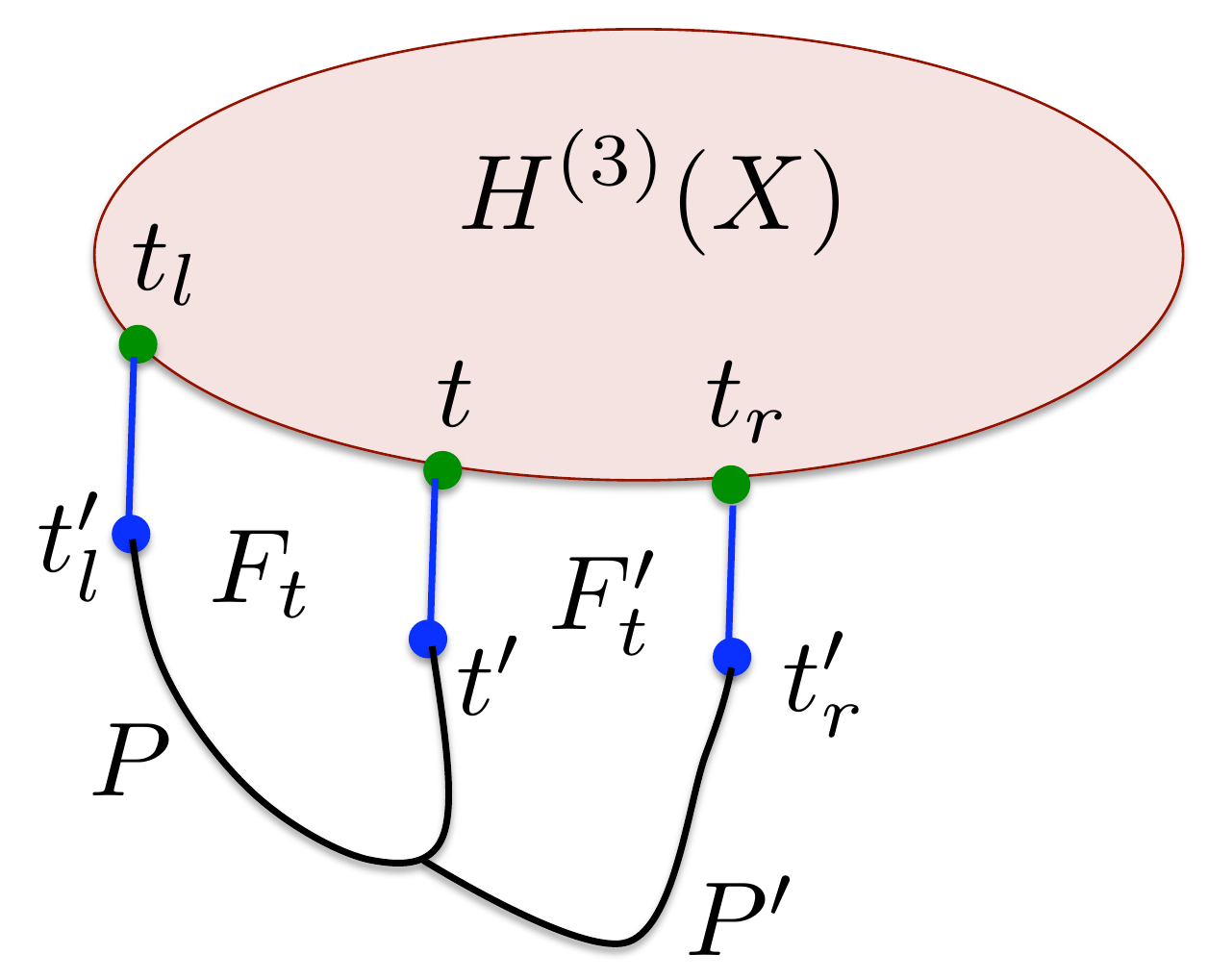}}} \caption{Illustration for ~\ref{lemma: bad vertices for X}.\label{fig: bad 2-sep-2}}
\end{figure}

 Let $W$ be the set of vertices and edges lying on the boundaries of the faces $F$ of $\tilde{\psi}$, such that $t\in \gamma(F)$. Notice that except for $F_t,F_t'$, all such faces are proper faces of $\psi_X$, distinct from $F_{out}$. Again, $W$ is homeomorphic to the wheel graph, and in any embedding where $t$ is irregular, a pair of edges $(e,e')$ must cross, where $e$ belongs to $\Hr3(X)\cap W$. We charge this crossing of $e$ for $t$. We now need to show that every crossing of every edge is charged $O(\beta^*)$ times. Again, if $|\Gamma'(X)|\leq 12\beta^*$, edge $e$ may only be charged $12\beta^*$ times. So we assume that $|\Gamma'(X)|>12\beta^*$. Consider some edge $e$ of $\Hr3(X)$, and let $F$ be a face of $\psi_X$ to which it belongs. Assume first that $F\neq F_{out}$. Then 
we can only charge $e$ for vertices $t$ of $\Gamma'(X)$ that appear on the boundary of $F$, that is, $t\in \gamma(F)\cap \gamma$. From Observation~\ref{observation: 2-cuts on face}, $F$ may contain at most $8\beta^*$ such vertices.  So edge $e$ may be charged at most $8\beta^*$ times for $F$. Finally, if $F=F_{out}$, then $e$ may only be charged for such vertices $t\in \Gamma'(X)$, for which $e$ lies on the boundary of $F_t$ or $F'_t$. Therefore, if $t,t'$ denote the vertices of $\Gamma'(X)$ lying immediately to the left and to the right of $e$ on $\gamma$, then $e$ may only be charged for these vertices as part of $F_{out}$.
 Therefore, the number of irregular vertices of the third type is bounded by $O(\beta^*)\cro_{\phi}(\Gr3(X),\Gr3)$.
\end{proof}

\begin{lemma}\label{lemma: bad edges for X}
For each $X\in \xset'$, $|\bad_E(X)|\leq O(\beta^*)\cro_{\phi}(\Gr3(X),\Gr3)$.
\end{lemma}
\begin{proof}
Fix some $X\in \xset'$.
We partition the edges of $\Hr3(X)$ into two types.
The first type is the edges that are not adjacent to any vertices in $S^{(2)}$ or $\Gamma'(X)$. Since these sets of vertices include all $2$-vertex separators of $\Hr3(X)$, if $e$ is an irregular edge of the first type, then $e\in \bad_E(\phi_X,\psi_X)\setminus E_2$, and by Lemma~\ref{lem:bad3}, their number is bounded by $O(\cro_{\phi_X}(\Hr3(X)))\leq O(\cro_{\phi}(\Hr3(X),\Hr3))$. 

The second type is the edges that are adjacent to vertices of $S^{(2)}\cup \Gamma'(X)$. Fix some such edge $e=(x,v)$, and assume that $x\in S^{(2)}\cup \Gamma'(X)$.  Then $x$ lies on the boundary $\gamma$ of the outer face $F_{out}$ of $\psi_X$, and we find two vertices $t,t'\in \Gamma'(X)$, as follows.
If $x\not\in \Gamma'(X)$, then $t$ and $t'$ are two vertices of $\Gamma'(X)$, lying immediately to the left and to the right of $x$ on $\gamma$ (it is possible that one of these vertices is $v$ itself if $v\in \Gamma'(X)$). Otherwise, if $x\in \Gamma'(X)$, then we let $t=x$. If $e\not\in \gamma$, then $t'$ is the vertex of $\Gamma'(X)$ lying immediately to the left of $x$ on $\gamma$. Finally, if $e\in \gamma$, then $t'$ is the vertex of $\Gamma'(X)$ closest to $v$, such that $v$ lies between $x$ and $t'$ on $\gamma$. Let $P$ be the path connecting $t$ to $t'$ in $\Gr3\setminus \Hr3(X)$, and let $\tilde{\psi}$ be the planar drawing of $\Hr3(X)\cup P$, obtained by adding the drawing of the path $P$ inside the outer face $F_{out}$ of $\psi_X$. This partitions the outer face $F_{out}$ into two sub-faces. If $e\in \gamma$, then we denote by $F_e$ the sub-face whose boundary contains $e$. Otherwise, we let $F_e$ be any one of the two sub-faces.
 
 Let $W$ be the union of the boundaries of the two faces containing the edge $e$ in $\tilde{\psi}$. Notice that while the first face may be $F_e$, the second face, $F'_e$ is a proper face of $\psi$, distinct from $F_{out}$. If $e$ is an irregular edge, then there must be two edges $e',e''$ of $W$, whose images cross in $\phi$, such that $e'\not\in P$. We charge this crossing of $e'$ for $e$. We now need to argue that each crossing of each edge of $\Hr3(X)$ is charged $O(\beta^*)$ times. Let $e'$ be any edge of $\Hr3(X)$, and let $F$ be one of the two faces on whose boundary $e'$ lies in $\psi_X$. Assume first that $F\neq F_{out}$. Then we can only charge $e'$ for edges $e$ adjacent to vertices of $S^{(2)}\cup \Gamma'(X)$ lying on $\gamma(F)$. As all such vertices belong to $\gamma$, by Observation~\ref{observation: 2-cuts on face}, their number is bounded by $O(\beta^*)$ (again, if $|\Gamma'(X)|<12\beta^*$, then $|S^{(2)}|\leq 2|\Gamma'(X)|\leq 24\beta^*$). If $F=F_{out}$, then we can only charge $e'$ for edges $e\in \gamma(F)$, such that $e'\in \gamma(F_e)$. Let $t,t'\in \Gamma'(X)$ be the new interface vertices lying immediately on the left and on the right of $e'$ (it is possible that $t$ or $t'$ are endpoints of $e'$). If some edge $e$ is charged to $e'$, then $e$ must also lie on the same segment of $\gamma$ between $t$ and $t'$, to which $e'$ belongs, and recall that one of the endpoints of $e$ must belong to $S_2$. From Observation~\ref{observation: num fo 2-cuts between subterminals}, the number of such edges is bounded by a constant. Therefore, the  number of times an edge may be charged is bounded by $O(\beta^*)$, and the total number of irregular edges of this type is $O(\beta^*)\cro_{\phi}(\Hr3(X),\Hr3)$.
\end{proof}

For each set $X\in \xset'$, we denote by $N(X)$ the set of all edges that either participate in crossings in the optimal drawing $\phi$ of $\Gr3$, or they belong to $\bad_E(X)$, or they are adjacent to vertices in $\bad_V(X)$. From Lemmas~\ref{lemma: bad vertices for X} and ~\ref{lemma: bad edges for X}, $|N(X)|\leq O(\beta^*\cdot \dmax)\cro_{\phi}(\Hr3(X),\Hr3)$.

\subsubsection{Part 3: finding the drawing}


Let $\phi$ be the optimal drawing of $\Gr3$. Recall that $\cro_{\phi}(\Gr3)\leq  \frac{20 \dmax^3}{\alpha^*}\cdot \opt_{cr}(\G)$.
We find a new drawing $\phi'$ of $\Gr3$, such that, for each $X\in \xset'$, no edges of $\Hr3(X)$ participate in crossings (and the matching edges will participate in crossings instead). We will show that $\cro_{\phi'}(\Gr3)\leq O(\poly(\dmax\cdot \log n\cdot \alpha^*))\cro_{\phi}(\Gr 3)$. After that, it is easy to show that there is a canonical drawing $\phi''$ of the contracted graph $H=\G_{|S}$, with $\cro_{\phi''}(H)\leq \cro_{\phi'}(\Gr3)$. The idea is that we simply replace $\Hr3(X)$ with the grid $Z_X$, for each $X\in \xset$, and since no edges of $\Hr3(X)$ participate in crossings in $\phi'$, and the ordering of the matching edges is identical in the planar drawings of $Z'_X$ and $\Gr3(X)$, we can do this transformation without increasing the number of crossings.

So from now on we can focus on finding such a drawing $\phi'$ of $\Gr3$. 
The following lemma is due to Anastasios Sidiropoulos~\cite{Tasos}. For completeness, we provide a slightly modified proof in Section \ref{sec:ell2}.

\begin{lemma}\label{lem:ell2}
Let $G=(V,E)$ be any $n$-vertex graph, and $S$ any subset of vertices of $G$, such that $G$ has property (P3) for $S$ with some parameter $\beta>0$. Moreover, assume that $G$ is $2$-connected, and it has a planar drawing $\psi$, in which the vertices of $S$ lie on the boundary of the outer face. Let $E'$ be any subset of edges of $G$. Then there is a vertex $v^*\in V$, and a collection $\pset$ of paths in $G$, such that for each $u\in S$, there is a path $P_u\in\pset$ connecting $u$ to $v^*$, and $\sum_{e\in E'}c^2(e)\leq O(\beta^2\cdot \log n\cdot |E'|)$, where $c(e)$ is the number of paths in $\pset$ containing $e$.
\end{lemma}

Recall that from Claim~\ref{claim: graph H3 has property P3}, for each $X\in \xset'$, the graph $\Hr3(X)$ has property (P3) for $\Gamma'(X)$, with parameter $O(\beta^*)$. Fix some $X\in \xset'$, and consider the subset $N(X)$ of edges of $\Hr3(X)$. Using Lemma~\ref{lem:ell2}, we can find a vertex $v^*_X$, and a collection $\pset_X$ of paths in $\Hr3(X)$, such that for each vertex $t\in \Gamma'(X)$, there is a path $P_t\in \pset_X$ connecting $t$ to $v^*_X$ in $\Hr3(X)$, and $\sum_{e\in N(X)}c^2(e)\leq O((\beta^*)^2 \cdot \log n\cdot |N(X)|)$, where $c(e)$ is the number of paths in $\pset_X$ containing $e$.

Consider the planar drawing $\psi_X$ of $\Hr3(X)$, with all vertices of $\Gamma'(X)$ on the boundary $\gamma$ of the outer face. Recall that $\Hr3(X)$ is $2$-connected, so $\gamma$ is a simple cycle. Denote by $\sigma_X$ the ordering of the vertices of $\Gamma'(X)$ along $\gamma$. Observe that $\psi_X$ induces a drawing of the paths in $\pset_X$, and we assume w.l.o.g. that the paths in $\pset_X$ are uncrossed w.r.t. this drawing, so that the paths $\set{P_t}_{t\in \Gamma'(X)}$ arrive at vertex $v^*$ in the same order as in $\sigma_X$ (this can be assumed w.l.o.g. since otherwise we can uncross the paths in $\pset_X$ to ensure this, without increasing the congestion on edges). Therefore, we obtain a planar drawing $\psi'_X$ of the paths in $\pset_X$. For each edge $e$ in $\Hr3(X)$, this drawing induces an ordering $\pi_e$ on all the paths in $\pset_X$ containing $e$. For each vertex $v$ of $\Hr3(X)$, this drawing induces a ``local planar drawing'' $\pi_v$ of all paths in $\pset_X$ going through vertex $v$. Consider the graph $H^*_X=\Hr3(X)\setminus N(X)$, and let $C_X$ be the connected component of this graph, containing $v^*$. Then the drawing of $C_X$ induced by $\phi$ is exactly the same as the drawing of $C_X$ induced by $\psi_X$, because $C_X$ does not contain irregular vertices or edges, or edges participating in crossings in $\phi$. 
Let $E'$ be the set of edges incident on $C_X$. For each edge $e\in E'$, let $T(e)\sse \Gamma'(X)$ be the set of vertices $t\in \Gamma'(X)$, such that $P'_{t'}$ contains $e$. If the edges in $E'$ appear in the order $e_1,e_2,\ldots,e_k$ along the boundary of the drawing of $C_X\cup E'$, induced by $\psi_X$, then each set $T(e_i)$ of vertices appears consecutively in $\sigma_X$, and the vertices in sets $T(e_1),T(e_2),\ldots,T(e_k)$ are ordered correctly  between the sets (w.r.t. $\sigma_X$).
For each path $P_t\in \pset_X$, let $P'_t$ be the portion of the path between $t$ and $C_X$. The ordering in which paths $P'_t$ hit $C_X$ in the drawing $\psi'_X$ is exactly the same as $\sigma_X$.

We are now ready to transform the drawing $\phi$. Start with the drawing $\phi$ of $\Gr3$. For each $X\in \xset'$, draw a closed curve $\gamma(X)$ around the drawing of $C_X$ in $\phi$. Since $C_X$ is a connected graph, whose edges do not participate in crossings in $\phi$, this can be done so that no other vertices or edges of $\Gr3$ appear inside $\gamma(X)$. Erase all edges and vertices of $\Hr3(X)$ from $\phi$, and instead place the embedding $\psi_X$ of $\Hr3(X)$ inside the curve $\gamma_X$, with the images of the vertices $\Gamma'(X)$ lying on $\gamma_X$. Let $e=(t,t')$ be any matching edge of $\Gr3(X)$, where $t\in \Gamma(X)$, $t'\in \Gamma'(X)$. Recall that $t$ does not belong to any other graph $\Hr3(X')$, for $X'\in \xset'$ (because of Step 2, where we have introduced new interface vertices). Therefore, the image of the vertex $t$ remains the same as in $\phi$, while the image of $t'$ now lies on $\gamma_X$. Let $v\neq t'$ be the other endpoint of $P'_{t'}$, and let $P$ be the concatenation of $P'_{t'}$ with the matching edge $e=(t,t')$. In the original drawing, $\phi$, the path $P$ connected the images of $t$ and $v$, where the image of $v$ lies just inside the curve $\gamma_X$. We use the image $\phi(P)$ to draw the edge $e$ in the new embedding. We perform this operation for each one of the matching edges of $\Gr3(X)$. We need to specify how these drawings interact with each other, in order to avoid large number of crossings. In particular, when a number of such paths $P$ share the same vertex $v$ or the same edge $e'$, we need to specify how the corresponding matching edges are drawn along the original image of $e'$, or around the original image of $v$. If $v\not\in C_X$ is a regular vertex, then the local drawing of the paths $P'_t$ that contain $v$ is the same as in $\pi_v$. Similarly, if $e\not\in C_X$ is a regular edge, then the local drawing of the paths containing $e$ is the same as in $\pi_e$. If $e$ is an irregular edge, then we allow all paths that use $e$ to cross at most once with each other, so the number of crossings due to $e$ is bounded by $c^2(e)$. Similarly, if $v$ is an irregular vertex, then we allow all paths that use $v$ to cross at most once with each other, so the number of crossings due to $v$ is bounded by $\left (\sum_{e: v\in e}c(e)\right )^2\leq \dmax^2\sum_{e:v\in e}c^2(e)$. Finally, whenever a pair of edges $e,e'$ in the drawing $\phi$ of graph $\Gr3$ cross, the images of the paths that contain $e$ will cross the images of the paths containing $e'$. The number of all such new crossings is bounded by $\sum_{X\in \xset}\sum_{e\in N(X)}c^2(e)$.
These are the only possible new crossings in the new drawing.

 It now only remains to re-order the images of the matching edges, so they enter the circle $\gamma(X)$ in the same order as they appear in the drawing $\psi_X$.
Recall that all vertices and edges of $C_X$ are regular, and $E'=\set{e_1,\ldots,e_k}$ is the set of edges incident on $C_X$, that appear in this order along the boundary of the drawing of $C_X\cup E'$ in $\psi_X$. The vertices in sets $T(e_1),T(e_2),\ldots,T(e_k)$ are ordered correctly between the sets (w.r.t. $\sigma_X$), but may not be ordered correctly within each set. However, re-ordering the paths within each set only introduces at most $\sum_{e\in N(X)}c^2(e)$ crossings. 

To summarize, the total number of new crossings due to the above transformation of $\phi$ is bounded by:

\[\begin{split}
\sum_{X\in \xset'}\sum_{e\in N(X)}O(\dmax^2c^2(e))&\leq \sum_{X\in \xset'}O(\dmax^2(\beta^*)^2 \cdot \log n\cdot |N(X)|)\\
&\leq  \sum_{X\in \xset'}O(\dmax^2(\beta^*)^2 \cdot \log n\cdot \beta^*\cdot \dmax\cro_{\phi}(\Hr3(X),\Hr3))\\
&\leq O(\dmax^3(\beta^*)^3\cdot \log n\cdot \cro_{\phi}(\Hr3))\\
&\leq O(\dmax^6 \log^{17/2}n\cdot (\log\log n)^3\cro_{\phi}(\Hr3))\\
&\leq O(\dmax^6 \log^{17/2}n\cdot (\log\log n)^3\cdot \frac{20\dmax^3}{\alpha^*}\cdot \optcro{G})\\
&\leq O(\dmax^9\log^{10}n\cdot (\log\log n)^4\cdot \optcro{G})
\end{split}
\]

It now only remains to prove Lemma~ \ref{lem:ell2}.


\subsubsection{Proof of Lemma \ref{lem:ell2}}\label{sec:ell2}
Let $v^*$ be a vertex in $G$, and $\pset$ a collection of paths connecting every vertex $u\in S$ to $v^*$. Given any edge $e\in E$, the congestion of $e$ w.r.t. $\pset$, denoted by $c_{\pset}(e)$, is the number of paths in $\pset$ containing $e$. In order to prove the lemma, it is enough to show a distribution $\dset$ on pairs $(v^*,\pset)$, such that for each edge $e\in E$, $\expect[(v^*,\pset)\in \dset]{c^2_{\pset}(e)}\leq  O(\beta^2\cdot \log n)$. In the rest of the proof, we focus on finding such a distribution.

The proof consists of two parts. In the first part, we prove a slightly stronger version of the lemma for the special case where $G$ is the grid, and $S$ is the set of vertices in the last row of $G$. In the second part, we extend this proof to general graphs. Throughout the proof, given a $k\times k$ grid $Z$, we denote by $(i,j)$ the vertex that lies in the $i$th row, $j$th column of $Z$.

\begin{claim}\label{claim: routing in big grid} Let $H$ be a $k\times k$ grid, where $k$ is a power of $2$, and let $S$ be the set of vertices in the last row of $G$. Then there is a distribution $\dset$ on pairs $(u^*,\qset)$, where $u^*\in V(H)$, and $\qset$ is a collection of paths connecting every vertex in $S$ to $u^*$, such that for each edge $e\in E(H)$, $\expect[(u^*,\qset)\in \dset]{c^2_{\qset}(e)}=O(\log k)$.
\end{claim}

\begin{proof}
Let $Z$ be the $k'\times k'$ grid, where $k'=k/2$. We draw a number of rectangles in $Z$, that will define a partition of the edges of $Z$. For $1\leq i\leq \log k'$, rectangle $R_i$ contains, as its top boundary, row $k'/2^i$ of $Z$, bottom boundary row $k'/2^{i-1}$, left boundary column $1$, and right boundary column $k'/2^{i-1}$. (See Figure~\ref{fig: rectangles}). We say that edge $e$ belongs to set $E_i$, iff it either lies inside the rectangle $R_i$, or on its bottom, left, or right boundaries. 
We need the following claim.

\begin{figure}[h]
\scalebox{0.3}{\includegraphics{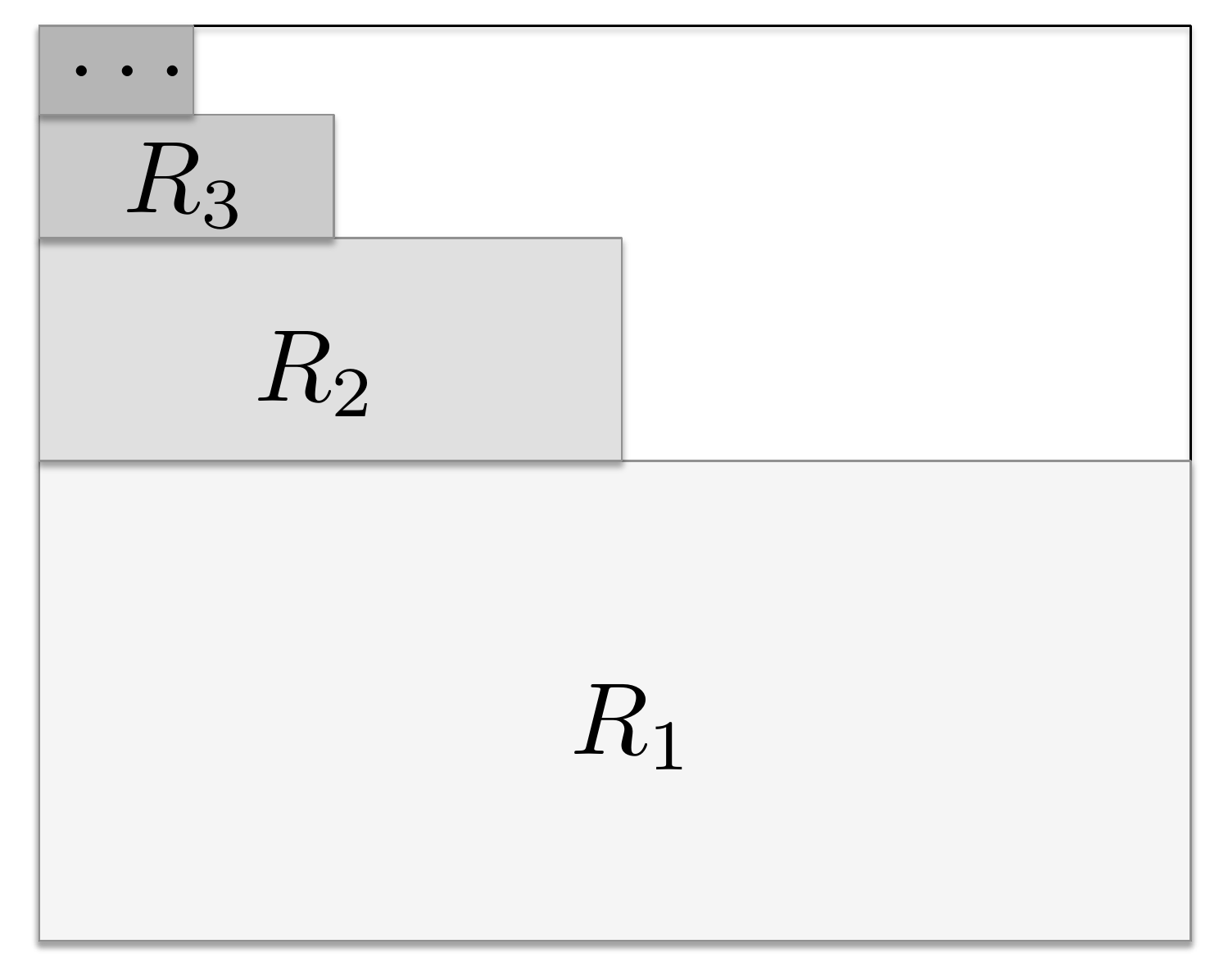}} \caption{Rectangles for grid $Z$.}
\label{fig: rectangles}
\end{figure}

\begin{claim}\label{claim: routing in small grid}
There is a collection $\pset'$ of paths in $Z$, that connect every vertex in the last row of $Z$ to the vertex $(1,1)$, such that for each $1\leq i\leq \log k'$, the congestion on any edge $e\in E_i$, is at most $O(2^i)$.
\end{claim}

\begin{proof}
It is enough to show that there is a flow $F$ in $Z$, where every vertex in the last row of $Z$ sends one flow unit to vertex $(1,1)$, and for each $1\leq i\leq \log k'$, the congestion on any edge $e\in E_i$ is $O(2^i)$. Since this is a single-sink flow, the claim will then follow from the integrality of flow.

Fix $1\leq i\leq \log k'$, and consider the rectangle $R_i$. Let $A$ be the set of the vertices lying on its bottom boundary, $|A|=k'/2^{i-1}$, and let $B$ be the set of the first $k'/2^{i}$ vertices lying on its top boundary. We show that there is a collection $\pset_i$ of paths, contained in $R_i$, 
connecting the vertices of $A$ to the vertices of $B$, such that every vertex in $A$ is an endpoint of exactly one such path, every vertex in $B$ an endpoint of exactly two paths, and the congestion on any edge in $E_i$ is bounded by $2$. Once we obtain such routing inside every rectangle $R_i$, in order to obtain the final flow $F$, we concatenate the paths in sets $\pset_i$, for $1\leq i\leq \log k_i$, sending $2^{i-1}$ flow units along each path in $\pset_i$.

We now show how to find the desired routing inside $R_i$. Let $v_1,\ldots,v_{k'/2^{i-1}}$ be the vertices of $A$, appearing in this order on the bottom boundary of $R_i$, and let $u_1,\ldots,u_{k'/2^{i}}$ be the first $k'/2^{i}$ vertices of $B$. First, for each $1\leq j\leq k'/2^i$, we define the path $P_{2j}$, connecting $v_{2j}$ to $u_j$, as follows. The path will follow column $2j$ of $R_i$ up to the $j$th row of $R_i$. Then it will follow row $j$ to column $j$, and finally column $j$ to $u_j$. In order to define the paths $P_{2j-1}$, connecting $v_{2j-1}$ to $u_j$, for $1\leq j\leq 2^i$, we simply concatenate the edge 
 $(v_{2j-1},v_{2j})$ with the path $P_{2j}$. This gives the desired routing in $R_i$, with congestion $2$.
 \end{proof}
 
We are now ready to define the distribution $\dset$ for the grid $H$. Let $Z'$ be the $k'\times k'$ sub-grid of $H$, where vertex $(1,1)$ of grid $Z'$ coincides with the vertex $(1,1)$ of $H$. We choose $u^*$ uniformly at random from the vertices of $Z'$. Once vertex $u^*$ is chosen, let $Z''$ be the $k'\times k'$ sub-grid of $H$, where the vertex $(1,1)$ of $Z''$ coincides with $u^*$ (see Figure~\ref{fig: shifted grid}).

\begin{figure}[h]
\scalebox{0.3}{\includegraphics{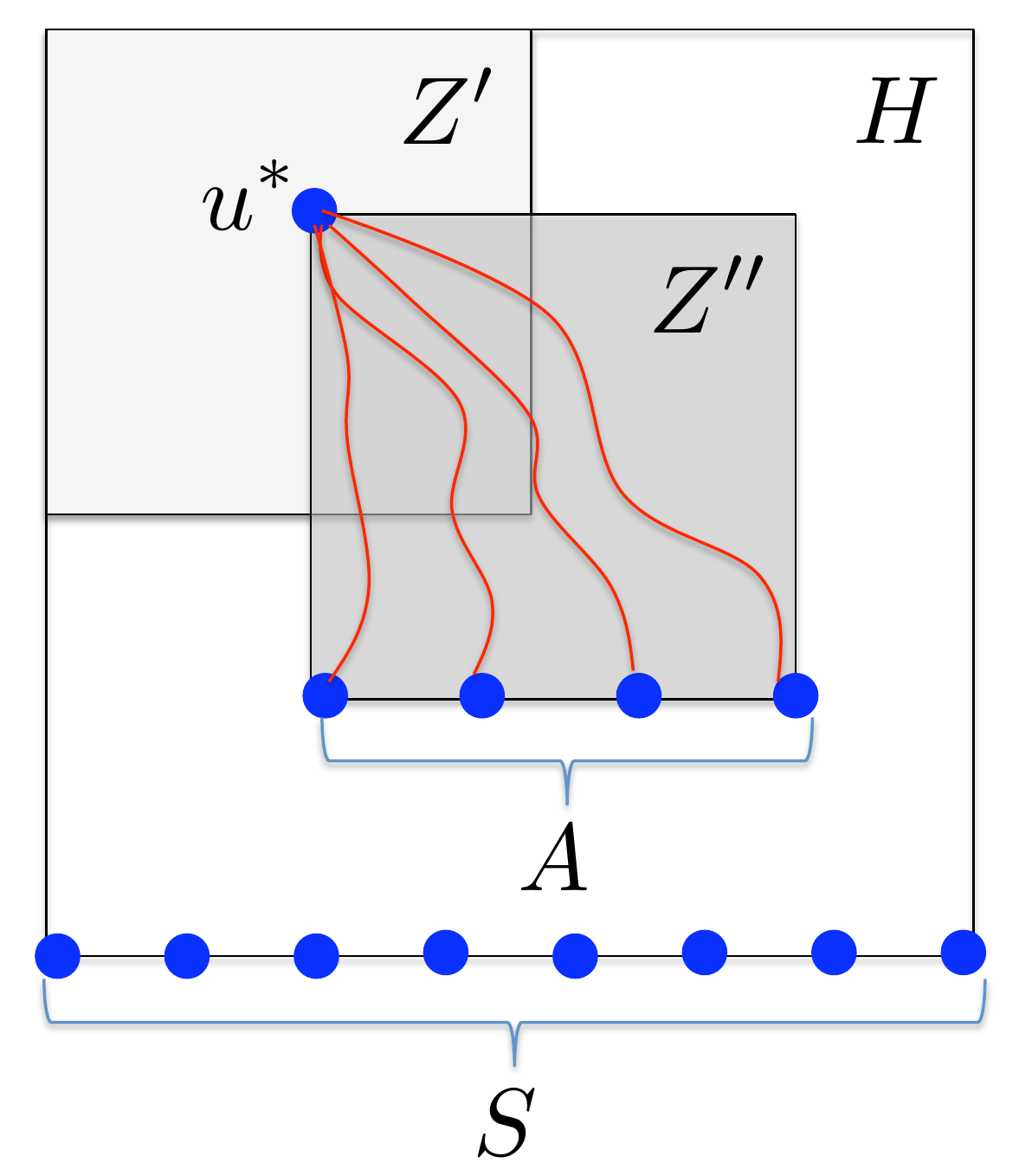}} \caption{Sub-grids $Z',Z''$ of $H$.}
\label{fig: shifted grid}
\end{figure}

Let $A$ denote the set of vertices in the last row of $Z''$, and let $\qset'$ be the collection of paths connecting vertices of $A$ to $u^*$ in $Z''$, as in Claim~\ref{claim: routing in small grid}. 
We now define the collection $\qset$ of paths, connecting the vertices of $S$ to $u^*$.
For each vertex $v\in S$, we define a path $P_1(v)$, connecting $v$ to some vertex $v'\in A$, and we let $P_2(v)\in \pset'$ be the path connecting $v'$ to $u^*$ inside $Z''$.  We choose the collection $\set{P_1(v)}_{v\in S}$ of paths, so that every vertex in $A$ is an endpoint of exactly two such paths, and the total congestion on the edges of $H$ due to paths  $\set{P_1(v)}_{v\in S}$ is bounded by a constant. For each vertex $v\in S$, we then let $P_v$ be the concatenation of $P_1(v)$ and $P_2(v)$, and we set $\qset=\set{P_v}_{v\in S}$.

Fix any edge $e\in E(H)$. We now bound the expected value of $c_{\qset}^2(e)$. If edge $e$ does not fall inside $Z''$, then the congestion on $e$ is bounded by a constant. Otherwise, if $e$ falls inside the rectangle $R_i$ of $Z''$, the congestion on $e$ is $O(2^i)$. The probability that $e$ belongs to $R_i$ is $O(1/2^{2i})$: indeed, the probability that both endpoints of $e$ belong to rows $k'/2^i,\ldots,k'/2^{i-1}$ is at most $1/2^{i}$, and the probability that they belong to columns $1,\ldots,k'/2^{i-1}$ is at most $1/2^{i-1}$. Therefore, $\expect{c^2_{\qset}(e)}\leq \sum_{i=1}^{\log k'} O(2^{2i}/2^{2i})=O(\log k)$.
 \end{proof}

In order to prove the lemma for a general graph $G$, and a subset $S$ of vertices of $G$, we first embed a $k\times k$ grid into $G$, where $k=O(|S|/\beta)$. We then route the vertices of $S$ to the last row of this grid, and use Claim~\ref{claim: routing in big grid} for routing inside the grid.

\begin{definition}
An embedding of a grid $H$ into a graph $G$ is a mapping $\pi$, where the vertices of $H$ are mapped to vertices of $G$, and edges $e=(u,v)\in E(H)$ are mapped to paths $\pi(e)$ connecting $\pi(u)$ to $\pi(v)$ in $G$. The congestion of the embedding is the maximum, over all edges $e'\in E(G)$, of the number of paths $\pi(e)$ containing $e'$, for all $e\in E(H)$.
\end{definition}

\begin{claim}\label{claim: embed grid into graph}
Let $G$ be any graph, $S$ a subset of vertices of $G$, such that $G$ has property (P3) for $S$ with parameter $\beta$, and there is a planar drawing $\psi$ of $G$, in which all vertices of $S$ lie on the boundary $\gamma$ of the outer face $\fout$ of the drawing. Then there is an embedding of the $k\times k$ grid $H$ into $G$ with congestion $2$, where $k$ is a power of $2$, $k=O(|S|/\beta)$. Moreover, the vertices of the last row of $H$ are mapped to a subset $A$ of $k$ distinct vertices of $G$, and there is a collection $\pset'$ of completely disjoint paths, connecting $k$ distinct vertices of $S$ to distinct vertices of $A$.
\end{claim}

\begin{proof}
Let $k$ be the largest power of $2$, smaller than $\lfloor |S|/(16\beta)\rfloor$.

Consider the planar drawing $\psi$ of $G$, and let $\sigma$ be the circular ordering of the vertices of $S$ on the boundary $\gamma$ of $\fout$. Let $S_1,S_2,S_3$ and $S_4$ be four disjoint subsets of $S$, such that for each $1\leq i\leq 4$, $|S_i|=\lfloor |S|/4\rfloor$, the vertices of $S_i$ appear consecutively in $\sigma$, and the ordering between the sets in $\sigma$ is $(S_1,S_2,S_3,S_4)$. 

We claim that there is a collection of $k$ vertex-disjoint paths in $G$, connecting the vertices of $S_1$ to the vertices of $S_3$. Assume otherwise. Then there is a $k$-vertex separator $C$ in $G$, separating the vertices of $S_1$ from the vertices of $S_3$. However, due to property (P3), it must be possible to send at least $\min\set{|S_1|,|S_3|}\cdot |S|/2>|S|^2/16$ flow units across the cut $C$, with congestion at most $|S|\cdot \beta$ on vertices. Therefore, the minimum cut separating $S_1$ from $S_3$ must contain more than $|S|/(16\beta)\geq k$ vertices. 
Let $\pset_1$ denote this collection of $k$ disjoint paths. Let $S_1'=\set{a_1,\ldots,a_k}$ be the subset of vertices of $S_1$ participating in these paths, and assume that they appear on $\sigma$ in this order. Let $S_3'=\set{a_1',\ldots,a_k'}$ be the subset of vertices of $S_3$ participating in these paths, and assume that they appear on $\sigma$ in the reverse order, $a_k',\ldots,a_1'$. Since $\psi$ is a planar drawing, set $\pset_1$ contains  a collection $(P_1,\ldots,P_k)$ of paths, where path $i$, for $1\leq i\leq k$, connects $a_i$ to $a_i'$.

Similarly, we can find a collection $\pset_2$ of vertex-disjoint paths $(Q_1,\ldots,Q_k)$, connecting a subset $S'_2=(b_1,\ldots,b_k)$ of vertices of $S_2$, to a subset $S_4'=(b_1',\ldots,b_k')$, where path $Q_i$ connects $b_i$ to $b_i'$, and the vertices $b_1,\ldots,b_k$ appear on $\sigma$ in this order. We are now ready to define the embedding $\pi$ of $H$ into $G$. 
For each $1\leq i\leq k$, $1\leq j\leq k$, let $v_{i,j}$ be the first vertex of $Q_i$ that belongs to $P_j$ (it is easy to verify that such a vertex exists, because of the planar drawing of $G$, in which the vertices of $S$ lie on $\gamma$).  Since the paths $\set{Q_i}_{i=1}^k$ are vertex-disjoint, it is easy to verify, that for every path $P_j$, the vertices $v_{1,j},v_{2,j},\ldots,v_{k,j}$ are distinct vertices of $P_j$, that appear on $P_j$ in this order.
For $1\leq i\leq k$, $1\leq j\leq k$, we map the vertex $(i,j)$ of $H$ to $v_{i,j}$. This concludes the definition of the mapping $\pi$ for the vertices of $H$. We now define the mappings of edges of $H$. Let $((i,j),(i,j+1))$ be any horizontal edge of $H$. We map this edge to the segment of $Q_i$ lying between $\pi((i,j))$ and $\pi((i,j+1))$. Similarly, if $((i,j),(i+1,j))$ is any vertical edge of $H$, we map it to the segment of $P_j$ connecting the images of vertices $(i,j)$ and $(i+1,j)$. Since the sets $\pset_1$, $\pset_2$ of paths are each vertex disjoint, the congestion on edges is at most $2$. Finally, we let $A=\set{v_{k,i}}_{i=1}^k$, $S'=\set{a_1',\ldots,a_k'}$, and we let $\pset'$ contain, for each path $P_i$, for $1\leq i\leq k$, the segment of $P_i$ between $a_i'$ and $v_{k,i}$.
\end{proof}

We are now ready to define the distribution $\dset$ over pairs $(v^*,\pset)$ in graph $G$. Let $H$ be the $k\times k$ grid, where $k=O(|S|/\beta)$,  with the embedding $\pi$ of $H$ into $G$, and the collection $\pset'$ of paths, connecting vertices in the subset $S'\sse S$, $|S'|=k$, to the vertices of $A$, as in Claim~\ref{claim: embed grid into graph}.

We will use the distribution $\dset'$ over pairs $(u^*,\qset)$, from Claim~\ref{claim: routing in big grid}. Given a pair $(u^*,\qset)$, we let $v^*=\pi(u^*)$, and we define a collection $\pset$ of paths, connecting the vertices in $S$ to $v^*$, as follows. We  start by defining three collections of paths. The first collection, $\pset_1$, connects all vertices of $S$ to vertices of $S'$. The second collection, $\pset_2$, is precisely the set $\pset'$ of paths, connecting each vertex of $S'$ to a vertex of $A$. The third collection, $\pset_3$, connects vertices of $A$ to $v^*$. The final set $\pset$ of paths is obtained by concatenating the paths in $\pset_1,\pset_2$, and $\pset_3$.

We now formally define each path set. Set $\pset_1$ contains $|S|$ paths. For each vertex $v\in S$, there is a path $P_v\in \pset_1$, connecting $v$ to some vertex of $S'$. We ensure that the total edge and vertex congestion due to these paths is at most $2\beta$, and each vertex in $S'$ serves as endpoint of at most $2\beta$ such paths. The problem of finding such paths can be cast as the problem of finding $s$-$t$ flow in the graph, where the vertices in $S\setminus S'$ serve as sources, and vertices in $S'$ serve as sinks. In order to show that such a flow exists, it is enough to show that for any collection $C$ of vertices, separating $S\setminus S'$ from $S'$, $|C|\geq |S'|/(2\beta)$ must hold. Let $C$ be any such separator. Notice that due to property (P3), the amount of flow sent across this cut is at least $|S'|\cdot |S|/2$, and the maximum vertex congestion is $|S|\cdot \beta$. Therefore, $C$ must contain at least $|S'|/(2\beta)$ vertices.

We obtain the second collection, $\pset_2$ of paths, from the set $\pset'$ of vertex-disjoint paths, connecting the vertices of $S'$ to the vertices of $A$. The only change is that if some vertex $v\in S'$ serves as an endpoint of $n_v$ paths in $\pset_1$, then we add $n_v$ copies of the corresponding path in $\pset'$ to set $\pset_2$. It is easy to see that the edge congestion due to paths in $\pset_2$ is bounded by $2\beta$.

Finally, we obtain the set $\pset_3$ of paths, connecting the vertices of $A$ to $v^*$, as follows.
Recall that $\qset$ is a collection of paths in the grid $H$, connecting every vertex of the last row of $H$ to $u^*$, and the expected value of $c^2_{\qset}(e)$ on any edge $e$ of $H$ is bounded by $O(\log k)$. We use the mapping $\pi$ of the edges of $H$ to paths of $G$, to define, for every path $Q\in \qset$, a corresponding path $P\in \pset''$. If path $Q_a\in \qset$ connects a vertex $a\in A$ to $u^*$, then the corresponding path $P_a\in \pset''$ will connect the vertex $\pi(a)$ to the vertex $v^*$ in graph $G$. If vertex $a$ serves as an endpoint of $n_a$ paths in $\pset_2$, then we add $n_a$ copies of the path $P_a$ to $\pset_3$. 

The final set, $\pset$ of paths, is obtained by concatenating the paths in $\pset_1,\pset_2$ and $\pset_3$. It is easy to see that $\pset$ contains a path connecting every vertex $v\in S$ to $v^*$. It now only remains to bound the expected value of $c_{\pset}^2(e)$ on edges $e\in E(G)$.

Let $e\in E(G)$ be any edge of $G$. Recall that the congestion of the embedding $\pi$ of $H$ into $G$ is at most $2$. Let $e_1,e_2$ be the two edges of $H$ whose images contain $e$ (one or both of these edges may be undefined). We then have:

\[\begin{split}
c_{\pset}^2(e)&=(c_{\pset_1}(e)+c_{\pset_2}(e)+c_{\pset_3}(e))^2\\
&\leq (c_{\pset_1}(e)+c_{\pset_2}(e)+2\beta\cdot c_{\qset}(e_1)+2\beta\cdot c_{\qset}(e_2))^2\\
&\leq 16\left (c^2_{\pset_1}(e)+c^2_{\pset_2}(e)+4\beta^2c_{\qset}^2(e_1)+4\beta^2c_{\qset}^2(e_2)\right )\\
&\leq O(\beta)^2+64\beta^2c_{\qset}^2(e_1)+64\beta^2c_{\qset}^2(e_2)
\end{split}\]

Therefore, $\expect[(v^*,\pset)]{c_{\pset}^2(e)}=O(\beta)^2+O(\beta^2)\expect[(u^*,\qset)]{c_{\qset}^2(e_1)}+O(\beta^2)\expect[(u^*,\qset)]{c_{\qset}^2(e_2)}=O(\beta^2\log |S|)$.

\subsection{Proof of Theorem~\ref{thm: nasty canonical set to contraction}}

Let $R$ be a nasty canonical set in the graph $H=\G_{|S}$, so it has properties (P1), (P2) in $H$, and $|R|\geq \frac{2^{16}\cdot \dmax^6}{(\alpha^*)^2}\cdot|\Gamma_H(R)|^2$. Throughout this proof, for each $X\in \xset$, $\Gamma(X)$ refers to the set $\Gamma_{\G}(X)$ of the interface vertices in the original graph, $\G$.

 We first perform the following clean-up step. For every set $Z_X\in \zset$, with $Z_X\sse R$, denote $T_R(Z_X)=T(Z_X)\cap R$, and $T_{\notr}(Z_X)=T(Z_X)\setminus R$ (see Figure~\ref{fig: cleanup}). If $|T_R(Z_X)|\leq |T(Z_X)|/4$, we remove the set $Z_X$ of vertices from $R$. 

\begin{figure}[h]
\scalebox{0.4}{\rotatebox{0}{\includegraphics{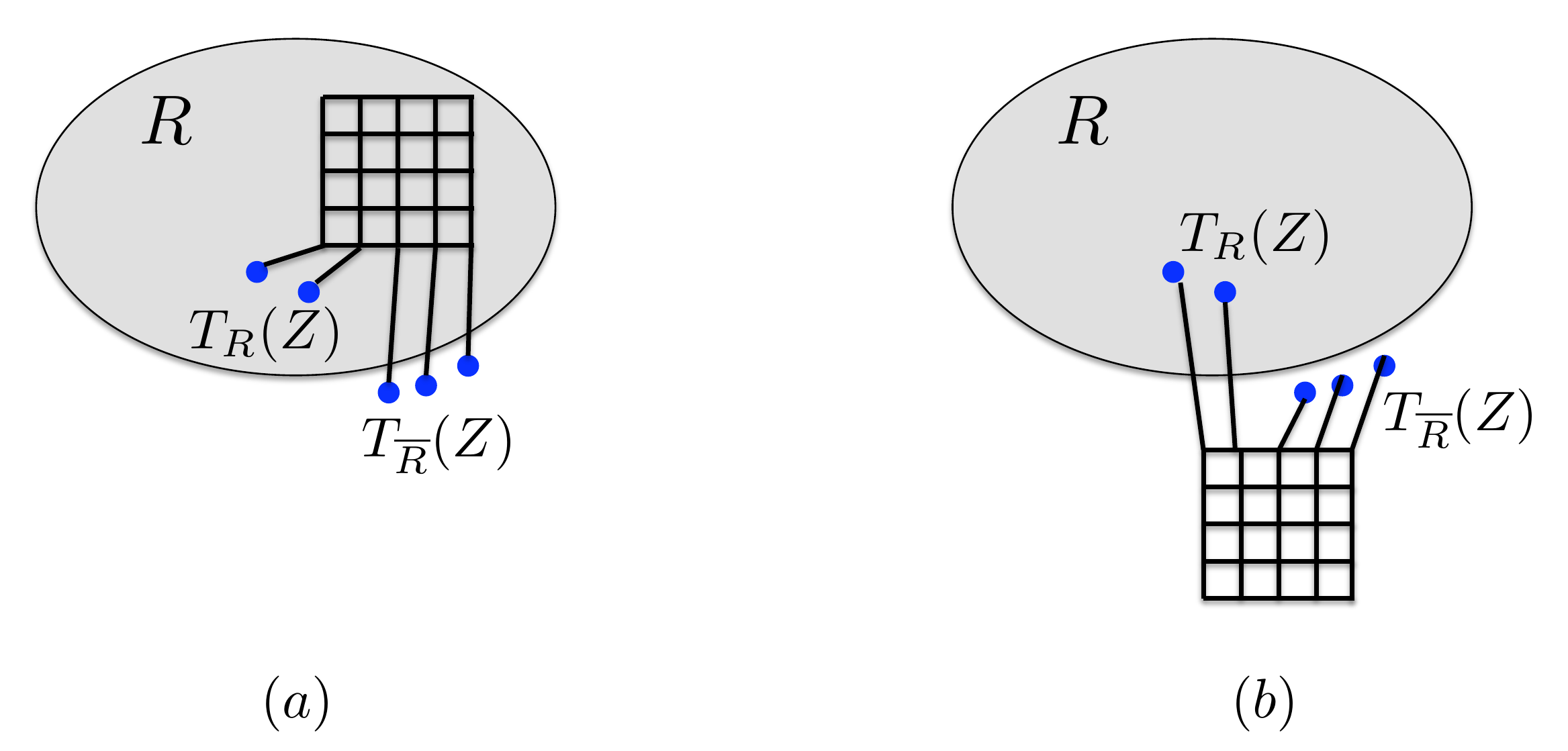}}} \caption{One round of the cleanup step.} \label{fig: cleanup}
\end{figure}

\begin{claim}
Let $R$ be any canonical nasty set in $H$, and let $R'=R\setminus Z_X$ be the set obtained after one round of the cleanup step. Then $R'$ is a nasty canonical set.
\end{claim}
\begin{proof}
It is obvious that $R'$ is a canonical set, and it is straightforward to verify that Properties (P1) and (P2) continue to hold for it. We only need to check that 
 $|R'|\geq \frac{2^{16}\cdot \dmax^6}{(\alpha^*)^2}\cdot|\Gamma(R')|^2$. For simplicity, denote $M=\frac{2^{16}\cdot \dmax^6}{(\alpha^*)^2}$, $A=|\Gamma(R)|$, $B=\Gamma(Z_X)$. Then $|R'|\geq |R|-B^2$. Moreover, observe that $\Gamma(R')$ is obtained from $\Gamma(R)$ by removing the vertices of $\Gamma(Z_X)$ that are adjacent to the vertices of $T_{\notr}(Z_X)$ from it, and possibly adding the vertices in $T_{R}(Z_X)$ (if they do not already belong to $\Gamma(R)$). Since $|T_R(Z_X)|\leq |T(Z_X)|/4$, we get that $|\Gamma(R')|\leq |\Gamma(R)|-|T_{\notr}(Z_X)|+|T_{R}(Z_X)|\leq |\Gamma(R)|-|B|/2$.
We then have that:

\[\begin{split}
|R'|&\geq |R|-B^2\\
&\geq MA^2-B^2\\
&\geq M(A-B/2)^2\\
&\geq  \frac{2^{16}\cdot \dmax^6}{(\alpha^*)^2}\cdot|\Gamma(R')|^2
\end{split} \]

(We have used the fact that $A\geq |T_{\notr}(Z_X)|\geq 3B/4$.)
\end{proof}

We perform the above cleanup step while possible, and we let $R$ denote the final nasty canonical set.
Notice that for all $Z\in \zset$ with $Z\sse R$, $|T_R(Z)|>|T(Z)|/4$.

We partition $\zset$ into two sets: $\zset_1$ contains all sets $Z_X$ with $Z_X\cap R=\emptyset$, and $\zset_2$ contains all sets $Z_X\sse R$. We also denote $\xset_1=\set{X\sse V(\G)\mid Z_X\in\zset_1}$, and $\xset_2=\set{X\sse V(\G)\mid Z_X\in\zset_2}$.

We now define a set $S'$ of vertices in the original graph $\G$. It consists of two subsets, $S_1'$ and $S_2'$. Subset $S_1'$ contains all vertices in sets $X\in\xset_1$, so $\xset_1$ is a partition of $S_1'$. Recall that the partition $\xset_1$ has Properties~(\ref{property: subsets-first})--(\ref{property: size-last}). We will use these properties later.

We now turn to define the subset $S_2'$, by first defining a set $S_2^*\sse V(\G)$. We start with the set $R\sse V(H)$ of vertices. If a vertex $v\in R$ is also a vertex of $\G$, then we add $v$ to $S_2^*$. Otherwise, $v\in Z_X$ for some $Z_X\in \zset_2$ must hold. Fix some such set $X$. Recall that the vertices of $\Gamma(X)=T(Z_X)$ belong to $Z'_X$, and therefore to $V(H)\cap V(\G)$. Some of these vertices may lie in $R$, and some of them outside of $R$. The vertices of $\Gamma(X)\cap R$ have been added to $S_2^*$, and the vertices of $\Gamma(X)\setminus R$ are not added to $S_2^*$. Finally, we add the vertices of $X\setminus \Gamma(X)$ to $S_2^*$. We have thus obtained a set $S_2^*\sse V(\G)$. We then set $S_2'=S_2^*\setminus \Gamma_{\G}(S_2^*)$, that is, $S_2'$ is obtained from $S_2^*$, after we remove all interface vertices from it. Finally, we set $S'=S_1'\cup S_2'$.

Observe that $\G[S_1']$ and $\G[S_2']$ are completely disjoint, with no direct edges connecting between them (that was the purpose of removing the vertices of $\Gamma(S^*_2)$ from $S^*_2$). Therefore, we can perform the decomposition for the graph contraction step separately for both sets. We will argue below that set $S_2'$ has properties (P1) and (P2). Assuming this is true, let $\xset_2'$ be the partition of the vertices of $S_2'$ guaranteed by Theorem~\ref{thm: graph contraction}. We then let $\xset'=\xset_1\cup\xset_2'$ be the final partition of the set $S'$. Notice that partition $\xset'$ has Properties~~(\ref{property: subsets-first})--(\ref{property: size-last}).  Let $H'$ be the contracted graph obtained from $\G$, after we replace each sub-graph $\G[X]$ with the grid $Z'_X$, for all $X\in \xset'$. Then from Theorem~\ref{thm: canonical drawing}, there is a canonical drawing $\phi'$ of the resulting graph $H'$, with $\cro_{\phi'}(H') = O(d_{\max}^9 \cdot \log^{10} n \cdot (\log\log n)^4 \cdot \optcro{\G})$.

It now only remains to prove two things. First, we need to show that $|V(H')|< |V(H)|$, and second, we need to show that $S_2'$ has properties (P1) and (P2), so that Theorem~\ref{thm: graph contraction} can be applied to it. These will complete the proof of the theorem.

\paragraph{Bounding the Size of $\mathbf{V(H')}$.}
Recall that since $R$ is a nasty set,  $|R|\geq \frac{2^{16}\cdot \dmax^6}{(\alpha^*)^2}\cdot|\Gamma(R)|^2$.
We can transform the graph $H$ into the graph $H'$ in the following two simple step. First, for each $Z_X\in \zset_2$, we replace $Z'_X$ with the graph $\G[X]$, using the vertices of $\Gamma(X)$ as the interface. Let $H^*$ denote this resulting graph. Notice that $S_2^*, S_2'$ are subsets of vertices of $H^*$. Next, for each set $X\in \xset_2'$, we replace $H^*[X]$ by $Z'_X$. This gives the final graph $H'$. We now analyze the number of vertices in these graphs.

The set $V(H)$ of vertices can be partitioned into two sets $(R,\notr)$, where $\notr=V(H)\setminus R$. Consider now the set $V(H^*)$ of vertices. By the definition of $S_2^*$, the two sets $(\notr,S_2^*)$ define a partition of $V(H^*)$.

Recall that we have obtained $S_2^*$ from $R$, by replacing each set $Z_X\in \zset_2$ with $X\setminus \Gamma(X)$. Since each set $X\in \xset_2$ has Property~(\ref{property: size-last}), $|X|\geq \frac{(\alpha^*|\Gamma(X)|)^2}{64\dmax^2}\geq \frac{(\alpha^*)^2}{64\dmax^2}|Z_X|$. Since at least $1/4$ of the vertices of $\Gamma(X)=T_{H}(Z_X)$ belong to $S_2^*$, $|X\cap S_2^*|\geq \frac{(\alpha^*)^2}{2^8\dmax^2}|Z_X|$. Finally, we observe that some vertices of $\Gamma(X)\cap S_2^*$ may be shared by up to $\dmax$ sets $X'\in \xset_2$. Therefore, in total,

\begin{equation}\label{eq: bound on S2star}
|S_2^*|\geq |R|\cdot \frac{(\alpha^*)^2}{2^8\dmax^3}\geq 2^8\dmax^3|\Gamma_H(R)|^2
\end{equation}

It is easy to see that $|\Gamma_{H^*}(S_2^*)|\leq \dmax |\Gamma_H(R)|$. 
The only difference between the two sets, is that for each $X\in \xset_2$, $\Gamma_H(R)$ contains $|T_{\notr}(Z_X)|$ vertices of $\Gamma(Z_X)$, that are adjacent to the vertices of $T_{\notr}(Z_X)$. These vertices are replaced by at most $\dmax|T_{\notr}(Z_X)|$ vertices in $\Gamma_{H^*}(S_2^*)$, because the maximum vertex degree is bounded by $\dmax$. Therefore, 

\begin{equation}\label{eq: bound on Gamma of S2star}
|\Gamma_{H^*}(S_2^*)|\leq \dmax |\Gamma_H(R)|.
\end{equation}

Recall that the set $S_2'$ is obtained by removing all vertices of $\Gamma(S_2^*)$ from $S_2^*$. From Equations~(\ref{eq: bound on Gamma of S2star}) and (\ref{eq: bound on S2star}), $|\Gamma_{\G}(S_2^*)|=|\Gamma_{H^*}(S_2^*)|<|S_2^*
|/2$, and so $|S_2'|\geq |S_2^*|/2\geq 2^5\dmax^3|\Gamma_H(R)|^2$, and $|\Gamma_{\G}(S_2')|\leq \dmax |\Gamma_{H^*}(S_2^*)|\leq \dmax^2|\Gamma_H(R)|$.

In the final step, we replace each set $X\in \xset_2'$ by set $Z_X$. The set of vertices $V(H')$ can then be partitioned into $(\notr,\Gamma_{H^*}(S_2^*),Y)$, where $Y$ is obtained from $S_2'$, after we replace each set $X\in \xset_2'$ with $Z_X$. From Equation~(\ref{eq: final size}) in Section~\ref{sec: graph contraction}, $|Y|\leq 162\dmax^2|\Gamma_{\G}(S_2')|^2\leq 162 \dmax^4|\Gamma_H(R)|^2<|R|/2$. Therefore, $|V(H')|=|\notr|+|\Gamma_{\G}(S_2^*)|+|Y|<
|\notr|+\dmax^2|\Gamma_H(R)|+|R|/2<|\notr|+|R|=|V(H)|$.

\paragraph{Property (P1)}
Recall that the original set $R$ we have started from had property (P1). It is immediate to see that the clean-up step has no affect on this property. Also, after we replace each graph $Z'_X$, for $Z_X\in \zset_2$, with graph $\G[X]$, the resulting set $S^*_2$, and consequently $S'_2$ have this property.

\paragraph{Property (P2)} 
Recall that set $R$ we have started from had property (P2). That is, there is a planar drawing $\psi_R$ of the graph $H[R]$, in which the vertices of $\Gamma(R)$ lie on the boundary of the outer face. It is easy to see that the clean-up step does not affect property (P2).

Let $R$ be the graph obtained after the clean-up step, and let $\psi$ be a planar drawing of $H[R]$, in which the vertices of $\Gamma(R)$ lie on the boundary $\Gamma$ of the outer face of $\psi$.

We now consider the sets $X\in\xset_2$ one-by-one.
 For each such set $X$, we replace the graph $Z'_X$ with the graph $\G[X]$. Let $\tilde{H}$ denote this new graph. We obtain the set $\tilde{R}$ of vertices in this new graph from the set $R$, by removing the vertices of $Z_X$, and adding the vertices of $X\setminus\Gamma(X)$ instead to $R$.
 
We need to prove the following. 

\begin{claim}
Let $\tilde{H}$ be the current graph, $\tilde{R}$ the current set of vertices, that has property (P2) in the graph $\tilde{H}$. Let $X\in \xset_2$ be any set with $Z_X\sse \tilde{R}$. Let $\tilde{H}'$ be the graph obtained from $\tilde{H}$ after we replace $Z_X'$ with $\G[X]$, and let $\tilde{R'}$ be the set of vertices obtained from $\tilde{R}$ by replacing the vertices in $Z_X$ with the vertices in $X\setminus\Gamma_X$. Then $\tilde{R}' $ has property (P2) in $\tilde{H}'$.
\end{claim}

\begin{proof}
The difficulty comes from the fact that not all vertices of $\Gamma(X)=T_{\tilde H}(Z_X)$ belong to $\tilde{R}$. Let $\Gamma_1(X)=\Gamma(X)\setminus\tilde{R}$, and let $\Gamma_2(X)=\Gamma(X)\setminus\Gamma_1(X)$. Let $\psi$ be a planar drawing of $\tilde{H}[\tilde{R}]$, in which the vertices of $\Gamma(\tilde{R})$ lie on the boundary $\gamma$ of the outer face $F_{out}$. Draw a closed curve $c$ around the boundary $\gamma$, inside $F_{out}$. We can now augment this drawing, by adding the matching edges, corresponding to vertices in $\Gamma_1(X)$, so that they cross $c$, and the resulting drawing of $Z'_X$ is planar (see Figure~\ref{fig: augmenting the drawing}). In this drawing, the vertices of $\Gamma_1(X)$ lie outside the circle $c$, and all other vertices lie inside it. 

\begin{figure}[h]
\scalebox{0.4}{\rotatebox{0}{\includegraphics{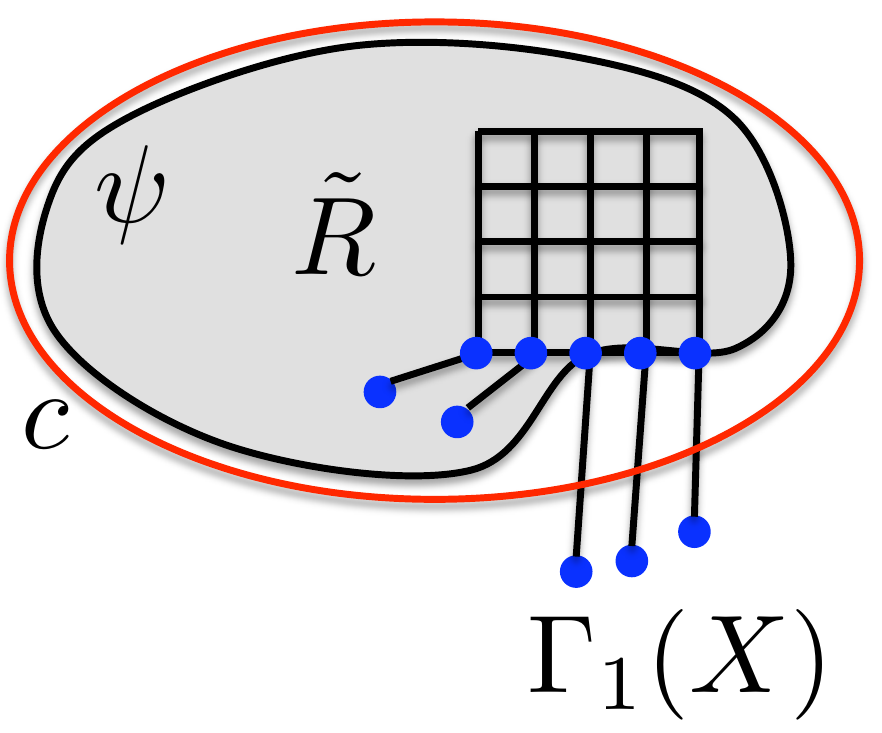}}} \caption{Augmenting the drawing $\psi$.} \label{fig: augmenting the drawing}
\end{figure}

This drawing gives a planar drawing of $Z'_X$, and the ordering of the vertices of $\Gamma(X)$ in this drawing is identical to their ordering along the curve $\gamma_X$ in the planar drawing $\pi_X$ of $\G_X$. Therefore, we can replace the drawing of $Z'_X$ with the drawing $\pi_X$ of $\G_X$. We do so without changing the drawing of the vertices in $\Gamma(X)$, so the vertices in $\Gamma_1(X)$ remain outside $c$, while all other vertices are drawn inside $c$. This can be done so that the only edges that cross the circle $c$ are the edges adjacent to vertices in $\Gamma_1(X)$. We then erase all vertices in $\Gamma_1(X)$, and all their adjacent edges from this drawing. This gives the desired planar drawing of $\tilde{H}'[\tilde{R}']$, in which all interface vertices lie on the boundary of the outer face. \end{proof} 

After we have processed all sets $X\in \xset_2$, the resulting graph that we obtain is exactly $H^*$, and the resulting set of vertices is $S_2^*$, which, from the above claim, has property (P2). In the final step, we remove the vertices of $\Gamma(S_2^*)$ from $S_2^*$ to obtain the set $S_2'$. It is easy to see that this step preserves property (P2). Since $H^*[S_2']=\G[S_2']$, set $S_2'$ has property (P2) in $\G$.


\section{Proof of Theorem~\ref{thm: stopping condition}}

\label{------------------------------------------------D: theorem for end of recursion-----------------------------------------------------------------------}

We start with the following theorem, whose proof uses the Planar Separator Theorem (Theorem~\ref{thm: planar separator}), the approximation algorithm for the Balanced Cut problem (Theorem~\ref{thm: ARV}), and Claim~\ref{claim: numbers}.

\begin{theorem}\label{theorem: balanced cut}
Let $G$ be any $n'$-vertex graph with maximum degree at most $\dmax$, and a collection $\zset$ of disjoint subsets of vertices of $G$, such that for each $Z\in \zset$, $G[Z]$ is a grid, $\Gamma(Z)$ is the set of vertices in the first row of $Z$, and $\out(Z)$ contains exactly one edge incident on each vertex in the first row of $Z$. Let $\opt$ denote the cost of the optimal canonical solution to the \MP problem on $G$, with respect to $\zset$, and assume that $\opt<\sqrt {n'}$. Then there is an efficient algorithm to partition the vertices of $G$ into two {\bf canonical} sets $\tilde A,\tilde C$, with $|\tilde A|,|\tilde C|\geq \Omega(n'/\dmax^2)$, such that $|E(\tilde A,\tilde C)|\leq O(\dmax\sqrt{n'\log n'})$.
\end{theorem}

\begin{proof}
If there is any grid $Z\in \zset$, with $|Z|\geq n'/(48q\dmax)^2$, where $q$ is the constant from Theorem~\ref{thm: planar separator}, then we return the partition $(\tilde A,\tilde C)$, where $\tilde A=Z$ and $\tilde C=V(G)\setminus Z$. Since $|\Gamma(Z)|\leq \sqrt {n'}$, this is a valid partition. We assume from now on that for each $Z\in \zset$, $|Z|< n'/(48q\dmax)^2$.

It is enough to show that there is a partition $(A',C')$ of $V(G)$, such that both $A'$ and $C'$ are canonical, $|A'|,|C'|\geq n'/20$, and $|E(A',C')|\leq O(\dmax\sqrt{n'})$. We can then apply Theorem~\ref{thm: ARV} to obtain the desired partition $(\tilde A,\tilde C)$; in order to ensure that the final partition is canonical, we can simply assign the edges of the grids $Z\in \zset$ infinite costs.

We now show that there is a partition $(A',C')$ of $V(G)$ with the above properties. Let $E'$ be the optimal canonical solution to the \MP problem on $G$, $|E'|=\opt$. Then graph $G\setminus E'$ is planar, and therefore, from Theorem~\ref{thm: planar separator}, there is a subset $B$ of vertices, $|B|\leq q\sqrt{n'}$, whose removal separates the two subsets $(A,C)$ of vertices, with $|A|,|C|\geq n'/3$. Assume w.l.o.g. that $|A|\leq |C|$, and consider the partition $(A'',C'')$, where $A''=A\cup B$, and $C''=C$. Then $|A''|,|C''|\geq n'/3$, and $|E(A'',C'')|\leq \opt+q\dmax \sqrt {n'}\leq 2q\dmax\sqrt {n'}$, since all edges in $E(A'',C'')$ either belong to $E'$, or are incident on the vertices of $B$, and since we have assumed that $\opt<\sqrt{n'}$. However, it is possible that the sets $A'',C''$ are not canonical. 

Consider some grid $Z\in \zset$, and let $\Gamma_1(Z)=\Gamma(Z)\cap A''$, $Z_1=Z\cap A''$, $\Gamma_2(Z)=\Gamma(Z)\cap C''$, and $Z_2=Z\cap C''$. We now define a partition $(\zset_1,\zset_2)$ of $\zset$, as follows: for each grid $Z\in \zset$, if $|\Gamma_1(Z)|>|\Gamma_2(Z)|$, then $Z$ belongs to $\zset_1$, and otherwise it belongs to $\zset_2$. For each grid $Z\in \zset_1$, we move all vertices of $Z_2$ to $A''$, and for each grid $Z\in \zset_2$, we move all vertices of $Z_1$ to $C''$, to obtain the final partition $(A',C')$.

Notice first that from Claim~\ref{claim: cut of grids}, for each $Z\in\zset_1$, $|E(Z_1,Z_2)|\geq |\Gamma_2|$, and for each $Z\in \zset_2$, $|E(Z_1,Z_2)|\geq |\Gamma_1|$. Therefore, the cut size does not increase, and $|E(A',C')|\leq |E(A'',C'')|\leq 2q\dmax\sqrt{n'}$.

It is now enough to show that $|A'|,|C'|\geq n'/20$. We show this for $A'$, and the proof for $C'$ is symmetric. Recall that $|A''|\geq n'/3$. Therefore, it is enough to show that $\sum_{Z\in \zset_2}|Z_1|\leq n'/4$. We prove this using Claim~\ref{claim: numbers}, as follows. For each grid $Z\in \zset_2$, we define a number $x_Z=|Z_1|$, and $y_{x_Z}=|E(Z_1,Z_2)|$. From Claim~\ref{claim: cutting the grid}, $|Z_1|\leq 4|E(Z_1,Z_2)|^2$. 
Therefore, $x_Z\leq 4y_{x_Z}$ for all $Z\in \zset_2$, and we can set the parameter $\beta=4$. Since we have assumed that for all $Z\in \zset$, $|Z|< n'/(48q\dmax)^2$, we can set $M= n'/(48q\dmax)^2$. Finally, we can set $S=\sum_{Z\in \zset_2}y_{x_Z}=\sum_{Z\in \zset_2}|E(Z_1,Z_2)|\leq |E(A'',C'')|\leq 2q\dmax\sqrt{n'}$. From Claim~\ref{claim: numbers}, 

\[\begin{split}
\sum_{Z\in \zset_2}|Z_1|&\leq 2S\sqrt{\beta M}+\frac M 4\\
&\leq 4q\dmax \sqrt{n'}\cdot \sqrt{\frac {4n'}{(48q\dmax)^2}}+\frac{n'}{4\cdot(48q\dmax)^2}\\
&\leq \frac{n'} 6+\frac{n'}{4\cdot (48q\dmax)^2}\leq \frac {n'} 4\\
\end{split}
\]
\end{proof}

We now turn to prove Theorem~\ref{thm: stopping condition}.
In order to simplify notation, we denote $\overline{\opt}$ by $\opt$ in this proof. Notice that in general, it is possible that $G$ contains much more vertices than $n'$, because the bounding box $X$ may contain many vertices.  In order to avoid this, we perform the following simple transformation. If $P\sse X$ is a path, all of whose vertices, except for endpoints, have degree $2$ in $G$, then we replace $P$ by an edge $e$. We perform this procedure while possible, and we let $G'$ denote the resulting graph, and $X'$ the resulting bounding box. Notice that $|V(G')|\leq \dmax\cdot n'$, and any weak feasible solution for $\pi(G,X,\zset')$ gives a weak feasible solution for $\pi(G',X',\zset')$, and vice versa. The same holds for strong solutions for $\pi(G,X,\zset')$ and $\pi(G',X',\zset')$. From now on we focus on finding a weak feasible solution for $\pi(G',X',\zset')$, and we denote $|V(G')|=n''$.

Let $\phi$ be the optimal strong solution for problem $\pi(G',X',\zset')$. Recall that $\phi$ contains at most $\opt$ crossings.
The algorithm consists of $O(\dmax^2\cdot \log n)$ iterations, and in each iteration $i$, we are given a collection $G_1^i,\ldots,G_{k_i}^i$ of disjoint sub-graphs of $G'$, such that for each $1\leq j\leq k_i$, $V(G_{j}^i)$ is canonical for $\zset'$, and $k_i\leq 2\OPT$. The number of vertices in each such sub-graph $G_j^i$ is bounded by $n_i=(1-\alpha)^{i-1}n''$, where $\alpha=\Omega(1/\dmax^2)$. In the input to the first iteration, $k_1=1$, and $G_1^1=G'$.

Iteration $i$, for $i\geq 1$ is performed as follows. We construct a new family $\gset_{i+1}$ of sub-graphs of $G'$, some of which will serve as the input to the next iteration, and a set $E^i$ of edges. At the beginning, $\gset_{i+1}=\emptyset$ and $E^i=\emptyset$.
Consider some graph $G^i_j$, for $1\leq j\leq k_i$. Let $\opt_j^i$ be the cost of the strong optimal canonical solution to problem $\pi(G_j^i,\emptyset,\zset')$. Notice that $\sum_{j=1}^{k_i}\opt_j^i$ is bounded by the cost of the strong optimal solution to problem $\pi(G',\emptyset,\zset')$ (since this solution implies solutions for each one of the sub-problems), which is in turn bounded by $\opt$.

We apply Theorem~\ref{theorem: balanced cut} to each such graph $G^i_j$. Recall that if $\opt_j^i< \sqrt{|V(G_j^i)|}$, then we obtain a partition $(\tilde A,\tilde C)$ of the vertices of $G^i_j$, where both subsets $\tilde A$ and $\tilde C$ are canonical, $|\tilde A|,|\tilde C|\geq \Omega\left (\frac{|V(G_j^i)|}{\dmax^2}\right )\geq \alpha n_i$, and so $|\tilde A|,|\tilde C|\leq (1-\alpha)n_i=n_{i+1}$. Moreover,  $|E_{G_j^i}(\tilde A,\tilde C)|\leq O(\dmax\sqrt {n_i\log n_i})$. We distinguish between three cases. The first case happens when the algorithm successfully finds such a  partition. The second case happens when the algorithm does not return a valid partition, but $|V(G_j^i)|\leq n_{i+1}$. The third case is when the algorithm does not return a valid partition, and $|V(G_j^i)|> n_{i+1}>n_i/2$.

Assume first that for each $1\leq j\leq k_i$, either Case 1 or Case 2 happens. Fix any such graph $G_j^i$. If Case 1 happens for $G_j^i$, denote by $H_j,H'_j$ the two sub-graphs of $G^i_j$ induced by $\tilde A$ and $\tilde C$, respectively, and denote by $E^i_j$ the corresponding set of edges $E_{G^i_j}(\tilde A_j,\tilde C_j)$. We add $H_j,H'_j$ to $\gset_{i+1}$, and the edges of $E^i_j$ to $E^i$. If Case 2 happens for $G^i_j$, 
then we add $G_j^i$ to $\gset_{i+1}$, and we let $E_j^i=\emptyset$.

Let $E^i=\bigcup_{j=1}^{k_i}E^i_j$. Since for all $j$, $|E^i_j|\leq O(\dmax\sqrt {n_i\log n_i})$, we get that 
\[|E^i|\leq O(k_i\dmax\sqrt{n_i\log n_i})\leq O(\opt\cdot\dmax\sqrt{n_i\log n''}),\] 

as $k_i\leq 2\opt$. Finally, consider the collection $\gset_{i+1}$ of sub-graphs of $G'$. Each sub-graph in $\gset_{i+1}$ is canonical and contains at most $n_{i+1}$ vertices.

Let $H\in \gset_{i+1}$ be any graph in this collection, and let $X''=V(X)\cap H$. Notice that the set $X''$ of vertices defines a partition $\Sigma_{H}$ of the cycle $X'$ into consecutive segments, where each segment contains two vertices of $X''$ as its endpoints, and no vertices of $X''$ as its inner vertices. We say that graph $H$ is good iff there is a planar canonical drawing $\psi_H$ of $H$ and a closed simple curve $\gamma_H$, such that the vertices of $X''$ appear on $\gamma_H$, in the same order in which they appear in $X'$, and the whole graph $H$ is embedded inside $\gamma_H$ in $\psi_H$. Moreover,  no other graph $H'\in \gset_{i+1}$ contains a path connecting a pair of vertices $v,v'\in X'$, such that $v$ and $v'$ are {\bf inner} vertices that belong to two distinct segments $\sigma,\sigma'\in \Sigma_{H}$. Notice that if $H$ is not a good graph, then at least one of its edges participates in a crossing in $\phi$.
We can efficiently check whether graph $H$ is good.

Let $\gset'_{i+1}\sse \gset_{i+1}$ contain all good graphs in $\gset_{i+1}$, and let $\gset''_{i+1}$ contain all bad graphs.
Notice that $|\gset''_{i+1}|\leq 2\opt$. The graphs in $\gset''_{i+1}$ become the input to the next iteration, $(i+1)$. 

Assume now that for some graph $G_j^i$, for $1\leq j\leq k_i$, Case 3 happens.  
Then $\opt\geq \opt_j^i\geq \sqrt{|V(G_j^i)|}\geq \sqrt {n_i/2}$. In this case, $i$ becomes the last iteration, and we denote its index by $i^*=i$. Notice that $i^*=O(\dmax^2\log n'')$ 
We then add the edges of all sub-graphs $G_j^{i^*}$ for all $1\leq j\leq k_{i^*}$ to set $E^{i^*}$, except for the edges that participate in the grids $Z\in\zset'$. We claim that $|E^{i^*}|\leq (\opt)^3$ must hold: indeed, for all $1\leq j\leq k_{i^*}$, if
 $|E(G_j^{i^*})|>4|V(G_j^{i^*})|$, then by Theorem~\ref{thm: large average degree large crossing number}, $\opt^i_j\geq \Omega(|V(G_j^{i^*})|)$, and  $|E(G_j^{i^*})|\leq |V(G_j^{i^*})|^2=O((\opt^i_j)^2)$ must hold. The total number of edges in such graphs is then bounded by $O(\opt^2)$. On the other hand, for indices $j$ for which  $|E(G_j^{i^*})|\leq 4|V(G_j^{i^*})|\leq 4n_{i^*}\leq O(\opt^2)$, we get that $|E(G_j^{i^*})|\leq O(\opt^2)$, and all such graphs can contribute at most $O(\opt^3)$ edges to 
$E^{i^*}$.

Our final solution is $E'=\left (\bigcup_{i=1}^{i^*}E^i\right )\setminus E(X')$, and its cost is bounded by 
\[|E'|\leq \sum_{i=1}^{i^*}|E^i|\leq \sum_{i=1}^{i^*-1} O(\dmax\opt\sqrt{n_i\log n''})+O(\opt^3)\leq O(\dmax\opt\sqrt{n''\log n''}+\opt^3),\]

 since the values $\sqrt{n_i}$ form a decreasing geometric series for $i\geq 1$. Since $n''\leq O(n'\dmax)$, $|E'|\leq O(\opt\cdot \sqrt{n'}\cdot\poly(\dmax\cdot \log n')+\opt^3)$ as required.

It now remains to show that $E'$ is a weak feasible solution to problem $\pi(G',X',\zset')$. Let $\cset$ be the set of all connected components in graph $G'\setminus \left (\bigcup_{i=1}^{i^*}E^i\right )$. Each graph $H\in \cset$ is either an isolated vertex, or a grid $Z\in \zset'$, or it belongs to some set $\gset'_i$, for some $1\leq i<i^*$. Therefore, each such graph $H\in \cset$ has a planar drawing $\psi_H$, inside a simple closed curve $\gamma_H$, such that the vertices of $X'\cap V(H)$ all appear on $\gamma_H$, in the same order as in $X'$. It is also easy to verify that no graph in $\cset$ contains a path connecting a pair of inner vertices of two distinct segments $\sigma,\sigma'\in \Sigma_{H}$. We can then draw a closed curve $\gamma$ in the plane, and compose the planar drawings $\psi_H$ of all graphs $H\in \cset$ together, so that all vertices of $X'$ all appear on  $\gamma$, in exactly the same order as in $X$, and all other vertices and edges are drawn inside $\gamma$, with no crossings. We can then add the edges of $X'$ to this drawing, thus obtaining a planar drawing of $G'\setminus E'$.

\section{Proofs Omitted from Section~\ref{sec: iteration}}
We start with the following simple claim, that we use repeatedly in this section.

\begin{claim}\label{claim: well-linkedness in extended graph}
Let $G=(V,E)$ be any graph, $A\sse A'\sse V$ any subsets of vertices, such that $A$ is $\alpha$-well-linked, for any parameter $0<\alpha<1$. Moreover, assume that the graph $H=G[A']\cup \out_G(A')$ contains a collection $\pset$ of edge-disjoint paths, connecting the edges in $\out_G(A')\setminus \out_G(A)$ to the vertices of $A$, such that each edge in $\out_G(A')\setminus \out_G(A)$ participates in exactly one such path. Then set $A'$ is $\alpha$-well-linked.
\end{claim}

\begin{proof}
We will omit the sub-script $G$ in $\out_G(A),\out_G(A')$ for convenience from now on.
We can assume w.l.o.g. that each path $P\in\pset$ contains exactly one vertex of $A$, so the last edge on path $P$ belongs to $\out(A)$.

Assume for contradiction that $A'$ is not $\alpha$-well-linked, and let $(X',Y')$ be the violating partition. This gives a partition $(X,Y)$ of $A$, with $X=X'\cap A$, $Y=Y'\cap A$. Let $T_X=\out(A)\cap \out(X)$, and $T_Y=\out(A)\cap \out(Y)$, and assume w.l.o.g. that $|T_X|\leq |T_Y|$. Since $A$ is $\alpha$-well-linked, $|E(X,Y)|\geq \alpha |T_X|$. Let $T_{X'}=\out(A')\cap \out(X')$ (see Figure~\ref{fig: claim 4}). Recall that each edge $e\in T_{X'}\setminus \out(A)$ is associated with a path $P_e\in \pset$, connecting $e$ to a vertex of $A$, and these paths are edge-disjoint. Therefore, for each edge $e\in T_{X'}\setminus \out(A)$, either the endpoint of the path $P_e$ belongs to $X$, and so the last edge on the path belongs to $T_X\setminus T_{X'}$, or an edge of $P_e$ belongs to $E(X',Y')\setminus E(X,Y)$. Therefore:
$T_{X'}\setminus \out(A)\leq |E(X',Y')\setminus E(X,Y)|+|T_X\setminus T_X'|$. Altogether, 
since $T_{X'}\cap \out(A)= T_{X'}\cap T_X$:

 \[|T_{X'}|\leq |E(X',Y')\setminus E(X,Y)|+|T_X|\leq |E(X',Y')\setminus E(X,Y)|+|E(X,Y)|/\alpha\leq |E(X',Y')|/\alpha,\]
 
contradicting the fact that $(X',Y')$ is a violating cut.
\begin{figure}[h]
\scalebox{0.4}{\rotatebox{0}{\includegraphics{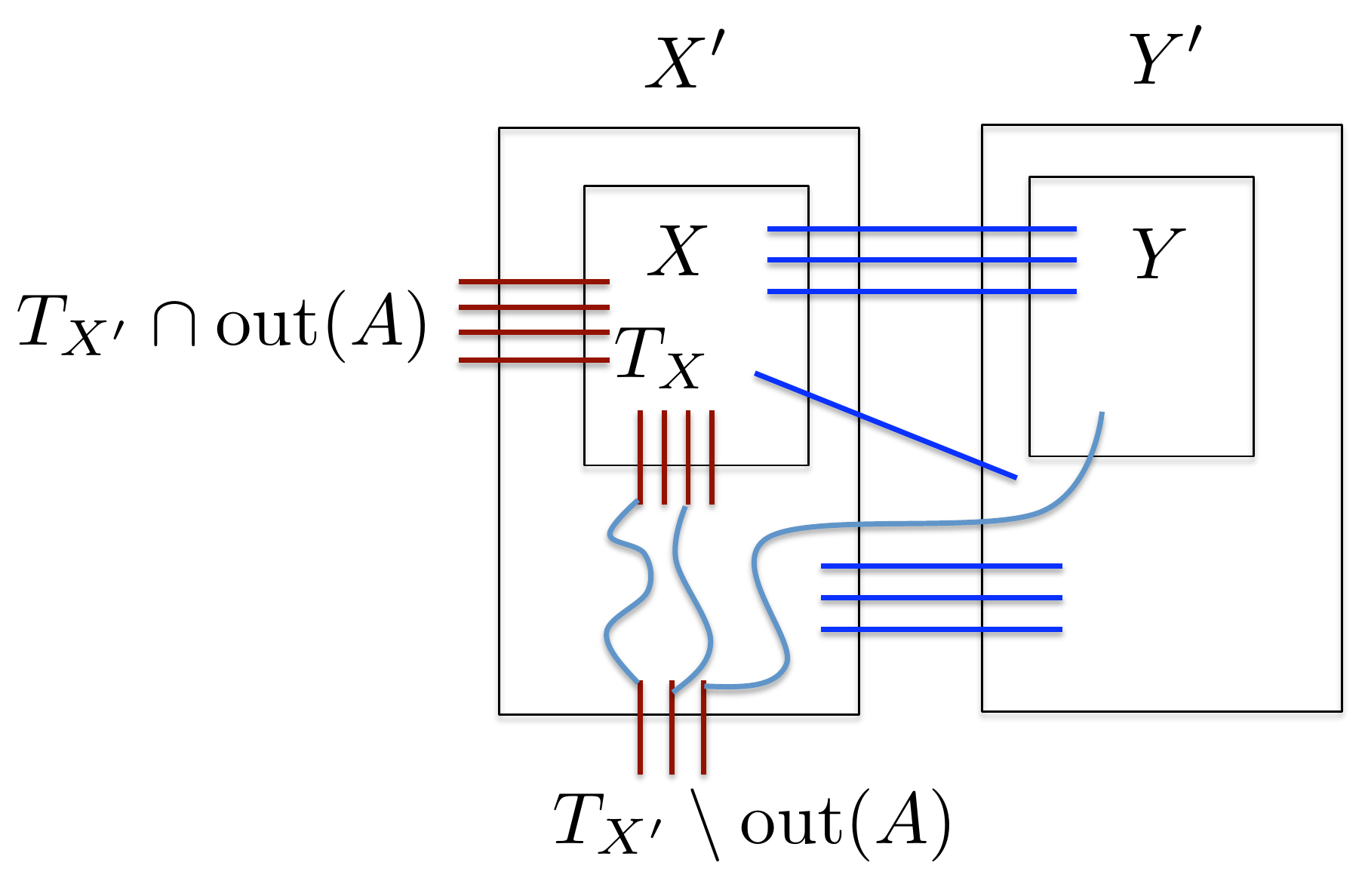}}} \caption{Illustration for Claim~\ref{claim: well-linkedness in extended graph}. \label{fig: claim 4}}
\end{figure}
\end{proof}

\subsection{Proof of Theorem~\ref{thm: initial partition}}

We first show that there is a partition $(X,Y)$ of $V(H)$, with $|X|,|Y|\geq n'/3$, and $|E(X,Y)|=O(\dmax\sqrt{n'})$. Indeed, let $E^*$ be the optimal planarizing set for the graph $H$. Observe that $|E^*|\leq \opt\leq O(\sqrt n')$, 
since we have assumed that $n'\geq m^*\rho^2> \opt^2$ in Equations~(\ref{eq: upper bound on rho in terms of n'}) and (\ref{eq: value of rho 2}).
Then graph $H\setminus E^*$ is planar, and so from Theorem~\ref{thm: planar separator}, there is a partition $(A',B',C')$ of $V(H)$, where $|A'|,|C'|\leq n'/3$, $|B'|\leq O(\sqrt{n'})$, and no edges connect vertices of $A'$ to vertices of $C'$ in $H\setminus E^*$. Assume w.l.o.g. that $|A'|\leq |C'|$, and consider the partition $(A'\cup B',C')$ of $V(H)$. Then $|A'\cup B'|,|C'|\geq n'/3$, and $|E_H(A'\cup B',C')|\leq |E^*|+\dmax|B'|\leq \opt +O(\dmax\sqrt{n'})\leq O(\dmax\sqrt{n'})$. We can then use Theorem~\ref{thm: ARV} to find a partition $(X_0,Y_0)$ of $V(H)$, with $|X_0|,|Y_0|\geq \eps n'$ for some constant $\eps$, and $|E(X_0,Y_0)|\leq O(\dmax\sqrt{n'\log n'})$.
We set the parameters  $N=|E(X_0,Y_0)|\leq O(\dmax\sqrt{n'\log n'})$, and $\lambda=\frac{\epsilon^2n'}{25N^2}=\Omega\left(\frac 1{\dmax^2\log n'}\right )$.

If there is any set $Z\in \zset'$, with $|Z|\geq \lambda n'$, then we can simply set let $Z=A$, $B=V(H)\setminus A$, to obtain the final partition $(A,B)$. It is easy to see that this partition satisfies the conditions of the theorem, since $|E(Z,V(H)\setminus Z)|=|\Gamma(Z)|\leq \sqrt{n'}$, and we have assumed that $|Z|\leq n'/2$, so $|B|\geq \lambda n'$ will also hold.

From now on we assume that there are no sets $Z\in \zset'$, with $|Z|\geq \lambda n'$.
 We perform a number of iterations. At the beginning of each iteration, we are given a partition $(X,Y)$ of $V(H)$, with $|X|,|Y|\geq \epsilon  n'$, and $|E(X,Y)|\leq N$.  The input to the first iteration is the partition $(X_0,Y_0)$. We then try to perform one of the  next two steps, while possible. 
 
Let $A\in \set{X,Y}$ be the set containing more vertices, and let $B$ be the other set. 
We set up a non-uniform sparsest cut problem on graph $H'=H[A]$, as follows. For each vertex $v\in \Gamma(A)$, we set its weight to be the number of edges in $\out(A)$ incident on $v$. For all other vertices, the weight is $0$. We then run the algorithm $\algsc$ on the resulting non-uniform sparsest cut instance. Assume first that we obtain a cut $(A_1,A_2)$ of sparsity less than $1$, and assume w.l.o.g. that $|A_1|\leq |A_2|$. Then $|E(A_1,A_2)|<|\out(A_1)\cap E(A,B)|$. We then define the new partition $(X,Y)$ to be $(A\setminus A_1,B\cup A_1)$. It is easy to see that the size of each resulting set is at least $\epsilon  n'$, and the number of the edges in the cut decreases.  
Otherwise, if the sparsity of the cut that we obtain is at least $1$, then, since $\algsc$ gives a factor $\alphasc$-approximation for the non-uniform sparsest cut problem, set $A$ is $1/\alphasc$ well-linked. We then perform the following step. 

Let $Z\in \zset'$, such that $Z\cap A\neq \emptyset$, and $Z\not\sse A$. Denote by $\Gamma_A(Z)=\Gamma(Z)\cap A$, and $\Gamma_B(Z)=\Gamma(Z)\cap B$. If $|\Gamma_A(Z)|<|\Gamma_B(Z)|$, then we define a new partition $(A\setminus Z, B\cup Z)$. It is easy to see that the new partition remains balanced, that is, $|A|,|B|\geq \epsilon  n'$, since we have assumed that $|Z|<\lambda n'<n'/4$. From Claim~\ref{claim: cut of grids}, $|E(A\cap Z,B\cap Z)|>|\Gamma_A(Z)|$. Since each vertex in $\Gamma(Z)$ has exactly one outgoing edge connecting it to vertices in $T(Z)$, it follows that the cut size strictly decreases.

We will perform one of the above two steps, while applicable. As observed above, the cut size can only go down, and eventually we will obtain a partition $(A,B)$, with $|A|\geq |B|\geq \epsilon  n'$, $|E(A,B)|\leq N$, such that $A$ is $1/\alphasc$-well-linked. Let $\zset''\sse \zset'$ denote the set of grids $Z\in \zset'$ with $A\cap Z\neq \emptyset$ and $Z\not\subseteq A$. Then $|\Gamma_A(Z)|\geq |\Gamma_B(Z)|$ must hold for each $Z\in \zset''$, and moreover, there is a collection $\pset(Z)$ of $|\Gamma_B(Z)|$ edge-disjoint paths in $H[Z]$, connecting the vertices in $\Gamma_B(Z)$ to the vertices in $\Gamma_A(Z)$, where each vertex in $\Gamma_B(Z)$ participates in exactly one such path. In our final step, we add all such sets $Z\in \zset''$ to the set $A$, to obtain the set $A'$. Let $B'=V(H)\setminus A'$. Our final partition is $(A',B')$. We now argue that it has all required properties. First, from Claim~\ref{claim: cut of grids},  for all $Z\in \zset''$,  $|E(A\cap Z,B\cap Z)|> |\Gamma_B(Z)|$, and so the size of the cut does not increase, and remains bounded by $N$, as required. Next, from Claim~\ref{claim: well-linkedness in extended graph}, set $A'$ is $1/\alphasc$-well-linked. Since $1/\alphasc=\Omega(1/(\sqrt{\log n}\cdot \log \log n))$, and $\alpha^*=\Omega(1/(\log^{3/2} n\log\log n))$, this implies that $A'$ is also $\alpha^*$-well-linked, as required. Finally, it is clear that $|A'|\geq \epsilon  n'$, and it is enough to show that $|B'|\geq \eps n'/10\geq \lambda n'$. Since $|B|\geq \eps n'$, it is enough to show that $\sum_{Z\in \zset''}|Z\cap B|\leq 0.9\epsilon  n'$.
 
We will use Claim~\ref{claim: numbers}, and define the set $\rset$ to contain, for each $Z\in \zset''$, the number $x(Z)=|Z\cap B|$, with the associated number $y_{x(Z)}=|E(Z\cap B,Z\cap A)|$. From Claim~\ref{claim: cutting the grid}, $x(Z)\leq 4y_{x(Z)}^2$, so we can set $\beta=4$. Since $|Z|<\lambda n'$ for all $Z\in \zset'$, we can set $M=\lambda n'$. Moreover, $\sum_{Z\in \zset''}|E(A\cap Z,B\cap Z)|\leq |E(A,B)|\leq N$, so we can set $S=N$. By Claim~\ref{claim: numbers}, $\sum_{Z\in \zset''}|Z\cap B|\leq 2S\sqrt{\beta M}+M/4\leq 2N\sqrt{4\lambda n'}+\lambda n'/4=4N\sqrt{\lambda n'}+\lambda n'/4$.
We now show that this number is bounded by $0.9\eps n'$.

Recall that $\lambda=\frac{\epsilon^2 n'}{25 N^2}\geq \Omega(\frac{1}{\dmax^2 \log n'})$. Therefore, $\lambda n'/4<\epsilon  n'/10$, and $4N\sqrt{\lambda n'}\leq 0.8\epsilon  n'$, giving the desired bound.

\subsection{Proof of Lemma~\ref{lemma: decomposition for Case 2}}
Recall that our starting point is two disjoint subsets $A,C$ of vertices of $H$, both of which are canonical w.r.t. $\zset'$, $\alpha^*$-well-linked in graph $H$, and contain at least $n'/\rho$ vertices. 

Construct a new graph from $H$, where all vertices in $A$ are contracted into a single source vertex $s$, and all vertices in $C$ are contracted into a single sink vertex $t$. Consider the minimum $s$-$t$ cut $\tilde{A},\tilde{B}$ in this graph, where $s\in \tilde{A}$, and let $(A',B')$ be the corresponding partition of the vertices of $H$, obtained from $(\tilde{A},\tilde B)$ after we un-contract $s$ and $t$. From Corollary~\ref{corollary: canonical s-t cut}, both sets $A'$ and $B'$ are canonical w.r.t. $\zset'$, and clearly both sets contain at least $n'/\rho$ vertices. From the min-cut max-flow theorem, there is a collection $\pset$ of $|E(A',B')|$ disjoint paths in $H$, where each such path connects an edge in $\out_H(A)$ to an edge in $\out_H(C)$, and contains exactly one edge in $E(A',B')$. Moreover, each edge in $E(A',B')$ is contained in exactly one path in $\pset$.
We now apply Claim~\ref{claim: well-linkedness in extended graph} to both sets $A'$ and $B'$ to conclude that they are both $\alpha^*$-well-linked.

\subsection{Proof of Lemma~\ref{lemma: decomposition for Case 3}}
The proof is very similar to the proof of Lemma~\ref{lemma: decomposition for Case 2}. Let $C$ be the union of all type-2 clusters. Recall that $|A|,|C|\geq \lambda n'/2\geq n'/\rho$ (by Equation~(\ref{eq: value of rho 2})), both sets $A$ and $C$ are canonical w.r.t. $\zset'$, $A$ is $\alpha^*$-well-linked w.r.t. $\out_H(A)$, and $H[C]\cup \out_G(C)$ contains a collection $\pset$ of edge disjoint paths, connecting the edges in $\out_H(C)$ to the edges in $\out^X(C)=\out_G(C)\setminus \out_H(C)$, such that every edge in $\out_H(C)$ is an endpoint of exactly one path.

As before, we construct a new graph from $H$, where all vertices in $A$ are contracted into a single source vertex $s$ and all vertices in $C$ are contracted into a single sink vertex $t$. Let $(\tilde A,\tilde B)$ denote the minimum $s$-$t$ cut in this network, as before, and let $(A',B')$ denote the corresponding partition of $H$. As before, both $A'$ and $B'$ are canonical w.r.t. $\zset'$, and they both contain at least $n'/\rho$ vertices. Similarly to the proof of Lemma~\ref{lemma: decomposition for Case 2}, set $A'$ is $\alpha^*$-well-linked w.r.t. $E_H(A',B')$. It now only remains to show the existence of the required paths in graph $H[B']\cup \out_G(B')$. Because of the min-cut max-flow theorem, there is a collection $\pset'$ of $|E_H(A',B')|$ edge-disjoint paths in the graph $H$, connecting edges in $\out_H(A)=E_H(A',B')$ to the edges of $\out_H(C)$, such that each edge of $E_H(A',B')$ appears on exactly one such path. This gives a collection $\pset''$ of paths in graph $(H[B']\setminus C)\cup \out_G(B'\setminus C)$, connecting the edges of  $E_H(A',B')$ to the edges of $\out_H(C)$, where every edge of $E_H(A',B')$ participates in one such path. The edges in $\out_H(C)$ are in turn connected by the paths in $\pset$ to the edges of $\out_X(C)$ in $H[B]$. The concatenation of paths in $\pset''$ and $\pset$ gives the desired collection of paths.

\subsection{Proof of Theorem~\ref{thm: case 4}}
We start with the set $\tset_1$ of all non-empty type-$1$ clusters. Consider the {\bf original contracted graph} $\H$ (which is different from the graph $H$ we have been working with so far). For each $C\in \tset_1$, the set $\out_{\H}(C)$ of edges consists of the following three subsets:

\begin{itemize}
\item edges in $E_1(C)$: recall that $C$ has property (P1) w.r.t. these edges, and in particular, for each such edge $e$, there is a path $P\sse \H\setminus C$, connecting $e$ to some vertex of set $A$. Moreover, for each $C\in \tset_1$, $E_1(C)\neq \emptyset$ (see Figure~\ref{fig: step 1 last partition}).

\item edges in $E_2(C)$: this set includes all edges in $\out_G(C)$ that do not belong to $E_1(C)$. Since the graph $H[C]\cup \out_G(C)$ contains a collection $\pset(C)$ of edge disjoint paths, connecting the edges of $E_1(C)$ to the edges of $E_2(C)$, and each edge in $E_2(C)$ participates in at least one path in $\pset(C)$, $|E_2(C)|\leq |E_1(C)|$. As observed before, each edge $e\in E_2(C)$ can reach the set $V(X)$ of vertices (the vertices of the bounding box) in the graph $G\setminus C$. Also, if $E_2(C)\neq \emptyset$, then there is a path $P(C)$, connecting a vertex of $C$ to a vertex of $X$ in graph $\H$, such that $P(C)$ does not contain vertices of any other clusters $C'\in \tset_1$.

\item edges in $\edges1,\ldots,\edges{h-1}$. From Invariant~(\ref{invariant 2: proper subgraph}), each such edge connects a vertex of $C$ to a vertex in $\H\setminus V(H)$, and from Invariant~(\ref{invariant 4.5: number of edges removed so far}), their number is bounded by $2m^*\cdot \rho\cdot \log n$, since the total number of iterations is at most $2\rho\cdot \log n$.
\end{itemize}

Since $\sum_{C\in \tset_1}|E_1(C)|\leq 2N$, and the total number of edges in sets $\edges1,\ldots,\edges{h-1}$ incident on vertices of $G$ is bounded by $2m^*\cdot \rho\cdot \log n\leq 2\sqrt{n'}\leq N$ (by Equation~(\ref{eq: upper bound on rho in terms of n'})), we have that $\sum_{C\in \tset_1}|\out_{\H}(C)|\leq 5N$. At the same time, since $|B|\geq \lambda/n'$, and Case 4 did not happen, $\sum_{C\in \tset_1}|C|\geq \lambda n'/2$, and for all $C\in \tset_1$, $|C|< n'/\rho$ (since Case 3 did not happen). In general, under these conditions, we could apply Claim~\ref{claim: numbers} to prove that there must be some set $C\in \tset_1$, with $|C|\geq \frac{2^{16}\cdot \dmax^6}{(\alpha^*)^2}\cdot|\Gamma(C)|^2$ --- the condition necessary for $C$ being a nasty set. However, we also need to ensure that $C$ has properties (P1) and (P2) in graph $\H$, which may not be true for all $C\in \tset_1$. Therefore, we first perform a transformation of sets $C\in \tset_1$ to ensure these properties. In this transformation, we will replace each cluster $C\in \tset_1$ with an augmented cluster $U_C$. We will then show that all but at most $2\opt+1$ augmented clusters have properties (P1) and (P2), and so these clusters still contain a large enough number of vertices, so that we can apply Claim~\ref{claim: numbers} to prove that one of them must be nasty.

Consider the graph $\H$, and remove the vertices in sets $A$, $X$, and all sets $C\in \tset_1$ from it. Let $\lset$ be the set of all connected components in the resulting graph. For each $S\in \lset$, if there is a cluster $C\in \tset_1$, such that $\out_{\H}(S)$ only contains edges connecting the vertices of $S$ to the vertices of $C$, we add $S$ to $C$. Let $U_C$ be the resulting extended cluster, for each $C\in \tset_1$.

\begin{claim}
There is at most one set $C\in \tset_1$, such that $U_C$ does not have the property (P1) in graph $\H$.
\end{claim}
\begin{proof}
Fix some $C\in \tset_1$. Observe that $\out_{\H}(U_C)\sse \out_{\H}(C)$. We first show that each edge $e\in \out_{\H}(U_C)$ can reach either $A$ or $X$ in the graph $\H\setminus U_C$.

Let $e\in \out_{\H}(U_C)$, and let $t\in T_{\H}(U_C)$ be its corresponding terminal. If $e\in E_1(C)$, then we have already established that there is a path connecting $e$ to $A$ in $H\setminus C$ and hence in $\H\setminus U_C$. Similarly, if $e\in E_2(C)$, then we have already established that there is a path connecting $e$ to $X$ in $H\setminus C$ and hence in $\H\setminus U_C$. Assume now that $e\not\in E_1(C)\cup E_2(C)$, so $e\not\in E_H(C)$. From Invariant~(\ref{invariant 2: proper subgraph}), $t\not \in V(H)$ must hold. Clearly, if $t\in X$, then there is nothing to prove.

Let $S\in \lset$ be the connected component to which $t$ belongs. Since $S$ was not added to $U_C$, there is another cluster, $C'\in \tset_1$, such that an edge $e'\in E(S,C')$ belongs to graph $\H$. Let $e''$ be any edge in $E_H(C')$. Then we can build a path $P$, connecting $t$ to $e''$, that does not contain any vertices of $U_C$. Since $e''\in E_H(C')$, there is a path $P'\sse H$, connecting $e''$ to some vertex of $A\cup X$. If $P'$ does not contain any vertices of $C$, then the concatenation of $P$ and $P'$ gives a path connecting $t$ to a vertex of $A\cup X$. Otherwise, we get a path connecting $t$ to some edge $e^*\in \out_H(C)$. Since all such edges can reach $A\cup X$ in the graph $H\setminus C$, this will also give the desired path. We therefore conclude that for each $C\in\tset_1$, for each $e\in \out_{\H}(U_C)$, edge $e$ can reach $A\cup X$ in the graph $\H\setminus U_C$.

Let $\tset'\sse \tset_1$ be the subset of clusters $C$, such that $U_C$ does not have property (P1). It is now enough to prove that $|\tset'|\leq 1$.

Assume otherwise, and let $C,C'\in \tset'$ be any two clusters. The high-level idea is that we show that each one of these clusters must contain some interface vertices that can reach $A$, and some interface vertices that can reach $X$. Therefore, one of these clusters can use the other cluster in order to connect all its interface terminals to each $A$.

We first claim that either $E_2(C)\neq \emptyset$, or $E_{\H}(C,X)\neq \emptyset$. Assume otherwise. Then every edge $e\in \out_H(C)$ can reach $A$ in the graph $\H\setminus U_C$. If an edge $e\in \out_{\H}(U_C)$ cannot reach $A$ in the graph $\H\setminus U_C$, then the terminal $t\in T_{\H}(C)$, that serves as the endpoint of $e$, cannot belong to $V(H)$, and must belong to some set $S\in \lset$ (we have assumed that it cannot belong to $X$). Since $S$ has not been added to $U_C$, there is an edge $e'$ connecting $S$ to some other cluster $C''$. Cluster $C''$, in turn, has an edge $e''\in E_1(C'')$. Therefore, we have a path $P$ connecting $e$ to $e''$ in $\H\setminus U_C$. Edge $e''$ is in turn connected by a path $P'\sse H$ to a vertex of $A$. If $P'$ does not contain vertices of $C$, then the concatenation of $P$ and $P'$ will give a path connecting $t$ to $A$ in $\H\setminus U_C$. Otherwise, path $P'$ must contain an edge in $E_H(C)=E_1(C)$, so we can obtain a path connecting $t$ to some edge in $E_1(C)$, and then to $A$ in graph $\H\setminus U_C$. Therefore, $U_C$ must have property (P1).

We can therefore assume that either $E_2(C)\neq \emptyset$, or $E_{\H}(C,X)\neq \emptyset$. In either case, there is a path $\tilde P(C)\sse \H$, connecting a vertex of $C$ to a vertex of $X$, such that $\tilde P(C)\cap V(C')=\emptyset$. Similarly, we can conclude that there is a path $\tilde P(C')\sse \H$, connecting a vertex of $C'$ to a vertex of $X$, such that $\tilde P(C')\cap V(C)=\emptyset$.

Assume now that there is a path $P^*(C)$, connecting a vertex of $C$ to a vertex of $A$ in $\H$, such that $P^*(C)$ does not contain any vertices of $C'$. We claim that in this case, $U_{C'}$ must have property (P1). Indeed, if $e\in \out_{\H}(U_{C'})$, such that there is a path $P$ connecting $e$ to $X$ in $\H\setminus U_{C'}$, then we can take the union $U$ of graphs $P,X,\tilde P(C),\H[U_C],P^*(C)$. No vertex of $U_{C'}$ belongs to $U$, and it is a connected sub-graph of $\H$, containing both $X$ and $A$. Therefore, edge $e$ can reach $A$ in graph $\H\setminus U_{C'}$, implying that $U_{C'}$ has property (P1).

Similarly, if there is a path $P^*(C')$, connecting a vertex of $C'$ to a vertex of $A$ in $\H$, such that $P^*(C')$ contains no vertices of $C$, then $U_C$ has property (P1).

The only remaining case is when every path connecting $C$ to $A$ in $\H$ contains a vertex of $C'$, and every path connecting $C'$ to $A$ in $\H$ contains a vertex of $C$. It is easy to see that this is impossible: recall that $E_1(C)\neq \emptyset$, and so there is an edge $e\in E_1(C)$, with a path $P$ connecting $e$ to $A$ in graph $H\setminus C$. Let $v$ be the last vertex on $P$ that belongs to $C'$, and let $P'$ be the portion of $P$ starting from $v$. Then $P'$ connects $C'$ to $A$ and contains no vertices of $C$.
\end{proof}

Let $\tset''\sse \tset_1$ be the collection of all clusters $C$, such that $U_C$ has properties (P1) and (P2). Since all clusters $U_C$, for $C\in \tset_1\setminus \tset'$ have property (P1), if any such cluster does not have property (P2), its edges must participate in crossings in any drawing. Therefore, the total number of clusters in $\tset_1\setminus \tset''$ is at most $2\opt+1$, and they contain at most $\frac{n'}{\rho}\cdot (2\opt+1)<\frac{\lambda n'} 4$  vertices in total (we have used Equation~(\ref{eq: value of rho 2})).
Therefore, $\sum_{C\in \tset''}|C|\geq \lambda n'/4$.

This leaves us with the set $\tset''$ of clusters. For each cluster $C\in \tset''$, set $U_C$ is canonical, and has properties (P1) and (P2) in $\H$.
It now only remains to show that, for at least one such set, $|U_C|\geq \frac{2^{16}\cdot \dmax^6}{(\alpha^*)^2}\cdot|\Gamma(U_C)|^2$. Since $\Gamma(U_C)\sse \Gamma(C)$, and $|\Gamma(C)|\leq |\out_{\H}(C)|$, it is enough to prove that for some $C\in \tset''$, $|C|\geq \frac{2^{16}\cdot \dmax^6}{(\alpha^*)^2}\cdot|\out_{\H}(C)|^2$. This would imply that the corresponding set $U_C$ is a nasty canonical set. Assume for contradiction that for all $C\in\tset''$, $|C|< \frac{2^{16}\cdot \dmax^6}{(\alpha^*)^2}\cdot|\out_{\H}(C)|^2$. We will show that $\sum_{C\in \tset''}|C|<\lambda n'/4$ in this case, leading to a contradiction.

We use Claim~\ref{claim: numbers} by defining, for each cluster $C\in \tset''$, value $x(C)=|C|$, and $y_{x(C)}=|\out_{\H}(C)|$. We can then use $\beta= \frac{2^{16}\cdot \dmax^6}{(\alpha^*)^2}$, $M=n'/\rho$, and $S=5N$. From Claim~\ref{claim: numbers}, we get that $\sum_{C\in \tset'}|C|\leq 2S\sqrt{\beta M}+M/4=10N\sqrt{\frac{n'}{\rho}\cdot\frac{2^{16}\cdot \dmax^6}{(\alpha^*)^2}}+\frac{n'}{4\rho}$. Since, by Equation~(\ref{eq: value of rho 2}), $\rho>8/\lambda$, the second term is less than $\lambda n'/8$, at it is enough to show that the first term is bounded by $\lambda n'/8$ as well. This is equivalent to:

\[\rho>\frac{25\cdot 2^{24}\cdot \dmax^6 N^2}{n'\cdot \lambda^2\cdot (\alpha^*)^2}\]

which is guaranteed by Equation~(\ref{eq: value of rho 1}).

\subsection{Proof of Theorem~\ref{thm: initial skeleton for Cases 3 and 4}}\label{sec: proof of theorem skeleton construction for cases 3 and 4}

Since Case 2 does not happen, we assume throughout the proof that:

\begin{equation}\label{eq: bound on T}
 |T|> \frac{10^7\opt^2\cdot \rho\cdot\log^2n\cdot \dmax^2\cdot \betaFCG }{\alpha^*}
 \end{equation}
 
 The plan of the proof is as follows. We start by analyzing some structural properties of the graph $H$ and of the drawing $\phi'$ of $G$ for Cases 3 and 4. Next, we use a simple randomized algorithm to construct a skeleton $K^0$, which may not be connected. We will then show that at least one connected component $K'$ of $K^0$ has all the desired properties, except that it may contain $2$-vertex cuts, and therefore may not be rigid. Finally, we show how to turn $K'$ into a rigid graph, while preserving all other properties.
  
\subsubsection*{Structural Properties of $H$ and $\phi'$}
Recall that we are given a partition $(A,B)$ of $H$. 
For convenience, we will denote $H[A]=H_r$, and call its vertices and edges ``red'', and $H[B]=H_b$, and call its vertices and edges ``blue''. We will also denote $H_r=(V_r,E_r)$ and $H_b=(V_b,E_b)$. Then $T=E(V_r,V_b)$.

Let $\phi_r,\phi_b$ be the drawings of $H_r$ and $H_b$, respectively, induced by $\phi'$.
Since the graphs $H_r,H_b$ are not necessarily $2$-vertex connected, the face boundaries in $\phi_r,\phi_b$ are not necessarily simple. For a face $F$ in any embedding, we denote by $\gamma(F)$ the boundary of the face.

\begin{definition}
Let $H'\sse G$ be any sub-graph of $G$, and let $F$ be any face of the drawing $\phi_{H'}$ of $H'$ induced by $\phi'$. Let $e\not\in E(H')$ be any edge. We say that $e$ is embedded inside $F$, iff the image of the edge $e$ in $\phi'$ is completely contained inside the face $F$ (we allow one or both endpoints of $e$ to lie on $\gamma(F)$.)
\end{definition}


\begin{claim}\label{claim: most terminals embedded inside same face} Assume that Case 3 or Case 4 happen. Then
there is a subset $T^0\sse T$ of terminals, $|T^0|\geq |T|-7\opt\cdot \betaFCG /\alpha^*$, a face $F_b$ of $\phi_b$ and a face $F_r$ of $\phi_r$, such that all terminals in $T^0$ are embedded inside the face $F_r$ of $\phi_r$, and inside the face $F_b$ of $\phi_b$.  Moreover, the drawings of these terminals in $\phi'$ do not cross each other (or any other edge).
\end{claim}
(See figure~\ref{fig: red-blue-terminals} for an illustration)

\begin{proof}
Remove from $T$ all terminals whose embeddings participate in crossings in $\phi'$, and let $T^*$ be the set of remaining terminals, $|T^*|\geq |T|-2\opt$.
\ From Equation~\ref{eq: bound on T},  $|T|\geq 14\betaFCG \cdot \opt/\alpha^*$.

Assume for contradiction that there is no face $F_b$ in the embedding $\phi_b$ of $H_b$, such that at least $|T^*|-3\opt\betaFCG /\alpha^*$ terminals of $T^*$ are embedded inside $F_b$.  We will find a set $M$ of $3\betaFCG \cdot \opt/\alpha^*$ disjoint pairs of terminals of $T^*$, such that each pair contains two terminals whose red endpoints lie inside distinct faces of $\phi_b$.  Since graph $H_r$ is $\alpha^*$-well-linked w.r.t. $T$, from Observation~\ref{observation: existence of flow in well-linked instance}, we can find a flow $F$ in $H_r$, that routes these pairs of terminals with congestion at most $2\betaFCG /\alpha^*$ in $H_r$. Each flow-path in $F$ must contain a red edge whose drawing crosses a blue edge. With the maximum congestion being bounded by $2\betaFCG /\alpha^*$, and the total amount of flow in $F$ being $3\betaFCG \cdot \opt/\alpha^*$, we get that there are more than $\opt$ red edges, whose drawings cross those of blue edges in $\phi'$, a contradiction. 

To find the desired set $M$ of terminal pairs, consider two cases. First, if any face $F'_b$ of $\phi_b$ contains at least $3\betaFCG \cdot \opt/\alpha^*$ terminals, then we can simply match  $3\betaFCG \cdot \opt/\alpha^*$ terminals lying outside $F'_b$ to $3\betaFCG \cdot \opt/\alpha^*$ terminals lying inside $F'_b$. Otherwise, we can find a set $\fset$ of faces, containing between $3\betaFCG \cdot \opt/\alpha^*$ and $6\betaFCG \cdot \opt/\alpha^*$ terminals in total. We can then match terminals lying outside these faces to terminals lying inside these faces.

We conclude that in both Case 3 and Case 4, there is a subset $T'\sse T^*$ of terminals, $|T'|\geq |T^*|-3\opt\betaFCG /\alpha^*$, such that all terminals in $T'$ are embedded inside the same face $F_b$ of $\phi_b$. 

We now show that there is a subset $T''\sse T^*$ of terminals, $|T''|\geq |T^*|-3\opt\betaFCG /\alpha^*$, such that all terminals of $T''$ are embedded inside the same face $F_r$ of $\phi_r$. In the end, we will set $T^0=T'\cap T''$.

If Case 3 happens, then $H_b$ is $\alpha^*$-well-linked w.r.t. $T$, so the same argument as for $H_r$ works. Assume now that Case 4 happens. Since the edges of $X$ do not participate in any crossings in $\phi'$, all vertices of $X$ are embedded inside the same face $F_r$ of $\phi_r$. Recall that we have a collection $\pset$ of edge-disjoint paths, connecting the terminals in $T$ to the vertices of $X$. If a terminal $e\in T$ is not embedded inside $F_r$, then the corresponding path has to cross the boundary of $F_r$. Therefore, all but $\opt$ terminals of $T^*$ must be embedded inside $F_r$. 

We obtain the final set of terminals, by setting $T^0=T'\cap T''$.
\end{proof}

\begin{figure}[h]
\scalebox{0.4}{\rotatebox{0}{\includegraphics{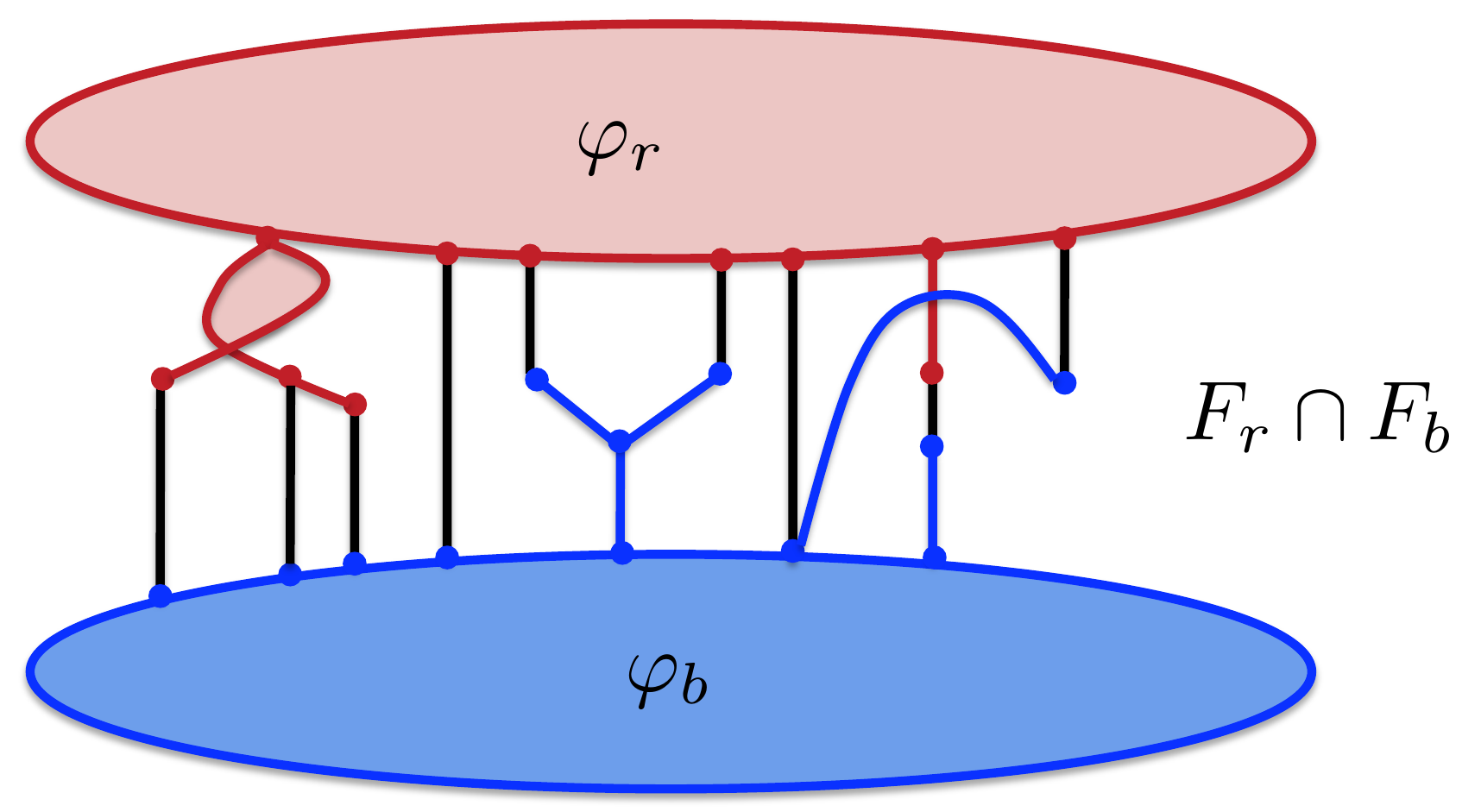}}} \caption{Drawing of $H_r,H_b$ and $T^0$ in $\phi$. Terminals in $T^0$ are shown in black. \label{fig: red-blue-terminals}}
\end{figure}

Let $N'=|T^0|$. From the above claim, $N'\geq |T|-7\opt\cdot \betaFCG /\alpha^*\geq |T|/2$.
Since $H_r$ is a connected graph, we can draw a simple closed curve $\tau$ along the boundary $\gamma(F_r)$ of $F_r$, inside $F_r$. This can be done so that $\tau$ does not cross any red edges or the edges of $X$, and it crosses every terminal in $T^0$ exactly once. Moreover, $\tau$ does not cross any blue edges, except for those blue edges $e$ whose drawing crosses the drawing of $\gamma(F_r)$ (See Figure~\ref{fig: red-blue-with-curve}). For each terminal $t\in T^0$, we let $p_t$ denote the point where the drawing of $t$ crosses the curve $\tau$. The order of the points $\set{p_t}_{t\in T^0}$ along the curve $\tau$ defines a circular ordering $\pi$ of the terminals in $T^0$. We will use the ordering $\pi$ of the terminals in $T^0$ throughout the proof.
We will sometimes say that some subset of terminals appears on $\tau$ in a certain order, to mean that their corresponding points $p_t$ lie on $\tau$ in this order.

\begin{figure}[h]
\scalebox{0.4}{\rotatebox{0}{\includegraphics{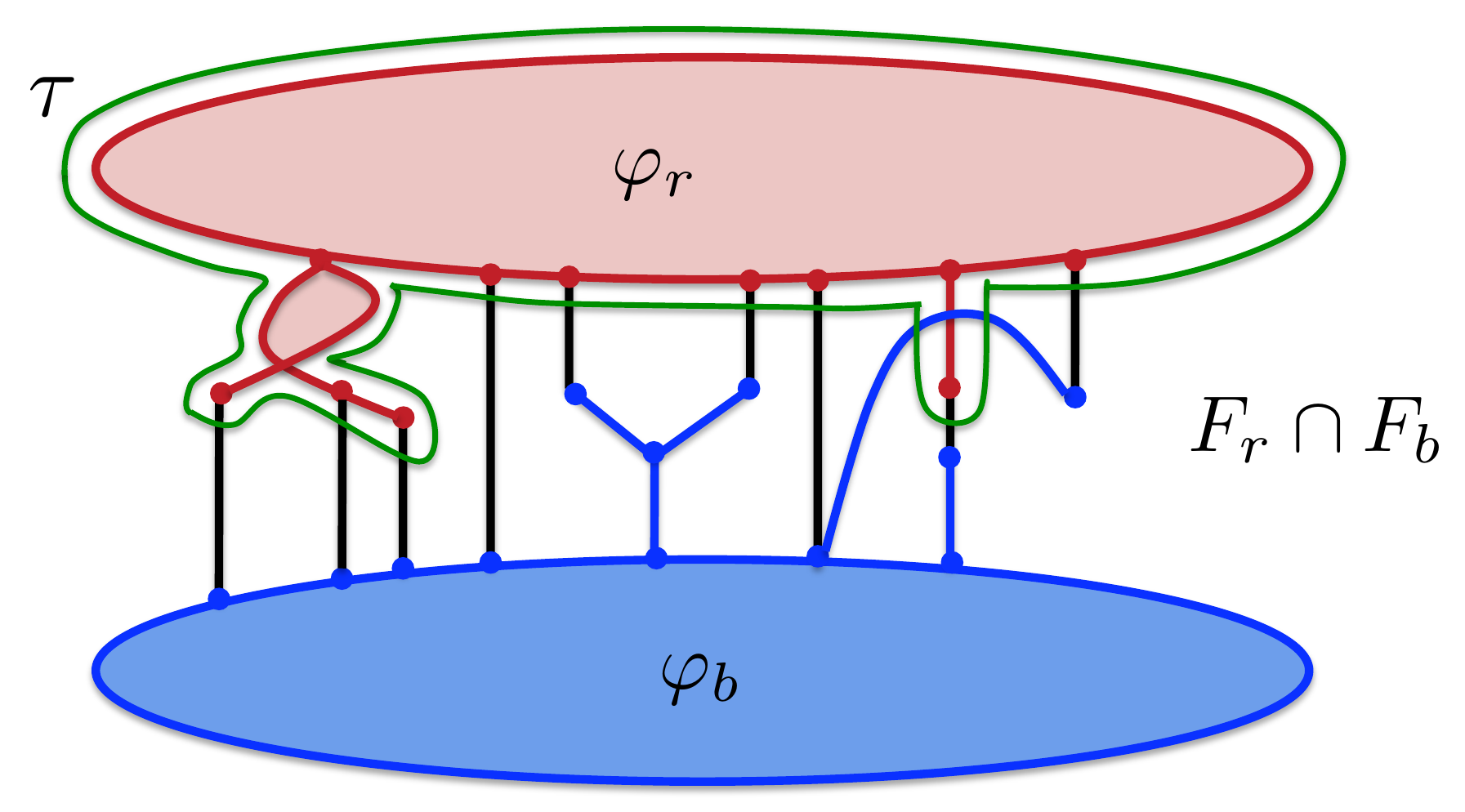}}} \caption{Curve $\tau$ is show in green. \label{fig: red-blue-with-curve}}
\end{figure}


Let $P=\frac{1}{2\rho\cdot \opt}$ be the allowed failure probability. Since we have assumed that $\opt\cdot \dmax^{6}\leq \sqrt n$, and $\rho=\Theta(\log n\cdot \dmax^2)\cdot \max\set{\dmax^{10}\log^5 n(\log\log n)^2,\opt}$, we can assume that $P\geq 10/{n^2}$.

Assume that Case 3 or Case 4 happen. In either case, the set $V_r$ of vertices is $\alpha^*$-well-linked w.r.t. the set $T$ of terminals. Therefore, from Observation~\ref{observation: existence of flow in well-linked instance}, there is a flow $F_r$ in $H_r\cup T$, where every pair $t,t'\in T$ of terminals sends one flow unit to each other, and the edge congestion is at most  $|T|\betaFCG /\alpha^*$. For each pair $t,t'\in T$ of terminals, we randomly choose a path $P_r(t,t')\sse H_r\cup\set{t,t'}$ connecting them, from the distribution induced by the flow $F_r$, and we let $\pset_r$ denote the set of all these selected paths.

Consider now Case 3.  Set $V_b$ is also $\alpha^*$-well-linked, and so there is a flow $F_b$ in $H_b\cup T$, where each pair $t,t'\in T$ of terminals sends one flow unit to each other, and the total congestion is at most $|T|\betaFCG/\alpha^*$. As before, each pair $t,t'\in T$ of terminals chooses a random path $P_b(t,t')$ according to the distribution induced by flow $F_b$, and we let $\pset_b$ be the set of all these chosen paths.

For Case 4, we construct the collection $\pset_b$ of paths as follows. Recall that we have a collection $\pset$ of edge-disjoint paths in graph $H_b\cup \out_G(V_b)$, connecting every terminal in $T$ to a vertex of $X$. For every pair $t,t'\in T$ of terminals, we construct a path $P_b(t,t')$, by concatenating $P_t,P_{t'}$ and one of the two segments of $X$ connecting $t$ to $t'$. Consider the resulting collection $\pset_b$ of paths. The congestion on the edges of $H_b$ is at most $|T|$ (but the congestion on the edges of $X$ may be higher). 

Let $\rset$ be the set of all pairs of distinct terminals in $T$. For each pair $(t,t')\in \rset$, the concatenation of the paths $P_r(t,t')$ and $P_b(t,t')$ gives a simple cycle $L(t,t')$ that contains both terminals. Let $\lset=\set{L(t,t')\mid (t,t')\in \rset}$. 

\subsubsection*{Constructing and analyzing $K^0$}
We select a random collection $\rset^*\sse \rset$ of pairs of terminals, where every pair is added to $\rset^*$ with probability $P'=\frac{P\cdot \alpha^*}{30\cdot \opt\cdot \betaFCG \cdot |T|}$. Let $\lset^*=\set{L(t,t')\mid (t,t')\in \rset^*}$.
The initial skeleton $K^0$ consists of the union of all cycles in $\lset^*$. Note that $K^0$ is not necessarily connected. 

Let $\event_1$ be the bad event that either a bad edge, or a terminal in $T\setminus T^0$ belong to $K^0$, or some edge of $H$
belongs to more than $4\log n$ cycles in $\lset^*$.
\begin{claim} $\prob{\event_1}\leq P/2$.
\end{claim}
\begin{proof}
Recall that the number of bad edges is at most $2\opt$. Let $e$ be any such bad edge. Recall that $e$ cannot belong to $X$, so assume first that it belongs to $H_r$. For each pair $(t,t')$ of terminals, let $f_e(t,t')$ be the total amount of flow in $F_r$, that is routed between $t$ and $t'$ via flow-paths that contain $e$. Clearly, $f_e(t,t')\leq 1$, and since the edge congestion is bounded by $|T|\betaFCG/\alpha^*$, $\sum_{t,t'\in T}f_e(t,t')\leq |T|\betaFCG/\alpha^*$. For each pair $t,t'\in T$ of terminals, the probability that the path $P_r(t,t')$ contains $e$ is precisely $f_e(t,t')$, and the probability that pair $(t,t')$ belongs to $R^*$ is $P'$. Therefore, the probability that edge $e$ belongs to $K^0$ is at most:

\[\sum_{t,t'\in T}P'\cdot f_e(t,t')=P'\cdot |T|\betaFCG/\alpha^*\]

This analysis works for any bad edge $e\in E_r$, and for bad edges in $E_b$ in Case 3. If Case 4 happens, then for each bad edge $e\in E_b$, there are at most $|T|$ paths in $\pset_b$ containing $e$, and so the probability that $e$ belongs to $K^0$ is bounded by $|T|\cdot P'<P'\cdot |T|\betaFCG/\alpha^*$.

Using the union bound, the probability that any bad edge belongs to $K^0$ is then at most:

\[2\opt\cdot P'\cdot |T|\betaFCG/\alpha^*  \leq 2\opt\cdot \frac{P\cdot \alpha^*}{30\cdot \opt\cdot \betaFCG \cdot |T|}\cdot |T|\betaFCG/\alpha^*\leq P/15\]

The number of terminals in $T\setminus T^0$ is at most $7\opt\cdot \betaFCG /\alpha^*$. Each terminal belongs to $K^0$ with probability  at most $P'\cdot |T|=\frac{P\cdot \alpha^*}{30\cdot \opt\cdot \betaFCG }$. Using the union bound, the probability that a terminal in $T\setminus T^0$ belongs to $K^0$ is at most $7P/30$.

Finally, fix some edge $e\in E(H)$. For every pair $t,t'\in T$ of terminal, we define $f_e(t,t')$ as before. Using the same reasoning as above, the expected number of cycles in $\lset^*$ containing $e$ is at most:

\[P'\cdot |T|\betaFCG/\alpha^*<1\]

Using the standard Chernoff bound, together with the union bound, the probability that the congestion on any edge of $H$ is more than $4\log n$ is bounded by $1/n^3<P/10$.
\end{proof}

(Notice that the congestion on edges of $X$ can be large for Case 4.) If event $\event_1$ does not happen, then the only edges of $K^0$ whose images cross the curve $\tau$ are the terminals in $T^0$. This is since the red edges do not cross $\tau$, and the only blue edges that may cross $\tau$ are bad edges. 

Consider the fixed ordering $\pi$ of the terminals in $T^0$, given by the order in which the terminals in $T^0$ cross the curve $\tau$. Recall that $N'=|T^0|\geq |T|/2$.
For every pair $(t,t')\in T^0$ of terminals, let $\Delta(t,t')$ be the length of the shorter distance between the two terminals in the circular ordering $\pi$. Let $\rset'=\set{(t,t')\mid t,t'\in T^0, \Delta(t,t')\geq N'/4}$. Notice that each terminal $t\in T^0$ participates in at least $N'/2$ pairs in $\rset'$. 
Let $\rset^{**}=\rset'\cap \rset^*$ be the set of pairs that have been selected to  $\rset^*$, and let $\lset^{**}=\set{L(t,t')\mid (t,t')\in \rset^{**}}$. Let $K^{**}$ be the graph consisting of the union of the cycles in $\lset^{**}$. Notice that $K^{**}\sse K^0$.

We now define two additional bad events and prove that with high probability they do not happen. We then show that this implies that with high probability, graph $K^{**}$ is connected, and each connected component of $H\setminus K^{**}$ contains a small number of terminals. Therefore, there is a connected component $K'$ of $K^0$ that contains $K^{**}$. We will then use $K'$ as our skeleton, after removing all $2$-vertex cuts from it.

\begin{definition}
Let $T'\sse T^0$ be any subset of terminals in $T^0$. We say that $T'$ is a consecutive set, iff the terminals in $T'$ appear consecutively in the circular ordering $\pi$.
\end{definition}

Let $\event_2$ be the bad event that there are two disjoint consecutive sets $T_1,T_2\sse T^0$ of terminals (but $T_1\cup T_2$ is not necessarily a consecutive set), with $|T_1|,|T_2|\geq N'/4$, such that the number of pairs $(t,t')\in T_1\times T_2$ of terminals belonging to $\rset^*$ is at most $4\log n\cdot \dmax$. Let $\event_3$ be the bad event that there is a consecutive set $T'\sse T^0$ of $\frac{720\cdot \opt^2\cdot \rho\cdot \log n\cdot \betaFCG }{\alpha^*}$ terminals, such that $T'\cap E(K^{**})=\emptyset$.

\begin{claim} $\prob{\event_2}\leq P/10$ and $\prob{\event_3}\leq P/10$.
\end{claim}
\begin{proof}
Let $T_1,T_2$ be any pair of disjoint consecutive subsets of terminals in $T^0$, with $|T_1|,|T_2|=\lceil N'/4\rceil$.
Recall that every pair of terminals is selected to $\rset^*$ with probability $P'=\frac{P\cdot \alpha^*}{30\cdot \opt\cdot \betaFCG \cdot |T|}$. There are at least $\left (\frac{N'} 4\right )^2\geq \left (\frac{|T|} 8\right )^2= \frac{|T|^2}{64}$ pairs of terminals in $T_1\times T_2$.
Therefore, the expected number of pairs $(t,t')\in T_1\times T_2$ of terminals in $\rset^*$ is at least:

\[\frac{|T|^2}{64}\cdot \frac{P\cdot \alpha^*}{30\cdot \opt\cdot \betaFCG \cdot |T|}=\frac{|T|\cdot \alpha^*}{3840\cdot\rho\cdot \opt^2\cdot \betaFCG }\geq 48 \log n\]

since $|T|\geq  \frac{10^7\opt^2\cdot \rho\cdot\log^2n\cdot \dmax^2\cdot \betaFCG }{\alpha^*}$ from Equation~\ref{eq: bound on T}.
Using Chernoff bound, the probability that at most $4\log n$ such pairs are selected is bounded by $e^{-48\log n/8}= 1/n^6$.
The number of possible choices for sets $T_1,T_2$ is bounded by $|T|^2\leq n^2$. Using the union bound, we get that $\prob{\event_2}\leq 1/n^4<P/10$.

In order to bound the probability of $\event_3$, recall that the probability that a terminal $t$ belongs to $E(K^{**})$ is the probability that a pair $(t,t')\in \rset'$ is selected. There are at least $N'/2\geq |T|/4$ such pairs for each terminal $t\in T^0$, and each pair is selected with probability $\frac{P\cdot \alpha^*}{30\cdot \opt\cdot \betaFCG \cdot |T|}$. Therefore, the probability that terminal $t$ does not belong to $E(K^{**})$ is at most $\left (1-\frac{P\cdot \alpha^*}{30\cdot \opt\cdot \betaFCG \cdot |T|}\right )^{|T|/4}\leq e^{-\frac{P\cdot \alpha^*}{120\cdot \opt\cdot \betaFCG}}$.

Let $T'$ be any consecutive set of $\frac{720\cdot \opt^2\cdot \rho\cdot \log n\cdot \betaFCG }{\alpha^*}$ terminals. Since $\frac{720\cdot \opt^2\cdot \rho \cdot \log n\cdot \betaFCG }{\alpha^*}<\frac{|T|}2\leq \frac{N'} 4$ (from Equation~\ref{eq: bound on T}), the events that different terminals $t\in T$ belong to $E(K^{**})$ are mutually independent. Therefore, the probability that no terminal of $T'$ belongs to $E(K^{**})$ is at most: 

\[e^{-\frac{P\cdot \alpha^*}{120\cdot \opt\cdot \betaFCG}\cdot \frac{720\cdot \opt^2\cdot\rho\cdot \log n \cdot\betaFCG }{\alpha^*}}\leq e^{-3\log n}\leq 1/n^3\]

Since there are at most $|T^0|\leq |T|\leq n$ possible choices for set $T'$, using the union bound finishes the proof.
\end{proof}

We denote $K^{**}=(V^{**},E^{**})$, and we denote by $K^{**}_r$ and $K^{**}_b$ the sub-graphs of $K^{**}$ induced by $V_r\cap V^{**}$ and $(V_b\cup X)\cap V^{**}$, respectively.

\begin{claim}
If none of the events $\event_1,\event_2,\event_3$ happens, then both graphs $K^{**}_r$ and $K^{**}_b$ are connected.
\end{claim}
\begin{proof}
Assume for contradiction that $K^{**}_r$ is not connected, and let $\cset$ be the set of all connected components of $K^{**}_r$. For each connected component $C\in \cset$, let $\pi(C)$ be the smallest segment of $\pi$ containing the terminals whose red endpoints belong to $C$. Let $C\in \cset$ be the connected component for which $\pi(C)$ contains the smallest number of terminals, and let $T'(C)$ be the set of terminals contained in $\pi(C)$.

We first show that no other connected component $C'\in\cset$ contains red endpoints of terminals in $\pi(C)$. Assume for contradiction that this is not true, and let $C\neq C'$ be some connected component in $\cset$, that contains some terminal $t^*\in \pi(C)$. Let $t,t'$ be the first and the last terminals of $\pi(C)$. Then $C$ contains a simple path $Q$ connecting $t$ to $t'$. The union of the image of path $Q$ in $\phi'$ with the segment of $\tau$ spanning $\pi(C)$, and the portions of the images of $t$ and $t'$ between their red endpoints and $p_t$, $p_t'$, respectively, forms a simple closed curve, that we denote by $\tau'$. This curve partitions the plane into two faces $F,F'$. The red endpoints of all terminals in $\pi(C)$ lie inside one of the faces (say $F$), or on $\tau'$, while the red endpoints of all terminals in $T^0\setminus \pi(C)$ must lie strictly inside the other face, $F'$. Since component $C'$ does not share any vertices with $C$, and the drawings of the two components do not cross, the red endpoint of $t^*$ must lie strictly inside $F$, and so the image of $C'$ must also lie strictly inside $F$. It follows that all terminals whose red endpoints belong to $C'$ are contained in $\pi(C)$, and so $|\pi(C')|<|\pi(C)|$ must hold, contradicting the minimality of $\pi(C)$.

We conclude that for each connected component $C'\in \cset$, if $C'\neq C$, then $C'$ does not contain red endpoints of terminals in $\pi(C)$. In particular, $|\pi(C)|\leq N'/2$ must hold: since $\cset$ contains at least two connected components, let $C'\in \cset$ be any component, $C'\neq C$. Then all terminals of $C'$ are contained in $T^0\setminus \pi(C)$. So if $|\pi(C)|>N'/2$, then $\pi(C')\sse T^0\setminus \pi(C)$, and so $|\pi(C')|< N'/2<|\pi(C)|$, contradicting the minimality of $\pi(C)$.

Since $C$ is a connected component of $K^{**}_r$, there is a pair $t,t'\in T'(C)$ of terminals, such that $L(t,t')\in \lset^{**}$. Therefore, $\Delta(t,t')\geq N'/4$, and $|T'(C)|\geq N'/4$.

 Let $T_1$ be any consecutive subset of terminals that belong to $\pi(C)$, such that $|T_1|=\lceil N'/4\rceil$. Since $|T'(C)|\leq N'/2$, we can find another consecutive set $T_2$ of $\lceil N'/4\rceil$ terminals, such that $T_2\sse T^0$, $T_2\cap T'(C)=\emptyset$, and for each pair $(t,t')\in T_1\times T_2$ of terminals, $\Delta(t,t')\geq N'/4$. Since all edges of $K^{**}$ are good, and $C$ is a connected component of $K^{**}_r$, no cycle $L(t,t')$ for $(t,t')\in T_1\times T_2$ belongs to $\lset^{*}$ (since any such cycle would have belonged to $\lset^{**}$), and therefore no pair $(t,t')\in T_1\times T_2$ of terminals belongs to $\rset^{*}$, contradicting the assumption that Event $\event_2$ did not happen.


The proof that $K^{**}_b$ is connected is similar. We only need to note that since all edges of $K^{**}_b$ are good, no edges of $K^{**}_b$ cross the curve $\tau$, so the same reasoning as for $K^{**}_r$ works.
\end{proof}

We conclude that if events $\event_1,\event_2,\event_3$ do not happen, there is a connected component $K'$ of $K^0$ that contains $K^{**}$. 
The next lemma will imply that every connected component of $G\setminus K^{**}$ contains at most $O\left (\frac{\opt^2\cdot \rho\cdot \log n\cdot \betaFCG }{\alpha^*}\right )$ terminals. Since we also use this lemma later, we state it in a more general form here.

Let $\tilde{K}=(\tilde V,\tilde E)$ be any sub-graph of $G$. Let $\tilde K_r$ be the sub-graph of $\tilde K$ induced by $V_r\cap \tilde{V}$ and let  $\tilde K_b$ be the sub-graph induced by $(V_b\cup V(X))\cap \tilde V$. Let $\tilde{T}=T\cap \tilde{E}$ be the set of terminals contained in $\tilde K$, and let $\phi_{\tilde K}'$ be the drawing of $\tilde{K}$ induced by the drawing $\phi'$ of $G$. For each face $F$ of $\phi_{\tilde K}'$, we denote by $T_F\sse T^0\setminus \tilde T$ the subset of terminals $t\in T^0\setminus \tilde T$, which are embedded inside $F$ in $\phi'$.

\begin{lemma}\label{lemma: each face contains consecutive set of terminals}
Let $\tilde K$ be as above, such that all edges of $\tilde{K}$ are good, both $\tilde K_r$ and $\tilde K_b$ are connected graphs, and $\tilde{T}\sse T^0$. Then for each face $F$ of $\phi_{\tilde K}$, the set $T_F$ is a consecutive set.
\end{lemma}
\begin{proof}
Assume otherwise, and let $F$ be the violating face of $\phi_{\tilde K}'$. Let $T_1,T_2\sse T_F$ be two subsets of terminals, such that each subset is consecutive, but $T_1\cup T_2$ is not consecutive. Moreover, assume that both $T_1,T_2$ are maximal w.r.t. inclusion. Let $t,t'$ be the two terminals appearing immediately before and immediately after $T_1$ in $\pi$, $t,t'\not\in T_F$. 
We claim that $t$ and $t'$ must belong to $\tilde T$. Otherwise, since they do not lie in $F$, some edge of $\tilde K$ must cross the closed curve $\tau$ between $T_1$ and $t$, or between $T_1$ and $t'$, which is impossible, since all edges of $\tilde K$ are good. Therefore, $t,t'\in \tilde T$. But then there is a path $P_r\sse \tilde K_r$ connecting $t$ to $t'$, and a path $P_b\sse \tilde K_b$ connecting them. Both paths are disjoint and only contain good edges. Therefore, the union of their drawings with $t$ and $t'$ in $\phi_{\tilde K}'$ forms a simple closed curve, that crosses $\tau$ exactly twice: at $p_t$ and at $p_{t'}$. Therefore, $T_1,T_2$ lie on different sides of this curve, contradicting the fact that both sets belong to the same face $F$ of $\phi_{\tilde{K}}'$.
\end{proof}

\begin{corollary}\label{corollary: small number of terminals in connected component}
Let $\tilde K$ be as in Lemma~\ref{lemma: each face contains consecutive set of terminals}. Assume that the maximum size of any consecutive set $T'\sse T^0$ of terminals with $T'\cap\tilde T=\emptyset$ is $\beta$. Then every connected component of $G\setminus \tilde K$ contains at most $\beta+|T\setminus T^0|$ terminals.
\end{corollary}

The corollary simply follows from the fact that, since all edges of $\tilde K$ are good, each connected component $C$ of $G\setminus \tilde K$ must be completely contained in some face $F$ of $\phi_{\tilde K}$ in the drawing $\phi'$. Since $|T_F|$ is bounded by $\beta$, and there are at most $|T\setminus T^0|$ additional terminals that do not belong to $T^0$, the corollary follows.

From the above corollary, if events $\event_1,\event_2$ and $\event_3$ do not happen, the skeleton $K^{**}$ we have constructed has the property that every connected component of $G\setminus K^{**}$ contains at most $\frac{720\cdot \opt^2\cdot \rho\cdot \log n\cdot \betaFCG }{\alpha^*}+|T\setminus T^0|$ terminals. Therefore, skeleton  $K'$ has the same property, as each connected component of $G\setminus K'$ is a subgraph of some connected component of $G\setminus K^{**}$. It is aso easy to see that graph $K'$ does not have any $1$-vertex cuts, because $K'$ is a union of simple cycles. However, it may still have $2$-vertex cuts, and therefore, it may be non-rigid. We take care of this problem next, by getting rid of all $2$-vertex cuts in $K'$. We will argue that the resulting graph still has the properties necessary to bound the number of terminals in each connected component of $G\setminus K'$, and it will serve as the final skeleton.

\subsubsection*{Handling $2$-vertex Cuts}
Let $K'_r$ be the sub-graph of $K'$ induced by $V(K')\cap V_r$, and $K'_b$ the sub-graph of $K'$ induced by $V(K')\cap (V_b\cup V(X))$.
We denote by $\tilde{T}$ the subset of terminals contained in $E(K')$. If Event $\event_1$ did not happen, then $\tilde{T}\sse T^0$. We denote by $\tilde{\pi}$ the circular ordering of the terminals in $\tilde T$ induced by $\pi$. Let $t_0$ be any terminal in $\tilde{T}$.

Let $(u,v)$ be any $2$-vertex cut in $K'$. Denote by $\cset'_{u,v}$ be the set of all connected components in graph $K'\setminus (u,v)$, and for each $S\in\cset'_{u,v}$, we let $S'$ denote the sub-graph of $K'$ induced by $V(S)\cup\set{u,v}$. Let $\cset_{u,v}$ denote the set of components $S'$, for $S\in \cset'_{u,v}$.

Let $C_{u,v}\in \cset_{u,v}$ be the component containing the largest number of terminals. If $\cset_{u,v}$ contains exactly two components, and both of them contain the same number of terminals, then we let $C_{u,v}$ be the component containing $t_0$. Otherwise, we let $C_{u,v}$ be any component containing the largest number of terminals.
Let $C'_{u,v}$ be the union of all other components. If the edge $(u,v)$ belongs to $K'$, we add it to $C'_{u,v}$. If $C'_{u,v}$ does not contain any terminals, then we simply replace $C'_{u,v}$ with any path $P_{u,v}\sse C'_{u,v}$ connecting $u$ to $v$ in $K'$. 

Therefore, we can assume from now on that for every $2$-vertex cut $(u,v)$ in $K'$, $C'_{u,v}$ contains terminals. Then one of the two vertices (say $u$) must be a $1$-vertex cut in $K'_r$, and the other ($v$) is a $1$-vertex cut in $K'_b$. The two clusters, $C_{u,v}$ and $C'_{u,v}$ define a partition $(\tilde{T}_{u,v},\tilde{T}'_{u,v})$ of $\tilde{T}$, where $\tilde{T}_{u,v}\sse E(C_{u,v})$, and $\tilde{T'}_{u,v}\sse E(C'_{u,v})$. 
We first show that sets $\tilde{T}_{u,v}$ and $\tilde{T}'_{u,v}$ are consecutive sets w.r.t. $\tilde{\pi}$.

\begin{claim} Let $(u,v)$ be any $2$-vertex cut in $K'$. Then both $\tilde{T}_{u,v}$ and $\tilde{T}'_{u,v}$ are consecutive sets w.r.t. $\tilde{\pi}$. 
\end{claim}
\begin{proof}
Assume otherwise.
Let $\pi'$ be the smallest segment of $\pi$ spanning the terminals in $\tilde{T}_{u,v}$, and let $\pi''$ be the complement of $\pi'$. Then both $\pi'$ and $\pi''$ contain terminals in  $\tilde{T}'_{u,v}$. Let $t_1,t_1'\in \tilde{T}'_{u,v}$, $t_1\in \pi'$, $t_1'\in \pi''$. 
Notice that $C'_{u,v}$ contains a simple cycle $C^1$, containing both $t_1$ and $t_1'$.

Let $
t_2$, $t_2'$ be the first and the last terminals in $\pi'$. Then both these terminals belong to $C_{u,v}$, and moreover $C_{u,v}$ must contain a simple cycle $C^2$ containing both $t_2$ and $t_2'$. The drawing of this cycle crosses the curve $\gamma$ exactly twice: once at $t_2$, and once at $t_2'$. Therefore, $t_1$ and $t_1'$ lie on different sides of the drawing of $C^2$. 
Similarly, $t_2$ and $t_2'$ lie on different sides of the drawing of $C^1$.

The two cycles, $C^1$ and $C^2$, may only share two vertices: $u$ and $v$, and their drawings do not cross. Therefore, $(u,v)$ must be a $2$-separator for $C_{u,v}$, which is impossible by the definition of $C_{u,v}$.
\end{proof}

We show in the next claim that the sets $\tilde{T}'_{u,v}$ of terminals, for all $2$-cuts $(u,v)$ in $K'$, form a laminar family.

\begin{claim}\label{claim: laminar}
For every pair $(u,v)$ and $(u',v')$ of $2$-vertex cuts in $K'$, either $\tilde{T}'_{u,v}\subseteq \tilde{T}'_{u',v'}$, or $\tilde{T}'_{u',v'}\subseteq \tilde{T}'_{u,v}$, or $\tilde{T}'_{u,v}\cap \tilde{T}'_{u',v'}=\emptyset$. Therefore, the sets $\tilde{T}'_{u,v}$ for all $2$-cuts $(u,v)$ in $K'$ form a laminar family.
\end{claim}
\begin{proof}
Recall that for any $2$-cut $(u,v)$ in $K'$, one of the two vertices must be a $1$-cut in $K'_r$, and the other a $1$-cut in $K'_b$. We assume w.l.o.g. that $u,u'\in V(K'_r)$, and $v,v'\in V(K'_b)$. Assume first that $v\neq v'$. If $v'$ belongs to $C'_{u,v}$, then it must belong to some component $C\in \cset_{u,v}$, $C\neq C_{u,v}$, and the number of terminals contained in $C$ is at most $|\tilde{T}_{u,v}|$. Since both $v$ and $v'$ must be $1$-cuts in $K'_b$, it follows that some component $C'\in\cset_{u',v'}$ contains all terminals in $\tilde{T}_{u,v}$, while the terminals of all other components are contained in $\tilde{T}'_{u,v}$. So $C'=C_{u',v'}$ must hold, and
 $\tilde{T}'_{u',v'}\subseteq \tilde{T}'_{u,v}$. Otherwise, if $v'$ belongs to $C_{u,v}$, then either $\tilde{T}'_{u',v'}\sse \tilde{T}_{u,v}$, and so $\tilde{T}'_{u',v'}\cap \tilde{T}'_{u,v}=\emptyset$; or $\tilde{T}'_{u,v}\subseteq \tilde{T}'_{u',v'}$. Similar reasoning works for the case where $u\neq u'$. Since $(u,v)\neq (u',v')$, the claim follows.
\end{proof}

We say that $C'_{u,v}$ is maximal iff the set $\tilde{T}'_{u,v}$ is maximal inclusion-wise. Let $M'$ denote the set of pairs $(u,v)$, such that $(u,v)$ is a $2$-cut in $K'$, and $\tilde T'_{u,v}$ is maximal (if there are several pairs $(u,v)$ with identical sets $\tilde T'_{u,v}$ of terminals, only one of them is added to $M'$ - the pair $(u,v)$ for which $C'_{u,v}$ contains most vertices).
We call vertices in $V(C'_{u,v})\setminus \set{u,v}$ the \emph{inner vertices} of $C'_{u,v}$ and $(u,v)$ the \emph{endpoints} of $C'_{u,v}$
We need the following claim:

\begin{claim}
Let $(u,v),(u',v')\in M'$, $(u,v)\neq (u',v')$. Then $C'_{u,v}$ and $C'_{u',v'}$ are completely vertex disjoint, except for possibly sharing one inner vertex (that is, $u\in \set{u',v'}$ or $v\in \set{u',v'}$).
\end{claim}

\begin{proof}
From the proof of Claim~\ref{claim: laminar}, since $\tilde{T}'_{u,v}\cap \tilde{T}'_{u',v'}=\emptyset$, 
it is impossible that $u$ or $v$ are inner vertices of $C'_{u',v'}$, and similarly, $u'$ and $v'$ are not inner vertices of $C'_{u,v}$. So the only possibility that the claim is wrong is that some vertex $x$ is an inner vertex for both $C'_{u,v}$ and $C'_{u',v'}$. Let $t\in \tilde{T}'_{u,v}$ be some terminal, lying in the same connected component of $\cset_{u,v}$ as $x$ (such a terminal must exist since all components that contain no terminals have been replaced by edges). Recall that $t\not\in C'_{u',v'}$. Then there is a path $Q$, connecting $t$ to $r$ in $C'_{u,v}\setminus\set{u,v}$. Since path $Q$ connects a vertex that does not belong to $C'_{u',v'}$ to a vertex that belongs to $C'_{u',v'}$, it follows that it contains either $u'$ or $v'$ as an inner vertex. Hence, either $u'$ or $v'$ must be inner vertices of $C'_{u,v}$, which we have already ruled out.
\end{proof}

 For each $(u,v)\in M'$, we select an arbitrary terminal $t'\in \tilde T'_{u,v}$, and replace $C'_{u,v}$ with any path $P_{u,v}\sse C'_{u,v}$ that connects $u$ to $v$ and contains $t'$. Let $\tilde{K}'$ denote the resulting graph. This is our final skeleton. It is easy to see that $\tilde{K}'$ is rigid: consider the graph obtained from $\tilde{K}'$, after we replace, for each pair $(u,v)\in M'$, the $2$-path $P_{u,v}$ with the edge $(u,v)$. We claim that the resulting graph $\tilde{K}''$ is $3$-vertex connected, with no parallel edges or self-loops. Indeed, it is easy to verify that $\tilde{K}''$ does not contain parallel edges, from the definition of sets $C'_{u,v}$ for $(u,v)\in M'$. It is also immediate to see that it does not contain self-loops.  Assume now for contradiction that $(x,y)$ is a $2$-vertex cut in $\tilde{K}''$. Then $(x,y)$ is also a $2$-vertex cut in $K'$, $(x,y)\not\in M'$, and moreover, every connected component of $\cset_{x,y}$ must contain terminals. Then there must be some $2$-cut $(u,v)\in M'$, such that either $\tilde{T}'_{x,y}\subsetneq \tilde{T}'_{u,v}$, or $\tilde{T}'_{x,y}=\tilde{T}'_{u,v}$, but $C'_{u,v}$ contains more vertices than $C'_{x,y}$. However, since $x$ and $y$ belong to $\tilde{K}''$, none of these vertices can be an inner vertex of $C'_{u,v}$. It is then impossible that either $\tilde{T}'_{x,y}\subsetneq \tilde{T}'_{u,v}$, or $\tilde{T}'_{x,y}=\tilde{T}'_{u,v}$, but $C'_{u,v}$ contains more vertices than $C'_{x,y}$.
 We conclude that $\tilde{K'}$ is rigid.

 We need to argue that every connected component of $G\setminus \tilde{K}'$ only contains a small number of terminals. We define the last bad event, $\event_4$. We first show that w.h.p. this event does not happen, and then prove that if Events $\event_1,\ldots,\event_4$ do not happen, then every connected component of $G\setminus \tilde{K}'$ contains a small number of terminals. 

Let $T'\sse T^0$ be any consecutive set of terminals in $\pi$. Let $T_1'\sse T'$ be the subset of terminals that participate in pairs in $\rset^*$, and let $\lambda(T')$ be the number of pairs of terminals $(t,t')\in \rset^{**}$ with $t\in T',t'\not\in T'$. We say that the bad event $\event_4$ happens, iff for some consecutive set $T'\sse T^0$ with $|T'|\leq N'/4$, $|T_1'|>256\log n\cdot\dmax$, but $\lambda(T')\leq 4\log n\cdot\dmax$.

\begin{claim}
$\prob{\event_4}\leq P/10$.
\end{claim}
\begin{proof}
Fix some consecutive set $T'$ with $|T'|\leq N'/4$. The expected size of $T'_1$ is bounded by $\mu=|T'|\cdot |T|\cdot P'$. 

If $\mu\leq 128\log n\cdot \dmax$, then $\prob{|T_1'|>256\log n\cdot \dmax}\leq e^{-128\log n\cdot \dmax/8}\leq e^{-8\log n}\leq 1/n^8$, by the Chernoff bound.

Assume now that $\mu>128\log n\cdot \dmax$.
The expected number of pairs of terminals $(t,t')\in \rset^{**}$ with $t\in T'$, $t'\not\in T'$ is at least $P'\cdot |T'|\cdot N'/2\geq P'\cdot |T'|\cdot |T|/4\geq \mu/4>32\log n\cdot \dmax$. The probability that at most 
$4\log n\cdot \dmax$ such pairs are selected is again bounded by $1/n^4$, by the Chernoff bound.

We have therefore shown, that for any consecutive set $T'$, with $|T'|< N'/4$, the probability that $|T_1'|>256\log n\cdot \dmax$, but $\lambda(T')\leq 4\log n\cdot\dmax$, is at most $1/n^4$. Since the number of such consecutive sets is at most $n^2$, the claim follows from the union bound.
\end{proof}

Let $C'_{u,v}$ be some cluster with $(u,v)\in M'$. Recall that $\tilde T'_{u,v}$ is the set of terminals in $C'_{u,v}$. Let $\pi(u,v)$ be the smallest segment of $\pi$ that contains all terminals in $\tilde T'_{u,v}$, and let $T'_{u,v}$ be the set of all terminals that appear on $\pi(u,v)$. Clearly, set $T'_{u,v}$ is a consecutive set for $\pi$. Moreover, for any pair $(u,v),(u',v')\in M'$, the sets $\tilde{T}'_{u,v}$ and $\tilde{T}'_{u',v'}$ are completely disjoint, and so $\pi_{u,v}$ and $\pi_{u',v'}$ are disjoint, and $T'_{u,v}\cap T'_{u',v'}=\emptyset$.

In the next claim, we show that for each $2$-vertex cut $(u,v)$ of $K'$, $|T'_{u,v}|\leq \frac{3600\cdot \opt^2\cdot \rho\cdot \log^2 n\cdot \dmax \cdot \betaFCG }{\alpha^*}$. We defer the proof of the claim to the next subsection, and show first that the proof of Theorem~\ref{thm: initial skeleton for Cases 3 and 4} follows from it.

\begin{claim}\label{claim: bounding T'u,v} If events $\event_1,\ldots,\event_4$ do not happen, then
for every $2$-vertex cut $(u,v)\in M'$:  

\[|T'_{u,v}|\leq \frac{3600\cdot \opt^2\cdot \rho\cdot \log^2 n\cdot \dmax \cdot \betaFCG }{\alpha^*}.\]
\end{claim}

 We perform some transformations to the skeleton $K^{**}$, to obtain a new graph $\tilde{K}^{**}$. These transformations will reflect the changes that have been done to graph $K'$, so in the end $\tilde K^{**}\sse \tilde K'$ will hold.
 
Consider again some $2$-vertex separator $(u,v)$ of $K'$. Since graph $K^{**}$ is also a union of simple cycles, either both $u,v\not \in K^{**}$, or both $u,v\in K^{**}$ must hold. In the former case, we do not perform any changes to $K^{**}$. For the latter case, if $(u,v)\in M'$, we simply replace $C'_{u,v}\cap K^{**}$ with the path $P_{u,v}$ -- the same path that replaced $C'_{u,v}$ in $K'$. 

Consider the final graph $\tilde K^{**}$. It is easy to see that both the subgraphs of $\tilde K^{**}$ induced by $V_r\cap V(\tilde K^{**})$, and by $(V_b\cup V(X))\cap V(\tilde K^{**})$ remain connected. Therefore, we can apply Corollary~\ref{corollary: small number of terminals in connected component}, once we suitably bound the size of maximum consecutive set of terminals that do not belong to $\tilde K^{**}$.

Consider the ordering $\pi$ of the terminals in $T^0$. Any consecutive set $T'\sse T^0$, such that $T'\cap E(K^{**})=\emptyset$ is called clean. Our goal is to bound the size of maximal consecutive clean set. At the beginning, if event $\event_3$ does not happen, the size of the maximal consecutive clean set is bounded by $\frac{720\cdot \opt^2\cdot \rho\cdot \log n\cdot \betaFCG }{\alpha^*}$. Consider now the family $\tset$ of subsets $T'_{u,v}$ of terminals, corresponding to the clusters $C'_{u,v}$ with $(u,v)\in M'$. Each subset $T'_{u,v}$ is a consecutive set, and we have replaced $C'_{u,v}$ with some path $P_{u,v}$, containing some terminal in $T'_{u,v}$. Moreover, all sets $T'_{u,v}\in \tset$ are disjoint, and from Claim~\ref{claim: bounding T'u,v}, the size of each such set is bounded by $\frac{3600\cdot \opt^2\cdot \rho\cdot \log ^2n\cdot \dmax \cdot \betaFCG }{\alpha^*}$. Therefore, the maximum size of any consecutive set  of terminals, that does not contain any edges of $\tilde{K}^{**}$ is bounded by $2\cdot \frac{3600\cdot \opt^2\cdot \rho\cdot \log^2 n\cdot \dmax \cdot \betaFCG }{\alpha^*}+\frac{720\cdot \opt^2\cdot \rho\cdot \log n\cdot \betaFCG }{\alpha^*}=O\left ( \frac{\opt^2\cdot \rho\cdot \log^2 n\cdot \dmax \cdot \betaFCG }{\alpha^*}\right )$. From Corollary~\ref{corollary: small number of terminals in connected component}, every connected component of $G\setminus \tilde{K}^{**}$ contains at most $O\left ( \frac{\opt^2\cdot \rho\cdot \log^2 n\cdot \dmax \cdot \betaFCG }{\alpha^*}\right )+|T\setminus T^0|=O\left ( \frac{\opt^2\cdot \rho\cdot \log^2 n\cdot \dmax \cdot \betaFCG }{\alpha^*}\right )$ terminals.

We conclude that if Events $\event_1,\ldots,\event_4$ do not happen, then the final skeleton $\tilde{K}'$ is good and rigid. Moreover, since $\tilde{K}^{**}\sse \tilde K'$, and every connected component of $G\setminus \tilde{K}^{**}$ contains at most $O\left ( \frac{\opt^2\cdot \rho\cdot \log^2 n\cdot \dmax \cdot \betaFCG }{\alpha^*}\right )$ terminals, every connected component of $G\setminus \tilde K'$ contains at most $O\left ( \frac{\opt^2\cdot \rho\cdot \log^2 n\cdot \dmax \cdot \betaFCG }{\alpha^*}\right )$ terminals.

We now summarize the algorithm for finding the skeleton. First, we build the graph $K^0$, using the randomized procedure as above. Since we do not know which connected component of $K^0$ contains $K^*$, we go over every connected component $K''$ of $K^0$. For each such connected component, we try to take care of all $2$-vertex cuts, as in the above procedure, and then check whether the resulting graph $\tilde K''$ is rigid, and whether every connected component of $G\setminus \tilde K''$ contains  at most $O\left ( \frac{\opt^2\cdot \rho\cdot \log^2 n\cdot \dmax \cdot \betaFCG }{\alpha^*}\right )$ terminals. By the above argument, if Events $\event_1,\ldots,\event_4$ do not happen, one of the connected components of $K^0$ will have these properties. We output this component as $K'$.

In order to finish the proof of Theorem~\ref{thm: initial skeleton for Cases 3 and 4}, it is now enough to prove Claim~\ref{claim: bounding T'u,v}.

\subsubsection{Proof of Claim~\ref{claim: bounding T'u,v}}
Let $(u,v)\in M'$, and assume for contradiction that $|T'_{u,v}|> \frac{3600\cdot \opt^2\cdot \rho\cdot \log^2 n\cdot \dmax \cdot \betaFCG }{\alpha^*}$. To simplify notation, we denote $T'_{u,v}$ by $T'$ for the rest of the proof. Recall that one of the two vertices $u,v$ must be a $1$-vertex cut in $K'[V_r]$. We assume w.l.o.g that it is $u$. Since we assume that event $\event_1$ does not happen, vertex $u$ may belong to at most $4\log n\cdot \dmax$ cycles in $\lset^*$. In particular, there are at most $4\log n\cdot \dmax$ pairs $(t,t')\in \rset^*$ of terminals, with $t\in T',t'\not\in T'$. Recall also that $T'$ is a consecutive set for $\pi$.

We consider three cases. The first case is when $|T'|\leq N'/4$. In this case, for every terminal $t\in T'$ that belongs to $K^{**}$, there must be another terminal $t'\not\in T'$, such that $(t,t')\in \rset^{**}$. Since each such cycle $L_{t,t'}$ must contain $u$, $|T'\cap K^{**}|\leq 4\log n\cdot \dmax$.
On the other hand, since we assume that event $\event_3$ does not happen, the maximum size of any consecutive subset of $T'$ with no terminals in $E(K^{**})$ is $\frac{720\cdot \opt^2\cdot \rho\cdot \log n\cdot \betaFCG }{\alpha^*}$. Therefore:

\[|T'|\leq (4\log n\cdot \dmax+2)\cdot  \frac{720\cdot \opt^2\cdot \rho\cdot \log n\cdot \betaFCG }{\alpha^*}+4\log n\cdot \dmax\leq\frac{3600\cdot \opt^2\cdot \rho\cdot \log^2 n\cdot \dmax\cdot \betaFCG }{\alpha^*}\]

The second case is that $N'/4<|T'|<3N'/4$. In this case, we have two consecutive subsets $T',T^0\setminus T'$, of at least $N'/4$ terminals each, such that the number of pairs $(t,t')\in T'\times (T^0\setminus T')$ that belong to $\rset^*$ is at most $4\log n\cdot \dmax$. Since we have assumed that Event $\event_2$ does not happen, this is impossible.

Finally, the third case is that $|T'|\geq 3N'/4$. Then $|T^0\setminus T'|<N'/4$, and all terminals of $\tilde{T}_{u,v}$ belong to $T^0\setminus T'$.
Every pair of terminals $(t,t')\in T'\times (T^0\setminus T')$, with $(t,t')\in \rset^{**}$ contributes a cycle containing $u$ to $\lset^{**}$, so the number of such pairs is at most $4\log n\cdot \dmax$. Since we assume that Event $\event_4$ does not happen, this means that $|\tilde{T}_{u,v}|\leq 256\log n\cdot \dmax$. Observe that there are at most $\dmax$ clusters in $\cset_{u,v}$, and by the definition of $C_{u,v}$, each such cluster contains at most $|\tilde{T}_{u,v}|\leq 256\log n\cdot \dmax$ terminals. Therefore, the total number of terminals in $K'$ is bounded by $256\log n\cdot \dmax^2$, and in particular, the number of terminals in $E(K^{**})$ is also at most $256\log n\cdot \dmax^2$. Since we assume that Event~$\event_3$ does not happen, the total number of terminals in $T^0$ must be at most:

\[256\log n\cdot \dmax^2\cdot \frac{720\cdot \opt^2\cdot \rho\cdot \log n\cdot \betaFCG }{\alpha^*}+256\log n\cdot \dmax^2<\frac{10^7\cdot \opt^2\cdot \rho\cdot \log^2n\cdot \betaFCG\cdot \dmax^2}{2\alpha^*}\]

and the total number of terminals in $T$ is less than $\frac{10^7\cdot \opt^2\cdot \rho\cdot \log^2n\cdot \betaFCG\cdot \dmax^2}{\alpha^*}$, contradicting our assumption in Equation~(\ref{eq: bound on T}).

\end{document}